\documentclass[a4paper]{article}

    \usepackage{amssymb,amsmath,amsfonts,eurosym,geometry,ulem,graphicx,caption,color,setspace,sectsty,comment,footmisc,caption,natbib,pdflscape,subfigure,array,url}
    \usepackage{hyperref}
    \hypersetup{
        colorlinks=true,        
        linkcolor=blue,         
        citecolor=blue,         
        urlcolor=blue,          
        filecolor=blue          
    }
    \usepackage{chngcntr}
    \usepackage{booktabs}
    \usepackage[figuresright]{rotating}
    \usepackage{adjustbox}
    \usepackage{multirow}
    \usepackage{float}
    \usepackage{arydshln}
    \usepackage{pifont}
    \usepackage{pdflscape}
    \usepackage[auth-lg,affil-sl]{authblk}
    \usepackage{enumitem}
    \usepackage{mathabx}

    \usepackage{thmtools}
    \usepackage{thm-restate}
    \usepackage{cleveref}

    \newenvironment{proofthm}[1]
    {\noindent\textbf{Proof of \textbf{Theorem \ref{#1}}.} }
    {\ \rule{0.5em}{0.5em} \vspace{\baselineskip}}
    \newenvironment{prooflmm}[1]
    {\noindent\textbf{Proof of \textbf{Lemma \ref{#1}}.} }
    {\ \rule{0.5em}{0.5em} \vspace{\baselineskip}}
    \newenvironment{proofcor}[1]
    {\noindent\textbf{Proof of \textbf{Corollary \ref{#1}}.} }
    {\ \rule{0.5em}{0.5em} \vspace{\baselineskip} }

    \newenvironment{proof}[1][Proof]{\noindent\textbf{#1.} }{\ \rule{0.5em}{0.5em} \vspace{\baselineskip}}

    \newtheorem{lemma}{Lemma}
    \newtheorem{theorem}{Theorem}
    \newtheorem{corollary}{Corollary}
    
    \newtheorem{assumption}{Assumption}
    \newtheorem{algorithm}{Algorithm}
    \newtheorem{example}{Example}
    \newtheorem{remark}{Remark}
    \newenvironment{remark*}{\vspace{0.5em}\noindent \textbf{{Remark.}} \itshape}{\vspace{0.5em}}

    \DeclareMathOperator*{\argmax}{argmax}
    \DeclareMathOperator*{\argmin}{argmin}

    \newcommand{\bs}{\boldsymbol}
    \newcommand{\mc}{\mathcal}
    
    \newcommand{\mb}{\mathbb}
    \newcommand{\mr}{\mathrm}

    \newcommand{\la}{\langle}
    \newcommand{\ra}{\rangle}

    \newcommand{\cp}{\stackrel{p}{\longrightarrow}}

    \newcommand{\cd}{\stackrel{d}{\longrightarrow}}

    \newcommand{\leqtext}[1]{\stackrel{\text{#1}}{\leq}}    
      
    \newcommand{\geqtext}[1]{\stackrel{\text{#1}}{\geq}}    
    \newcommand{\eqtext}[1]{\stackrel{\text{#1}}{=}}    
    \newcommand{\lesssimtext}[1]{\stackrel{\text{#1}}{\lesssim}}  
    \newcommand{\gtrsimtext}[1]{\stackrel{\text{#1}}{\gtrsim}}

    \makeatletter
    \newsavebox{\@brx}
    \newcommand{\llangle}[1][]{\savebox{\@brx}{\(\m@th{#1\langle}\)}%
    \mathopen{\copy\@brx\kern-0.5\wd\@brx\usebox{\@brx}}}
    \newcommand{\rrangle}[1][]{\savebox{\@brx}{\(\m@th{#1\rangle}\)}%
    \mathclose{\copy\@brx\kern-0.5\wd\@brx\usebox{\@brx}}}
    \makeatother

    \newcommand{\vt}[1]{{\vert\kern-0.25ex\vert #1 
    \vert\kern-0.25ex\vert}}

    \numberwithin{equation}{section}

\begin{document}

\title{\textsc{Tractable Estimation of Nonlinear Panels\\ with Interactive Fixed Effects}\thanks{
We thank Iv\'{a}n Fern\'{a}ndez-Val, Aureo de Paula, Martin Weidner and the participants of the 2nd UCL–CeMMAP–IFS Ph.D. Econometrics Research Day~(2024) and UCL Econometrics Brownbag Seminar for their valuable comments. We thank Martin Weidner for sharing the codes and data from \citet{chen2021nonlinear}, and thank Wei Miao for careful testing of the R package and for helpful comments on the implementation. %
    }}
\author{Andrei \textsc{Zeleneev}\thanks{University College London and CeMMAP: \textsf{a.zeleneev@ucl.ac.uk}.} \quad \quad Weisheng \textsc{Zhang}%
\thanks{%
University College London: \textsf{weisheng.zhang.21@ucl.ac.uk}.}
}

    \maketitle

\begin{abstract}
    Interactive fixed effects are routinely controlled for in linear panel models. While an analogous fixed effects (FE) estimator for nonlinear models has been available in the literature \citep{chen2021nonlinear}, it sees much more limited use in applied research because its implementation involves solving a high-dimensional non-convex problem. In this paper, we complement the theoretical analysis of \citet{chen2021nonlinear} by providing a new computationally efficient estimator that is asymptotically equivalent to their estimator. Unlike the previously proposed FE estimator, our estimator avoids solving a high-dimensional non-convex optimization problem and can be feasibly computed in large nonlinear panels. Our proposed method involves two steps. In the first step, we convexify the optimization problem using nuclear norm regularization (NNR) and obtain preliminary NNR estimators of the parameters, including the fixed effects. Then, we find the global solution of the original optimization problem using a standard gradient descent method initialized at these preliminary estimates. To make our method readily applicable in practice, we also propose specific numerical algorithms for solving the involved optimization problems, establish their convergence, and provide their efficient implementation in our R package~\texttt{NNRPanel}.

\end{abstract}

\thispagestyle{empty}
\clearpage

\setcounter{page}{1} 

\clearpage

    \section{Introduction}\label{sec:intro}
The importance of accounting for interactive unobserved heterogeneity in panel and network models is well recognized. For example, in linear panel models, interactive fixed effects are routinely controlled for using the seminal approaches of \citet{bai2009panel} or \citet{pesaran2006estimation}.  While analogous methods for nonlinear models have been developed in the literature (e.g., \citealp{chen2021nonlinear}), they see much more limited use in empirical research due to their rapidly growing computational complexity or the lack of inferential theory.\footnote{For example, \citet{zeleneev2020identification} proposes a method for estimating network models with (nonparametric) interactive unobserved heterogeneity that does not require solving a high-dimensional nonconvex problem. However, unlike \citet{chen2021nonlinear}, \citet{zeleneev2020identification} focuses on identification and consistent estimation and does not provide inference tools.}

The main goal of this paper is to bridge the gap between the recent theoretical developments by \citet{chen2021nonlinear} and empirical work by providing a new computationally efficient estimator that can be feasibly implemented in a wide range of nonlinear (semiparametric) settings with unobserved effects following a linear factor structure. We demonstrate that our estimator has two important properties. First, unlike the approach of \citet{chen2021nonlinear}, our method does not require solving a high-dimensional non-convex optimization problem, so our estimator can be efficiently computed when the cross-sectional and time dimensions, $N$ and $T$, are large. Second, we argue that our estimator is asymptotically equivalent to the fixed effects (FE) estimator of \citet{chen2021nonlinear}. This means that, in practice, one can combine our computationally efficient estimator with the inferential theory provided in \citet{chen2021nonlinear} to construct confidence intervals for various objects of interest, including structural parameters and average partial effects.

Our proposed estimation procedure involves the following two steps. In the first step, we construct preliminary estimators of the parameters of interest, including the loadings and the factors, by solving a convex relaxation of the original (non-convex) optimization problem in \citet{chen2021nonlinear}. Following the literature, we convexify the original problem by replacing the low-rank constraint imposed on the unobserved effects by the factor model with a nuclear norm penalty. Then, we compute our final estimator by solving the original optimization problem using a standard gradient descent method initialized at the preliminary nuclear norm regularized (NNR) estimator obtained in the first step. %

To demonstrate that our final estimator is asymptotically equivalent to the FE estimator of \citet{chen2021nonlinear} defined as the global solution of the original high-dimensional and non-convex optimization problem, we show that the original problem is locally convex in a \emph{shrinking} neighborhood around the true values of the parameters. Importantly, in the general nonlinear setting studied in this paper (with a growing number of factors and loadings as $N,T \rightarrow \infty$), the size of this neighborhood \emph{shrinks} at a certain rate. To establish the desired result, we characterize the rate of convergence of our preliminary NNR estimator, and demonstrate that this rate is \emph{sufficiently fast} to ensure that our NNR estimator, as well as the FE estimator, falls into that shrinking neighborhood with probability approaching one.%

The idea of using a preliminary NNR estimator to initialize local optimization in (globally) non-convex problems has been previously explored in the econometrics literature. For example, \citet{moon2018nuclear} originally proposed an analogous two-step approach for estimating linear panel models with interactive fixed effects. In particular, \citet{moon2018nuclear} also demonstrate that their two-step estimator is asymptotically equivalent to the LS estimator of \citet{bai2009panel}. However, extending these ideas and formally establishing an analogous equivalence result in the general nonlinear setting of \citet{chen2021nonlinear} is a non-trivial task involving additional technical challenges.

As highlighted above, the main conceptual and technical difference is that, in the general nonlinear case, the objective function is locally convex only in a shrinking neighborhood of the true parameters value. In particular, unlike in the linear case, one cannot simply profile out the fixed effects using the singular value decomposition, and demonstrate that the profiled objective function (only depending on the common parameters $\beta$) is locally convex. Since, in the nonlinear case, we cannot work with the profiled objective function directly, we establish local convexity of the original objective function by inspecting its Hessian taken with respect to all of the parameters including the loadings and the factors. The analysis is further complicated by the fact that the dimension of the parameter space and hence the dimension of the Hessian grows with $N,T \rightarrow \infty$. As a result, local convexity of the objective function can only be established in a shrinking neighborhood of the true parameters. Establishing local convexity in that neighborhood and characterizing at which rate it shrinks is a technical innovation of the paper that has important practical implications. Specifically, it imposes an additional requirement on the preliminary estimator's rate of convergence: unless the preliminary estimator falls into that shrinking convexity region with probability approaching one, one cannot guarantee that the second step local optimization finds the global solution. In particular, it turns out that the rate obtained by \citet[Theorem~A.1]{moon2018nuclear} for the NNR estimator in single-index models is \emph{not} sufficiently fast to satisfy this requirement.

To take advantage of the local convexity result described above, we provide a new improved error bound for the NNR estimator in nonlinear models with interactive fixed effects. Following the literature, we derive this result under a version of the restricted strong convexity (RSC) condition. While various variations of the RSC condition are routinely employed for deriving analogous results in low-rank models (e.g., \citealp{moon2018nuclear,chernozhukov2019inference,ma2022detecting}), these conditions are often difficult to verify. Unlike most previous studies, we provide a set of primitive conditions which can be used to verify the RSC condition in a wide range of panel models allowing, in particular, for predetermined covariates.

To further facilitate the practicality of our approach, we supplement it with concrete implementation details. In particular, we provide specific optimization algorithms, which can be used to efficiently compute the preliminary NNR and the final estimators, and establish their convergence. We also propose data-driven ways of choosing the regularization parameter involved in the first step and determining the unknown number of factors. To make the proposed method readily applicable, we provide its efficient implementation in an optimized R package, \texttt{NNRPanel}.\footnote{Available at: \url{https://github.com/wszhang-econ/NNRPanel}.}

Finally, we study the finite sample properties of our estimator in a number of numerical experiments, documenting its excellent performance and computational efficiency even in large panels with ${(N,T) = (1000,200)}$, and revisit the empirical application of \citet{chen2021nonlinear}.

\bigskip

This paper contributes to the literature on estimation of panel (and network) models with interactive fixed effects in two important ways.

First, we complement the theoretical analysis of nonlinear panel models provided by \citet{chen2021nonlinear} by proposing a new estimator that is asymptotically equivalent to their FE estimator. Importantly, unlike their FE estimator, our estimator does not involve solving a high-dimensional non-convex optimization problem, making it an attractive, if not the only available, computationally efficient alternative, which can be feasibly implemented even when both $N$ and $T$ are large. Our two-step approach to solving a non-convex optimization problem essentially extends the proposal of \citet{moon2018nuclear} to nonlinear settings. However, as explained above, establishing the asymptotic equivalence between the two-step and FE estimators in nonlinear settings is more nuanced: it involves careful establishing of local convexity of the criterion function in a shrinking neighborhood of the true parameters value, resulting in additional requirements imposed on the preliminary NNR estimator's convergence rate absent in the linear case studied by \citet{moon2018nuclear}.

Second, we also contribute to the literature on nuclear norm regularized estimation of low-rank models by extending the previously available results established by \citet{moon2018nuclear} and \citet{chernozhukov2019inference} for linear panel models to nonlinear settings.
By verifying the RSC condition in a wide class of nonlinear panel models, we improve on the result of \citet[Theorem~A.1]{moon2018nuclear}, which extends their original analysis to single-index models. Importantly, our analysis allows for predetermined covariates such as the outcome's lags, which are routinely used in panels, whereas the existing studies providing error bounds for the NNR estimator either only consider strictly exogenous covariates or do not verify the RSC condition at all. %

The idea of using nuclear norm regularization to turn the estimation of low-rank models, such as factor models, into a convex problem has been extensively applied in various settings in statistics and econometrics. In econometrics, its numerous recent applications include estimation of pure factor models  \citep{BaiNg2017,bai2019rank}, estimation of linear \citep{moon2018nuclear,chernozhukov2019inference,beyhum2019square,mugnier2025simple} and quantile panel regressions \citep{belloni2019high,wang2022low,feng_2023}, and treatment effect estimation \citep{athey_matrix_2021,fernandez2021low}. Nuclear norm relaxations have also been proved instrumental in constructing estimation and inference methods robust to weak factors \citep{armstrong2022robust} and missing data \citep{su2025estimation}. Other recent applications of nuclear norm regularization also include, among others, network recovery and community detection (\citealp{alidaee2020recovering} and \citealp{ma2022detecting}), and estimation of panel threshold models and high-dimensional VARs (\citealp{miao2020panel} and \citealp{miao2023high}).

Recent papers employing similar two-step estimation procedures in nonlinear panels also include \citet{chen2025high} and \citet{yao2025low}. \citet{chen2025high} study estimation of logistic panel models without covariates, whereas we (mainly) focus on the common parameters $\beta$.
In a closely related and independently developed paper, \citet{yao2025low} also studies estimation of nonlinear panel models with interactive fixed effects. In particular, \citet{yao2025low} allows for nonconvex model specifications such as the important class of random-coefficients models. This generality, however, comes at a cost: even after the low-rank constraint is convexified using the nuclear norm penalty, the first step of \citet{yao2025low}'s procedure still requires solving a nonconvex problem. Since the goal of our paper is to provide a theoretically justified and computationally feasible estimator that avoids nonconvex optimization, we focus on single-index models with convex link functions. As a result, we are able to verify a \emph{global} version of the RSC condition, allowing us to formally study the global properties of the NNR estimator. We also provide specific numerical algorithms for solving both first and second step optimization problems and establish their convergence.

    \bigskip    

    \noindent {\bf Notation} For any vector $u\in \mb{R}^n$, its Euclidean norm is denoted as $\|u\| = \left(u'u\right)^{\frac{1}{2}}$.  For any matrix $A\in \mb{R}^{m\times n}$, we use $A'$ to denote the transpose of $A$, and use  $\|A\|_{\mr{F}} = \left(\mr{trace}( A'A )\right)^{\frac{1}{2}}$ to denote the Frobenius norm. Furthermore, the singular values of $A$ are arranged in non-increasing order: $\psi_1\left(A\right)\geq \psi_2\left(A\right) \geq \ldots \geq \psi_{\min\{m, n\}}\left(A\right) \geq 0$. 
    The $\ell^2$ operator norm, $\|A\|_{\mr{op}} = \psi_1\left(A\right)$, is the maximum singular value of the matrix, and the nuclear norm is the sum of all singular values:
    $\|A\|_{\mr{nuc}} = \sum_{i=1}^{\min\{m, n\}}\psi_i\left(A\right)$. We also use $\|A\|_{\max} = \max_{i,j} |A_{ij}| $ to denote the element-wise max norm.  
    When $A$ is a square matrix, we use $\sigma_i(A)$ to denote $A$'s $i$-th largest eigenvalue. We also use $\psi_{\max}$, $\psi_{\min}$, $\sigma_{\max}$,  $\sigma_{\min}$ to denote the max/min singular values and max/min eigenvalues respectively. For any matrix $A$, define the coprojection matrix as $M_{A}: = \mb{I} - A(A'A)^{\dagger }A'$, where $\mb{I}$ denotes the identity matrix of appropriate size and the super-script~$\dagger$ denotes the Moore-Penrose generalized inverse. Finally, for any two square matrices $A$ and $B$ of the same dimension, we use $A \geq B$ to denote that $A - B$ is positive semi-definite, and $A > B$ to denote that $A - B$ is positive definite. We use the abbreviation wpa1 instead of with probability approaching one. 

    \bigskip

    The remainder of the paper is organized as follows. Section~\ref{sec:model} introduces the model, highlights the computational challenges  of the FE estimator, and describes the proposed two-step estimator. Section~\ref{sec:theory} presents the asymptotic equivalence between our two-step estimator and the FE estimator. Section~\ref{sec:implementation} provides practical implementation details, including optimization algorithms, data-driven selection of the regularization parameter, and determination of the number of factors. Section~\ref{sec:MC} provides numerical and empirical illustrations. Additional technical and numerical results, along with all proofs, are contained in the appendix.

    \section{The Model and Two-Step Estimation}\label{sec:model}

\subsection{The model}
    We observe data $\{(Y_{it}, X_{it})\}_{ 1 \leq  i \leq N, 1 \leq t\leq T}$, where $Y_{it}$ is a scalar outcome variable and $X_{it} \in \mathbb R^{d_X}$ is a vector of covariates. For concreteness, we adopt the standard panel notation with $i$ indexing units and $t$ indexing time periods, but it should be understood that the considered framework applies to general two-way settings. For example, in a directed network $i$ and $t$ could index senders and receivers (e.g., exporters and importers in an international trade network). The covariates $X_{it}$ could be strictly exogenous or predetermined, e.g., our framework also accommodates lagged outcomes as covariates in panels.

    Following \citet{chen2021nonlinear}, we assume that the (conditional) distribution of $Y_{it}$ belongs to a known family of distributions and is determined by the latent index $Y_{it}^*$, i.e., we assume that the (conditional) log-likelihood takes the form
    \begin{align}\label{eq:true_model}
        \log f (Y_{it}|X_{it},\lambda_{0,i},\gamma_{0,t}) = \ell (Y_{it}|Y^*_{it}), \quad Y^*_{it}  = X'_{it}\beta_0  + \lambda_{0, i}' \gamma_{0, t},
    \end{align}
    where $\ell(\cdot|Y_{it}^*)$ is a known log-likelihood function, and $\beta_0 \in \mathbb R^{d_X}$ is the parameter of interest. Here, ${\lambda_{0,i} \in \mathbb R^R}$ and $\gamma_{0, t} \in \mathbb R^R$ are unobserved interactive unit and time effects, commonly referred to as loadings and factors. This formulation is substantially more flexible than the routinely employed two-way fixed effects (TWFE) model, $\lambda_{0, i} + \gamma_{0, t}$, because it allows incorporating multidimensional heterogeneous individual responses $\lambda_{0, i}$ to time-varying aggregate shocks $\gamma_{0, t}$.\footnote{\citet{chen2021nonlinear} also argue that the interactive fixed effects model is sufficiently flexible to allow for homophily based on unobservables (as well as for degree heterogeneity) in network settings.}  In particular, the TWFE model corresponds to the special case of the interactive fixed effects model with $R = 2$,  where $\lambda_{0, i} = (\lambda_{i1}, 1)'$ and $\gamma_{0, t} = (1, \gamma_{ t1})'$.

    While the single-index formulation~\eqref{eq:true_model} is restrictive, it covers a number of important nonlinear models,  including binary response models such as Probit and Logit, and Poisson regression.

    \begin{example}[Binary response panel model]
        Consider the following binary response panel model
        \begin{align}
            \label{eq: binary choice example}
            Y_{it} = \bs{1}(X_{it}'\beta  + \lambda_i'\gamma_t- u_{it} \geq 0),
        \end{align}
        where the idiosyncratic errors $\{u_{it}\}$ are independent draws from a distribution with cumulative distribution function $F(\cdot)$.\footnote{E.g., $F(\cdot)$ can stand for the logistic or standard normal CDF in Logit and Probit models, respectively.}
        
        Special variations of \eqref{eq: binary choice example} with unobserved heterogeneity fully captured by additive unit and/or time effects are routinely used in empirical studies of panels with binary outcomes such as labor force participation. As explained above, the general specification \eqref{eq: binary choice example} allows not only for TWFE but also for heterogeneous responses of units to aggregate shocks. E.g., in the labor force participation example, $\lambda_i$ and $\gamma_t$ can represent multidimensional human capital and its market prices that vary over time, respectively. Importantly, our framework accommodates predetermined covariates $X_{it}$ such as lagged outcomes $Y_{i,t-1}$, $Y_{i,t-2}$, etc., which are often included in binary response panel models either as primary variables of interest or as controls.
        
        In this example, the conditional distribution of $Y_{it}$ given $Y_{it}^*$ takes the form
        \begin{align*}
            \mb{P}(Y_{it} = y\mid Y^*_{it}) = F(Y^*_{it})^{y} (1-F(Y^*_{it}))^{(1-y)},\quad y\in \{0, 1\}, 
        \end{align*}
        and the conditional log-likelihood is given by
        \begin{align*}
            \ell (Y_{it}|Y_{it}^*) = Y_{it} \log F(Y_{it}^*) + (1 - Y_{it}) \log (1 - F(Y_{it}^*)).
        \end{align*}
    \end{example}

    \begin{example}[Poisson network regression]
        Our analysis applies to network settings. For example, the studied framework generalizes the TWFE Poisson regression model (e.g., \citealp{SantosSilva2006}) commonly employed to analyze international trade data. Specifically, consider
        \begin{align*}
           Y_{ij} \sim \mr{ Poisson }\left(\exp\left( X_{ij}'\beta + \lambda_i'\gamma_j\right)\right),   
        \end{align*}
        where $Y_{ij}$ represents the export volume from country $i$ to $j$, and $X_{ij}$ is a collection of trade determinants used in the gravity analysis, e.g., the geographic distance between countries $i$ and $j$. This specification allows for rich patterns of unobserved heterogeneity, including degree heterogeneity and homophily based on latent characteristics. For example, homophily based on the quadratic distance $(\xi_i - \xi_j)^2$ can be captured by $\lambda_{i} \propto (\xi_i^2, 1, -2\xi_i)'$ and $\gamma_{j} \propto (1, \xi_j^2, \xi_j)'$.
        
        In this example, the conditional distribution of $Y_{ij}$ given $Y_{ij}^*$ takes the form
        \begin{align*}
            \mb{P}(Y_{ij} = y\mid Y^*_{ij}) = \frac{\exp (-\exp(Y^*_{ij}))(\exp(Y^*_{ij}))^{y}}{y!}, \quad y = 0,1,2,\ldots, 
        \end{align*}
        and the conditional log-likelihood is given by
        \begin{align*}
            \ell (Y_{ij}|Y_{ij}^*) = Y_{ij} Y_{ij}^* - \exp (Y_{ij}^*) - \log (Y_{ij}!). 
        \end{align*}
        We will revisit this specification in our empirical application in Section~\ref{ssec: empirical}.
    \end{example}

    As in \citet{chen2021nonlinear}, we consider the so-called large $N,T$ asymptotics with $N,T \rightarrow \infty$, whereas we treat both $d_X$ and $R$ as fixed. For now, we will also assume that the number of factors $R$ is known; we will discuss the estimation of $R$ in Section~\ref{sec:implementation}. Finally, we do not put additional restrictions on the relationship between the covariates and the unobserved effects, i.e., we adopt the fixed effects approach.

\subsection{Fixed Effects MLE Estimator and Computational Challenges}\label{ssec:challenges}

    \citet{chen2021nonlinear} proposed estimating the studied model by the fixed effects~(FE) MLE estimator maximizing the conditional log-likelihood jointly over the common parameters $\beta$, loadings  $\{\lambda_{i}\}_{ 1\leq i\leq N}$ and factors $\{\gamma_{ t}\}_{1\leq t\leq T}$. Specifically, the FE estimator $(\hat{\beta}_{\mr{FE}}, \hat{\Lambda}_{\mr{FE}}, \hat{\Gamma}_{\mr{FE}})$ solves 
    \begin{equation}\label{eq:FE_estimator}
        (\hat{\beta}_{\mr{FE}}, \hat{\Lambda}_{\mr{FE}}, \hat{\Gamma}_{\mr{FE}}) \in \argmin_{ \beta, \Lambda, \Gamma} \underbrace{-\frac{1}{NT} \sum_{i=1}^{N}\sum_{t=1}^{T} \ell(Y_{it}\mid X_{it}'\beta + \lambda_{ i}'\gamma_{ t})}_{\mc{L}_{NT}(\beta, \Lambda, \Gamma)},%
    \end{equation}
    where, for notational simplicity,  we collect the unobserved effects $\{\lambda_{i}\}_{ 1\leq i\leq N}$ and $\{\gamma_{ t}\}_{1\leq t\leq T}$ into matrices $\Lambda = \left(\lambda_{1}, \lambda_{2},\ldots, \lambda_{N} \right)'\in \mb{R}^{N\times R}$ and $\Gamma = \left(\gamma_{1}, \gamma_{2},\ldots, \gamma_{T} \right)'\in \mb{R}^{T\times R}$. Note that problem~\eqref{eq:FE_estimator} does not have a unique solution for $ \hat{\Lambda}_{\mr{FE}}$ and $\hat{\Gamma}_{\mr{FE}}$ and thus requires a normalization. We will abstract from this issue for now and discuss it in more detail in Section~\ref{sec:theory}.

    \citet{chen2021nonlinear} showed that the FE estimator of $\beta_0$  is $\sqrt{NT}$-consistent and asymptotically normal, with an asymptotic incidental parameter bias that can be corrected using various bias reduction methods.  
    However, despite these well-established theoretical properties, implementing the FE estimator remains a significant computational challenge.

    The key computational difficulty is the non-convexity of the objective function $\mc{L}_{NT}(\beta, \Lambda, \Gamma)$. To better understand this issue, we reformulate the original optimization problem into an alternative but equivalent form.  Let $\theta_{it} = \lambda_i'\gamma_t$ and collect $\theta_{it}$ into a matrix $\Theta\in \mb{R}^{N\times T}$. Note that since matrices $\Lambda$ and $\Gamma$ have at most rank $R$, the rank of $\Theta = \Lambda \Gamma'$ is also at most $R$. Likewise, any matrix $\Theta \in \mathbb R^{N \times T}$ such that $\mr{rank}(\Theta)\leq R$ can be represented as $\Lambda \Gamma'$ for some $\Lambda \in \mathbb R^{N \times R}$ and $\Gamma \in \mathbb R^{T \times R}$.\footnote{The representation of $\Theta$ with $\mr{rank}(\Theta)\leq R$ as $\Theta = \Lambda \Gamma'$ is not unique. If $\Theta = \Lambda \Gamma'$ for some $\Lambda$ and $\Gamma$, we also have $\Theta = \tilde \Lambda \tilde \Gamma'$ for $\tilde \Lambda = \Lambda G'$ and $\tilde \Gamma = \Gamma G^{-1}$ for any invertible matrix $G \in \mathbb R^{R \times R}$. This non-uniqueness manifests itself in the necessity of normalizing $\Lambda$ and $\Gamma$ in problem~\eqref{eq:FE_estimator} in order to ensure the uniqueness of $\hat{\Lambda}_{\mr{FE}}$ and $\hat{\Gamma}_{\mr{FE}}$.} Thus, problem~\eqref{eq:FE_estimator} can be equivalently reformulated as
    \begin{align}\label{eq:FE_rank}
        (\hat{\beta}_{\mr{FE}}, \hat{\Theta}_{\mr{FE}}) \in  \argmin_{\beta\in \mb{R}^{d_X}, \Theta \in \mathbb{R}^{N\times T} } \underbrace{ -\frac{1}{NT} \sum_{i=1}^{N}\sum_{t = 1}^{T} \ell(Y_{it}\mid X_{it}'\beta + \theta_{it})}_{\mc{L}_{NT}(\beta, \Theta )}, \quad \text{s.t. } \mr{rank}(\Theta)\leq R.
    \end{align}
    The non-convexity arises  from the rank constraint $\mr{rank}(\Theta)\leq R$: the set of matrices satisfying it is not convex since the sum of two rank-$R$ matrices could have a rank up to $2R$.

    The high-dimensional parameter space further exacerbates the computational challenges.  When dealing with non-convex optimization problems, it is common practice to start the optimization process with multiple initial values and select the solution that minimizes the objective function. This approach is generally considered effective for finding the global minimum with sufficient trials. However, for problems \eqref{eq:FE_estimator} and \eqref{eq:FE_rank} involving $d_X + R (N+T)$ parameters, this approach becomes intractable even for moderate values of $N$ and $T$.

    \begin{remark}
        \citet{chen2021nonlinear} proposed solving optimization problem~\eqref{eq:FE_estimator} using the EM-algorithm of \citet{chen2016estimation} initialized at multiple starting values. Unfortunately, this method does not overcome the computational challenge discussed above because the EM-algorithm of \citet{chen2016estimation} as well as EM-algorithms in general does not have global convergence guarantees in non-convex problems.
    \end{remark}

\subsection{Two-step estimation}
    To overcome the computational challenges faced by the FE estimator, we propose an alternative two-step estimation procedure. Our procedure does not involve solving a non-convex problem and can be efficiently computed even for large values of $N$ and $T$. Importantly, in Section~\ref{sec:theory}, we demonstrate that, under standard regularity conditions, our two-step estimator is asymptotically equivalent to the FE estimator, whose asymptotic properties have been established in \citet{chen2021nonlinear}. This means that, instead of trying to solve the non-convex and high-dimensional optimization problem~\eqref{eq:FE_estimator} directly, one could compute our two-step estimator and then combine it with the asymptotic theory developed by \citet{chen2021nonlinear} to construct confidence intervals for various parameters of interest, including $\beta$ and policy relevant counterfactuals such as average partial effects~(APEs).

    \bigskip

    Our estimation procedure involves the following two steps.

    \bigskip
        
    \noindent {\bf Step 1: Nuclear Norm Regularized (NNR) estimation}
    
    \noindent The goal of the first step is to construct an easily computable preliminary estimator of $(\beta_0, \Lambda_0, \Gamma_0)$ that is sufficiently close to the global minimizer in \eqref{eq:FE_estimator}. To this end, we consider a convex relaxation of problem~\eqref{eq:FE_rank} of the form
    \begin{align}\label{eq:nnr_definition}
        \left(\hat{\beta}_{\mr{nuc}}, \hat{\Theta}_{\mr{nuc}}\right) = \argmin_{\beta \in \mb{R}^{d_X}, \Theta\in \mb{R}^{N\times T}} \left\{ \mc{L}_{NT} \left(  \beta, \Theta\right) + \frac{\varphi_{NT}}{\sqrt{NT}} \|\Theta\|_{\mr{nuc}}\right\},
    \end{align}
    where $\|\Theta\|_{\mr{nuc}}$ denotes the nuclear norm of matrix $\Theta$, and $\varphi_{NT} > 0$ is a regularization parameter. We will refer to the solution of this problem $(\hat{\beta}_{\mr{nuc}}, \hat{\Theta}_{\mr{nuc}})$ as the nuclear norm regularized (NNR) estimator.
    
    Since $\|\Theta\|_{\mr{nuc}}$ is a convex function of $\Theta$, problem~\eqref{eq:nnr_definition} is convex when $\mc{L}_{NT} \left(  \beta, \Theta\right)$ is a convex function of $\beta$ and $\Theta$. This condition is satisfied in important nonlinear models such as Logit, Probit, and Poisson models. Thanks to the convexity of problem~\eqref{eq:nnr_definition}, the NNR estimator can be efficiently computed using, for example, a proximal gradient descent method (e.g., \citealp{hastie2015statistical}) even when the parameter space is high-dimensional. We provide a specific optimization algorithm and a data-dependent recommendation for choosing the regularization parameter $\varphi_{NT}$ in Section~\ref{sec:implementation}.

    Notice that problem~\eqref{eq:nnr_definition} can be equivalently rewritten as
    \begin{align*}
        \left(\hat{\beta}_{\mr{nuc}}, \hat{\Theta}_{\mr{nuc}}\right) = \argmin_{\beta \in \mb{R}^{d_X}, \Theta\in \mb{R}^{N\times T}} \mc{L}_{NT} \left(  \beta, \Theta\right), \quad \text{s.t. } \|\Theta\|_{\mr{nuc}} \leqslant C_{\varphi_{NT}}   
    \end{align*}
    for an appropriately chosen $C_{\varphi_{NT}} > 0$ determined by $\varphi_{NT}$. Thus, problem~\eqref{eq:nnr_definition} can be seen as a convexification of problem~\eqref{eq:FE_rank}, where the non-convex rank constraint is replaced by the slightly looser yet convex constraint $\|\Theta\|_{\mr{nuc}} \leqslant C_{\varphi_{NT}}$. Analogously to LASSO using the $\ell_1$-regularization to induce sparsity of the solution in a high-dimensional regression, the nuclear norm regularization (i.e., the $\ell_1$-regularization of the singular values of $\Theta$) induces $\hat \Theta_{\mr{nuc}}$ to have low rank (i.e., sparsity of its singular values).

    Finally, the nuclear norm regularized estimators $(\hat{\Lambda}_{\mr{nuc}}, \hat{\Gamma}_{\mr{nuc}})$ are obtained through the singular value decomposition of $\hat{\Theta}_{\mr{nuc}}$. Specifically, let 
    $\hat{\Theta}_{\mr{nuc}}/\sqrt{NT} = \hat{U}\hat{D}\hat{V}'$,
    where $\hat{U}\in \mb{R}^{N\times \min\{N, T\}}$ and $\hat{V}\in \mb{R}^{T\times \min\{N, T\}}$ are matrices of left and right (orthonormal) singular vectors of $\hat{\Theta}_{\mr{nuc}}$, and $\hat{D}$ is a diagonal matrix with  the singular  values of $\hat{\Theta}_{\mr{nuc}}/\sqrt{NT}$ (arranged in non-increasing order) on its diagonal. Let $\hat{U}_{[:, 1:R]}$ and $\hat{V}_{[:, 1:R]}$ denote the matrices containing the first $R$ columns of $\hat{U}$ and $\hat{V}$, respectively, and $\hat{D}_{[1:R, 1:R]}$ denote the upper-left $R\times R$ diagonal block of $\hat{D}$. We compute  $(\hat{\Lambda}_{\mr{nuc}}, \hat{\Gamma}_{\mr{nuc}})$ as follows: 
    \begin{gather}\label{eq:space_definition}
        \hat{\Lambda}_{\mr{nuc}} =  \sqrt{N} \hat{U}_{[:, 1:R]}\hat{D}^{1/2}_{[1:R, 1:R]}, \quad \hat{\Gamma}_{\mr{nuc}} = \sqrt{T} \hat{V}_{[:, 1:R]}\hat{D}^{1/2}_{[1:R, 1:R]}.
    \end{gather}

    \bigskip
    
    \noindent {\bf Step 2: Local estimation}
    
    \noindent While, with appropriately chosen $\varphi_{NT}$, the NNR estimator is consistent for the true values $(\beta_0, \Lambda_0, \Gamma_0)$, it suffers from the regularization bias. To improve on the NNR estimator, in the second step, we solve the original optimization problem~\eqref{eq:FE_estimator} using a standard gradient descent method with $(\hat{\beta}_{\mr{nuc}}, \hat{\Lambda}_{\mr{nuc}}, \hat{\Gamma}_{\mr{nuc}})$ as the initial values. While the original problem~\eqref{eq:FE_estimator} is non-convex, availability of the NNR estimator allows us to guarantee that standard local optimization methods initialized at $(\hat{\beta}_{\mr{nuc}}, \hat{\Lambda}_{\mr{nuc}}, \hat{\Gamma}_{\mr{nuc}})$ converge to the global solution $(\hat{\beta}_{\mr{FE}}, \hat{\Lambda}_{\mr{FE}}, \hat{\Gamma}_{\mr{FE}})$. In particular, in Section~\ref{sec:implementation}, we provide a specific gradient descent algorithm and establish its convergence guarantees.\footnote{In principle, instead of using a gradient descent method, it is possible also employ an EM-algorithm (see, e.g., \citealp{chen2016estimation,chen2021nonlinear}) initialized at $(\hat{\beta}_{\mr{nuc}}, \hat{\Lambda}_{\mr{nuc}}, \hat{\Gamma}_{\mr{nuc}})$.}

    Specifically, to demonstrate that our two-step estimator is (asymptotically) equivalent to the FE estimator, in Section~\ref{sec:theory}, we show that, with probability approaching one, (i) the objective function $\mathcal L_{NT} (\beta, \Lambda, \Gamma)$ is strictly convex in a \emph{shrinking} neighborhood around the true values $(\beta_0, \Lambda_0, \Gamma_0)$, after normalization of the loadings and factors, and (ii) both the NNR and the FE estimators fall into this neighborhood. The technical difficulty here is that, since the dimension of the parameter space grows with $N,T \rightarrow \infty$, the size of the local neighborhood, in which $\mathcal L(\beta,\Lambda,\Gamma)$ remains convex, shrinks at a certain rate. To establish the desired result, we characterize this rate, and demonstrate that the NNR and the FE estimator have sufficiently fast rates of convergence to fall into that neighborhood with probability approaching one.

    \bigskip
    
    \noindent {\bf Bias correction and inference}

    \noindent Since our two-step estimator is asymptotically equivalent to the FE estimator, it also follows the same asymptotic distribution previously derived by \citet{chen2021nonlinear}. In particular, similar to the FE estimator, the two-step estimator of $\beta_0$ also suffers from the incidental parameter bias caused by the necessity to estimate a large number of nuisance parameters. Implementation details for the analytical and split-panel jackknife bias correction methods are provided in Appendix~\ref{appendix:bias_correction}. For a more detailed discussion of various bias correction and inference methods available for the FE estimator, we refer the reader to \citet{chen2021nonlinear} and \citet{xu2026bootstrap}.   %

    \bigskip
    
    \noindent {\bf Practical implementation}

    \noindent We provide additional practical implementation details in Section~\ref{sec:implementation}, including specific optimization algorithms used for both steps, as well as data=driven procedures for choosing the regularization parameter~$\varphi_{NT}$ and the number of factors $R$. An optimized implementation of our two-step estimator in R, including bias correction and calculation of standard errors, is also readily available in our package~\texttt{NNRPanel}.

    \section{Asymptotic Analysis}\label{sec:theory}%
    
    In this section, we establish the consistency of the NNR estimator and derive its convergence rate. We also demonstrate that the original optimization problem~\eqref{eq:FE_estimator} is locally convex in a (shrinking) neighborhood of the true values of the parameters. Combining these results, we argue that our two-step estimator is asymptotically equivalent to the FE estimator.

\subsection{Consistency of NNR Estimator}

    In this section, we establish the consistency of $(\hat{\beta}_{\mr{nuc}}, \hat{\Theta}_{\mr{nuc}})$ as in~\eqref{eq:nnr_definition}, as well as of the associated nuisance estimators $(\hat{\Lambda}_{\mr{nuc}}, \hat{\Gamma}_{\mr{nuc}})$ defined in~\eqref{eq:space_definition}. To fix ideas, we will first focus on strictly exogenous covariates $X_{it}$. We formally extend our analysis to more general settings allowing for predetermined covariates (e.g., the lagged outcomes) in Appendix~\ref{ssec: predetermined covariates}.

    For notational simplicity,  we collect $X_{it, d}$ into covariate matrices $X_{d} \in \mb{R}^{N\times T}$ for each $d= 1,2,\ldots, d_X$, and let $X$ be the collection of all covariate matrices $X= \{X_{1}, \ldots, X_{d_X}\}$.  Whenever it does not cause confusion, for any $Y^*$, we abbreviate $\ell_{it}(Y^*) = \ell (Y_{it}\mid Y^* )$ to denote the  log-likelihood evaluated at index~$Y^*$. We further denote the derivatives of the log-likelihood with respect to the index by $\dot{\ell}_{it}(\cdot), \ddot{\ell}_{it}(\cdot), \ldots$. In addition, we use $\mb{P}_{X, \Lambda_0, \Gamma_0}(\cdot) = \mb{P}(\cdot\mid X, \Lambda_0, \Gamma_0)$ and $\mb{E}_{X, \Lambda_0, \Gamma_0}(\cdot) = \mb{E}(\cdot\mid X, \Lambda_0, \Gamma_0)$ to denote the conditional probability and expectation given $X, \Lambda_0, \Gamma_0$, respectively.

    \bigskip
    
    We will maintain the following regularity conditions throughout this section.

    \begin{assumption}[Regularity conditions]\label{assumption:regularity}
        \leavevmode
        \begin{enumerate}[label=(\roman*)]
            \item \label{item:sampling} \textbf{(Sampling)} For all $i = 1,2,\ldots, N$ and $t=1,2,\ldots, T$, conditional on $(X,  \Lambda_{0}, \Gamma_{0})$, $\{Y_{it}\}_{1\leq i \leq N, 1\leq t \leq T} $ are independent across $i$ and $t$ and distributed as in \eqref{eq:true_model}.
            
            \item \label{item:boundedness} \textbf{(Boundedness)} The parameter spaces for $\beta$, $\lambda_i$, and $\gamma_t$ are uniformly bounded  for all $i, t, N, T$.  In addition,  there exists a constant $\rho_X>0$ such that $\max_{d=1,\ldots, d_X}\|X_d\|_{\max} \leq \rho_X$ for all $i, t, N, T$.
            
            \item \label{item:smoothness} \textbf{(Smoothness and convexity)} $-\ell_{it}(\cdot)$ is four times continuously differentiable.
            Furthermore, $-\ell_{it}(\cdot)$ is strictly convex  with $0 < b_{\min} \leq -\ddot{\ell}_{it}( X_{it}' \beta + \lambda_i' \gamma_t) \leq b_{\max} <\infty$  almost surely for all $\beta, \lambda_i, \gamma_t$ in the parameter space  uniformly over $i, t, N, T$.
            
            \item \label{item:strong_factors} \textbf{(Strong factors)} $\frac{1}{N} \sum_{i=1}^{N}\lambda_{0, i}\lambda_{0, i}' \cp \Sigma_{\lambda}$ and $\frac{1}{T} \sum_{t=1}^{T}\gamma_{0, t}\gamma_{0, t}'\cp \Sigma_{\gamma}$, where $ \Sigma_{\lambda} >0$ and  $\Sigma_{\gamma} >0$. In addition, the eigenvalues of $\Sigma_{\lambda}\Sigma_{\gamma}$ are distinct. 
            
            \item \label{item:X_generalized_nonlinearity} \textbf{(Generalized non-collinearity)} For any $\Gamma\in \mb{R}^{T\times R}$, let $M_{\Lambda_0}$ and $M_{\Gamma}$ be coprojection matrices of $\Lambda_0$ and $\Gamma$, respectively.  The $d_{X}\times d_{X}$ matrix $D(\Gamma)$ with elements 
            \begin{align*}
                D(\Gamma)_{d_1, d_2} = \frac{1}{NT} \mr{Tr} \left(M_{\Lambda_0}X_{d_1} M_{\Gamma}X_{d_2}'\right), \quad d_{1}, d_{2}  = 1, \ldots, d_X, 
            \end{align*} 
            satisfies $\inf_{\Gamma \in \mb{R}^{T\times R}} \sigma_{\min} (D(\Gamma) )>0$ wpa1. 
        \end{enumerate}
    \end{assumption}

    Assumption~\ref{assumption:regularity}\ref{item:sampling} imposes the conditional independence of $Y_{it}$ across $i$ and $t$. This aligns with the sampling assumption in \citet{chen2021nonlinear} and is  primarily applicable in contexts where $X_{it}$ is strictly exogenous. It is also natural in network settings when the nature of interactions is primarily bilateral and the interdependencies of the outcomes are captured by the fixed effects. We consider a more general version of the sampling process allowing for predetermined covariates in Appendix~\ref{ssec: predetermined covariates}.

    Assumption~\ref{assumption:regularity}\ref{item:boundedness} imposes  boundedness on the parameter spaces for $\beta$,  $\Lambda$,  and $\Gamma$, as well as the boundedness on covariates. This assumption is widely adopted in the literature to derive concentration bounds (see, for example, \citealp{chernozhukov2019inference}, \citealp{chernozhukov2023inference}, and \citealp{ma2022detecting}). It is worth noting that \citet{fernandez2016individual}  and \citet{chen2021nonlinear} do not impose the boundedness of $\lambda_i$, $\gamma_t$, or $X_{it}$  because their analyses focus on the local properties of the objective function. In contrast, we impose stronger assumptions on the parameter space to study global properties and establish global convergence guarantees for the NNR estimator.\footnote{The boundedness of $\Lambda$ and $\Gamma$ could, in principle, be replaced by assuming that $\{\lambda_i\}_{1\leq i\leq N}$ and $\{\gamma_t\}_{1\leq t\leq T}$ are sub-Gaussian sequences that are independent across $i$ and weakly dependent over $t$, respectively.}

    Assumption~\ref{assumption:regularity}\ref{item:smoothness} is commonly adopted  in the nonlinear panel regression literature (see \citealp{fernandez2016individual} and \citealp{chen2021nonlinear})  and is satisfied by Logit, Probit, and Poisson models. 
    
    Assumption~\ref{assumption:regularity}\ref{item:strong_factors} requires the factors to be strong. This condition is standard in the factor model literature.\footnote{Developing estimation and inference methods robust to weak factors is an important but highly nontrivial problem, even in linear panels; see \citet{armstrong2022robust}. In this paper, we simply follow the set-up of \citet{chen2021nonlinear} and leave the important problem of allowing for weak factors in nonlinear models for future research.} Additionally, the boundedness of the nuisance parameter space ensures that the maximum eigenvalues of $\Sigma_{\lambda}$ and $\Sigma_{\gamma}$ are bounded. The additional assumption that $\Sigma_{\lambda}\Sigma_{\gamma}$  has distinct eigenvalues is not crucial and is introduced purely to simplify the discussion of technical aspects in the main text. In Appendix~\ref{ssec: distinct eigenvalues}, we demonstrate that this assumption can be relaxed without affecting our main results.

    Assumption~\ref{assumption:regularity}\ref{item:X_generalized_nonlinearity} is a generalized non-collinearity condition that rules out covariates that do not display variation in the individual and time dimensions, such as time-invariant or individual-invariant regressors. This condition is identical to Assumption 1(vii) in \citet{chen2021nonlinear}, and we the reader to that paper for further discussion. In fact, we do \emph{not} use this assumption in our proofs but we impose it here in order to be able to invoke the results of \citet{chen2021nonlinear} characterizing the asymptotic properties of the FE estimator.

    \bigskip 

    We now turn to a key condition for establishing the consistency of our NNR estimator, the restricted strong convexity (RSC) condition.\footnote{Originally introduced by \citet{negahban2012unified}, the RSC condition (in its various forms) has been widely applied in problems involving estimation of low-rank matrices, including matrix completion \citep{negahban2012restricted}, reduced-rank regression \citep{rohde2011estimation}, latent community detection \citep{ma2022detecting}, and econometric analysis of panel models with interactive unobserved heterogeneity \citep{moon2018nuclear, chernozhukov2019inference}.} As in models with a fixed number of parameters, establishing error bounds for NNR estimators requires analyzing the Hessian of the objective function $\mc{L}(\beta, \Theta)$ and demonstrating that it is positive definite, implying that the objective function is strongly convex. However, in the context of problem~\eqref{eq:nnr_definition}, it is not possible for the Hessian matrix to be positive definite, as the number of parameters grows with $N$ and  $T$ and exceeds the number of observations. Nonetheless, thanks to the nuclear norm penalty in \eqref{eq:nnr_definition}, it is sufficient to require $\mc{L}(\beta, \Theta)$ to be strictly convex only over a restricted part of the parameter space in which $\Theta$ is approximately low-rank.

    To elaborate, for any $(\beta, \Theta)$,  the second-order Taylor remainder of the objective function $\mc{L}(\beta, \Theta)$ around the true parameter $(\beta_0, \Theta_0)$ satisfies 
    \begin{equation}\label{eqref:second_order_convexity}
        \frac{1}{NT}\sum_{i=1}^{N}\sum_{t=1}^{T} \frac{-\ddot{\ell}_{it}(X_{it}'\tilde{\beta} + \tilde{\theta}_{it})}{2} (X_{it}' \Delta_{\beta} + \Delta_{\theta_{it}} )^2   \geq \frac{b_{\min}}{2}  \underbrace{\frac{1}{NT}\sum_{i=1}^{N}\sum_{t=1}^{T} (X_{it}' \Delta_{\beta} + \Delta_{\theta_{it}} )^2}_{\mc{E}_{NT} (\Delta_{\beta}, \Delta_{\Theta}) }, 
    \end{equation}
    where $(\tilde{\beta}, \tilde{\Theta})$ lies between $(\beta, \Theta)$ and $(\beta_0, \Theta_0)$, $\Delta_{\beta} = \beta - \beta_0$, $\Delta_{\theta_{it}} = \theta_{it} - \theta_{0,it}$, and the inequality follows from the fact that the second-order derivative $-\ddot{\ell}_{it}(\cdot)$ is uniformly bounded below by $b_{\min}$ (Assumption~\ref{assumption:regularity}\ref{item:smoothness}). Thus, it suffices to study the properties of $\mc{E}_{NT}(\Delta_{\beta}, \Delta_{\Theta})$ to ensure that the objective function is strongly convex over a restricted parameter space. Specifically, we introduce the following assumption.

    \begin{assumption}[Restricted strong convexity (RSC)]\label{assumption:RSC}
        For any $c_0 > 0$, there exist constants ${\kappa, \eta >0}$, independent of $N, T$,  such that for any $(\Delta_{\beta}, \Delta_{\Theta})\in \mc{C}_1 \cap \mc{C}_2$, 
        \begin{align}\label{eq:RSC}
            \frac{1}{NT}\sum_{i=1}^{N}\sum_{t=1}^{T}(X_{it}' \Delta_{\beta} + \Delta_{\theta_{it}} )^2 \geq \kappa \left(\|\Delta_{\beta}\|^2 + \frac{1}{NT}\|\Delta_{\Theta}\|_{\mr{F}}^2\right) - \eta \frac{N+T}{NT} (\log(NT))^2, \quad \text{wpa1},
        \end{align}
        where
        \begin{gather*}
            \mc{C}_1 := \left\{(\Delta_{\beta}, \Delta_{\Theta})\in (\mb{R}^{d_X}\times \mb{R}^{N\times T})\mid \|M_{\Lambda_0}\Delta_{\Theta}M_{\Gamma_0}\|_{\mr{nuc}} \leq c_0 (\sqrt{NT}\|\Delta_{\beta}\| + \|\Delta_{\Theta} - M_{\Lambda_0}\Delta_{\Theta}M_{\Gamma_0}\|_{\mr{nuc}})\right\}, \\
            \mc{C}_2 := \left\{ (\Delta_{\beta}, \Delta_{\Theta})\in (\mb{R}^{d_X}\times \mb{R}^{N\times T})\mid \|\Delta_{\beta}\|^2 + \frac{1}{NT}  \|\Delta_{\Theta}\|_{\mr{F}}^2 \geq  \sqrt{\frac{\log (NT)}{NT}}\right\}. 
        \end{gather*}
    \end{assumption}

    Assumption~\ref{assumption:RSC} controls the behavior of $\mc{E}_{NT} (\Delta_{\beta}, \Delta_{\Theta})$ defined in \eqref{eqref:second_order_convexity} over the restricted parameter space of interest $\mc{C}_1\cap \mc{C}_2$.\footnote{Notice that $\kappa$, $\eta$, and $\mathcal C_1$ introduced in Assumption~\ref{assumption:RSC} all depend on $c_0$,  but we suppress this dependence for ease of notation.} $\mc{C}_1$ can be viewed as an approximately low-rank space. The term ${\Delta_{\Theta} - M_{\Lambda_0}\Delta_{\Theta}M_{\Gamma_0}}$ on the right-hand side represents the component that can be explained by $\Lambda_0$ and $\Gamma_0$,  serving as a  low-rank approximation. In contrast, the  left-hand side, $M_{\Lambda_0}\Delta_{\Theta}M_{\Gamma_0}$,  corresponds to the residual  of $\Delta_{\Theta}$ that cannot be explained by $\Lambda_0$ and $\Gamma_0$, interpreted as the low-rank approximation residual. Therefore, $\mc{C}_1$ consists of matrices whose low-rank approximation residuals  (in terms of nuclear norm) are  small compared to their low-rank approximation (along with the estimation error of $\beta$). Importantly, it is sufficient to control the behavior of $\mc{E}_{NT} (\Delta_{\beta}, \Delta_{\Theta})$ over $\mathcal C_1$ because the appropriately chosen nuclear norm penalty in \eqref{eq:nnr_definition} enforces $(\Delta \hat{\beta}_{\mr{nuc}}, \Delta \hat{\Theta}_{\mr{nuc}}) \in \mathcal C_1$ for some $c_0$ wpa1.\footnote{See Lemma~\ref{lemma:cone} for a formal result and discussion.}
    
    The set $\mc{C}_2$ is introduced to restrict our attention to scenarios of primary interest. Since we can directly obtain the bounds $\|\Delta_{\beta}\|_{2}^2 \leq  \sqrt{\frac{\log (NT)}{NT}}$ and $\frac{1}{NT}  \|\Delta_{\Theta}\|_F^2 \leq  \sqrt{\frac{\log (NT)}{NT}}$ for matrices that do not belong to this space, focusing on $\mc{C}_2$ simplifies the analysis without affecting the results of this section. %

    Notice that inequality~\eqref{eq:RSC} also introduces an additional tolerance term $\eta \frac{N+T}{NT} (\log(NT))^2$ to account for the randomness in $X_{it}$. 
    Combined with~\eqref{eqref:second_order_convexity}, \eqref{eq:RSC} ensures that $\mc{L}(\beta, \Theta)$ is strongly convex (up to the tolerance term) over the restricted space $\mc{C}_1 \cap \mc{C}_2$, even in the high-dimensional setting.

    \bigskip
    
    While various versions of the RSC condition are routinely employed in the low-rank estimation literature (e.g., \citealp{moon2018nuclear,chernozhukov2019inference,ma2022detecting}), they are typically introduced as high-level assumptions and can be challenging to verify. To address this issue, \citet{chernozhukov2019inference} provided sufficient conditions for verifying RSC in panels with strictly exogenous covariates. In this paper, we extend the related results of \citet{chernozhukov2019inference} and provide a set of easily verifiable primitive conditions sufficient for the RSC to hold in settings with \emph{predetermined} $X_{it}$, including the outcome's lags. Since predetermined regressors are ubiquitous in applied work, we believe that this result, formalized by the lemma below, is an important technical contribution of the paper, broadening the applicability of low-rank estimators and enabling empirical researchers to adopt them with greater confidence.

        \begin{lemma}[Sufficient conditions for RSC]\label{lemma:sufficient_RSC}
            Under Assumptions~\ref{assumption:regularity_pre} and \ref{assumption:conditional_independence_RSC} provided in Appendix~\ref{sec:extension_theory}, Assumption~\ref{assumption:RSC} is satisfied. %
        \end{lemma}

    To streamline the exposition, we will still state our main results of this section using Assumption~\ref{assumption:RSC} and defer the discussion of Lemma~\ref{lemma:sufficient_RSC} and the related primitive conditions to Appendix~\ref{appendix:RSC}.

    \bigskip

    For notational simplicity, we impose the sign constraint (e.g., fixing the sign of each factor) and the normalization constraint on the nuisance parameters $(\Lambda_0, \Gamma_0)$ such that $\Lambda_0'\Lambda_0/N$ and $\Gamma_0'\Gamma_0/T$ are diagonal, with $\Lambda_0'\Lambda_0/N = \Gamma_0'\Gamma_0/T$. These constraints are consistent with the construction of $(\hat{\Lambda}_{\mr{nuc}}, \hat{\Gamma}_{\mr{nuc}})$ in~\eqref{eq:space_definition}. The feasibility of these constraints  is provided in Appendix~\ref{appendix:normalization}.
    
    The following theorem establishes the convergence rates of the NNR estimators.
    
    \begin{theorem}\label{thm:consistency}
        For any $\alpha>0$ such that  $\varphi_{NT} \geq (1+\alpha) \max\{ \|\nabla_{\beta}\mc{L}_{NT}\left(\beta_0 , \Theta_0\right)\|_{2}, \sqrt{NT}\|\nabla_{\Theta}\mc{L}_{NT}\left(\beta_0 , \Theta_0\right)\|_{\mr{op}} \}$, under Assumptions~\ref{assumption:regularity} and \ref{assumption:RSC}, there exist constants  $c_1, c_2>0$ that do not depend on $N, T$ such that, as $N, T\rightarrow\infty$, wpa1, 
        \begin{align*}
            \| \hat{\beta}_{\mr{nuc}} - \beta_0\| & \leq c_1 \left(\varphi_{NT} + \log (NT)/\sqrt{\min\{N, T\}}\right), \\
            \frac{1}{\sqrt{NT}}\|\hat{\Theta}_{\mr{nuc}} - \Theta_0\|_{\mr{F}} & \leq  c_1 \left(\varphi_{NT} + \log (NT)/\sqrt{\min\{N, T\}}\right).  
        \end{align*}
        In addition, wpa1, 
        \begin{align*}      
            \frac{1}{\sqrt{N}}\|\hat{\Lambda}_{\mr{nuc}} - \Lambda_0 \|_{\mr{F}} & \leq c_2  \left(\varphi_{NT} + \log (NT)/\sqrt{\min\{N, T\}} \right),  \\
            \frac{1}{\sqrt{T}}\|\hat{\Gamma}_{\mr{nuc}} - \Gamma_0 \|_{\mr{F}} & \leq c_2  \left(\varphi_{NT} + \log (NT)/\sqrt{\min\{N, T\}} \right). 
        \end{align*} 
    \end{theorem}

    Theorem~\ref{thm:consistency} establishes that, for a sufficiently large regularization parameter $\varphi_{NT}$, the NNR estimator $\hat{\beta}_{\mr{nuc}}$ converges to $\beta_0$ at a rate of order $\varphi_{NT} + \log(NT)/\sqrt{\min\{N,T\}}$. This convergence rate coincides with that obtained by \citet{moon2018nuclear} for the linear panel model. The estimator $\hat{\Theta}_{\mr{nuc}}$ achieves the same convergence rate under the normalized Frobenius norm, and the associated estimators of the latent factors $(\hat{\Lambda}_{\mr{nuc}}, \hat{\Gamma}_{\mr{nuc}})$ inherit this rate as well.
    
    The convergence rates in Theorem~\ref{thm:consistency} depend on the choice of the regularization parameter $\varphi_{NT}$. While $\varphi_{NT}$ must be sufficiently large to guarantee consistency, choosing $\varphi_{NT}$ larger than $O\left(\log (NT)/ \sqrt{  \min\{N, T\}}\right)$ would lead to slower rates of convergence. To further investigate this issue, we demonstrate that choosing $\varphi_{NT} = O\left(\log (NT)/ \sqrt{  \min\{N, T\}}\right)$ is compatible with the restriction imposed by Theorem~\ref{thm:consistency}, and thus achieves the fastest rate implied by the theorem. This result is formalized by the following corollary.

    \begin{corollary}\label{corollary:consistency}
        Suppose that the hypotheses of Theorem~\ref{thm:consistency} hold, and let $\varphi_{NT} = O\left(\log (NT)/ \sqrt{  \min\{N, T\}}\right)$. Then there exist constants  $c_3, c_4>0$ that do not depend on $N, T$ such that, as $N, T\rightarrow\infty$, wpa1, 
        \begin{align*}
            \| \hat{\beta}_{\mr{nuc}} - \beta_0\| & \leq c_3 \log (NT)/\sqrt{\min\{N, T\}},  \\
            \frac{1}{\sqrt{NT}}\|\hat{\Theta}_{\mr{nuc}} - \Theta_0\|_{\mr{F}} & \leq   c_3 \log (NT)/\sqrt{\min\{N, T\}}. 
        \end{align*}
        In addition, wpa1, 
        \begin{align*}      
            \frac{1}{\sqrt{N}}\|\hat{\Lambda}_{\mr{nuc}} - \Lambda_0 \|_{\mr{F}} & \leq c_4   \log (NT)/\sqrt{\min\{N, T\}}, \\
            \frac{1}{\sqrt{T}}\|\hat{\Gamma}_{\mr{nuc}} - \Gamma_0 \|_{\mr{F}} & \leq c_4  \log (NT)/\sqrt{\min\{N, T\}}. 
        \end{align*}
    \end{corollary}

    Corollary~\ref{corollary:consistency} establishes the convergence rate for the studied NNR estimator under the optimal choice of $\varphi_{NT}$. It is instructive to compare this result with the existing rates for NNR estimators previously obtained in the literature for different settings. In particular, up to an additional logarithmic factor, our rate coincides with the rates established by \citet{moon2018nuclear}, \citet{chernozhukov2019inference}, and \citet{su2025estimation} for linear panel models. In nonlinear models, \citet{ma2022detecting} established the same rate in a network formation model. Importantly, Corollary~\ref{corollary:consistency} improves on \citet[Theorem~A.1]{moon2018nuclear},  which extends their original analysis to single-index models. As we explain in the following section, this improvement in the convergence rate is crucial for demonstrating that the first-step estimator $(\hat{\beta}_{\mr{nuc}}, \hat{\Lambda}_{\mr{nuc}}, \hat{\Gamma}_{\mr{nuc}})$ is sufficiently close to the true values of the parameters, so the second-step optimization problem initialized at the NNR estimator is (locally) convex.

\subsection{Local convexity}
     
    In this section, we establish the asymptotic equivalence between our two-step estimator and the FE estimator. This result builds on the consistency of the NNR estimator and the local convexity of the objective function $\mc{L}_{NT}(\beta, \Lambda, \Gamma)$ around the true values of the parameters.

    Since  the NNR  estimator $(\hat{\beta}_{\mr{nuc}}, \hat{\Lambda}_{\mr{nuc}}, \hat{\Gamma}_{\mr{nuc}})$  is used to initialize the second-step optimization problem~\eqref{eq:FE_estimator},  its consistency ensures that the starting values lie within a shrinking neighborhood of the true parameters as $N, T \rightarrow \infty$. Thus, to establish the desired result, it suffices to focus on studying the local properties of  $\mc{L}_{NT}(\beta, \Lambda, \Gamma)$ within these shrinking neighborhoods. To formalize the argument, let $\{\delta_{NT}\}$ be a sequence of shrinking radii such that $\delta_{NT}\rightarrow 0$ as $N, T\rightarrow \infty$, 
    and define the shrinking neighborhood around the true parameters $(\beta_0, \Lambda_0, \Gamma_0)$ as follows:
    \begin{align}\label{eq:neighborhood}
        \mc{B}_{\delta_{NT}} := \bigg\{(\beta, \Lambda, \Gamma)  \mid  & \|\beta-\beta_0\|, \frac{1}{\sqrt{N}}\|\Lambda - {\Lambda}_0 \|_{\mr{F}} , \frac{1}{\sqrt{T}}\|\Gamma - {\Gamma}_0 \|_{\mr{F}} \leq \delta_{NT} \bigg\}. 
    \end{align} 
    The neighborhood $\mc{B}_{\delta_{NT}}$ consists of parameters whose distances from the true values $(\beta_0, \Lambda_0, \Gamma_0)$ are less than $\delta_{NT}$.

    We establish the desired asymptotic equivalence result by demonstrating that with a properly chosen $\delta_{NT}$, the following conditions hold:  (1) the NNR estimator falls within the shrinking neighborhood, i.e., $(\hat{\beta}_{\mr{nuc}}, \hat{\Lambda}_{\mr{nuc}}, \hat{\Gamma}_{\mr{nuc}})\in \mc{B}_{\delta_{NT}}$ wpa1, (2) the FE estimator, as the global minimizer of problem~\eqref{eq:FE_estimator},  also lies within the shrinking neighborhood (up to rotation) wpa1, and (3) the objective function $\mc{L}_{NT}(\beta, \Lambda, \Gamma)$ is strictly convex within the neighborhood $\mc{B}_{\delta_{NT}}$. Together, these conditions ensure that, wpa1, (1) any locally convergent solver initialized at the NNR estimator is guaranteed to find the global minimum of problem~\eqref{eq:FE_estimator}, and, consequently, (2) our two-step estimator is asymptotically equivalent to the FE estimator.

    Let $\delta_{NT} = \log(NT) \min\{N^{-3/8}, T^{-3/8}\}$. The NNR estimator falls within the corresponding shrinking neighborhood $\mc{B}_{\delta_{NT}}$ wpa1 by Corollary~\ref{corollary:consistency} (provided that $\varphi_{NT}$ is chosen appropriately). Similarly, the FE estimator also lies within $\mc{B}_{\delta_{NT}}$ wpa1, as shown in Lemma 1 in \citet{chen2021nonlinear}. Hence, to establish the desired equivalence result, it would be sufficient to establish local convexity of the original objective function $\mc{L}_{NT}(\beta, \Lambda, \Gamma)$ within $\mc{B}_{\delta_{NT}}$, which will be the main focus of this section.

    Specifically, consider the following local optimization problem
    \begin{equation}\label{eq:definition_local}
    \begin{aligned}
        (\hat{\beta}_{\mr{local}}, \hat{\Lambda}_{\mr{local}}, \hat{\Gamma}_{\mr{local}}) \in  \argmin_{(\beta, \Lambda, \Gamma) \in \mc{B}_{\delta_{NT}} }  \mc{L}_{NT}(\beta, \Lambda, \Gamma), 
    \end{aligned}
    \end{equation}
    where the parameter space is restricted to the shrinking neighborhood $\mc{B}_{\delta_{NT}}$. Notice that, as in the linear case, this problem does not have a unique solution for $\hat{\Lambda}_{\mr{local}}$ and $\hat{\Gamma}_{\mr{local}}$ and requires imposing $R^2$ additional constraints to identify these parameters. Building on the normalization proposed in \citet{chen2021nonlinear}, we introduce the following restricted parameter set
    \begin{align*}
        \Phi_{NT} = \left\{(\Lambda, \Gamma)\mid  \hat{\Lambda}_{\mr{nuc}}' \Lambda/N =   \Gamma' \hat{\Gamma}_{\mr{nuc}} /T  \right\}
    \end{align*}
    consistent with the construction in~\eqref{eq:space_definition}. While in practice researchers could employ alternative normalizations, focusing on $\Phi_{NT}$ greatly facilitates the following theoretical analysis thanks to the linearity of the imposed restriction. Specifically, instead of imposing $\Phi_{NT}$ directly, we convert it into a quadratic penalization term, $\| \hat{\Lambda}_{\mr{nuc}}' \Lambda/N -  \Gamma' \hat{\Gamma}_{\mr{nuc}} /T \|_{\mr{F}}^2$, 
    whose Hessian matrix only depends on  $(\hat{\Lambda}_{\mr{nuc}}, \hat{\Gamma}_{\mr{nuc}})$ due to the aforementioned linearity. In particular, consider the associated (penalized) optimization problem:
    \begin{equation}\label{eq:local_estimators}
        \begin{aligned}
            (\hat{\beta}_{\mr{local}}, \hat{\Lambda}_{\mr{local}}, \hat{\Gamma}_{\mr{local}}) =  \argmin_{(\beta, \Lambda, \Gamma) \in \mc{B}_{\delta_{NT}} }  \left\{\mc{L}_{NT}(\beta, \Lambda, \Gamma) +  \frac{1}{2} \| \hat{\Lambda}_{\mr{nuc}}' \Lambda/N -  \Gamma' \hat{\Gamma}_{\mr{nuc}} /T \|_{\mr{F}}^2  \right\}, 
        \end{aligned}
    \end{equation}
    which is equivalent to solving~\eqref{eq:definition_local} with the normalization constraint $(\Lambda, \Gamma)\in\Phi_{NT} $. The differentiability of the penalty simplifies the analysis of the associated (penalized) Hessian,  providing a more tractable alternative to directly studying \eqref{eq:definition_local} under the hard constraint. %

    To demonstrate that problem~\eqref{eq:local_estimators} is convex, we need to inspect the sample Hessian of its objective function with respect to $(\beta, \Lambda, \Gamma)$ given by
    \begin{align*}
        \mc{H}_{NT}(\beta, \Lambda, \Gamma) := \nabla^2 \mc{L}_{NT}(\beta, \Lambda, \Gamma) +  \frac{1}{2} \nabla^2 \| \hat{\Lambda}_{\mr{nuc}}' \Lambda/N -  \Gamma' \hat{\Gamma}_{\mr{nuc}} /T \|_{\mr{F}}^2. 
    \end{align*}
    It is a $(d_X + R(N+T))$-dimensional square matrix.\footnote{
        The sample Hessian is a matrix-valued function of a $d_X$-dimensional vector $\beta$, an  $N \times R$  parameter matrix $\Lambda$, and a  $T \times R$  parameter matrix $\Gamma$. These parameters are arranged as follows: 
        $$(\beta', \text{vec}(\Lambda')', \text{vec}(\Gamma')')', $$ where $\text{vec}(\cdot)$ denotes the vectorization operator, stacking the columns of a matrix into a vector.
    }   Establishing local convexity is equivalent to showing that $\mc{H}_{NT}(\beta, \Lambda, \Gamma)$ is positive definite uniformly over $(\beta, \Lambda, \Gamma)\in \mc{B}_{\delta_{NT}}$.
    Consider the following decomposition: %
    \begin{align}\label{eq:hessian decomposition main}
        \mc{H}_{NT}(\beta, \Lambda, \Gamma) = \underbrace{\mb{E}_{X, \Lambda_0, \Gamma_0}\mc{H}_{NT}(\beta_0, \Lambda_0, \Gamma_0)}_{\text{population Hessian at true parameters}} + \underbrace{\mc{H}_{NT}(\beta, \Lambda, \Gamma)- \mb{E}_{X, \Lambda_0, \Gamma_0}\mc{H}_{NT}(\beta_0, \Lambda_0, \Gamma_0)}_{\text{deviation}}. 
    \end{align} 
    The first term represents the population Hessian evaluated at the  true parameters $(\beta_0, \Lambda_0, \Gamma_0)$, whose smallest eigenvalue is strictly positive under general conditions.  The second term captures the deviation arising from the sampling error and the perturbations of $(\beta, \Lambda, \Gamma)$ from the  true parameters. Thus, by Weyl's theorem, we can demonstrate that problem~\eqref{eq:local_estimators} is locally convex provided that the deviation term is negligible (in terms of the operator norm) compared to $\mb{E}_{X, \Lambda_0, \Gamma_0}\mc{H}_{NT}(\beta_0, \Lambda_0, \Gamma_0)$.

    Following the literature, we control the behavior of the population Hessian using the condition below.%
    \begin{assumption}[Diagonal structure]\label{assumption:block}
        There exists a constant $C>0$ (that does not depend on $N, T$) such that
        \begin{align*}
            \mb{E}_{X, \Lambda_0, \Gamma_0}\mc{H}_{NT}(\beta_0, \Lambda_0, \Gamma_0)  \geq C   \mr{diag}\left\{\mb{I}_{d_X}, \frac{1}{N}\mb{I}_{N R}, \frac{1}{T}\mb{I}_{TR}\right\}. 
        \end{align*}
    \end{assumption}
    Assumption~\ref{assumption:block} is closely related to the asymptotic diagonal structure condition in \citet{chen2021nonlinear},  \citet{wang2022maximum}, and \citet{su2025estimation}. Notice that Assumption~\ref{assumption:block} is mild since it imposes conditions only on the population Hessian evaluated at the  true parameters and does not restrict the behavior of the sample Hessian $\mc{H}_{NT}(\beta, \Lambda, \Gamma)$ directly. Easily verifiable sufficient conditions for Assumption~\ref{assumption:block} are provided in  Lemma~\ref{lemma:sufficient_convexity} in Appendix~\ref{appendix:diagonal}.

    Equipped with Assumption~\ref{assumption:block}, we are now ready to inspect the behavior of the sample Hessian $\mc{H}_{NT}(\beta, \Lambda, \Gamma)$ using the decomposition in \eqref{eq:hessian decomposition main}. Specifically, to establish the desired result, we demonstrate that the contribution of the deviation term in \eqref{eq:hessian decomposition main} is asymptotically negligible uniformly over $(\beta, \Lambda, \Gamma)\in \mc{B}_{\delta_{NT}}$. This result is formalized by the following theorem.
    
    \begin{theorem}\label{thm:convexity_strong}
        Suppose that the hypotheses of Corollary~\ref{corollary:consistency} are satisfied, and that $N$ and $T$ have the same order. Let $ \delta_{NT}= \log (NT) \min\{N^{-3/8}, T^{-3/8}\} $, and $\mc{B}_{\delta_{NT}}$ be the neighborhood defined in ~\eqref{eq:neighborhood}. Then, under Assumption~\ref{assumption:block}, the following results hold: 
        \begin{enumerate}[label=(\roman*)]
            \item \label{item:1} the local optimization problem~\eqref{eq:local_estimators} is strictly convex wpa1;
            \item  \label{item:2} our two-step estimator is asymptotically equivalent to the FE estimator; 
            \item \label{item:3} the local optimization problem~\eqref{eq:local_estimators} is strongly  convex wpa1 uniformly over $\mc{B}_{\delta_{NT}}$, i.e., there exists a constant $c_5>0$ independent of $N, T$ such that for any $(\beta, \Lambda, \Gamma) \in \mc{B}_{\delta_{NT}}$,  
            \begin{align*}
                \mc{H}_{NT}(\beta, \Lambda, \Gamma) > c_5 \mr{diag}\left\{\mb{I}_{d_X}, \frac{1}{N}\mb{I}_{NR}, \frac{1}{T}\mb{I}_{TR}\right\}, \quad \text{wpa1}. 
            \end{align*}
        \end{enumerate}
    \end{theorem}

    Since both the NNR and the FE estimators fall into $\mc{B}_{\delta_{NT}}$ wpa1, Theorem~\ref{thm:convexity_strong}\ref{item:1} guarantees that any locally convergent solver initialized at $(\hat{\beta}_{\mr{nuc}}, \hat{\Lambda}_{\mr{nuc}}, \hat{\Gamma}_{\mr{nuc}})$ finds the global minimum, and thus it formally demonstrates that our two-step estimator is asymptotically equivalent to the FE estimator. Moreover, Theorem~\ref{thm:convexity_strong}\ref{item:3} further establishes strong local  convexity of~\eqref{eq:local_estimators}, implying that simple gradient descent can be effectively applied in the second step to find the global minimum, even in high-dimensional settings.

    \begin{remark}
        \citet{moon2018nuclear} and \citet{su2025estimation} obtained analogous asymptotic results for their post-NNR estimators in linear panel models by inspecting the second-step optimization problem and establishing its local convexity. Thus, Theorem~\ref{thm:convexity_strong} generalizes their analysis to nonlinear settings. However, this extension is technically nontrivial and requires a more elaborate treatment. In particular, in linear panel models, the original objective function is locally convex in $\mc{B}_{\delta_{NT}}$ whenever {$\delta_{NT}=o(1)$}, and thus consistency of the NNR estimator alone is sufficient to guarantee convergence of the second-step optimization to the global minimum. In contrast, establishing local convexity is more delicate in the studied nonlinear setting, requiring the neighborhood to shrink at a faster rate than in the linear cases. Therefore, unless the preliminary estimator has a sufficiently fast convergence rate to fall into the shrinking convexity region wpa1, the convergence of the second-step local optimization to the global solution cannot be guaranteed. In particular, it turns out that the convergence rate obtained by \citet[Theorem~A.1]{moon2018nuclear}  for the NNR estimator for single-index models is not sufficiently fast to satisfy this requirement.
    \end{remark}

    \section{Implementation Details}\label{sec:implementation}

    In this section, we provide practical implementation details for the proposed method, including specific optimization algorithms for both estimation steps (along with their convergence guarantees), as well as data-dependent procedures for selecting the regularization parameter $\varphi_{NT}$ and determining the number of factors. The specific algorithms and procedures described in this section are also implemented in our R~package~\texttt{NNRPanel}.

\subsection{Computation details for the NNR estimator}%

    We compute the NNR estimator defined in~\eqref{eq:nnr_definition} using the proximal gradient descent method following \citet{hastie2015statistical}. Specifically, given the k-step estimates $(\beta^{(k)}, \Theta^{(k)})$, the $k+1$-step~estimates are updated  by  solving   
    \begin{align*}
        \beta^{(k+1)}, \Theta^{(k+1)} \in \arg\min_{\beta, \Theta} \Big\{& \mc{L}_{NT}(\beta^{(k)}, \Theta^{(k)}) + \la \nabla_{\beta}\mc{L}_{NT}(\beta^{(k)}, \Theta^{(k)}), \beta - \beta^{(k)}\ra  + \la \nabla_{\Theta}\mc{\mc{L}}_{NT}(\beta^{(k)}, \Theta^{(k)}), \Theta - \Theta^{(k)} \ra  \\
        & + \frac{1}{2s_{\beta}}\|\beta - \beta^{(k)}\|^2 + \frac{1}{2s_{\theta}}  \|\Theta - \Theta^{(k)}\|_{\mr{F}}^2 + \frac{\varphi_{NT}}{\sqrt{NT}}\|\Theta\|_{\mr{nuc}}\Big\}, 
    \end{align*}
    where $\nabla_{\beta}\mc{L}_{NT}(\cdot, \cdot)$ is a $d_X$-dimensional vector of gradients with respect to $\beta$,  $\nabla_{\Theta}\mc{L}_{NT}(\cdot, \cdot)\in \mb{R}^{N\times T}$ is a matrix of gradients with respect to $\theta_{it}$, $\la \cdot, \cdot\ra$ denotes the inner product between two vectors or two matrices, and $s_{\beta}, s_{\theta}>0$ are step sizes.

    Importantly, the solution to the optimization problem is available in closed form. In particular, let ${\mc{S}^*_{s_{\theta}\frac{\varphi_{NT}}{\sqrt{NT}}}: \mb{R}^{N\times T}\mapsto \mb{R}^{N\times T}}$ denote the soft-thresholding operator applied to the singular values of an  $N \times T$  matrix, where  $s_{\theta}\frac{\varphi_{NT}}{\sqrt{NT}}$ is the threshold value. Specifically, for any matrix  $A\in \mb{R}^{N\times T}$ with singular value decomposition $A = U\Sigma V'$, let
     $$\mc{S}^*_{s_{\theta}\frac{\varphi_{NT}}{\sqrt{NT}}}(A) = U\mr{diag}\left\{\max\left\{\Sigma_{rr} - s_{\theta}\frac{\varphi_{NT}}{\sqrt{NT}}, 0\right\}_{r=1,\ldots, \min\{N, T\}}\right\}V'. $$
    Then, the NNR estimator in \eqref{eq:nnr_definition} can be obtained using the following algorithm.
    
    \begin{algorithm}[Proximal gradient descent]\label{algorithm:NNR}
        Compute the NNR estimator as follows:
        \begin{itemize}[label={}, leftmargin=2cm]
            \item[Step 1:] Fix the step sizes $(s_{\beta}, s_{\theta})$. Initialize $\beta^{(0)}$ and $\Theta^{(0)}$. Set $k = 0$. 
            \item[Step 2:] Let 
            \begin{equation*}
                \begin{aligned}
                    \beta^{(k+1)} & = \beta^{(k)} - s_{\beta} \nabla_{\beta}\mc{L}_{NT}\left(\beta^{(k)}, \Theta^{(k)}\right), \\
                    \Theta^{(k+1)} & = \mc{S}^*_{s_{\theta}\frac{\varphi_{NT}}{\sqrt{NT}}}\left(\Theta^{(k)} - s_{\theta}\nabla_{\Theta}\mc{L}_{NT}\left(\beta^{(k)}, \Theta^{(k)}\right)\right), 
               \end{aligned}
            \end{equation*}
            and set $k = k+1$. 
            \item[Step 3:] Repeat Step 2 until convergence. 
        \end{itemize}
    \end{algorithm}

    We establish the convergence of Algorithm~\ref{algorithm:NNR} using a proof strategy similar to the one provided in \citet{nesterov2013gradient}.
    \begin{theorem}\label{thm:algorithm_convergence}
        Suppose that the hypotheses of Theorem \ref{thm:consistency} are satisfied. Then, Algorithm~\ref{algorithm:NNR} is guaranteed to converge to the global minimizer if $0 < s_{\beta} < \frac{1}{L_{\beta}}$ and $0 < \frac{s_{\theta}}{NT} < \frac{1}{L_{\theta} }$, where 
        $ L_{\beta} = 2 d_X b_{\max} \rho_X^2$, $L_{\theta} =  2b_{\max}$, and ${\rho_X = \max_{1\leq d\leq d_X} \|X_d\|_{\max}}$.  
    \end{theorem}

    The theorem establishes the algorithm's convergence to the global minimizer provided that $(s_{\beta}, s_{\theta}/(NT))$ is sufficiently small. Notice that the step sizes scale differently: $s_{\beta}$ and $s_{\theta}/(NT)$ should be of the same order, reflecting their distinct influence on the objective function. This difference arises because a change in $\beta$ affects $\ell_{it}$ for all $(i,t)$, whereas a change in $\theta_{it}$ affects only the corresponding $\ell_{it}$.

    To apply Algorithm~\ref{algorithm:NNR} in practice, one needs to specify the initial values $(\beta^{(0)},\Theta^{(0)})$ and step sizes $(s_\beta, s_\theta)$. Since the optimization problem is convex, the algorithm converges to a global minimizer regardless of the choice of its initial values. Thus, researchers may simply initialize with  $\beta^{(0)} = 0$  and  $\Theta^{(0)} = 0$. While Theorem~\ref{thm:algorithm_convergence} characterizes how small $s_\beta$ and $s_\theta$ should be to guarantee the convergence, it offers limited practical guidance, since $b_{\max}$ is typically unknown. In practice, we recommend starting each iteration with step sizes $s_{\beta} = 1$ and $s_{\theta} = NT$. If the objective function increases, the step sizes are iteratively halved until a decrease in the objective function is achieved.

    Finally, we note that solving the nuclear norm-regularized problem is usually more computationally demanding than the second step. The computational bottleneck is  computing the singular value decomposition of an  $N \times T$  matrix at each iteration. Our R package implementation remains computationally efficient even for matrices of size $N, T = 1000$.  In principle, to speed up the computation for much larger values of  $N$  or  $T$, it is also possible to employ the accelerated proximal gradient method (see, e.g., \citealp{nesterov2013gradient}), but we do not explore this direction in this paper.%

\subsection{Computation details for the local estimator}%

    In the second step, we employ gradient descent to search for a minimizer of~\eqref{eq:local_estimators}. The following algorithm practically minimizes $\mc{L}_{NT}(\beta, \Lambda, \Gamma)$ using the NNR estimator as the initial value. Step 3 is introduced to address the rotational invariance of the loadings and factors. It is optional since the penalization introduced in~\eqref{eq:local_estimators} is only used to facilitate the theoretical analysis. %
   
    \begin{algorithm}[Gradient descent]\label{algorithm:local}
        Compute the local estimator as follows:
        \begin{itemize}[label={}, leftmargin=2cm]
            \item[Step 1:] Fix the step sizes $(s_{\beta}, s_{\lambda}, s_{\gamma})$. Initialize  $\beta^{(0)} = \hat{\beta}_{\mr{nuc}}$,  $\Lambda^{(0)} = \hat{\Lambda}_{\mr{nuc}}$, and $\Gamma^{(0)} = \hat{\Gamma}_{\mr{nuc}}$.  Set $k = 0$. 
            \item[Step 2:] Let 
            \begin{equation*}
                \begin{aligned}
                    \beta^{(k+1)} & = \beta^{(k)} - s_{\beta} \nabla_{\beta}\mc{\mc{L}}_{NT}(\beta^{(k)}, \Lambda^{(k)}, \Gamma^{(k)} ),  \\
                    \Lambda^{(k+1)} & = \Lambda^{(k)} - s_{\lambda}\nabla_{\lambda}\mc{L}_{NT}(\beta^{(k)}, \Lambda^{(k)}, \Gamma^{(k)} ),  \\
                    \Gamma^{(k+1)} & = \Gamma^{(k)} - s_{\gamma}\nabla_{\gamma}\mc{L}_{NT}(\beta^{(k)}, \Lambda^{(k)}, \Gamma^{(k)} ), 
               \end{aligned}
            \end{equation*}
            and set $k = k+1$.
            \item[Step 3:] (Optional) Normalize $\Lambda^{(k+1)}$ and $\Gamma^{(k+1)}$, for example, let $\frac{1}{N}\Lambda^{(k+1)\prime}\Lambda^{(k+1)} = \frac{1}{T}\Gamma^{(k+1)\prime}\Gamma^{(k+1)}$ and diagonal. 
            \item[Step 4:] Repeat Step 2 and Step 3 until convergence. 
        \end{itemize}
    \end{algorithm}

    The following theorem provides conditions under which Algorithm~\ref{algorithm:local} is guaranteed to converge.  %
    \begin{theorem}\label{thm:algorithm_convergence_local}
        Suppose that the hypotheses of Theorem~\ref{thm:convexity_strong} are satisfied. Then, Algorithm~\ref{algorithm:local} is guaranteed to converge to a global minimizer  when $0 < s_{\beta} < \frac{1}{L_{\beta}}$,  $0 < \frac{s_{\lambda}}{N} < \frac{1}{L_{\lambda} }$, and $0 < \frac{s_{\gamma}}{T} < \frac{1}{L_{\gamma} }$, where 
        $ L_{\beta}, L_{\lambda}, L_{\gamma}$ are sufficiently large constants independent of $N, T$. 
    \end{theorem} 
    Similarly to the analogous requirement of Theorem~\ref{thm:algorithm_convergence}, to guarantee the convergence, the step sizes need to be scaled differently. Specifically, we require $s_{\beta}\sim s_{\lambda}/N \sim s_{\gamma}/T $, reflecting their respective influence on the objective function. In practice, in each step, we recommend  starting with $s_{\beta} = 1$, $s_{\lambda} = N$, and $s_{\gamma} = T$. If the objective function increases, we iteratively halve the step sizes $(s_{\beta}, s_{\lambda}, s_{\gamma})$ until the  objective function decreases.

    \begin{remark}
        \citet{chen2021nonlinear} propose using an EM algorithm, a Newton-Raphson-type method that theoretically achieves faster convergence through second-order accuracy. Despite this potential advantage, implementation of their method requires calculating and inverting a high-dimensional Hessian matrix, which can computationally expensive but also numerically unstable in the studied nonlinear setting. Thus, we adopt a more robust, albeit slower, gradient-based algorithm. 
    \end{remark}

\subsection{Determining the number of factors $R$ and regularization parameter $\varphi_{NT}$}

    In this section, we propose a data-driven procedure for selecting the regularization parameter~$\varphi_{NT}$ and determining the number of factors $R$. %
    
    Recall that Theorem~\ref{thm:consistency} requires $\varphi_{NT} >  \sqrt{NT}\|\nabla_{\Theta}\mc{L}_{NT}(\beta_0, \Theta_0)\|_{\mr{op}} $ in order to achieve the desired consistency result.\footnote{Notice that 
        $\|\nabla_{\beta}\mc{L}_{NT}(\beta_0, \Theta_0)\|_2$ is negligible compared to $\sqrt{NT}\|\nabla_{\Theta}\mc{L}_{NT}(\beta_0, \Theta_0)\|_{\mr{op}}$ as $N,T\rightarrow \infty$. 
    } At the same time, Theorem~\ref{thm:consistency} also indicates that selecting an excessively large $\varphi_{NT}$ should be avoided, as it may induce substantial estimation error in the NNR estimator.
    Hence, a preferable choice for $\varphi_{NT}$ is one that slightly exceeds $\sqrt{NT}\|\nabla_{\Theta}\mc{L}_{NT}(\beta_0, \Theta_0)\|_{\mr{op}}$. This quantity, however, is not immediately available since $\beta_0$ and $\Theta_0$ are unknown. To address this, we propose the following three-step procedure for jointly selecting $\varphi_{NT}$ and $R$, building on the idea of \citet{chernozhukov2019inference}.

    \begin{algorithm}[Selecting $\varphi_{NT}$ and $R$]\label{algorithm:tuning}
        \leavevmode
        \begin{itemize}[label={}, leftmargin=2cm]
            \item[Step 1:] Solve the following optimization problem with additive fixed effects to obtain $(\tilde{\beta}_{1}, \tilde{\Lambda}_{1}, \tilde{\Gamma}_{1})$:  %
            \begin{align*}
                (\tilde{\beta}_{1}, \tilde{\Lambda}_{1}, \tilde{\Gamma}_{1}) \in \argmin_{\beta\in \mb{R}^{d_X}, \Lambda\in \mb{R}^{N}, \Gamma\in\mb{R}^{T}} -\frac{1}{NT}\sum_{i=1}^{N}\sum_{t=1}^{T} \ell(Y_{it}\mid X_{it}'\beta + \lambda_i + \gamma_t), 
            \end{align*}
            and calculate the initial guess for the regularization parameter as  
            $$\tilde{\varphi}_{NT}:= (1 + \alpha)\|\nabla_{\Theta}\mc{L}_{NT}(\tilde{\beta}_1, \tilde{\Lambda}_1, \tilde{\Gamma}_1)\|_{\mr{op}}, $$
             where $\alpha$ is a small positive constant (e.g., $\alpha = 0.05$).
            \item[{Step 2:}] Solve the nuclear-norm regularized optimization problem~\eqref{eq:nnr_definition} with the previously obtained~$\tilde{\varphi}_{NT}$ to get $(\tilde{\beta}_2, \tilde{\Theta}_2)$. To obtain a preliminary estimator of the number of factors, consider the singular value sequence of $\tilde{\Theta}_{2}$, i.e.,  $\psi_{1}(\tilde{\Theta}_{2})\geq \psi_{2}(\tilde{\Theta}_{2}) \geq \ldots \geq \psi_{\min\{N, T\}}(\tilde{\Theta}_{2})$.  Compute $\tilde{R}$ by the eigenvalue-ratio test \citep{ahn2013eigenvalue} with a preset upper bound on the number of factors $R_{\max}>0$, i.e.,  
            \begin{align*}
                \tilde{R} = \argmax_{r = 1,2,\ldots, R_{\max}}  \psi_{r}(\tilde{\Theta}_{2})/ \psi_{r+1}(\tilde{\Theta}_{2}). 
            \end{align*}
            Finally, compute $(\tilde{\Lambda}_2, \tilde{\Gamma}_2)$ as in~\eqref{eq:space_definition} using $\tilde{R}$. 
            \item[Step 3:] Calculate %
            \begin{align*}
            \varphi_{NT}:= \left(1 +\alpha\right)\|\nabla_{\Theta}\mc{L}_{NT}(\tilde{\beta}_2, \tilde{\Lambda}_2, \tilde{\Gamma}_2)\|_{\mr{op}}, 
            \end{align*} 
            and solve problem~\eqref{eq:nnr_definition} using the obtained $\varphi_{NT}$ to get the final NNR estimates $(\hat{\beta}_{\mr{nuc}},\hat{\Theta}_{\mr{nuc}})$. The~estimated number of factors $\hat{R}$ is given by   %
            \begin{align*}
                \hat{R} = \argmax_{r = 1,2,\ldots, R_{\max}}  \psi_{r}(\hat{\Theta}_{\mr{nuc}})/ \psi_{r+1}(\hat{\Theta}_{\mr{nuc}}). 
            \end{align*}
        \end{itemize}
    \end{algorithm}

    Algorithm~\ref{algorithm:tuning} is computationally simple to implement.
    Step~1 and Step~2 require solving convex optimization problems, for which efficient algorithms are readily available. Importantly, our procedure avoids using cross-validation for selecting $\varphi_{NT}$, which substantially reduces the computational cost.

    We investigate the performance of Algorithm~\ref{algorithm:tuning} in numerical experiments and document its consistently good performance for different specifications with sample sizes ranging from $(N,T) = (50,40)$ to $(N,T) = (1000,200)$. This routine is already implemented in our R~package~\texttt{NNRPanel} and only requires the user to specify an ex-ante upper bound on the number of factors $R_{\max}$ as an input, which is a standard requirement even in linear factor models.

\section{Numerical Evidence}\label{sec:MC}

    In this section, we investigate the finite sample properties of our method in numerical experiments and revisit the empirical application studied in \citet{chen2021nonlinear}.

\subsection{Simulation study}

    In this section, we evaluate the performance of our estimator in canonical binary response panel models. For brevity, here we will present and discuss results for specifications with strictly exogenous covariates (conditional on the fixed effects). Additional numerical results for dynamic specifications with predetermined covariates are provided and discussed in Appendix~\ref{ssec:additional MCs}.

    In the considered numerical experiments, the data are generated from the following binary response model with $R = 2$: 
    \begin{equation}\label{eq:logit_static}
        \begin{aligned}
            Y_{it} & = \bs{1}\left(\beta_{1}X_{it} + \lambda_{i}' \gamma_{t} + \epsilon_{Y, it}>0\right), \\
            X_{it} & =  \lambda_{i}' \gamma_{t} + \lambda_{i}' \iota  + \gamma_{ t}' \iota  + \lambda_{X, i} \gamma_{X, t} +  \epsilon_{X, it}. 
        \end{aligned}
    \end{equation} 
    The error terms $\{\epsilon_{Y, it}\}_{1\leq i\leq N, 1\leq t\leq T}$ are i.i.d. over both dimensions. We consider two canonical designs: (i) \emph{Probit},  where $\epsilon_{Y, it}$ follows the standard normal distribution $N(0, 1)$; and (ii) \emph{Logit}, where $\epsilon_{Y, it}$ follows the standard logistic  distribution.

    In both designs,  $ \lambda_{i} = (\lambda_{i1}, \lambda_{i2})'$, $ \gamma_{t} = (\gamma_{t1}, \gamma_{t2})'$.  $\{ \lambda_{ir}\}_{ 1\leq i\leq N,  r=1,2}$,  $\{ \gamma_{tr}\}_{1\leq t\leq T, r=1,2}$, $\{ \lambda_{X, i}\}_{ 1\leq i\leq N}$, and $\{ \gamma_{X, t}\}_{ 1\leq t\leq T}$ consist of independent random variables drawn from the standard normal distribution $N(0, 1)$. 
    In addition, these collections are mutually independent. For both designs, $\{\epsilon_{X, it}\}_{1\leq i\leq N, 1\leq t\leq T}$ are i.i.d. (over both dimensions) draws from $N(0, 4)$, and $\beta_1 = 0.2$. We vary the sample size from $(N,T) = (50,40)$ to $(N,T) = (1000,200)$ and perform 1000 replications for each design.

    We report results for two versions of our two-step estimator. The first is implemented using the correct number of factors $R = 2$ ($\mr{TS}^*$), whereas the second involves estimating the number of factors~($\mr{TS}$). For both of these estimators, we use Algorithm~\ref{algorithm:tuning} to choose $\varphi_{NT}$ (setting $\alpha = 0.05$), and we specify $R_{\max} = 5$ when we also need to estimate the number of factors.\footnote{For the $\mr{TS}^*$ estimator, we simply apply Algorithm~\ref{algorithm:tuning} using the true number of factors.} For both of these estimators, we also report results for their analytical and jackknife bias-corrected versions ($\mr{ABC}^*$ and $\mr{JBC}^*$ for $\mr{TS}^*$, and $\mr{ABC}$ and $\mr{JBC}$ for $\mr{TS}$).\footnote{See Appendix~\ref{appendix:bias_correction} for bias-correction implementation details.}  We also report results for the naive pooled estimator that ignores individual and time latent factors (POOL) and for our first-step nuclear norm-regularized estimator (NNR).

    Table~\ref{tab:probit_static} reports the results for the Probit specification. As expected, the pooled estimator (POOL) has a substantial bias that does not diminish as the sample size increases, whereas the bias of the NNR estimator decreases slowly towards zero as $N, T \to \infty$, supporting the findings of Theorem~\ref{thm:consistency}. The proposed two-step estimator using the true number of factors ($\mr{TS}^*$) substantially improves upon the NNR estimator: its bias is smaller and converges rapidly to zero as $N$ and $T$ increase. Nevertheless, due to the incidental parameter problem, the bias and standard deviation of the $\mr{TS}^*$ estimator remain of the same order even for large sample sizes, such as $(N, T) = (1000, 200)$. The analytical bias correction method effectively addresses this problem even in smaller samples, e.g., for $(N, T) = (50, 40)$. While both the analytical and jackknife correction methods effectively reduce the bias in larger sample sizes, the former might still be preferred to the latter because it does not inflate the (asymptotic) variance of the estimator. Finally, when the number of factors is estimated rather than known, the corresponding estimators ($\mr{TS}$, $\mr{ABC}$, and $\mr{JBC}$) continue to perform well.\footnote{The average number of the estimated factors is reported in the last column of the table and denoted by $\bar{R}$.}

    Table~\ref{tab:logit_static} reports the results for the Logit specification. These findings are very similar to those for the Probit specification, except for the  deterioration in the performance of the analytical bias correction method in smaller sample sizes. The analogous results for dynamic specifications with predetermined covariates are also qualitatively similar and discussed in Appendix~\ref{ssec:additional MCs}. Overall, the considered numerical experiments demonstrate that our two-step estimator has good finite sample properties and that it can be efficiently computed even in large panels with $(N,T) = (1000,200)$.

    \begin{landscape}
    \begin{table}[p]
            \centering
            \caption{Simulation Results: Probit Model}\label{tab:probit_static}
            \begin{tabular}{cccccccccc}
            \toprule
                    & POOL  & NNR   & $\mr{TS}^*$    & $\mr{ABC}^*$ & $\mr{JBC}^*$ & $\mr{TS}$  & ABC & JBC & $\bar{R}$ \\
                    & $( \times 10^{-2})$  & $( \times 10^{-2})$  & $(10^{-2})$  & $(\times 10^{-2})$  & $(\times 10^{-2})$  & $(\times 10^{-2})$  & $(\times 10^{-2})$  & $(\times 10^{-2})$  &  \\
            \midrule
            N = 50, T = 40 &       &       &       &       &       &       &       &       &  \\
        BIAS  & 2.14  & 3.69  & 3.78  & 0.78  & -3.05 & 3.78  & 0.85  & -2.87 & 1.962 \\
        STD   & (1.82)  & (1.85)  & (2.53)  & (2.15)  & (3.61)  & (2.51) & (2.20)  & (3.74)  &  \\
        N = 100, T = 40 &       &       &       &       &       &       &       &       &  \\
        BIAS  & 2.04  & 3.25  & 2.29  & 0.18  & -1.56 & 2.30  & 0.18  & -1.55 & 1.999 \\
        STD   & (1.46)  & (1.32)  & (1.59)  & (1.42)  & (2.13)  & (1.59)  & (1.42)  & (2.13)  &  \\
        N = 200, T = 40 &       &       &       &       &       &       &       &       &  \\
        BIAS  & 2.06  & 3.08  & 1.79  & 0.05  & -0.55 & 1.79  & 0.05  & -0.55 & 2.000 \\
        STD   & (1.35)  & (1.11)  & (1.06)  & (0.96)  & (1.36)  & (1.06)  & (0.96)  & (1.36)  &  \\
        N = 100, T = 100 &       &       &       &       &       &       &       &       &  \\
        BIAS  & 2.07  & 2.70  & 1.17  & 0.07  & -0.34 & 1.17  & 0.07  & -0.34 & 2.000 \\
        STD   & (1.10)  & (0.91)  & (0.86)  & (0.81)  & (0.97)  & (0.86)  & (0.81)  & (0.97)  &  \\
        N = 200, T = 100 &       &       &       &       &       &       &       &       &  \\
        BIAS  & 2.01  & 2.41  & 0.83  & 0.02  & -0.13 & 0.83  & 0.02  & -0.13 & 2.000 \\
        STD   & (0.94)  & (0.68)  & (0.60)  & (0.58)  & (0.65)  & (0.60)  & (0.58)  & (0.65)  &  \\
        N = 200, T = 200 &       &       &       &       &       &       &       &       &  \\
        BIAS  & 1.98  & 2.14  & 0.53  & 0.00  & -0.07 & 0.53  & 0.00  & -0.07 & 2.000 \\
        STD   & (0.75)  & (0.49)  & (0.41)  & (0.40)  & (0.44)  & (0.41)  & (0.40)  & (0.44)  &  \\
        N = 1000, T = 200 &       &       &       &       &       &       &       &       &  \\
        BIAS  & 1.98  & 1.93  & 0.31  & 0.00  & -0.01 & 0.31  & 0.00  & -0.01 & 2.000 \\
        STD   & (0.57)  & (0.28)  & (0.18)  & (0.18)  & (0.25)  & (0.18)  & (0.18)  & (0.25)  &  \\
        \bottomrule
        \end{tabular}%
        \vspace{0.5cm}
        \begin{minipage}{\textwidth}
            \footnotesize
            \textbf{Note:} Monte Carlo results based on $1000$ replications for the Probit  model with strictly exogenous covariates as in~\eqref{eq:logit_static}. We set $\alpha = 0.05$ to determine the regularization parameter $\varphi_{NT}$. We report bias and standard deviation for  pooled estimators (POOL), nuclear norm regularized estimators (NNR). Using the true number of factors, $R=2$, we report bias and standard deviation (measured in units of $\times 10^{-2}$)  for our proposed estimators ($\mr{TS}^*$),   analytical bias-corrected estimators based on $\mr{TS}^*$, ($\mr{ABC}^*$) and jackknife bias-corrected estimators based on $\mr{TS}^*$, ($\mr{JBC}^*$). Using estimated number of factors $\hat{R}$,  we report bias and standard deviation  (measured in units of $\times 10^{-2}$)  for our proposed estimators ($\mr{TS}$),   analytical bias-corrected estimators based on $\mr{TS}$, ($\mr{ABC}$) and jackknife bias-corrected estimators based on $\mr{TS}$, ($\mr{JBC}$). Furthermore, we report the average estimated number of factors $\bar{R}$.  
        \end{minipage}
        \end{table}  
        \end{landscape}

        \begin{landscape}
        \begin{table}[p]
        \centering
        \caption{Simulation Results: Logit Model}\label{tab:logit_static}
        \begin{tabular}{cccccccccc}
        \toprule
                & POOL  & NNR   & $\mr{TS}^*$    & $\mr{ABC}^*$ & $\mr{JBC}^*$ & $\mr{TS}$  & ABC & JBC & $\bar{R}$ \\
                & $( \times 10^{-2})$  & $( \times 10^{-2})$  & $(10^{-2})$  & $(\times 10^{-2})$  & $(\times 10^{-2})$  & $(\times 10^{-2})$  & $(\times 10^{-2})$  & $(\times 10^{-2})$  &  \\
        \midrule
        N = 50, T = 40 &       &       &       &       &       &       &       &       &  \\
        BIAS  & 7.51  & 6.89  & 5.40  & 3.64  & -3.32 & 5.45  & 3.70  & -3.21 & 1.802 \\
        STD   & (2.48)  & (2.40)  & (3.64)  & (3.48)  & (6.56)  & (3.68)  & (3.52)  & (6.75)  &  \\
        N = 100, T = 40 &       &       &       &       &       &       &       &       &  \\
        BIAS  & 7.57  & 6.64  & 3.09  & 1.70  & -2.25 & 3.10  & 1.71  & -2.23 & 1.959 \\
        STD   & (1.99)  & (1.88)  & (2.39)  & (2.33)  & (4.01)  & (2.40)  & (2.34)  & (4.02)  &  \\
        N = 200, T = 40 &       &       &       &       &       &       &       &       &  \\
        BIAS  & 7.46  & 6.26  & 1.89  & 0.69  & -1.24 & 1.89  & 0.69  & -1.24 & 1.999 \\
        STD   & (1.67)  & (1.50)  & (1.56)  & (1.50) & (2.33)  & (1.56)  & (1.50)  & (2.33)  &  \\
        N = 100, T = 100 &       &       &       &       &       &       &       &       &  \\
        BIAS  & 7.35  & 5.82  & 1.23  & 0.32  & -0.90 & 1.23  & 0.32  & -0.90 & 2.000 \\
        STD   & (1.40)  & (1.22)  & (1.17)  & (1.12)  & (1.51)  & (1.17)  & (1.12)  & (1.51)  &  \\
        N = 200, T = 100 &       &       &       &       &       &       &       &       &  \\
        BIAS  & 7.36  & 5.39  & 0.76  & 0.11  & -0.29 & 0.76  & 0.11  & -0.29 & 2.000 \\
        STD   & (1.23)  & (1.01)  & (0.88)  & (0.85)  & (1.07)  &(0.88)  & (0.85)  & (1.07)  &  \\
        N = 200, T = 200 &       &       &       &       &       &       &       &       &  \\
        BIAS  & 7.34  & 4.87  & 0.51  & 0.08  & -0.09 & 0.51  & 0.08  & -0.09 & 2.000 \\
        STD   & (0.95)  & (0.74)  & (0.59)  & (0.57)  & (0.68)  & (0.59)  & (0.57)  & (0.68)  &  \\
        N = 1000, T = 200 &       &       &       &       &       &       &       &       &  \\
        BIAS  & 7.36  & 4.07  & 0.27  & 0.02  & -0.01 & 0.27  & 0.02  & -0.01 & 2.000 \\
        STD   & (0.70)  & (0.45)  & (0.27)  & (0.26)  & (0.34)  & (0.27)  & (0.26)  & (0.34)  &  \\
        \bottomrule
        \end{tabular}%
        \vspace{0.5cm}
        \begin{minipage}{\textwidth}
            \footnotesize
            \textbf{Note:} Monte Carlo results based on $1000$ replications for the Logit  model with strictly exogenous covariates as in~\eqref{eq:logit_static}. We set $\alpha = 0.05$ to determine the regularization parameter $\varphi_{NT}$. We report bias and standard deviation (measured in units of $\times 10^{-2}$)  for  pooled estimators (POOL), nuclear norm regularized estimators (NNR). Using the true number of factors, $R=2$, we report bias and standard deviation for our proposed estimators ($\mr{TS}^*$),   analytical bias-corrected estimators based on $\mr{TS}^*$, ($\mr{ABC}^*$) and jackknife bias-corrected estimators based on $\mr{TS}^*$, ($\mr{JBC}^*$). Using estimated number of factors $\hat{R}$,  we report bias and standard deviation  (measured in units of $\times 10^{-2}$) for our proposed estimators ($\mr{TS}$),   analytical bias-corrected estimators based on $\mr{TS}$, ($\mr{ABC}$) and jackknife bias-corrected estimators based on $\mr{TS}$, ($\mr{JBC}$). Furthermore, we report the average estimated number of factors $\bar{R}$.  
        \end{minipage}
    \end{table}
    \end{landscape}

\subsection{Empirical application}
    \label{ssec: empirical}

    In this section, we revisit the empirical application studied by \citet{chen2021nonlinear}, who proposed estimating the gravity equation with interactive fixed effects to examine the determinants of international trade flows. Specifically, our goal is to investigate if our two-step approach can replicate the estimates reported in \citet{chen2021nonlinear}.

    The data, originally from \citet{helpman2008estimating}, include bilateral trade flows and other relevant variables for $N = 157$ countries. Following \citet{chen2021nonlinear}, we focus on the year $1986$ with the trade network including 157 countries, with the effective sample size $157 \times 156 = 24{,}492$ (the number of distinct country pairs). The outcome variable  $Y_{ij}$  represents the volume of trade from country  $i$  to country  $j$ (measured in thousands of constant $2000$ US dollars). The covariates $X_{ij}$ include key determinants of $Y_{ij}$, such as the logarithm of the distance between the capitals of the two countries (Log distance) and binary indicators for shared border (Border), legal system (Legal), common language (Language), colonial ties (Colony), currency union (Currency), regional free-trade agreement (FTA), and religion (Religion). The descriptive statistics are presented in Table~\ref{tab:summary} below.

    \begin{table}[H]
        \centering
        \caption{Summary Statistics}
    
          \begin{tabular}{lcc}
          \toprule
                & Mean & Standard deviation   \\
          \midrule
          \multicolumn{1}{l}{Trade Volume} & 84,542 & 1,082,219   \\
          \multicolumn{1}{l}{Log distance} & 4.18 & 0.78  \\
          \multicolumn{1}{l}{Border} & 0.02  & 0.13   \\
          \multicolumn{1}{l}{Legal} & 0.37  & 0.48   \\
          \multicolumn{1}{l}{Language} & 0.29 & 0.45  \\
          \multicolumn{1}{l}{Colony} & 0.01  & 0.10    \\
          \multicolumn{1}{l}{Currency} & 0.60  & 1.37    \\
          \multicolumn{1}{l}{FTA} & 0.01  & 0.08   \\
          \multicolumn{1}{l}{Religion} & 0.17 & 0.25   \\
          \bottomrule
          \end{tabular}%
      
          \vspace{0.5cm}
          \begin{minipage}{\textwidth}
            \centering
              \footnotesize
              \textbf{Note:}  The table is from \cite{helpman2008estimating}. 
          \end{minipage}
          \label{tab:summary}
    \end{table}%

    As in \citet{chen2021nonlinear}, we estimate the following Poisson regression model
    \begin{align*}
         Y_{ij}\mid  X_{ij}, \lambda_{1, i} , \gamma_{1, j}, \lambda_{2, i}, \gamma_{2, j}   \sim  \mr{Poisson}(\exp\{\beta'  X_{ij} + \lambda_{1, i} + \gamma_{1, j} + \lambda'_{2, i}\gamma_{2, j}\}), 
    \end{align*}
    where we explicitly include additive two-way fixed effects $\lambda_{1,i}+\gamma_{1,j}$ together with an interactive fixed effects component $\lambda_{2,i}'\gamma_{2,j}$, which allows for richer patterns of unobserved heterogeneity such as homophily based on latent characteristics of countries $i$ and $j$. We compute our two-step estimator and select the regularization parameter and the number of factors using the algorithms provided in Section~\ref{sec:implementation}.\footnote{The incorporation of additive fixed effects requires only minor modifications of the optimization algorithms presented in Section~\ref{sec:implementation}. These modifications and additional implementation details are provided in Appendix~\ref{appendix:empirical_algorithm}.}

    In Table~\ref{tab:trade} below, we provide the estimates produced by our two-step estimator (column $\mr{TS}$) and the estimates reported by \citet{chen2021nonlinear} (column $\mr{CFW}$). Notice that since, using our algorithm, the estimated number of factors is $\hat{R}=2$, we include the estimates of \citet{chen2021nonlinear} for $R = 2$ only. As a baseline benchmark, we also report the estimates produced by the estimator incorporating only two-way fixed effects $\lambda_{1,i}+\gamma_{1,j}$ (column $\mr{TWFE}$), which match the analogous results in \citet{chen2021nonlinear}.

    We find that the estimates produced by our two-step estimator and the ones obtained by \citet{chen2021nonlinear} are overall very similar. The main difference is in the estimated effects of countries having the same currency, but this difference is still small relative to the associated standard errors. We also compare the (average) log-likelihood achieved by the methods (multiplied by 100), and find that the value achieved by our two-step estimator, $0.6711$, is only slightly lower than the one reported by \citet{chen2021nonlinear}, $0.6714$. Thus, we conclude that our two-step estimator can effectively recover the FE estimates in fairly large-dimensional empirical settings, making it a computationally attractive approach to estimation of nonlinear interactive fixed effects models in practice.

    \begin{table}[H]
        \centering
        \caption{Empirical Application: Gravity Equation}
    
          \begin{tabular}{lccc}
          \toprule
                & TWFE & CFW & $\mr{TS}$ \\
                & R = 0 & R = 2 & R = 2 \\
          \midrule
          \multicolumn{1}{c}{Log distance} & -0.64 & -0.71 & -0.71 \\
                & (0.07)  & (0.06)  & (0.05) \\
          \multicolumn{1}{c}{Border} & 0.71  & 0.32  & 0.32 \\
                & (0.16)  & (0.05)  & (0.06) \\
          \multicolumn{1}{c}{Legal} & 0.30  & 0.26  & 0.26 \\
                & (0.06)  & (0.04)  & (0.04) \\
          \multicolumn{1}{c}{Language} & -0.17 & -0.02 & -0.02 \\
                & (0.10)  & (0.06)  & (0.06) \\
          \multicolumn{1}{c}{Colony} & 0.36  & 0.39  & 0.39 \\
                & (0.12)  & (0.09)  & (0.10) \\
          \multicolumn{1}{c}{Currency} & 0.60  & 1.37  & 1.25 \\
                &(0.09)  & (0.41)  & (0.34) \\
          \multicolumn{1}{c}{FTA} & 0.25  & 0.17  & 0.17 \\
                & (0.13)  & (0.07)  & (0.06) \\
          \multicolumn{1}{c}{Religion} & -0.25 & 0.24  & 0.24 \\
                & (0.12)  & (0.13)  & (0.08) \\
          \midrule
          Log-Likelihood & -0.44 & 0.67  & 0.67 \\
          \bottomrule
          \end{tabular}
      
          \vspace{0.5cm}
          \begin{minipage}{\textwidth}
              \footnotesize
              \textbf{Note:}  This table reports the estimates and their standard errors (in parentheses) produced by the naive two-way fixed effects (TWFE) and the proposed two-step method ($\mr{TS}$), using the estimated number of factors $\hat R = 2$. Column $\mr{CFW}$ provides the analogous results reported in \citet{chen2021nonlinear} for $R = 2$. The last row reports the average log-likelihood achieved by the methods (multiplied by 100).%
          \end{minipage}
          \label{tab:trade}
        \end{table}%

    \section*{Declaration of Generative AI and AI-assisted Technologies in the Manuscript Preparation Process}

During the preparation of this work, the authors used ChatGPT for assistance with language editing and \LaTeX{} formatting, and Refine.ink to identify typos in the proofs. The authors reviewed and edited the output as needed and take full responsibility for the content of the published article.

    \bibliographystyle{aea}
    \bibliography{ref}

@article{fernandez2021low,
  title={Low-rank approximations of nonseparable panel models},
  author={Fern{\'a}ndez-Val, Iv{\'a}n and Freeman, Hugo and Weidner, Martin},
  journal={The Econometrics Journal},
  volume={24},
  number={2},
  pages={C40--C77},
  year={2021},
  publisher={Oxford University Press}
}

@article{moon2018nuclear,
  title={Nuclear norm regularized estimation of panel regression models},
  author={Moon, Hyungsik Roger and Weidner, Martin},
  journal={arXiv preprint arXiv:1810.10987},
  year={2018}
}

@article{negahban2012restricted,
  title={Restricted strong convexity and weighted matrix completion: Optimal bounds with noise},
  author={Negahban, Sahand and Wainwright, Martin J},
  journal={The Journal of Machine Learning Research},
  volume={13},
  number={1},
  pages={1665--1697},
  year={2012},
  publisher={JMLR. org}
}

@article{zeleneev2020identification,
  title={Identification and estimation of network models with nonparametric unobserved heterogeneity},
  author={Zeleneev, Andrei},
  journal={Department of Economics, Princeton University},
  year={2019}
}

@article{ma2022detecting,
  title={Detecting latent communities in network formation models},
  author={Ma, Shujie and Su, Liangjun and Zhang, Yichong},
  journal={The Journal of Machine Learning Research},
  volume={23},
  number={1},
  pages={13971--14031},
  year={2022},
  publisher={JMLRORG}
}

@article{chen2021nonlinear,
  title={Nonlinear factor models for network and panel data},
  author={Chen, Mingli and Fern{\'a}ndez-Val, Iv{\'a}n and Weidner, Martin},
  journal={Journal of Econometrics},
  volume={220},
  number={2},
  pages={296--324},
  year={2021},
  publisher={Elsevier}
}

@article{chen2016estimation,
  title={Estimation of nonlinear panel models with multiple unobserved effects},
  author={Chen, Mingli},
  year={2016},
  publisher={University of Warwick. Department of Economics}
}

@article{chernozhukov2023inference,
  title={Inference for low-rank models},
  author={Chernozhukov, Victor and Hansen, Christian and Liao, Yuan and Zhu, Yinchu},
  journal={The Annals of statistics},
  volume={51},
  number={3},
  pages={1309--1330},
  year={2023},
  publisher={Institute of Mathematical Statistics}
}

@techreport{chernozhukov2019inference,
  title={Inference for heterogeneous effects using low-rank estimations},
  author={Chernozhukov, Victor and Hansen, Christian Bailey and Liao, Yuan and Zhu, Yinchu},
  year={2019},
  institution={CEMMAP working paper}
}

@article{yu2015useful,
  title={A useful variant of the Davis--Kahan theorem for statisticians},
  author={Yu, Yi and Wang, Tengyao and Samworth, Richard J},
  journal={Biometrika},
  volume={102},
  number={2},
  pages={315--323},
  year={2015},
  publisher={Oxford University Press}
}

@article{fernandez2016individual,
  title={Individual and time effects in nonlinear panel models with large N, T},
  author={Fern{\'a}ndez-Val, Iv{\'a}n and Weidner, Martin},
  journal={Journal of Econometrics},
  volume={192},
  number={1},
  pages={291--312},
  year={2016},
  publisher={Elsevier}
}

@article{bandeira2016sharp,
  title={Sharp nonasymptotic bounds on the norm of random matrices with independent entries},
  author={Bandeira, Afonso S and Van Handel, Ramon},
  year={2016}
}

@article{tropp2012user,
  title={User-friendly tail bounds for sums of random matrices},
  author={Tropp, Joel A},
  journal={Foundations of computational mathematics},
  volume={12},
  pages={389--434},
  year={2012},
  publisher={Springer}
}

@article{kanaya2017convergence,
  title={Convergence rates of sums of $\alpha$-mixing triangular arrays: With an application to nonparametric drift function estimation of continuous-time processes},
  author={Kanaya, Shin},
  journal={Econometric Theory},
  volume={33},
  number={5},
  pages={1121--1153},
  year={2017},
  publisher={Cambridge University Press}
}

@article{negahban2012unified,
  title={A unified framework for high-dimensional analysis of M-estimators with decomposable regularizers},
  author={Negahban, Sahand N and Ravikumar, Pradeep and Wainwright, Martin J and Yu, Bin},
  year={2012}
}

@article{samson2000concentration,
  title={Concentration of measure inequalities for Markov chains and $\Phi$-mixing processes},
  author={Samson, Paul-Marie},
  journal={The Annals of Probability},
  volume={28},
  number={1},
  pages={416--461},
  year={2000},
  publisher={Institute of Mathematical Statistics}
}

@article{yu1994rates,
  title={Rates of convergence for empirical processes of stationary mixing sequences},
  author={Yu, Bin},
  journal={The Annals of Probability},
  pages={94--116},
  year={1994},
  publisher={JSTOR}
}

@book{ledoux2013probability,
  title={Probability in Banach Spaces: isoperimetry and processes},
  author={Ledoux, Michel and Talagrand, Michel},
  year={2013},
  publisher={Springer Science  Business Media}
}

@article{wang2022maximum,
  title={Maximum likelihood estimation and inference for high dimensional generalized factor models with application to factor-augmented regressions},
  author={Wang, Fa},
  journal={Journal of Econometrics},
  volume={229},
  number={1},
  pages={180--200},
  year={2022},
  publisher={Elsevier}
}

@article{hastie2015statistical,
  title={Statistical learning with sparsity},
  author={Hastie, Trevor and Tibshirani, Robert and Wainwright, Martin},
  journal={Monographs on statistics and applied probability},
  volume={143},
  number={143},
  pages={8},
  year={2015}
}

@article{nesterov2013gradient,
  title={Gradient methods for minimizing composite functions},
  author={Nesterov, Yu},
  journal={Mathematical programming},
  volume={140},
  number={1},
  pages={125--161},
  year={2013},
  publisher={Springer}
}

@article{ahn2013eigenvalue,
  title={Eigenvalue ratio test for the number of factors},
  author={Ahn, Seung C and Horenstein, Alex R},
  journal={Econometrica},
  volume={81},
  number={3},
  pages={1203--1227},
  year={2013},
  publisher={Wiley Online Library}
}

@article{bai2009panel,
  title={Panel data models with interactive fixed effects},
  author={Bai, Jushan},
  journal={Econometrica},
  volume={77},
  number={4},
  pages={1229--1279},
  year={2009},
  publisher={Wiley Online Library}
}

@article{truquet2023strong,
  title={Strong mixing properties of discrete-valued time series with exogenous covariates},
  author={Truquet, Lionel},
  journal={Stochastic Processes and their Applications},
  volume={160},
  pages={294--317},
  year={2023},
  publisher={Elsevier}
}

@article{de2011dynamic,
  title={Dynamic time series binary choice},
  author={De Jong, Robert M and Woutersen, Tiemen},
  journal={Econometric Theory},
  volume={27},
  number={4},
  pages={673--702},
  year={2011},
  publisher={Cambridge University Press}
}

@book{wainwright2019high,
  title={High-dimensional statistics: A non-asymptotic viewpoint},
  author={Wainwright, Martin J},
  volume={48},
  year={2019},
  publisher={Cambridge university press}
}

@book{fan2008nonlinear,
  title={Nonlinear time series: nonparametric and parametric methods},
  author={Fan, Jianqing and Yao, Qiwei},
  year={2008},
  publisher={Springer Science \& Business Media}
}

@article{helpman2008estimating,
  title={Estimating trade flows: Trading partners and trading volumes},
  author={Helpman, Elhanan and Melitz, Marc and Rubinstein, Yona},
  journal={The quarterly journal of economics},
  volume={123},
  number={2},
  pages={441--487},
  year={2008},
  publisher={MIT Press}
}

@article{rohde2011estimation,
  title={Estimation of high-dimensional low-rank matrices},
  author={Rohde, Angelika and Tsybakov, Alexandre B},
  journal={The Annals of Statistics},
  volume={39},
  number={2},
  pages={887--930},
  year={2011}
}

@article{pesaran2006estimation,
  title={Estimation and inference in large heterogeneous panels with a multifactor error structure},
  author={Pesaran, M Hashem},
  journal={Econometrica},
  volume={74},
  number={4},
  pages={967--1012},
  year={2006},
  publisher={Wiley Online Library}
}

@article{su2025estimation,
  title={Estimation and inference for unbalanced panel data models with interactive fixed effects},
  author={Su, Liangjun and Wang, Fa and Wang, Yiren},
  journal={Journal of Econometrics},
  volume={255},
  pages={106222},
  year={2026},
  publisher={Elsevier}
}

@article{armstrong2022robust,
  title={Robust estimation and inference in panels with interactive fixed effects},
  author={Armstrong, Timothy B and Weidner, Martin and Zeleneev, Andrei},
  journal={arXiv preprint arXiv:2210.06639},
  year={2022}
}

@article{BaiNg2017,
	title={Principal Components and Regularized Estimation of Factor Models},
	author={Bai, Jushan and Ng, Serena},
	journal={arXiv preprint arXiv:1708.08137},
	year={2017}
}

@article{bai2019rank,
  title={Rank regularized estimation of approximate factor models},
  author={Bai, Jushan and Ng, Serena},
  journal={Journal of Econometrics},
  volume={212},
  number={1},
  pages={78--96},
  year={2019},
  publisher={Elsevier}
}

@article{belloni2019high,
  title={High-dimensional latent panel quantile regression with an application to asset pricing},
  author={Belloni, Alexandre and Chen, Mingli and Madrid Padilla, Oscar Hernan and Wang, Zixuan},
  journal={The Annals of Statistics},
  volume={51},
  number={1},
  pages={96--121},
  year={2023},
  publisher={Institute of Mathematical Statistics}
}

@article{wang2022low,
  title={Low-rank panel quantile regression: Estimation and inference},
  author={Wang, Yiren and Su, Liangjun and Zhang, Yichong},
  journal={arXiv preprint arXiv:2210.11062},
  year={2022}
}

@article{feng_2023,
title={Nuclear Norm Regularized Quantile Regression with Interactive Fixed Effects}, DOI={10.1017/S0266466623000129}, journal={Econometric Theory}, publisher={Cambridge University Press}, author={Feng, Junlong}, year={2023}, pages={1–31}}

@article{athey_matrix_2021,
	title = {Matrix {Completion} {Methods} for {Causal} {Panel} {Data} {Models}},
	volume = {0},
	issn = {0162-1459},
	url = {https://doi.org/10.1080/01621459.2021.1891924},
	doi = {10.1080/01621459.2021.1891924},
	number = {0},
	urldate = {2021-05-17},
	journal = {Journal of the American Statistical Association},
	author = {Athey, Susan and Bayati, Mohsen and Doudchenko, Nikolay and Imbens, Guido and Khosravi, Khashayar},
	year = {2021},
	pages = {1--15},
}

@article{beyhum2019square,
  title={Square-root nuclear norm penalized estimator for panel data models with approximately low-rank unobserved heterogeneity},
  author={Beyhum, Jad and Gautier, Eric},
  journal={arXiv preprint arXiv:1904.09192},
  year={2019}
}

@article{mugnier2025simple,
  title={A simple and computationally trivial estimator for grouped fixed effects models},
  author={Mugnier, Martin},
  journal={Journal of Econometrics},
  volume={250},
  pages={106011},
  year={2025},
  publisher={Elsevier}
}

@article{alidaee2020recovering,
  title={Recovering network structure from aggregated relational data using penalized regression},
  author={Alidaee, Hossein and Auerbach, Eric and Leung, Michael P},
  journal={arXiv preprint arXiv:2001.06052},
  year={2020}
}

@article{miao2023high,
  title={High-dimensional VARs with common factors},
  author={Miao, Ke and Phillips, Peter CB and Su, Liangjun},
  journal={Journal of Econometrics},
  volume={233},
  number={1},
  pages={155--183},
  year={2023},
  publisher={Elsevier}
}

@article{miao2020panel,
  title={Panel threshold models with interactive fixed effects},
  author={Miao, Ke and Li, Kunpeng and Su, Liangjun},
  journal={Journal of Econometrics},
  volume={219},
  number={1},
  pages={137--170},
  year={2020},
  publisher={Elsevier}
}

@article{yao2025low,
  title={Low-rank estimation of nonlinear panel data models},
  author={Yao, Kan},
  journal={arXiv preprint arXiv:2511.21948},
  year={2025}
}

@article{chen2025high,
  title={High Dimensional Discrete Choice Models With Interactive Fixed Effects Applied to Causal Inference},
  author={Chen, Ye and Miao, Ke and Su, Liangjun},
  journal={Journal of Applied Econometrics},
  year={2025},
  publisher={Wiley Online Library}
}

@article{xu2026bootstrap,
  title={Bootstrap Inference in Nonlinear Panel Data Models with Interactive Fixed Effects},
  author={Xu, Haoyuan and Miao, Wei and Dhaene, Geert and Beyhum, Jad},
  journal={arXiv preprint arXiv:2604.26826},
  year={2026}
}

@article{SantosSilva2006,
author = {{Santos Silva}, J. M.C. and Tenreyro, Silvana},
doi = {10.1162/rest.88.4.641},
issn = {00346535},
journal = {Review of Economics and Statistics},
month = {nov},
number = {4},
pages = {641--658},
title = {{The log of gravity}},
volume = {88},
year = {2006}
}

\clearpage 

\appendix

\numberwithin{proposition}{section}
\numberwithin{lemma}{section}
\numberwithin{theorem}{section}
\numberwithin{corollary}{section}
\numberwithin{example}{section}
\numberwithin{table}{section}
\numberwithin{algorithm}{section}

\numberwithin{assumption}{section}

    \section{Extensions and Technical Discussions}\label{sec:extension_theory}

    This section presents extensions and technical discussions related to the assumptions in the main text: (1) We discuss the normalization of fixed effects. (2) We show how to extend our analysis to the case where $X_{it}$ contains predetermined variables. (3) We give easily verifiable sufficient conditions for restricted strong convexity (Assumption~\ref{assumption:RSC}) and diagonal structure (Assumption~\ref{assumption:block}) when predetermined covariates are present. (4)  We demonstrate that our method remains applicable even when $\Sigma_{\lambda}\Sigma_{\gamma}$ contains repeated eigenvalues (relaxing Assumption~\ref{assumption:regularity}\ref{item:strong_factors}). 

\subsection{Normalization of  fixed effects}\label{appendix:normalization}

    Similar to the linear case, we need additional $R^2$ restrictions to identify fixed effects $(\Lambda_0, \Gamma_0)$\footnote{
        For any $R$-dimensional non-singular matrix $G$, the conditional distribution of $Y_{it}$ remains invariant under the transformations $\lambda_i \mapsto \lambda_i  G'$ and $\gamma_t \mapsto \gamma_t G^{-1}$. This invariance allows us to freely choose different normalization methods for different purposes without affecting the inference of $\beta_0$.  
    }. 
    In Section \ref{sec:theory}, we directly impose normalization constraints on fixed effects $(\Lambda_0, \Gamma_0)$ such that 
    $\Lambda_0'\Lambda/N$ and $\Gamma_0'\Gamma_0/T$ are diagonal and satisfy  $\Lambda_0'\Lambda_0 /N= \Gamma_0'\Gamma_0/T$. The feasibility of such normalization can be illustrated as follows: for any $(\Lambda_0, \Gamma_0)$ satisfying Assumption~\ref{assumption:regularity}\ref{item:strong_factors}, let $D$ be the diagonal matrix containing the square roots of the  eigenvalues of the matrix $(NT)^{-1} (\Lambda_0'\Lambda_0)^{1/2}\Gamma_0'\Gamma_0 (\Lambda_0'\Lambda_0)^{1/2}$, and let $\Upsilon$ be the matrix collecting the corresponding eigenvectors.  Then there exists a  unique  transformation matrix $G= D^{1/2} \Upsilon'(\Lambda_0\Lambda_0/N)^{-1/2}$, such that the normalized nuisance parameters, 
    \begin{align*}
        \Lambda^G_0 := \Lambda_0 G', \quad \Gamma^G_0 := \Gamma_0 G^{-1}, 
    \end{align*}
    lie in the normalized parameter space $\Phi_{NT}$\footnote{
        When the eigenvalues of $\Sigma_{\lambda}\Sigma_{\gamma}$ are distinct, $G$ is unique and does not depend on the sample $\{(Y_{it}, X_{it})\}_{1\leq i\leq N, 1\leq t\leq T}$. However, with possible repeated eigenvalues, $G$ is not unique and can be identified up to an orthogonal transformation.
    }.  Unlike in the main text, in the following analysis and proofs, we distinguish between  $(\Lambda_0, \Gamma_0)$  and  $(\Lambda_0^G, \Gamma_0^G)$ to ensure a rigorous argument.

\subsection{Predetermined covariates}\label{ssec: predetermined covariates}
 
    Our estimation method is applicable in scenarios where  $X_{it}$  includes predetermined variables (e.g., lags of $Y_{it}$ in dynamic panels). While the main text focuses on the simpler case where  $X_{it}$  is strictly exogenous, this section extends the analysis to the more complex case involving predetermined variables. Considering that $X_{it}$ includes predetermined variables is crucial in panel data because the time order is very important—--unlike in network data, where node order is irrelevant. Including these variables allows us to account for dynamic effects common in empirical research. 
    
    We partition $X_{it}$ into  $X_{it} := (W'_{it}, Z'_{it})'$, where $W_{it}$ represents $d_{W}$-dimensional predetermined covariates and $Z_{it}$ represents $d_{Z}$-dimensional exogenous covariates. We collect $Z_{it, d}$ into the covariate matrix $Z_{d} \in \mb{R}^{N\times T}$ for each $d= 1,2,\ldots, d_Z$, and let $Z$ be the collection of all strictly exogenous covariate matrices, $Z= \{Z_{1}, \ldots, Z_{d_Z}\}$.  For each $i = 1,2,\ldots, N$, denote $Y_i := (Y_{i1}, Y_{i2}, \ldots, Y_{iT})'$ and $W_i := \left(W_{i1}, W_{i2}, \ldots, W_{iT}\right)'$. 
    In addition, we use $\mb{P}_{Z, \Lambda_0, \Gamma_0}(\cdot) = \mb{P}(\cdot\mid Z, \Lambda_0, \Gamma_0)$ to denote the conditional probability and use $\mb{E}_{Z, \Lambda_0, \Gamma_0}(\cdot) = \mb{E}(\cdot\mid Z, \Lambda_0, \Gamma_0)$ to denote the conditional expectation.

    Let us now consider the regularity conditions appropriate for predetermined covariates.

    \begin{assumption}[Regularity conditions - predetermined covariates]\label{assumption:regularity_pre}
        Suppose that 
        \begin{enumerate}[label=(\roman*)]
            \item \label{item:sampling_pre} \textbf{(Sampling)} Conditional on $(Z,  \Lambda_{0}, \Gamma_{0})$, $\{(Y_{i},W_i)\}_{i=1,\ldots, N}$ are independent across $i$, and for each $i$, $((Y_{i1}, W_{i1}), (Y_{i2}, W_{i2}), \ldots, (Y_{iT}, W_{iT}))$ is $\phi$-mixing with mixing coefficient $\phi_i(\tau)\rightarrow 0$ as $\tau \rightarrow \infty$, where 
                \begin{align*}
                    \phi_i(\tau) = \sup_{t} \sup_{A \in \mc{A}^{i}_{t},  B\in \mc{B}^{i}_{t + \tau}}|\mb{P}_{Z, \Lambda_0, \Gamma_0}(B\mid  A) - \mb{P}_{ Z, \Lambda_0, \Gamma_0}(B)|. 
                \end{align*} 
                Here, $\mc{A}^{i}_{t}$ is the sigma-field generated by $\{ \ldots, (Y_{i,t-1}, W_{i,t-1}), (Y_{i,t},W_{i,t})\}$, and $ \mc{B}^{i}_{t+\tau}$ is the sigma-field generated by $\{(Y_{i,t+\tau},W_{i, t+\tau}),  (Y_{i,t+\tau+1},W_{i, t+\tau+1}), \ldots\}$. For  mixing coefficients $\phi_i(\tau)$, $i=1,2,\ldots, N$, we further assume  that they exhibit a uniformly exponential decay rate:  there exists  $\zeta_0 >0$  such that $\sup_{1\leq i\leq N}\phi_i(\tau) \leq e^{-\zeta_0 \tau }$.
              
            \item \label{item:boundedness_pre} \textbf{(Boundedness)} The parameter spaces of $\beta$, $\lambda_i$, and $\gamma_t$ are bounded uniformly for all $i, t, N, T$.  In addition,  there exists a constant $\rho_X>0$ such that $\max_{d=1,\ldots, d_X}\|X_d\|_\infty <\rho_X$ for all $i, t$ and $N, T$. 
            
            \item \label{item:smoothing_pre} \textbf{(Smoothness and Convexity)} $-\ell_{it}(\cdot)$ is four times continuously differentiable and strictly convex almost surely. Furthermore, we assume that $0 < b_{\min} \leq -\ddot{\ell}_{it}( X_{it}' \beta + \lambda_i' \gamma_t)\leq b_{\max}<\infty$  almost surely  for all $\beta, \lambda_i, \gamma_t$ in the parameter space  uniformly over $i, t, N, T$. 
            
            \item \label{item:strong_factors_pre} \textbf{(Strong factors)} $\frac{1}{N} \sum_{i=1}^{N}\lambda_{0, i}\lambda_{0, i}' \cp \Sigma_{\lambda}$ and $\frac{1}{T} \sum_{t=1}^{T}\gamma_{0, t}\gamma_{0, t}'\cp \Sigma_{\gamma}$, where $ \Sigma_{\lambda} >0$ and  $\Sigma_{\gamma} >0$. In addition, the eigenvalues of $\Sigma_{\lambda}\Sigma_{\gamma}$ are distinct. 
            \item \label{item:X_generalized_nonlinearity_pre} \textbf{(Generalized non-collinearity)} For any $\Gamma\in \mb{R}^{T\times R}$, let $M_{\Lambda_0}$ and $M_{\Gamma}$ be coprojection matrices of $\Lambda_0$ and $\Gamma$ respectively.  The $d_{X}\times d_{X}$ matrix $D(\Gamma)$ with elements 
            \begin{align*}
                D(\Gamma)_{d_1, d_2} = \frac{1}{NT} \mr{Tr} \left(M_{\Lambda_0}X_{d_1} M_{\Gamma}X_{d_2}'\right), \quad d_{1}, d_{2}  = 1, \ldots, d_X, 
            \end{align*} 
            satisfies $\inf_{\Gamma \in \mb{R}^{T\times R}} \sigma_{\min} (D(\Gamma) )>0$ wpa1. 
        \end{enumerate}
    \end{assumption}

    The new regularity assumption differs from the previous one (Assumption~\ref{assumption:regularity}) only in the sampling assumption. In Assumption~\ref{assumption:regularity_pre}\ref{item:sampling_pre}, we impose weak dependence on the sequence  $(Y_{it}, W_{it})$: conditional on  $(Z, \Lambda_0, \Gamma_0)$, each sequence is $\phi$-mixing with an exponential decay rate. Although this is stricter than necessary, we adopt it to align with the sufficient conditions for restricted strong convexity (RSC). $\phi$-mixing can be replaced with $\alpha$-mixing, and the exponential decay rate can be relaxed to a sufficiently fast polynomial decay, as in \citet{fernandez2016individual}. In addition,  we do not need identical distribution or stationarity assumptions on the sequence.  
    
    We now give  the new consistency results of NNR estimators,  extending Theorem~\ref{thm:consistency} to incorporate predetermined covariates:
    \begin{theorem}\label{thm:consistency_pre}
        For any $\alpha>0$ such that  $\varphi_{NT} \geq (1+\alpha) \max\{ \|\nabla_{\beta}\mc{L}_{NT}\left(\beta_0 , \Theta_0\right)\|_{2}, \sqrt{NT}\|\nabla_{\Theta}\mc{L}_{NT}\left(\beta_0 , \Theta_0\right)\|_{\mr{op}} \}$, under Assumption~\ref{assumption:regularity_pre} and~\ref{assumption:RSC}, as $N, T\rightarrow\infty$, there exist constants  $c_1, c_2>0$ that do not depend on $N, T$ such that wpa1, 
        \begin{align*}
            \| \hat{\beta}_{\mr{nuc}} - \beta_0\|_2 & \leq c_1 \left(\varphi_{NT} + \log (NT)/\sqrt{\min\{N, T\}}\right), \\
            \frac{1}{\sqrt{NT}}\vt{\hat{\Theta}_{\mr{nuc}} - \Theta_0}_F & \leq  c_1 \left(\varphi_{NT} + \log (NT)/\sqrt{\min\{N, T\}}\right). 
        \end{align*}
        In addition, wpa1, 
        \begin{align*}      
            \frac{1}{\sqrt{N}}\|\hat{\Lambda}_{\mr{nuc}} - \Lambda^G_0 \|_{\mr{F}} & \leq c_2  \left(\varphi_{NT} + \log (NT)/\sqrt{\min\{N, T\}} \right),  \\
            \frac{1}{\sqrt{T}}\vt{\hat{\Gamma}_{\mr{nuc}} - \Gamma^G_0  }_{\mr{F}} & \leq c_2  \left(\varphi_{NT} + \log (NT)/\sqrt{\min\{N, T\}} \right). 
        \end{align*} 
    \end{theorem}
    The following corollary extends Corollary~\ref{corollary:consistency} to include predetermined covariates. 
    \begin{corollary}\label{corollary:consistency_pre}
        Under the conditions of Theorem~\ref{thm:consistency_pre},  let $\varphi_{NT} = O\left(\log (NT)/ \sqrt{  \min\{N, T\}}\right)$. Then there exist constants  $c_3, c_4>0$ that do not depend on $N, T$ such that wpa1, 
        \begin{align*}
            \| \hat{\beta}_{\mr{nuc}} - \beta_0\|_2 & \leq c_3 \log (NT)/\sqrt{\min\{N, T\}},   \\
            \frac{1}{\sqrt{NT}}\vt{\hat{\Theta}_{\mr{nuc}} - \Theta_0}_{\mr{F}} & \leq   c_3 \log (NT)/\sqrt{\min\{N, T\}}.
        \end{align*}
        In addition,  wpa1, 
        \begin{align*}      
            \frac{1}{\sqrt{N}}\|\hat{\Lambda}_{\mr{nuc}} - \Lambda^G_0\|_{\mr{F}} & \leq c_4 \log (NT)/\sqrt{\min\{N, T\}},   \\
            \frac{1}{\sqrt{T}}\| \hat{\Gamma}_{\mr{nuc}} - \Gamma^G_0 \|_{\mr{F}} & \leq  c_4 \log (NT)/\sqrt{\min\{N, T\}}. 
        \end{align*}  
    \end{corollary}

    \bigskip 
    
    We state the new diagonal structure assumption which generalizes Assumption~\ref{assumption:block}. 
    \begin{assumption}[Diagonal structure]\label{assumption:block_pre}
        The population Hessian in~\eqref{eq:local_estimators} evaluated at the   true parameters, $\mb{E}_{Z, \Lambda_0, \Gamma_0}\mc{H}_{NT}(\beta_0, \Lambda_0, \Gamma_0)$, admits a diagonal  structure,  i.e.,  there exists a constant $C>0$ (does not depend on $N, T$) such that 
        \begin{align*}
            \mb{E}_{Z, \Lambda_0, \Gamma_0}\mc{H}_{NT}(\beta_0, \Lambda_0^G, \Gamma_0^G)  \geq C   \mr{diag}\left\{\mb{I}_{d_X}, \frac{1}{N}\mb{I}_{NR}, \frac{1}{T}\mb{I}_{TR}\right\}. 
        \end{align*}
    \end{assumption}
    The local convexity and the asymptotic equivalence can be extended to incorporate predetermined covariates.
    \begin{theorem}\label{thm:convexity_strong_pre}
        Under Assumption~\ref{assumption:block_pre} and the conditions in  Corollary~\ref{corollary:consistency_pre}, suppose $N, T$ have the same order. Let $ \delta_{NT}= \log (NT) \min\{N^{-3/8}, T^{-3/8}\} $, and let $\mc{B}_{\delta_{NT}}$ be the neighborhood defined in ~\eqref{eq:neighborhood}. The following results hold: 
        \begin{enumerate}[label=(\roman*)]
            \item \label{item:1_pre} the local optimization problem~\eqref{eq:local_estimators} is strictly  convex wpa1;
            \item \label{item:2_pre} our two-step estimator is asymptotically equivalent to the FE estimator; 
            \item \label{item:3_pre} the local optimization problem~\eqref{eq:local_estimators} is strongly  convex wpa1 uniformly on $\mc{B}_{\delta_{NT}}$, i.e., there exists a constant $c_5>0$ independent of $N, T$ such that for any $(\beta, \Lambda, \Gamma) \in \mc{B}_{\delta_{NT}}$,  
            \begin{align*}
                \mc{H}_{NT}(\beta, \Lambda, \Gamma) > c_5 \mr{diag}\left\{\mb{I}_{d_X}, \frac{1}{N}\mb{I}_{NR}, \frac{1}{T}\mb{I}_{TR}\right\}, \quad \text{wpa1}. 
            \end{align*}
        \end{enumerate}
    \end{theorem}

\subsection{Further discussion on RSC}\label{appendix:RSC}

    We aim to provide a set of easily verifiable sufficient conditions for the restricted strong convexity (RSC) condition in the panel data setting. Given the critical role of the RSC condition in low-rank estimation, establishing verifiable low-level conditions in the panel data setting is important and may be of independent interest. While \citet{chernozhukov2019inference} offers sufficient conditions for verifying the RSC condition in the panel data context, their approach is limited to strictly exogenous covariates. We extend their proof strategy to accommodate predetermined covariates, thereby broadening the applicability of the RSC condition in panel data. It is worth noting that although we are focusing on low-rank estimation with homogeneous slopes, our proof strategy is quite flexible and is applicable to more general settings. For example, it can be readily adapted to establish the RSC condition under heterogeneous slopes, as in the models considered by \citet{chernozhukov2019inference} and \citet{ma2022detecting}.

    \begin{assumption}\label{assumption:conditional_independence_RSC}
        There exists a sequence of random vectors $\{v_{it}\}_{ 1\leqslant i \leqslant N,  1\leqslant t \leqslant T}$, such that
        \begin{enumerate}[label=(\roman*)]
            \item \label{item:conditional_weak_dependence_RSC} \textbf{(Conditional weak dependence)}  Conditional on $\mc{V}: =\left\{v_{it} \mid 1\leq i\leq N, 1\leq t\leq T\right\}$, $\{X_{i}\}_{1\leq i\leq N}$ are independent across $i$, and for each $i$, $(X_{i1}, X_{i2}, \ldots, X_{iT})$ is $\phi$-mixing with mixing coefficient $\phi_i(\tau)\rightarrow 0$ as $\tau \rightarrow \infty$.  We assume   that  the mixing coefficients exhibit  a uniformly exponential decay rate, i.e., there exists $\zeta_1 >0$  such that $\sup_{1\leq i\leq N}\phi_i(\tau) \leq e^{-\zeta_1 \tau }$.
            \item \label{item:conditional_variability_RSC} \textbf{(Conditional variability)} Conditional on $\mc{V}$, there exists $\kappa_0 >0$ such that the following inequality holds for any $ N, T$, 
            \begin{align*}
                \inf_{1\leq i\leq N, 1\leq t \leq T}\sigma_{\min}\left(
                    \begin{pmatrix}
                        \mb{E}(X_{it}X_{it}' \mid \mc{V}) & \mb{E} (X_{it}\mid \mc{V}) \\
                        \mb{E}(X_{it}'\mid \mc{V}) & 1
                    \end{pmatrix}
                \right) \geq \kappa_0, 
            \end{align*}
            where $\sigma_{\min}(\cdot)$ denotes the minimum eigenvalue of a matrix. 
        \end{enumerate}
    \end{assumption}

    The first part of Assumption~\ref{assumption:conditional_independence_RSC} states that while it is unrealistic to assume weak dependence of  $\{X_{it}\}_{1\leq t\leq T}$  across  $t$  due to the presence of common factors,  we assume that  the serial correlation can be significantly reduced by conditioning on a set of latent factors, $\mathcal{V}$. The second part of Assumption~\ref{assumption:conditional_independence_RSC} makes sure that $X_{it}$ cannot be fully explained by $\mc{V}$.  
    
    More specifically, Assumption~\ref{assumption:conditional_independence_RSC}\ref{item:conditional_weak_dependence_RSC} imposes weak dependence on the sequence   $\{X_{it}\}_{1\leq t\leq T}$: conditional on the common factors $\mc{V}$, the sequence is $\phi$-mixing with an exponential decay rate. The conditional weak dependence assumption is much milder than it appears, and here are some examples: 
    \begin{example}[Factor model]
        Suppose $X_{it} = \lambda_x'\gamma_t + \epsilon_{it}$ where $\lambda_x'\gamma_t$ is the common factor and $\epsilon_{it}$ is the error term. Let $v_{it} = \lambda_x'\gamma_t$. The conditional weak dependence condition holds if $\{\epsilon_{it}\}_{1\leq i\leq N, 1\leq t\leq T}$ is independent across both $i$ and $t$ conditional on $\mc{V}$. 
    \end{example}
    \begin{example}[Autoregression]
        Allowing for serial correlation becomes necessary when the model includes predetermined covariates. Suppose $X_{it} = A X_{i,t-1} + B v_{it} + \epsilon_{it} $ 
        where $A$ and $B$ are the coefficient matrices and $\epsilon_{it}$ is an innovation term with uniform bounded support. We further assume that $\psi_{\max}(A) < 1$, where $\psi_{\max}(A)$ denotes the largest singular value of $A$.  Then conditional on $\mc{V}$, the sequence $\{X_{it}\}_{1\leq t\leq T}$ is $\phi$-mixing with exponential decay rate. 
    \end{example}
    \begin{example}[Non-separable model]
        A more interesting case is the non-separable model with serial correlation, i.e.,  $X_{it} = h(X_{i,t-1}, v_{it})$. For example, consider the dynamic Logit panel, where $X_{it} = Y_{i, t-1}$, and 
        \begin{align*}
            Y_{it} = 1(\beta Y_{i, t-1} + \lambda_{0, i}'\gamma_{0, t} + \epsilon_{it} > 0). 
        \end{align*}
        Here, $\epsilon_{it}$ follows a standard logistic  distribution. Let $v_{it} = \lambda_{0, i}'\gamma_{0, t} $, and assume that the innovation terms $\{\epsilon_{it}\}_{1\leq i\leq N, 1\leq t\leq T}$ are independent across both $i$ and   $t$ and are strictly exogenous.  By Assumption~\ref{assumption:regularity_pre}\ref{item:smoothing_pre}, we conclude that,  conditional on $\mc{V}$, $\{X_{it}\}_{1\leq t\leq T}$ is a time-inhomogeneous Markov process with transition matrices whose entries are uniformly bounded away from zero. Furthermore, by verifying Doeblin's condition, it follows that $\{X_{it}\}_{1\leq t\leq T}$ is $\phi$-mixing with exponential decay rate. We recommend the interested reader in nonlinear time series models to refer to \citet{de2011dynamic} and \citet{truquet2023strong} for more details. 
    \end{example}

    Assumption~\ref{assumption:conditional_independence_RSC}\ref{item:conditional_variability_RSC} imposes  that,  even after controlling for the set of latent factors $\mathcal{V}$,   $X_{it}$  preserves sufficient variation. Notably, by applying the Schur decomposition, we directly obtain that
    \begin{align*}
        \begin{pmatrix}
            \mb{E}(X_{it}X_{it}' \mid \mc{V}) & \mb{E} (X_{it}\mid \mc{V}) \\
            \mb{E}(X_{it}'\mid \mc{V}) & 1
        \end{pmatrix} > 0 \quad \Longleftrightarrow \quad \mr{Var}(X_{it}\mid \mc{V})>0. 
    \end{align*}
    If  $X_{it}$ given $\mc{V}$ is identically distributed, the conditional variability condition simplifies to $\mr{Var}(X_{it}\mid \mc{V})>0$. If we further assume   that $X_{it}$ admits an additive structure, this condition is satisfied if $\inf_{1\leq i\leq N, 1\leq t\leq T}\mr{Var}(\epsilon_{it}\mid \mc{V})>0$. Assumption~\ref{assumption:conditional_independence_RSC}\ref{item:conditional_variability_RSC} can be viewed as  a generalization of these simpler cases.

    \begingroup
        \def\thelemma{\ref{lemma:sufficient_RSC}}
        \begin{lemma}
            Under Assumptions~\ref{assumption:regularity_pre} and \ref{assumption:conditional_independence_RSC} provided in Appendix~\ref{sec:extension_theory}, Assumption~\ref{assumption:RSC} is satisfied. %
        \end{lemma}
        \addtocounter{lemma}{-1}
    \endgroup

\subsection{Further discussion on diagonal structure}\label{appendix:diagonal}

    We will provide sufficient conditions for verifying Assumption~\ref{assumption:block} and Assumption~\ref{assumption:block_pre}, which are critical in establishing local convexity in Theorem~\ref{thm:convexity_strong} and Theorem~\ref{thm:convexity_strong_pre}. 
    Assumption~\ref{assumption:block} and Assumption~\ref{assumption:block_pre} require that the population Hessian, evaluated at the normalized true parameters $(\beta_0, \Lambda_0^G, \Gamma^G_0)$, is positive definite, with a minimum eigenvalue on the order of $(\max\{N, T\})^{-1}$. We will show that these two assumptions are weak and that,  under the strong factor condition, it suffices for $X_{it}$ to exhibit weak dependence after accounting for the influence of common factors and to retain components that cannot be fully explained by these factors. 

    \begin{assumption}\label{assumption:conditional_independence_hessian}
        There exists a sequence of random vectors $\{u_{it}\}_{ 1\leqslant i\leqslant N, 1\leqslant t\leqslant T}$, such that
        \begin{enumerate}[label=(\roman*)]
            \item \label{item:conditional_weak_dependence_hessian} \textbf{(Conditional weak dependence)}  Conditional on $\mc{U}: =\left\{u_{it} \mid 1\leq i\leq N, 1\leq t\leq T\right\}$, $\{X_{i}\}_{1\leq i\leq N}$ is independent across $i$, and for each $i$, $(X_{i1}, X_{i2}, \ldots, X_{iT})$ is $\phi$-mixing with mixing coefficient $\phi_i(\tau)\rightarrow 0$ as $\tau \rightarrow \infty$. We assume   that the mixing coefficients exhibit a uniformly exponential  decay rate, i.e., there exists $\zeta_2 >0$  such that $\sup_{1\leq i\leq N}\phi_i(\tau) \leq e^{-\zeta_2 \tau }$.
            \item \label{item:conditional_variability_hessian} \textbf{(Conditional variability)} Conditional on $\mc{U}$,  there exists $0 <\nu <1$ such that the following inequality holds  uniformly for all $1\leq i\leq N$ and $ 1\leq t\leq T$, 
            \begin{align*}
                \mb{E}(\ddot{\ell}_{it}^0X_{it}\mid \mc{U})\mb{E}(\ddot{\ell}_{it}^0X_{it}'\mid \mc{U}) \leq \nu \mb{E} (\ddot{\ell}_{it}^0 \mid \mc{U} )\mb{E} (\ddot{\ell}_{it}^0X_{it}X_{it}'\mid \mc{U}),
            \end{align*}
        where $\ddot{\ell}_{it}^0 =\ddot{\ell}_{it}(X_{it}'\beta_0 + \lambda_{0, i}'\gamma_{0, t})$. 
        \end{enumerate}
    \end{assumption}

    The first part of this assumption is identical to the conditional weak dependence stated in Assumption~\ref{assumption:conditional_independence_hessian}\ref{item:conditional_weak_dependence_hessian}. It is worth noting that the condition requiring $\{X_{it}\}_{1 \leq t \leq T}$ to be $\phi$-mixing with an exponential decay rate uniformly across $i$ is stronger than necessary and is imposed here for consistency with Assumption~\ref{assumption:conditional_independence_RSC}\ref{item:conditional_weak_dependence_RSC}. In fact, this condition could be relaxed to $\{X_{it}\}_{1 \leq t \leq T}$ being $\alpha$-mixing with a polynomial decay rate, still uniformly across $i$.   Assumption~\ref{assumption:conditional_independence_hessian}\ref{item:conditional_variability_hessian}  imposes a weak additional restriction on the model. To elaborate, let us use $\mb{E}_{\mc{U}}(\cdot)$ to denote  the conditional expectation $\mb{E}(\cdot \mid \mc{U})$, and first consider the  simplest case where $\ddot{\ell}_{it}^0$ is a constant (corresponding to a linear panel model) with $d_X = 1$. In this case, it is straightforward to verify that Assumption~\ref{assumption:conditional_independence_hessian}\ref{item:conditional_variability_hessian} is equivalent to $(\mb{E}_{\mc{U}}{X}_{it})^2 \leq \nu \mb{E}_{\mc{U}}({X}^2_{it})$ uniformly for all $i, t, N, T$. When extending the analysis to the nonlinear model with $d_X = 1$, it follows (using Hölder's inequality) that  
    \begin{align*}
        \mb{E}_{\mc{U}}(\ddot{\ell}_{it}^0{X}_{it}) \leq  \mb{E}_{\mc{U}}(|\ddot{\ell}_{it}^0{X}_{it}|)  =  \mb{E}_{\mc{U}}(|\ddot{\ell}_{it}^0|^{1/2} |\ddot{\ell}_{it}^0{X}^2_{it}|^{1/2})   \leq  \left(\mb{E}_{\mc{U}}(-\ddot{\ell}_{it}^0)\right)^{1/2}  \left(\mb{E}_{\mc{U}}(-\ddot{\ell}_{it}^0{X}^2_{it})\right)^{1/2}. 
    \end{align*} 
    Equality holds if and only if $X_{it}$ is  constant given $\mc{U}$, i.e., $X_{it}$ can be fully explained by $\mc{U}$.  Thus, in this case, Assumption~\ref{assumption:conditional_independence_hessian}\ref{item:conditional_variability_hessian} can be viewed as a uniform version of Hölder's inequality.  Extending the  analysis from $d_X = 1$ to $d_X > 1$ is straightforward and yields the same  conclusions.

    \begin{lemma}[Sufficient conditions for diagonal structure]\label{lemma:sufficient_convexity}
        Under Assumption~\ref{assumption:regularity_pre} and Assumption~\ref{assumption:conditional_independence_hessian}, Assumption~\ref{assumption:block} and Assumption~\ref{assumption:block_pre} hold.
    \end{lemma}

\subsection{Relaxing distinct eigenvalues assumption}\label{ssec: distinct eigenvalues}

    We now demonstrate that our method remains applicable even when $\Sigma_{\lambda}\Sigma_{\gamma}$ contains repeated eigenvalues (relaxing Assumption~\ref{assumption:regularity}\ref{item:strong_factors} and Assumption~\ref{assumption:regularity_pre}\ref{item:strong_factors_pre}). When $\Sigma_{\lambda}\Sigma_{\gamma}$ has distinct eigenvalues, there exists a unique $R$-dimensional invertible matrix, $G= D^{1/2} \Upsilon'(\Lambda_0'\Lambda_0/N)^{-1/2}$, which does not depend on the sample $\{(Y_{it}, X_{it})\}_{1\leq i\leq N, 1\leq t\leq T}$, such that $\Lambda^{G\prime}_0\Lambda^G_0/N, \Gamma^{G\prime}_0\Gamma^G_0/T$ are diagonal and satisfy $\Lambda^{G\prime}_0\Lambda^G_0/N =  \Gamma^{G\prime}_0\Gamma^G_0/T$. However, when $\Sigma_{\lambda}\Sigma_{\gamma}$ has repeated eigenvalues, the transformation $G$ may depend on the sample. In this case, we denote it by $\hat{G}$: $\hat{G} = D^{1/2}\hat{O}\Upsilon'(\Lambda_0'\Lambda_0 /N)^{-1/2}$ (see the proof of Theorem~\ref{thm:consistency_pre}),  where $\hat{O}$ is an orthogonal matrix that may also depend on the sample. 

    The dependence of  $\hat{G}$  on the sample adds extra complexity to the proof. Therefore, we aim to construct an objective function that is independent of $G$, along with its corresponding Hessian, which allows us to directly apply the proof of Theorem~\ref{thm:convexity_strong_pre} to establish the asymptotic equivalence between our estimators and the FE estimators in \citet{chen2021nonlinear}. 
    
    Define 
    \begin{align*}
        G_{NT} := \mr{diag}\left\{\mb{I}_{d_X}, \underbrace{\hat{G}, \ldots, \hat{G}}_{N}, \underbrace{\hat{G}^{-1}, \ldots, \hat{G}^{-1}}_{T}\right\}. 
    \end{align*}
    It is straightforward to verify that (i) $\hat{G}$ is invertible wpa1,  (ii) the minimum eigenvalue of $\hat{G}$ is strictly positive wpa1 and does not depend on the sample, and (iii) the maximum eigenvalue of $\hat{G}$ is uniformly bounded wpa1 (see the proof of Theorem~\ref{thm:consistency_pre}) and does not depend on the sample. Therefore, it follows that $G_{NT}$ is invertible wpa1,  the minimum eigenvalue of $G_{NT}$ is strictly positive wpa1 and independent of the sample, and the maximum eigenvalue of $G_{NT}$ is uniformly bounded wpa1. In addition, one can verify that neither
    \begin{align*}
        \nabla^2 \mc{L}_{NT}(\beta_0, \Lambda_0, \Gamma_0)  = G_{NT}' \nabla^2 \mc{L}_{NT}(\beta_0, \Lambda_0^G, \Gamma_0^G) G_{NT} 
    \end{align*}
    or 
    $\nabla^2 \|\hat{\Lambda}_{\mr{nuc}}'\Lambda - \Gamma'\hat{\Gamma}_{\mr{nuc}} \|_{\mr{F}}^2 $ do not depend on $G_{NT}$.  
    Let $a := \min\left\{1, \sigma_{\min}(G_{NT}), 1/\sigma_{\max}(G_{NT})\right\}$. Then, 
    \begin{align*}
        \mc{H}_{NT}(\beta_0, \Lambda_0^G, \Gamma_0^G) & = G_{NT}^{-1\prime} \nabla^2 \mc{L}_{NT}(\beta_0, \Lambda_0, \Gamma_0)G_{NT}^{-1} + \frac{1}{2}G_{NT}^{\prime -1} G_{NT}^{\prime} \nabla^2\|\hat{\Lambda}_{\mr{nuc}}'\Lambda - \Gamma'\hat{\Gamma}_{\mr{nuc}} \|_{\mr{F}}^2 G_{NT} G_{NT}^{-1}  \\
        & = G_{NT}^{-1\prime} \left(\nabla^2 \mc{L}_{NT}(\beta_0, \Lambda_0, \Gamma_0) + \frac{1}{2} G_{NT}^{\prime} \nabla^2 \|\hat{\Lambda}_{\mr{nuc}}'\Lambda - \Gamma'\hat{\Gamma}_{\mr{nuc}}  \|_{\mr{F}}^2  G_{NT}\right) G_{NT}^{-1} \\
        & \geq G_{NT}^{-1\prime} \underbrace{\left(\nabla^2 \mc{L}_{NT}(\beta_0, \Lambda_0, \Gamma_0) + \frac{a^2}{2} \nabla^2 \|\hat{\Lambda}_{\mr{nuc}}'\Lambda - \Gamma'\hat{\Gamma}_{\mr{nuc}}  \|_{\mr{F}}^2  \right)}_{\text{does not depend on the sample}} G_{NT}^{-1}. 
    \end{align*}
    Therefore, $\mb{E}_{Z, \Lambda_0, \Gamma_0}\mc{H}_{NT}(\beta_0, \Lambda_0^G, \Gamma_0^G)$ satisfies Assumption~\ref{assumption:block_pre} (or Assumption~\ref{assumption:block}) if and only if 
    \begin{align*}
        \mb{E}_{Z, \Lambda_0, \Gamma_0}\left(\nabla^2 \mc{L}_{NT}(\beta_0, \Lambda_0, \Gamma_0) + \frac{a^2}{2} \nabla^2 \|\hat{\Lambda}_{\mr{nuc}}'\Lambda - \Gamma'\hat{\Gamma}_{\mr{nuc}} \|_{\mr{F}}^2 \right)
    \end{align*}
    also satisfies Assumption~\ref{assumption:block_pre} (or Assumption~\ref{assumption:block}). Using the same technique as in the proof of Lemma~\ref{lemma:sufficient_convexity},  Assumption~\ref{assumption:regularity_pre} together with Assumption~\ref{assumption:conditional_independence_hessian} is sufficient to establish  that $$\mb{E}_{Z, \Lambda_0, \Gamma_0}\left(\nabla^2 \mc{L}_{NT}(\beta_0, \Lambda_0, \Gamma_0) + \frac{a^2}{2} \nabla^2 \|\hat{\Lambda}_{\mr{nuc}}'\Lambda - \Gamma'\hat{\Gamma}_{\mr{nuc}} \|_{\mr{F}}^2\right)$$ has a diagonal structure. Therefore, by replacing $(\Lambda^G_0, \Gamma_0^G)$ with $(\Lambda_0, \Gamma_0)$ in the proof of Theorem~\ref{thm:convexity_strong_pre},  we establish that Theorem~\ref{thm:convexity_strong_pre} still holds even in the presence of repeated eigenvalues.

    \section{Discussions on Algorithm and Additional Simulations}\label{sec:extension_numerics}

In this section, (i) we characterize the two-step estimator and present the corresponding algorithm used in the empirical application in Section~\ref{sec:MC}; (ii) we provide procedures for the analytical bias correction and the split-panel Jackknife bias correction; (iii) we present Monte Carlo simulation results for nonlinear panel models with predetermined covariates.

\subsection{Implementation details for the empirical application}\label{appendix:empirical_algorithm}

    We provide the optimization problem and the associated algorithm used in the empirical application in Section~\ref{sec:MC}, where additive two-way fixed effects are explicitly incorporated into the index. Recall that we aim to estimate the following Poisson model:
    \begin{align*}
         Y_{ij}\mid  X_{ij}, \lambda_{1, i} , \gamma_{1, j}, \lambda_{2, i}, \gamma_{2, j}  \sim  \mr{Poisson}(\exp\{\beta'  X_{ij} + \lambda_{1, i} + \gamma_{1, j} + \lambda'_{2, i}\gamma_{2, j}\}). 
    \end{align*}
    For notational simplicity, define
    \begin{align*}
        \mc{L}_{N}\left(\beta, \lambda_1, \gamma_1, \Theta\right)& : = \frac{-1}{N(N-1)}\sum_{\substack{ i, j = 1, \ldots, N \\ i\neq j} } \log f(Y_{ij} \mid \beta'X_{ij} + \lambda_{1, i} + \gamma_{1, j} + \theta_{ij}),  \\
        \mc{L}_{N}\left(\beta, \lambda_1, \gamma_1, \Lambda_2, \Gamma_2 \right)& :  = \frac{-1}{N(N-1)}\sum_{\substack{ i, j = 1, \ldots, N \\ i\neq j} } \log f(Y_{ij} \mid \beta'X_{ij} + \lambda_{1, i} + \gamma_{1, j} + \lambda_{2, i}' \gamma_{2, j}). 
    \end{align*}
    \paragraph{NNR estimator} In the first step, the NNR estimator $(\hat{\beta}_{\mr{nuc}},\hat{\lambda}_{\mr{1, nuc}}, \hat{\gamma}_{\mr{1, nuc}}, \hat{\Theta}_{\mr{nuc}})$ solves 
    \begin{equation*}
    \begin{aligned}
        \left(\hat{\beta}_{\mr{nuc}},\hat{\lambda}_{\mr{1, nuc}}, \hat{\gamma}_{\mr{1, nuc}}, \hat{\Theta}_{\mr{nuc}}\right) = \argmin_{\substack{
        \beta \in \mb{R}^{d_X}, \Theta\in \mb{R}^{N\times N} \\
        \lambda_1 \in \mb{R}^{N}, \gamma_1 \in \mb{R}^{N} }}  \left\{ \mc{L}_{N}\left(\beta, \lambda_1, \gamma_1, \Theta\right)  + \frac{\varphi_{N}}{\sqrt{N(N-1)}} \|\Theta\|_{\mr{nuc}}\right\}, 
    \end{aligned}
    \end{equation*}
    where $\|\Theta\|_{\mr{nuc}}$ denotes the nuclear norm of matrix $\Theta$, and $\varphi_{N} > 0$ is a regularization parameter. The nuclear norm regularized estimators $(\hat{\Lambda}_{\mr{2, nuc}}, \hat{\Gamma}_{\mr{2, nuc}})$ are obtained through the singular value decomposition of $\hat{\Theta}_{\mr{nuc}}$ as in~\eqref{eq:space_definition}. 
    \paragraph{Local estimator} The local estimator $(\hat{\beta}_{\mr{local}},\hat{\lambda}_{\mr{1, local}}, \hat{\gamma}_{\mr{1, local}}, \hat{\Lambda}_{\mr{2, local}}, \hat{\Gamma}_{\mr{2, local}})$ solves 
    \begin{equation*}
    \begin{aligned}
         (\hat{\beta}_{\mr{local}},\hat{\lambda}_{\mr{1, local}}, \hat{\gamma}_{\mr{1, local}}, \hat{\Lambda}_{\mr{2, local}}, \hat{\Gamma}_{\mr{2, local}} ) = \argmin_{\substack{
        \beta \in \mb{R}^{d_X}, \lambda_1 \in \mb{R}^{N}, \gamma_1 \in \mb{R}^{N} , \\
        \Lambda_2 \in \mb{R}^{N\times R}, \Gamma_2 \in \mb{R}^{N\times R}
        }}  \mc{L}_{N}\left(\beta, \lambda_1, \gamma_1, \Lambda_2, \Gamma_2 \right)
    \end{aligned}
    \end{equation*}
    using the NNR estimator as the initial value. 

    \bigskip 
    We now provide the algorithm for obtaining the NNR estimators with additive fixed effects as follows. 
    \begin{algorithm}[Proximal gradient descent with additive fixed effects]
        Compute the nuclear norm regularized estimator as follows:
        \begin{itemize}[label={}, leftmargin=2cm]
            \item[Step 1:] Fix the step sizes $(s_{\beta}, s_{\lambda, \gamma}, s_{\theta})$. Initialize $\beta^{(0)}$, $\lambda^{(0)}_{1}$, $\gamma_{1}^{(0)}$, and $\Theta^{(0)}$. Set $k = 0$. 
            \item[Step 2:] Let 
            \begin{equation*} 
                \begin{aligned}
                    \beta^{(k+1)} & = \beta^{(k)} - s_{\beta} \nabla_{\beta}\mc{L}_{N}\left(\beta^{(k)}, \lambda^{(k)}_{1}, \gamma_{1}^{(k)}, \Theta^{(k)}\right), \\
                    \lambda^{(k+1)}_{1} & = \lambda^{(k)}_{1} - s_{\lambda, \gamma} \nabla_{\lambda_1}\mc{L}_{N}\left(\beta^{(k)}, \lambda^{(k)}_{1}, \gamma_{1}^{(k)}, \Theta^{(k)}\right), \\
                    \gamma^{(k+1)}_{1} & = \gamma^{(k)}_{1} - s_{\lambda, \gamma} \nabla_{\gamma_1}\mc{L}_{N}\left(\beta^{(k)}, \lambda^{(k)}_{1}, \gamma_{1}^{(k)}, \Theta^{(k)}\right), \\
                    \Theta^{(k+1)} & = \mc{S}^*_{s_{\theta}\frac{\varphi_{N}}{\sqrt{N(N-1)}}}\left(\Theta^{(k)} - s_{\theta}\nabla_{\Theta}\mc{L}_{N}\left(\beta^{(k)}, \lambda^{(k)}_{1}, \gamma_{1}^{(k)}, \Theta^{(k)}\right)\right), 
               \end{aligned}
            \end{equation*}
            and set $k = k+1$. 
            \item[Step 3:] Repeat Step 2 until convergence. 
        \end{itemize}
    \end{algorithm}
    Regarding the choice of step sizes,  we recommend starting each iteration with step sizes $s_{\beta} = 1$, $s_{\lambda, \gamma} = N$, and  $s_{\theta} = N^2$. If the objective function increases, the step sizes are iteratively halved until a decrease in the objective function is achieved.
    \begin{algorithm}[Gradient descent] 
        Compute the local estimator as follows:
        \begin{itemize}[label={}, leftmargin=2cm]
            \item[Step 1:] Fix the step sizes $(s_{\beta}, s_{\lambda, \gamma}, s_{\lambda_2}, s_{\gamma_2})$. Initialize  $\beta^{(0)} = \hat{\beta}_{\mr{nuc}}$, $\lambda^{(0)}_{1} = \hat{\lambda}_{1, \mr{nuc}}$, $\gamma^{(0)}_{1} = \hat{\gamma}_{1, \mr{nuc}}$,  $\Lambda^{(0)}_2 = \hat{\Lambda}_{\mr{nuc}, 2}$, and $\Gamma^{(0)}_2 = \hat{\Gamma}_{\mr{nuc}, 2}$.  Set $k = 0$. 
            \item[Step 2:] Let 
            \begin{equation*}
                \begin{aligned}
                    \beta^{(k+1)} & = \beta^{(k)} - s_{\beta} \nabla_{\beta}\mc{L}_{N}(\beta^{(k)}, \lambda^{(k)}_{1}, \gamma_{1}^{(k)}, \Lambda^{(k)}_2, \Gamma^{(k)}_2 ), \\
                    \lambda^{(k+1)}_1 & = \lambda^{(k)}_1 - s_{\lambda, \gamma}\nabla_{\lambda_1 }\mc{L}_{N}(\beta^{(k)}, \lambda^{(k)}_{1}, \gamma_{1}^{(k)}, \Lambda^{(k)}_2, \Gamma^{(k)}_2 ), \\
                    \gamma^{(k+1)}_1 & = \gamma^{(k)}_1 - s_{\lambda, \gamma}\nabla_{\gamma_1}\mc{L}_{N}(\beta^{(k)}, \lambda^{(k)}_{1}, \gamma_{1}^{(k)}, \Lambda^{(k)}_2, \Gamma^{(k)}_2 ), \\
                    \Lambda^{(k+1)}_2 & = \Lambda^{(k)}_2 - s_{\lambda_2}\nabla_{\lambda_2}\mc{L}_{N}(\beta^{(k)}, \lambda^{(k)}_{1}, \gamma_{1}^{(k)}, \Lambda^{(k)}_2, \Gamma^{(k)}_2 ), \\
                    \Gamma^{(k+1)}_2 & = \Gamma^{(k)}_2 - s_{\gamma_2}\nabla_{\gamma_2}\mc{L}_{N}(\beta^{(k)}, \lambda^{(k)}_{1}, \gamma_{1}^{(k)}, \Lambda^{(k)}_2, \Gamma^{(k)}_2 ), 
               \end{aligned}
            \end{equation*}
            and set $k = k+1$.
            \item[Step 3:] (Optional) Normalize $\Lambda^{(k+1)}_2$ and $\Gamma^{(k+1)}_2$, for example, let $\frac{1}{N}\Lambda^{(k+1)\prime}_2\Lambda^{(k+1)}_2 = \frac{1}{N}\Gamma^{(k+1)\prime}_2\Gamma^{(k+1)}_2$ and diagonal. 
            \item[Step 4:] Repeat Step 2 and Step 3 until convergence. 
        \end{itemize}
    \end{algorithm}
    In practice, in each step, we recommend  starting with $s_{\beta} = 1$, $s_{\lambda, \gamma} = s_{\lambda_2} = s_{\gamma_2} = N$. If the objective function increases, the step sizes are iteratively halved until a decrease in the objective function is achieved.

\subsection{Bias correction}\label{appendix:bias_correction}

    In this section, we provide implementation details for the analytical bias correction and the split-panel jackknife bias correction.
   
    \paragraph{Analytical bias correction}
    We follow \citet{chen2021nonlinear} to perform analytical bias correction. For any $d = 1,2,\ldots, d_X$, let 
    \begin{align*}
        (\Lambda^*_{d}, \Gamma^*_{d}) \in \argmin_{\Lambda_{d}\in \mb{R}^{N\times R}, \Gamma_{d}\in \mb{R}^{T\times R}}\sum_{i=1}^{N}\sum_{t=1}^{T}\mb{E}(-\ddot{\ell}^0_{it})\left(\frac{\mb{E}(\ddot{\ell}^0_{it}X_{it, d})}{\mb{E}(\ddot{\ell}^0_{it})} - \lambda_{d, i}' \gamma_{0, t} - \lambda_{0, i}' \gamma_{d, t}\right)^2, 
    \end{align*}
    and  define 
    \begin{align*}
        \Xi_{d, it} := \lambda_{d, i}^{*\prime} \gamma_{0, t} + \lambda_{0, i}' \gamma^*_{d, t}, \quad \tilde{X}_{d, it} := X_{d, it} - \Xi_{d, it}. 
    \end{align*}
    Here, $\ddot{\ell}_{it}^0: = \partial^2 \log f(\beta_0'X_{it} + \lambda_0'\gamma_0 )/\partial Y^{*2} $ denotes the second derivative of the log-likelihood $\ell_{it}(\cdot)$ with respect to the index, evaluated at the true parameter values. In addition, let $\hat{\Xi}_{it}$ denote the sample analogue of $\Xi_{it}$,  and define 
    \begin{align*}
        \widehat{B} & :=  -\frac{1}{N} \sum_{i=1}^{N} \sum_{t=1}^{T}\widehat{\gamma}'_{\mr{local}, t} \left(\sum_{\tau=1}^{T}\widehat{\gamma}_{\mr{local}, \tau}\widehat{\gamma}_{\mr{local}, \tau}'\widehat{\ddot{\ell}}_{i\tau}\right)^{-1}\widehat{\gamma}_{\mr{local}, t} \left( \widehat{\dot{\ell}}_{it}\widehat{\ddot{\ell}}_{it}(X_{it} - \widehat{\Xi}_{it}) + \frac{1}{2}\widehat{\dddot{\ell}}_{it}(X_{it} - \widehat{\Xi}_{it})\right),   \\
        \widehat{D} & := -\frac{1}{T} \sum_{t=1}^{T} \sum_{i=1}^{N} \widehat{\lambda}'_{\mr{local}, i} \left(\sum_{j=1}^{N}\widehat{\lambda}_{\mr{local}, j}\widehat{\lambda}_{\mr{local}, j}'\widehat{\ddot{\ell}}_{jt}\right)^{-1}\widehat{\lambda}_{\mr{local}, i}\left(\widehat{\dot{\ell}}_{it}\widehat{\ddot{\ell}}_{it}(X_{it} - \widehat{\Xi}_{it}) + \frac{1}{2}\widehat{\dddot{\ell}}_{it}(X_{it} - \widehat{\Xi}_{it})\right), \\
        \widehat{W} & := -\frac{1}{NT} \sum_{i=1}^{N}\sum_{t=1}^{T} \widehat{\ddot{\ell}}_{it}(X_{it} - \widehat{\Xi}_{it})(X_{it} - \widehat{\Xi}_{it})'. 
    \end{align*}
    Here, $\widehat{\dot{\ell}}_{it}$, $\widehat{\ddot{\ell}}_{it}$, $\widehat{\dddot{\ell}}_{it}$ denote the first, second, and third derivatives of the log-likelihood $\ell_{it}(\cdot)$ with respect to the index, evaluated at the local estimator $(\hat{\beta}_{\mr{local}}, \hat{\Lambda}_{\mr{local}}, \hat{\Gamma}_{\mr{local}})$.   The analytical bias correction estimator, $\hat{\beta}_{\mr{ABC}} := \hat{\beta}_{\mr{local}} - \frac{1}{T} \widehat{W}^{-1} \widehat{B} - \frac{1}{N} \widehat{W}^{-1} \widehat{D}$, follows
    \begin{align*}
        \sqrt{NT}\left(\hat{\beta}_{\mr{ABC}} - \beta_0\right)\cd N(0, \widehat{W}^{-1}). 
    \end{align*}

    \paragraph{Split-panel jackknife} The split-panel jackknife estimator in \citet{chen2021nonlinear} is given by 
    \begin{align*}
        \hat{\beta}_{JBC} := 3 \hat{\beta}_{\mr{local}} - \overline{\beta}_{N, T/2} - \overline{\beta}_{N/2, T}, 
    \end{align*}
    where $\overline{\beta}_{N, T/2}$ is the average of the estimators in the half-panels $\{(i, t)\mid i=1, \ldots, N, t = 1,\ldots, \lceil T/2 \rceil \}$ and $\{(i, t)\mid i=1, \ldots, N, t =  \lceil T/2 \rceil \ + 1,\ldots, T\}$, $\overline{\beta}_{N/2, T}$ is the average of the estimators in the half-panels $\{(i, t)\mid i=1, \ldots, \lceil N/2 \rceil , t = 1,\ldots, T\}$ and $\{(i, t)\mid i=\lceil N/2 \rceil + 1\ldots, N , t = 1,\ldots, T\}$.

\subsection{Additional Simulations}\label{ssec:additional MCs}
    
    We evaluate the performance of our estimator in a binary response model with predetermined covariates, where the data are generated from the following single-index model with $R = 2$, 
    \begin{equation}\label{eq:logit_dynamic}
        \begin{aligned}
            Y_{it} & = \bs{1}\left(\beta_{1} Y_{i,t-1} + \beta_{ 2}Z_{it} + \lambda_{i}' \gamma_{t} + \epsilon_{Y, it}>0\right),  \\
            Z_{it} & =  \lambda_{i}' \gamma_{t} + \lambda_{i}' \iota  + \gamma_{ t}' \iota  + \lambda_{z, i}  \gamma_{z, t} +  \epsilon_{Z, it}. 
        \end{aligned}
    \end{equation}  
    Here, $ \lambda_{i} = (\lambda_{i1}, \lambda_{i2})'$, $ \gamma_{t} = (\gamma_{t1}, \gamma_{t2})'$.  $\{ \lambda_{ir}\}_{ 1\leq i\leq N,  r=1,2}$,  $\{ \gamma_{tr}\}_{1\leq t\leq T, r=1,2}$, $\{ \lambda_{z, i}\}_{ 1\leq i\leq N}$, and $\{ \gamma_{z, t}\}_{ 1\leq t\leq T}$ consist of independent random variables drawn from the standard normal distribution $N(0, 1)$. 
    The error terms $\{\epsilon_{Y, it}\}_{1\leq i\leq N, 1\leq t\leq T}$ are i.i.d. random variables over both $i$ and $t$ following the standard logistic distribution, and $\{\epsilon_{Z, it}\}_{1\leq i\leq N, 1\leq t\leq T}$ are i.i.d. random variables over both $i$ and $t$ following the normal distribution  $N(0, 4)$.  We set $R_{\max} = 5$ and $(\beta_{1}, \beta_2) = (0.5, 0.2)$.  Conditional on $\Lambda_0$ and $\Gamma_0$, $Z_{it}$ is an exogenous covariate while the lagged outcome $Y_{it-1}$ is a predetermined variable. 

    We conduct $1000$ Monte Carlo replications to evaluate the finite-sample performance of our estimator across different sample sizes, ranging from $(N,T) = (50,40)$ to $(N,T) = (1000,200)$. We also report alternative estimators for comparison, as well as bias-corrected estimators to examine whether bias corrections based on our two-step estimator mitigate the incidental parameter problem. 
    
    The estimators considered in the Monte Carlo experiments include a pooled estimator that ignores individual and time latent factors (POOL); our first-step nuclear norm-regularized estimator (NNR), defined in~\eqref{eq:nnr_definition}; our two-step estimator using the true number of factors $R$ ($\mr{TS}^*$); and its analytical bias-corrected estimator  ($\mr{ABC}^*$). We also report the corresponding two-step estimator using an estimated number of factors $\hat{R}$ ($\mr{TS}$), along with its analytical bias-corrected counterpart ($\mr{ABC}$).
    We report both the bias and the standard deviation of each estimator. In addition, we provide the average estimated number of factors $\hat{R}$.

    It is worth noting that we do not report results for the split-panel jackknife bias correction method, as it requires an additional homogeneity condition (\citealp{fernandez2016individual}, Assumption 4.3), which we consider too restrictive in dynamic settings. For the analytical correction,  we directly apply the analytical bias correction method proposed in \citet{chen2021nonlinear} to the dynamic setting, despite their method not being specifically designed for dynamic models. Developing a valid analytical bias correction formula for dynamic nonlinear panel models with interactive fixed effects is challenging and beyond the scope of this paper. 

    Table~\ref{tab:logit_dynamic} reports results for dynamic Logit models. The first column (POOL) presents the pooled regression results, which exhibit substantial bias that does not diminish as the sample size increases. The second column (NNR) reports the performance of the NNR estimator, whose bias decreases slowly toward zero as $N, T \to \infty$, consistent with Theorem~\ref{thm:consistency}. 
    Table~\ref{tab:logit_dynamic} shows that our proposed two-step estimator using the true number of factors ($R = 2$), reported in column $\mr{TS}^*$, substantially improves upon the NNR estimator: its bias is smaller and converges rapidly to zero as $N$ and $T$ increase. Nevertheless, due to the incidental parameter problem, the bias and standard deviation of the $\mr{TS}^*$ estimator remain of the same order even for large sample sizes  such as $(N, T) = (1000, 200)$. The table also shows that applying analytical bias correction to our two-step estimator effectively reduces bias. The analytical bias-corrected estimator ($\mr{ABC}^*$) significantly reduces bias even in small samples, such as $(N, T) = (50, 40)$. In addition, we find that the analytical bias correction applied to the $\mr{TS}^*$ estimator is more effective in reducing the bias of the estimator of the dynamic effect $\beta_1$ than that of $\beta_2$.

    When the number of factors is estimated rather than known, the corresponding estimators ($\mr{TS}$, and $\mr{ABC}$) continue to perform well. Since the number of factors can be estimated with high accuracy (see the last column $\bar{R}$), the performance of the  estimators based on $\hat{R}$ is very similar to that obtained using the true $R$.  
    In some cases, the bias-corrected estimator using $\hat{R}$ even performs better than the bias-corrected estimator that uses the true $R$. 
    For example, when $N = 200, T = 40$, the bias of $\hat{\beta}_1$ for $\mr{ABC}^*$ is $-0.47 \times 10^{-2}$, whereas the bias for $\mr{ABC}$ is $0.09 \times 10^{-2}$. We attribute this observation to small-sample effects. 
    Overall, the two-step estimator delivers strong finite-sample performance and, when combined with bias correction, achieves substantial bias reduction and supports valid inference.

    \begin{table}[H]
        \centering
        \caption{Logit Model with Predetermined Covariate}\label{tab:logit_dynamic}
        \begin{tabular}{ccccccccc}
        \toprule
                &       & POOL  & NNR   &$\mr{TS}^*$    & $\mr{ABC}^*$ & $\mr{TS}$   & ABC & $\overline{R}$ \\
                &       & $(\times 10^{-2})$  & $(\times 10^{-2})$  & $(\times 10^{-2})$  &$(\times 10^{-2})$  & $(\times 10^{-2})$  & $(\times 10^{-2})$  &  \\
        \midrule
        N = 50, T = 40 &       &       &       &       &       &       &       &  \\
        BIAS  & $\beta_1$ & -11.99 & -10.65 & 2.25  & -0.90 & 2.68  & -0.37 & 1.94 \\
        STD   &       & (12.75) & (8.32) & (10.75) & (10.37) & (10.93) & (10.26) &  \\
        BIAS  & $\beta_2$ & 7.53  & 6.97  & 5.57  & 4.04  & 5.61  & 3.75  &  \\
        STD   &       & (2.38)  & (2.34)  & (3.69)  & (3.48)  & (3.81)  & (3.52)  &  \\
        N = 100, T = 40 &       &       &       &       &       &       &       &  \\
        BIAS  & $\beta_1$ & -12.32 & -10.36 & 2.03  & -0.53 & 2.44  & 0.00  & 1.997 \\
        STD   &       & (11.27) & (6.10)  & (6.76)  & (6.52)  & (6.74) & (6.42)  &  \\
        BIAS  & $\beta_2$ & 7.53  & 6.66  & 2.89  & 1.71  & 2.93  & 1.49  &  \\
        STD   &       & (1.98)  & (1.86)  & (2.30)  & (2.23)  & (2.27)  & (2.18)  &  \\
        N = 200, T = 40 &       &       &       &       &       &       &       &  \\
        BIAS  & $\beta_1$ & -12.12 & -10.28 & 1.55  & -0.47 & 1.89  & -0.09 & 1.988 \\
        STD   &       & (10.37) & (5.08)  & (4.79)  & (4.71)  & (4.82)  & (4.56)&  \\
        BIAS  & $\beta_2$ & 7.48  & 6.37  & 1.84  & 0.85  & 1.94  & 0.73  &  \\
        STD   &       & (1.66)  & (1.53)  & (1.70)  & (1.69)  & (1.68)  & (1.66)  &  \\
        N = 100, T = 100 &       &       &       &       &       &       &       &  \\
        BIAS  & $\beta_1$ & -12.12 & -9.81 & 1.53  & 0.09  & 1.53  & 0.09  & 2.000 \\
        STD   &       & (7.03)  & (3.89)  & (3.87)  & (3.77)  & (3.87)  & (3.77)  &  \\
        BIAS  & $\beta_2$ & 7.47  & 5.95  & 1.17  & 0.27  & 1.17  & 0.27  &  \\
        STD   &       & (1.45)  & (1.29)  & (1.31)  & (1.26)  & (1.31)  & (1.26)  &  \\
        N = 200, T = 100 &       &       &       &       &       &       &       &  \\
        BIAS  & $\beta_1$ & -12.25 & -9.25 & 1.42  & 0.18  & 1.42  & 0.18  & 2.000 \\
        STD   &       & (6.64)  & (3.02)  & (2.72)  & (2.63)  & (2.72)  & (2.63)  &  \\
        BIAS  & $\beta_2$ & 7.37  & 5.46  & 0.80  & 0.14  & 0.80  & 0.14  &  \\
        STD   &       & (1.19)  & (1.01)  & (0.86)  & (0.83)  & (0.86)  & (0.83)  &  \\
        N = 200, T = 200 &       &       &       &       &       &       &       &  \\
        BIAS  & $\beta_1$ & -11.95 & -8.38 & 1.02  & 0.15  & 1.02  & 0.15  & 2.000 \\
        STD   &       & (4.48)  & (1.97)  & (1.83)  & (1.79)  & (1.83)  & (1.79)  &  \\
        BIAS  & $\beta_2$ & 7.43  & 4.94  & 0.48  & 0.04  & 0.48  & 0.04  &  \\
        STD   &       & (0.93)  & (0.74)  & (0.60)  & (0.59)  & (0.60)  & (0.59)  &  \\
        N = 1000, T = 200 &       &       &       &       &       &       &       &  \\
        BIAS  & $\beta_1$ & -12.11 & -7.11 & 0.52  & -0.04 & 0.52  & -0.04 & 2.000 \\
        STD   &       & (4.09)  & (1.16)  & (0.80)  & (0.79)  & (0.80)  & (0.79)  &  \\
        BIAS  & $\beta_2$ & 7.40  & 4.13  & 0.28  & 0.03  & 0.28  & 0.03  &  \\
        STD   &       & (0.71)  & (0.47)  & (0.26)  & (0.26)  & (0.26)  & (0.26)  &  \\
        \bottomrule
        \end{tabular}%
        \vspace{0.5cm}
        \begin{minipage}{\textwidth}
            \footnotesize
            \textbf{Note:} Simulation results based on $1000$ replications for the dynamic Logit model with predetermined covariates as in~\eqref{eq:logit_dynamic}. We set $\alpha = 0.05$ to determine the regularization parameter $\varphi_{NT}$. We report bias and standard deviation for  pooled estimators (POOL), nuclear norm regularized estimators (NNR). Using the true number of factors, $R=2$, we report bias and standard deviation (measured in units of $\times 10^{-2}$)  for our proposed estimators ($\mr{TS}^*$) and the   analytical bias-corrected estimators based on $\mr{TS}^*$ ($\mr{ABC}^*$). Using estimated number of factors $\hat{R}$,  we report bias and standard deviation  (measured in units of $\times 10^{-2}$)  for our proposed estimators ($\mr{TS}$), and the analytical bias-corrected estimators based on $\mr{TS}$ ($\mr{ABC}$). Furthermore, we report the average estimated number of factors $\bar{R}$.  
        \end{minipage}
    \end{table}%

    \section{Proofs of Consistency of the NNR Estimator}

    For readability, the main text presents simplified definitions of the NNR estimators, FE estimators, and shrinking neighborhoods. We begin by providing their fully rigorous definitions here. The NNR estimators $(\hat{\beta}_{\mr{nuc}}, \hat{\Theta}_{\mr{nuc}})$ solve the following optimization problem with a nuclear-norm penalty:
    \begin{equation}\label{eq:nnr_definition_formal}
        \begin{gathered}
            \left(\hat{\beta}_{\mr{nuc}}, \hat{\Theta}_{\mr{nuc}}\right) = \argmin_{\beta \in \mb{R}^{d_X}, \Theta\in \mb{R}^{N\times T}} \left\{ \mc{L}_{NT} \left(  \beta, \Theta\right) + \frac{\varphi_{NT}}{\sqrt{NT}} \|\Theta\|_{\mr{nuc}}\right\}, \\
           \text{s.t. } \quad \|\beta\|_{\max} \leq \rho_{\beta}, \quad \|\Theta\|_{\max} \leq \rho_{\theta}. 
        \end{gathered}
    \end{equation}
    Here, we impose additional constraints $\|\beta\|_{\max} \leq \rho_{\beta}$ and $\|\Theta\|_{\max} \leq \rho_{\theta}$ compared to optimization~\eqref{eq:nnr_definition}. These constraints, which are standard in extremum estimation, ensure that optimization is conducted within a compact parameter space. In addition, the FE estimators $(\hat{\beta}_{\mr{FE}}, \hat{\Lambda}_{\mr{FE}}, \hat{\Gamma}_{\mr{FE}})$ solve
    \begin{equation}\label{eq:FE_estimator_formal}
        \begin{gathered}
            (\hat{\beta}_{\mr{FE}}, \hat{\Lambda}_{\mr{FE}}, \hat{\Gamma}_{\mr{FE}}) \in \argmin_{ \beta, \Lambda, \Gamma} \mc{L}_{NT}(\beta, \Lambda, \Gamma), \\
            \text{s.t. }\quad \|\beta\|_{\max} \leq \rho_{\beta}, \quad \|\Lambda\|_{\max} \leq \rho_{\lambda}, \quad \|\Gamma\|_{\max} \leq \rho_{\gamma}. 
        \end{gathered}
    \end{equation}
    in which we impose additional constraints to~\eqref{eq:FE_estimator},  $\|\Lambda\|_{\max} \leq \rho_{\lambda}$ and $\|\Gamma\|_{\max} \leq \rho_{\gamma}$, which are standard in extremum estimation and ensure that the optimization is conducted over a compact parameter space. We also impose such compactness conditions on the definition of neighborhoods $\mc{B}_{\delta_{NT}} $ introduced in~\eqref{eq:neighborhood}: 
    \begin{align}\label{eq:neighborhood_formal}
        \mc{B}_{\delta_{NT}} = \bigg\{(\beta, \Lambda, \Gamma)\mid  & \|\beta-\beta_0\|, \frac{1}{\sqrt{N}}\|\Lambda - {\Lambda}_0 \|_{\mr{F}} , \frac{1}{\sqrt{T}}\|\Gamma - {\Gamma}_0 \|_{\mr{F}} \leq \delta_{NT} \text{, }  \|\Lambda\|_{\max} \leq \rho_{\lambda}\text{, } \|\Gamma\|_{\max} \leq \rho_{\gamma}\bigg\}, 
    \end{align} 
    and we need to make sure the NNR estimator falls within the shrinking neighborhood $\mc{B}_{\delta_{NT}} $ wpa1. One potential concern of adding  compactness constraints in~\eqref{eq:neighborhood_formal} is that  the  entries of $(\hat{\Lambda}_{\mr{nuc}}, \hat{\Gamma}_{\mr{nuc}})$  are not necessarily  uniformly bounded. This is not a substantive issue, as we can  truncate and normalize $(\hat{\Lambda}_{\mr{nuc}}, \hat{\Gamma}_{\mr{nuc}})$ to obtain nuisance estimators that satisfy the uniform boundedness condition. The details of this procedure, along with a proof demonstrating that it does not affect the theoretical results, are provided later.

    We introduce some new notations for simplicity. For any two positive real sequences $\left\{a_n\right\}_{n \geq 1}$ and $\left\{b_n\right\}_{n\geq 1}$, we use $a_n\lesssim b_n$ ($a_n\gtrsim b_n$) to denote that there exists a positive constant $c$ such that $a_n\leq cb_n$ ($a_n \geq  cb_n$) for all $n$. We write  $a_n\asymp b_n$ if both $a_n\lesssim b_n$  and $a_n\gtrsim b_n$.  Recall that there exist constants $(\rho_{\beta}, \rho_{\lambda}, \rho_{\gamma}, \rho_{\theta})$ such that 
    \begin{align*}
        \|\beta\|_{\max}\leq \rho_{\beta}, \quad \|\Lambda_0\|_{\max} \leq \rho_{\lambda}, \quad \|\Gamma_0\|_{\max} \leq \rho_{\gamma}, \quad \|\Theta\|_{\max} \leq \rho_{\theta}.  
    \end{align*}
    We define the estimation errors as
    \begin{align*}
        \hat{\Delta}_{\beta } := \hat{\beta}_{\mr{nuc}} - \beta_0, \quad \hat{\Delta}_{\Lambda  } := \hat{\Lambda}_{\mr{nuc}} - \Lambda_0, \quad \hat{\Delta}_{\Gamma} := \hat{\Gamma}_{\mr{nuc}} - \Gamma_0, \quad
         \hat{\Delta}_{\Theta } := \hat{\Theta}_{\mr{nuc}} - \Theta_0.  
    \end{align*}
    It follows from \eqref{eq:nnr_definition_formal} that 
    \begin{align*}
        \|\hat{\Delta}_{\beta}\|_{\max} \leq 2\rho_{\beta}, \quad \|\hat{\Delta}_{\Theta}\|_{\max} \leq 2\rho_{\theta}. 
    \end{align*}

    The following lemma collects basic properties of low-rank projections that are used frequently in the subsequent proof. 
    \begin{lemma}\label{lemma:property_local_projection}
        Let $\Delta$ be any $N\times T$ matrix. The following properties hold:
        \begin{enumerate}[label=(\roman*)]
            \item \label{item:local_projection_1} $\|\Delta\|_{\mr{nuc}} = \|M_{\Lambda_0}\Delta M_{\Gamma_0} \|_{\mr{nuc}} + \|\Delta - M_{\Lambda_0}\Delta M_{\Gamma_0}\|_{\mr{nuc}}$.
            \item \label{item:local_projection_2} $\|\Delta\|_{\mr{F}}^2 = \|M_{\Lambda_0}\Delta M_{\Gamma_0} \|_{\mr{F}}^2 + \|\Delta - M_{\Lambda_0}\Delta M_{\Gamma_0}\|_{\mr{F}}^2 $. 
            \item \label{item:local_projection_3} $\mr{rank}(\Delta - M_{\Lambda_0}\Delta M_{\Gamma_0})\leq 2R$. 
        \end{enumerate}
    \end{lemma}
    The proof is omitted. Readers may refer to \citet[Lemma~D.2]{chernozhukov2019inference} and \citet[Chapter~10]{wainwright2019high} for details.

\subsection{Proofs of Theorem~\ref{thm:consistency} and Theorem~\ref{thm:consistency_pre}}

    We provide only the proof of Theorem~\ref{thm:consistency_pre}, as it extends Theorem~\ref{thm:consistency} by incorporating predetermined covariates.

    \begin{proofthm}{thm:consistency_pre}
        The proof is based on the following lemma:
        \begin{lemma}\label{lemma:cone}
            Under the conditions of Theorem~\ref{thm:consistency_pre},  
            \begin{align*}
                \| M_{\Lambda_0} \hat{\Delta}_{\Theta} M_{\Gamma_0}\|_{\mr{nuc}} \leq \frac{2+\alpha}{\alpha}\left(\sqrt{NT}\|\hat{\Delta}_{\beta}\| +  \| \hat{\Delta}_{\Theta} -  M_{\Lambda_0} \hat{\Delta}_{\Theta} M_{\Gamma_0} \|_{\mr{nuc}} \right). 
            \end{align*}
        \end{lemma}
        The proof of Lemma~\ref{lemma:cone} is standard and will be presented at the end of this subsection. Lemma~\ref{lemma:cone} states that,  when the penalization parameter $\varphi_{NT}$ is sufficiently large, the component of the estimation error of  $\Theta$ that cannot be explained by either $\Lambda_0$ or $\Gamma_0$ is relatively small compared to the part of the estimation error of $\Theta$ that can be explained by $\Lambda_0$  and $\Gamma_0$, along with a term that accounts for the estimation error of $\beta$. 

        When 
        $\|\hat{\Delta}_{\beta}\|^2  +  \frac{1}{NT}\|\hat{\Delta}_{\Theta}\|_\mr{F}^2 \leq  \sqrt{\frac{\log (NT)}{NT}}$, we directly obtain 
        \begin{align*}
            \|\hat{\Delta}_{\beta}\| & \leq \log (NT) /\sqrt{ \min\{N, T\}}, \quad 
            \frac{1}{\sqrt{NT}}\|\hat{\Delta}_{\Theta}\|_{\mr{F}} \leq \log (NT) /\sqrt{ \min\{N, T\}}. 
        \end{align*}
        When $\|\hat{\Delta}_{\beta}\|^2  +  \frac{1}{NT}\|\hat{\Delta}_{\Theta}\|_\mr{F}^2 >  \sqrt{\frac{\log (NT)}{NT}}$, the proof becomes more intricate and requires additional effort.  

        \paragraph{Step 1} Since $(\hat{\beta}_{\mr{nuc}}, \hat{\Theta}_{\mr{nuc}})$  solves optimization problem~\eqref{eq:nnr_definition_formal}, we have 
        \begin{align}\label{eq:thm_consistency_1}
            \mc{L}_{NT}( \beta_0 + \hat{\Delta}_{\beta} ,  \Theta_0 + \hat{\Delta}_{\Theta} ) - \mc{L}_{NT}(\beta_0, \Theta_0)  \leq \frac{\varphi_{NT}}{\sqrt{NT}} (\|\Theta_0\|_{\mr{nuc}} - \|\Theta_0 + \hat{\Delta}_{\Theta} \|_{\mr{nuc}} ). 
        \end{align}
        Consider the Taylor expansion of $\mc{L}_{NT}(\hat{\beta}_{\mr{nuc}}, \hat{\Theta}_{\mr{nuc}}) $ around $(\beta_0, \Theta_0)$: 
        \begin{equation}\label{eq:thm_consistency_2}
        \begin{aligned}
            & \mc{L}_{NT}( \beta_0 + \hat{\Delta}_{\beta} ,  \Theta_0 + \hat{\Delta}_{\Theta} ) - \mc{L}_{NT}(\beta_0, \Theta_0)  -  \nabla_{\beta}\mc{L}_{NT}\left(\beta_0 , \Theta_0\right)'\hat{\Delta} _{\beta}  - \la \nabla_{\Theta}\mc{L}_{NT}\left(\beta_0 , \Theta_0\right), \hat{\Delta}_{\Theta} \ra \\
            \leqtext{(i)} & \frac{\varphi_{NT}}{\sqrt{NT}} \underbrace{(\|\Theta_0\|_{\mr{nuc}} - \| \Theta_0 + \hat{\Delta}_{\Theta} \|_{\mr{nuc}} ) }_{\leq \|\hat{\Delta}_{\Theta}\|_{\mr{nuc}}}  -  \underbrace{\nabla_{\beta}\mc{L}_{NT}\left(\beta_0 , \Theta_0\right)'\hat{\Delta} _{\beta}}_{\leq \|\nabla_{\beta}\mc{L}_{NT}\left(\beta_0 , \Theta_0\right)\| \|\hat{\Delta} _{\beta}\| }  - \underbrace{\la \nabla_{\Theta}\mc{L}_{NT}\left(\beta_0 , \Theta_0\right), \hat{\Delta}_{\Theta} \ra}_{\leq \|\nabla_{\Theta}\mc{L}_{NT}\left(\beta_0 , \Theta_0\right)\|_{\mr{op}}\| \hat{\Delta}_{\Theta}\|_{\mr{nuc}}} \\
            \leqtext{(ii)} & \frac{\varphi_{NT}}{\sqrt{NT}}\|\hat{\Delta}_{\Theta}\|_{\mr{nuc}} + \|\nabla_{\beta}\mc{L}_{NT}\left(\beta_0 , \Theta_0\right)\| \|\hat{\Delta} _{\beta}\| + \|\nabla_{\Theta}\mc{L}_{NT}\left(\beta_0 , \Theta_0\right)\|_{\mr{op}}\| \hat{\Delta}_{\Theta}\|_{\mr{nuc}} \\
            \leqtext{(iii)} & \frac{\varphi_{NT}}{\sqrt{NT}}\|\hat{\Delta}_{\Theta}\|_{\mr{nuc}} + \frac{ \varphi_{NT}}{1+\alpha} \|\hat{\Delta} _{\beta}\| +  \frac{\varphi_{NT}}{(1+\alpha)\sqrt{NT}} \| \hat{\Delta}_{\Theta}\|_{\mr{nuc}}  \\
            \leq & \frac{2 \varphi_{NT}}{\sqrt{NT}}\left( \|\hat{\Delta}_{\Theta}\|_{\mr{nuc}} + \sqrt{NT}  \|\hat{\Delta}_{\beta}\| \right), 
        \end{aligned}
        \end{equation}
        where inequality (i) follows from inequality (\ref{eq:thm_consistency_1}),  inequality (ii) employs the triangle inequality, Cauchy-Schwarz inequality, and Hölder's inequality (since the spectral norm is the dual norm of the nuclear norm), and inequality (iii) holds because  of $\varphi_{NT} \geq (1+\alpha)  \max\{\|\nabla_{\beta}\mc{L}_{NT}\left(\beta_0 , \Theta_0\right)\|, \sqrt{NT}\|\nabla_{\Theta}\mc{L}_{NT}\left(\beta_0 , \Theta_0\right)\|_{\mr{op}}\}$,  as stated in Theorem~\ref{thm:consistency_pre}. In addition, we have the following inequalities for the  nuclear norm:
        \begin{align*}
            \|\hat{\Delta}_{\Theta}\|_{\mr{nuc}} & \eqtext{(i)} \|  M_{\Lambda_0 }\hat{\Delta}_{\Theta}M_{\Gamma_0} \| _{\mr{nuc}} + \|   \hat{\Delta}_{\Theta} -  M_{\Lambda_0 }\hat{\Delta}_{\Theta}M_{\Gamma_0} \| _{\mr{nuc}} \\
            & \leqtext{(ii)}  \frac{2+\alpha}{\alpha}\left(\sqrt{NT}\|\hat{\Delta}_{\beta}\| +  \| \hat{\Delta}_{\Theta} -  M_{\Lambda_0} \hat{\Delta}_{\Theta} M_{\Gamma_0} \|_{\mr{nuc}} \right)  + \|   \hat{\Delta}_{\Theta} -  M_{\Lambda_0 }\hat{\Delta}_{\Theta}M_{\Gamma_0} \| _{\mr{nuc}} \\
            & \leqtext{(iii)} \frac{2+\alpha}{\alpha} \sqrt{NT}\|\hat{\Delta}_{\beta}\| + \frac{2(1+\alpha)\sqrt{2R}}{\alpha}\|\hat{\Delta}_{\Theta} -  M_{\Lambda_0 }\hat{\Delta}_{\Theta}M_{\Gamma_0}\|_{\mr{F}} \\
            & \leqtext{(iv)} \frac{2+\alpha}{\alpha} \sqrt{NT}\|\hat{\Delta}_{\beta}\| + \frac{2(1+\alpha)\sqrt{2R}}{\alpha}\|\hat{\Delta}_{\Theta} \|_{\mr{F}} \\
            & \leq \frac{2(1+\alpha)\sqrt{2R}}{\alpha}( \sqrt{NT}\|\hat{\Delta}_{\beta}\| + \|\hat{\Delta}_{\Theta} \|_{\mr{F}}), 
        \end{align*}
        where equality (i) holds due to Lemma~\ref{lemma:property_local_projection}\ref{item:local_projection_1},  inequality (ii) follows from Lemma~\ref{lemma:cone}, inequality (iii) follows from Lemma~\ref{lemma:property_local_projection}\ref{item:local_projection_3}, 
        and inequality (iv) is based on Lemma~\ref{lemma:property_local_projection}\ref{item:local_projection_2}. 

        Therefore, combining the nuclear norm inequality with  inequality~\eqref{eq:thm_consistency_2} yields 
        \begin{equation}\label{eq:thm_consistency_3}
            \begin{aligned}
                & \mc{L}_{NT}( \beta_0 + \hat{\Delta}_{\beta} ,  \Theta_0 + \hat{\Delta}_{\Theta} ) - \mc{L}_{NT}(\beta_0, \Theta_0)  -  \nabla_{\beta}\mc{L}_{NT}\left(\beta_0 , \Theta_0\right)'\hat{\Delta} _{\beta}  - \la \nabla_{\Theta}\mc{L}_{NT}\left(\beta_0 , \Theta_0\right), \hat{\Delta}_{\Theta} \ra \\
                \leq &  \frac{8(1+\alpha)\sqrt{2R}}{\alpha}\frac{\varphi_{NT}}{\sqrt{NT}} ( \sqrt{NT}\|\hat{\Delta}_{\beta}\| + \|\hat{\Delta}_{\Theta} \|_{\mr{F}}) \\
                \leq & \frac{16(1+\alpha)\sqrt{2R}}{\alpha} \varphi_{NT} \sqrt{ \|\hat{\Delta}_{\beta}\|^2 + \frac{1}{NT}\|\hat{\Delta}_{\Theta} \|^2_\mr{F}}. 
            \end{aligned}
        \end{equation}

        \paragraph{Step 2} By the convexity of $\mc{L}_{NT}(\cdot, \cdot)$, we have 
        \begin{equation}\label{eq:thm_consistency_4}
        \begin{aligned}
            & \mc{L}_{NT}(\beta_0 + \hat{\Delta}_{\beta}, \Theta_0 + \hat{\Delta}_{\Theta}) - \mc{L}_{NT}(\beta_0, \Theta_0) - \nabla_{\beta}\mc{L}_{NT}(\beta_0, \Theta_0)'\hat{\Delta}_{\beta} - \la \nabla_{\Theta}\mc{L}_{NT}(\beta_0, \Theta_0), \hat{\Delta}_{\Theta} \ra \\
            \eqtext{(i)} & \frac{1}{2NT}\sum_{i=1}^{N}\sum_{t=1}^{T}(-\ddot{\ell}_{it}(X_{it}'\tilde{\beta} + \tilde{\theta}_{it} ))(X_{it}' \hat{\Delta}_{\beta} + \hat{\Delta}_{\theta_{it}} )^2 \\
            \geqtext{(ii)} & \frac{{b}_{\min}}{2} \frac{1}{NT}\sum_{i=1}^{N}\sum_{t=1}^{T}(X_{it}' \hat{\Delta}_{\beta} + \hat{\Delta}_{\theta_{it}} )^2 \\
            \geqtext{(iii)} &  \frac{{b}_{\min}}{2} \left(\kappa \left(\|\hat{\Delta}_{\beta}\|^2  +  \frac{1}{NT}\|\hat{\Delta}_{\Theta}\|_{\mr{F}}^2\right) - \eta \frac{N+T}{NT}(\log(NT))^2\right), 
        \end{aligned}
        \end{equation} 
        where equality (i) follows from the second-order Taylor expansion of $\mc{L}_{NT}(\cdot, \cdot)$, inequality (ii) holds because $-\ddot{\ell}_{it}(\cdot)\geq b_{\min}$ uniformly as stated in Assumption~\ref{assumption:regularity_pre}\ref{item:smoothing_pre}, and inequality (iii) follows from  Lemma~\ref{lemma:cone} and the RSC (Assumption~\ref{assumption:RSC}).
        
        By combining~\eqref{eq:thm_consistency_3} and~\eqref{eq:thm_consistency_4}, we obtain 
        \begin{align*}
            \frac{{b}_{\min}}{2} \left(\kappa \left(\|\hat{\Delta}_{\beta}\|^2  +  \frac{1}{NT}\|\hat{\Delta}_{\Theta}\|_{\mr{F}}^2\right) - \eta \frac{N+T}{NT} (\log(NT))^2 \right) &\leq \frac{16(1+\alpha)\sqrt{2R}}{\alpha} \varphi_{NT} \sqrt{ \|\hat{\Delta}_{\beta}\|^2 + \frac{1}{NT}\|\hat{\Delta}_{\Theta} \|^2_{\mr{F}}}  \\
            \Rightarrow \sqrt{ \|\hat{\Delta}_{\beta}\|^2 + \frac{1}{NT}\|\hat{\Delta}_{\Theta}\|_{\mr{F}}^2 } & \leq \frac{a_1 \varphi_{NT} + \sqrt{a_1^2  \varphi_{NT}^2  + 4a_2^2 \frac{N+T}{NT}(\log(NT))^2}}{2} \\
            & \leq a_1 \varphi_{NT}  + a_2   \sqrt{\frac{N+T}{NT}}(\log(NT))^2 \\
            & \leq a_1 \varphi_{NT} + \sqrt{2} a_2 \frac{\log(NT)}{\sqrt{\min\{N, T\}}}. 
        \end{align*}
        Here,  $a_1 := \frac{32(1+\alpha)\sqrt{2R}}{\alpha b_{\min}\kappa } > 0$ and $a_2 := \sqrt{\frac{\eta}{\kappa}} >0$. Let $c_1 := \max\{a_1, \sqrt{2} a_2\}$. It is straightforward to show that,  wpa1, 
        \begin{align*}
            \| \hat{\beta}_{\mr{nuc}} - \beta_0\| & \leq c_1 \left(\varphi_{NT} + \log (NT)/\sqrt{\min\{N, T\}}\right), \\
            \frac{1}{\sqrt{NT}}\|\hat{\Theta}_{\mr{nuc}} - \Theta_0\|_{\mr{F}} & \leq  c_1 \left(\varphi_{NT} + \log (NT)/\sqrt{\min\{N, T\}}\right) . 
        \end{align*}

        \paragraph{Step 3} 
        
        We aim to establish an estimation error bound for $\hat{\Lambda}_{\mr{nuc}}$. The error bound for $\hat{\Gamma}_{\mr{nuc}}$  can be directly obtained by the same  method.  In addition, our proof below is based on the situation where $\Sigma_{\lambda}^{1/2}\Sigma_{\gamma}\Sigma_{\lambda}^{1/2}$ may have repeated eigenvalues. 
        For notational simplicity, let $\Upsilon$  be the $R\times R$-dimensional matrix containing the eigenvectors of $\frac{1}{NT}(\Lambda_0'\Lambda_0)^{1/2}\Gamma_0'\Gamma_0(\Lambda_0'\Lambda_0)^{1/2}$, let $D$ be the $R\times R$-dimensional diagonal matrix containing the square roots of the eigenvalues of $\frac{1}{NT}(\Lambda_0'\Lambda_0)^{1/2}\Gamma_0'\Gamma_0(\Lambda_0'\Lambda_0)^{1/2}$, $\Omega$ be the $R$-dimensional vector containing the square roots of the  eigenvalues of $\Sigma_{\lambda}^{1/2}\Sigma_{\gamma}\Sigma_{\lambda}^{1/2}$ in non-increasing order, and  $U_0$ be the matrix of left singular vectors of $\Theta_0$. One can verify that $  U_0 = \Lambda_0(\Lambda_0'\Lambda_0)^{-1/2} \Upsilon$.  In addition, let $\hat{U}$ be the matrix containing the left singular vectors  of $\hat{\Theta}_{\mr{nuc}}$ such that $ \hat{U}' \hat{U} = \mb{I}_{R}$. Then we have $\hat{\Lambda}_{\mr{nuc}} =  \sqrt{N}\hat{U}\hat{D}^{1/2}_{[1:R, 1:R]}$. 

        We first  establish the bound for the distance between two spaces spanned by $\hat{U}$ and $U_0$, respectively.  Using Davis-Kahan Theorem (Theorem 4 in \citealp{yu2015useful}), there exists an orthogonal matrix $O^*$ such that, wpa1,  
        \begin{equation}\label{eq:thm_consistency_5}
        \begin{aligned}
            \|\hat{U}  - U_0  O^{*\prime}\|_{\mr{F}}
            \leqtext{(i)} &   \frac{2^{\frac{3}{2}}(2\|\Theta_0\|_{\mr{op}} + \|\hat{\Theta}_{\mr{nuc}} - \Theta_0 \|_{\mr{op}}) \|\hat{\Theta}_{\mr{nuc}} - \Theta_0 \|_{\mr{F}} }{\psi^2_{R}\left(\Theta_0\right)} \\
            \leqtext{(ii)} &   \frac{8 \|\Theta_0\|_{\mr{op}}  \|\hat{\Theta}_{\mr{nuc}} - \Theta_0 \|_{\mr{F}} }{\psi^2_{R}\left(\Theta_0\right)},   
        \end{aligned}
        \end{equation}
        where inequality (i) follows from the fact that $\psi_{R+1}(\Theta_0) = 0$, inequality (ii) is based on the estimation error bound of $\hat{\Theta}_{\mr{nuc}}$, implying that  $\|\hat{\Theta}_{\mr{nuc}} - \Theta_0 \|_{\mr{op}} / \|\Theta_0\|_{\mr{op}}  \cp 0$. The matrix  $O^*$  arises due to the possible  multiplicity of eigenvalues and depends only on 
         $(\Lambda_0, \Gamma_0)$.  In addition, the strong factor assumption (Assumption~\ref{assumption:regularity_pre}\ref{item:strong_factors_pre}) implies that
        \begin{equation}\label{eq:thm_consistency_6}
            \begin{gathered}
                \frac{1}{\sqrt{NT}}\|\Theta_0\|_{\mr{op}}  \cp \Omega_1, \quad \frac{1}{\sqrt{NT}}\psi_{R}\left(\Theta_0\right) \cp \Omega_{R}. 
            \end{gathered}
        \end{equation}
        Combining inequality (\ref{eq:thm_consistency_5}) and   (\ref{eq:thm_consistency_6}) yields that, wpa1,  
        \begin{align}\label{eq:thm_consistency_7}
            \|\hat{U} - U_0  O^{*\prime}\|_{\mr{F}} 
            &  \leq   \frac{16\Omega_1}{\sqrt{NT}\Omega_R^2} \|\hat{\Theta}_{\mr{nuc}} - \Theta_0 \|_{\mr{F}}  \leq \frac{16 c_1 \Omega_1}{\Omega_R^2 }  \left(\varphi_{NT} + \frac{\log(NT)}{\sqrt{\min\{N, T\}}}\right) . 
        \end{align}

        Now we turn to the error bound for $\hat{\Lambda}_{\mr{nuc}}$.  Note  that 
        \begin{align*}
            \left\|\hat{\Lambda}_{\mr{nuc}} -  \sqrt{N}\underbrace{\Lambda_0(\Lambda_0'\Lambda_0)^{-1/2} \Upsilon }_{U_0}   O^{*\prime} D^{1/2} \right\|_{\mr{F}} = & \sqrt{N} \left\|\hat{U} \hat{D}^{1/2}_{[1:R, 1:R]}  - U_0  O^{*\prime} D^{1/2}   \right\|_{\mr{F}}  \\ 
            \leq & \sqrt{N}\left\|\hat{U}   - U_0  O^{*\prime} D^{1/2}\hat{D}^{-1/2}_{[1:R, 1:R]}   \right\|_{\mr{F}} \left\|\hat{D}^{1/2}_{[1:R, 1:R]}\right\|_{\mr{F}}\\
            \leq & \sqrt{N} \left(\underbrace{\left\|\hat{U}  - U_0 O^{*\prime}  \right\|_{\mr{F}}}_{A_1}\left\|\hat{D}^{1/2}_{[1:R, 1:R]}\right\|_{\mr{F}} + \underbrace{\left\| U_0O^{*\prime}  \left(\hat{D}^{1/2}_{[1:R, 1:R]} - D^{1/2}  \right)  \right\|_{\mr{F}}}_{A_2}\right) . 
        \end{align*}
        Since we have already established the bound for $A_1$ (see~\eqref{eq:thm_consistency_7}), we focus on  the upper bound of $A_2$.  Note that 
        \begin{align*}
            \left\|\hat{D}^{1/2}_{[1:R, 1:R]} - D^{1/2} \right\|_{\mr{F}} \leq & \sqrt{R} \frac{\left\|\hat{D}_{[1:R, 1:R]} - D\right\|_{\mr{op}}}{2\min\{\psi_{R}^{1/2} \left(\hat{D}_{[1:R, 1:R]}\right), \psi_{R}^{1/2} \left(D\right)\}} \\
            \leq & \sqrt{R}\frac{\left\|\hat{D}_{[1:R, 1:R]} - D\right\|_{\mr{F}}}{2\min\{\psi_{R}^{1/2} \left(\hat{D}_{[1:R, 1:R]}\right), \psi_{R}^{1/2} \left(D\right)\}} \\
            \leq & \frac{\sqrt{R}\|\hat{\Theta}_{\mr{nuc}} - \Theta_0\|_{\mr{F}}}{\sqrt{NT}\Omega_R^{1/2}}. 
        \end{align*}
        Then, the following inequality holds wpa1 
        \begin{align*}
            A_2\leq \underbrace{\|U_0\|_{\mr{F}}}_{= \sqrt{R}}\left\|\hat{D}^{1/2}_{[1:R, 1:R]} - D^{1/2} \right\|_{\mr{F}} & \leq \frac{R\|\hat{\Theta}_{\mr{nuc}} - \Theta_0\|_{\mr{F}}}{\sqrt{NT}\Omega_R^{1/2}}  
            \leq \frac{ c_1 R}{\Omega_R^{1/2} }  \left(\varphi_{NT} + \frac{\log(NT)}{\sqrt{\min\{N, T\}}}\right). 
        \end{align*}
        Let $G = D^{1/2} O^* \Upsilon'(\Lambda_0'\Lambda_0/N)^{-1/2}$. We obtain   
        \begin{align*}
            \left\|\hat{\Lambda}_{\mr{nuc}} -   \Lambda_0 G' \right\|_{\mr{F}} \leq \sqrt{N}  \underbrace{\left(\frac{16c_1\Omega_1^{3/2}}{\Omega_R^2} +  \frac{ c_1 R}{\Omega_R^{1/2} }\right)}_{B_1}  \left(\varphi_{NT} + \frac{\log(NT)}{\sqrt{\min\{N, T\}}}\right), \quad \text{wpa1}. 
        \end{align*}
        By the same method, we also show that 
        \begin{align*}
            \|\hat{\Gamma}_{\mr{nuc}} - \Gamma_0 G^{-1}\|_{\mr{F}} \leq \sqrt{T}  B_1  \left(\varphi_{NT} + \frac{\log(NT)}{\sqrt{\min\{N, T\}}}\right), \quad \text{wpa1}. 
        \end{align*}
        In addition, since 
        \begin{align*}
             GG'  =  \underbrace{D^{1/2}}_{\cp \mr{diag}(\Omega)^{1/2}} O^* \Upsilon'\underbrace{(\Lambda_0'\Lambda_0/N)^{-1}}_{\cp \Sigma_{\lambda}^{-1}}\Upsilon O^{*\prime}  \underbrace{D^{1/2}}_{\cp \mr{diag}(\Omega)^{1/2}} \cp \mr{diag}(\Omega)^{1/2} O^* \Upsilon'  \Sigma_{\lambda}^{-1} \Upsilon O^{*\prime} \mr{diag}(\Omega)^{1/2}, 
        \end{align*}
        it immediately follows that, wpa1, (i) the maximum singular value of $G$ is uniformly bounded, and  (ii) the minimum singular value of $G$ is strictly greater than zero. Therefore, this completes the proof of the theorem.

        In the following text, we discuss how to construct nuisance estimators with uniformly bounded entries to ensure that $(\hat{\Lambda}_{\mr{nuc}}, \hat{\Gamma}_{\mr{nuc}}) \in \mc{B}_{\delta_{NT}}$. Although this property is not directly related to the current theorem, it will be used frequently in subsequent theoretical discussions.

        \paragraph{Construct nuisance estimators with uniformly bounded entries} As we have discussed before, one potential concern with $(\hat{\Lambda}_{\mr{nuc}}, \hat{\Gamma}_{\mr{nuc}})$ as in \eqref{eq:space_definition} is that  the  entries of $(\hat{\Lambda}_{\mr{nuc}}, \hat{\Gamma}_{\mr{nuc}})$  are not necessarily  uniformly bounded. In the following text, we show that  after  truncating and normalizing the estimators in \eqref{eq:space_definition},  we can obtain new nuisance estimators $(\tilde{\Lambda}_{\mr{nuc}}, \tilde{\Gamma}_{\mr{nuc}})$ that satisfy the uniform boundedness condition, and consequently, $(\tilde{\Lambda}_{\mr{nuc}}, \tilde{\Gamma}_{\mr{nuc}})\in \Phi_{NT}$.

        It should be noted that constructing uniformly bounded nuisance estimators is solely for the convenience of theoretical analysis. In practice, applied researchers do not need to perform this step.
        
        Our construction is based on the following observation: 
        Since   (i) $(\Lambda_0, \Gamma_0)$ is uniformly bounded,  (ii) the maximum singular value of $G$ is uniformly bounded, and  (iii) the minimum singular value of $G$ is strictly greater than zero, each entry of $\Lambda_0^G$ and $\Gamma_0^G$ is uniformly bound. Thus, there exists a constant $M>0$ that is sufficiently large but independent of $N, T$ such that $\|\Lambda_0^G\|_{\max}, \|\Gamma_0^G\|_{\max} \leq M$,  wpa1. 
        
        We first describe how to obtain the new estimators:
        \begin{enumerate}[label={}, leftmargin=2cm]
            \item [Step 1:] For a sufficiently large constant $M>0$ (independent of $N, T$), compute truncated estimators $(\bar{\Lambda}_{\mr{nuc}}, \bar{\Gamma}_{\mr{nuc}})$, defined as
            \begin{align*}
                &\bar{\Lambda}_{\mr{nuc}, ir} = 
                \left\{
                    \begin{array}{ll}
                        \hat{\Lambda}_{\mr{nuc}, ir}, & \text{if} \quad  |\hat{\Lambda}_{\mr{nuc}, ir}|\leq M \\
                        M\mr{sign}(\hat{\Lambda}_{\mr{nuc}, ir}),   & \text{otherwise }
                    \end{array}
                \right. \\
                &\bar{\Gamma}_{\mr{nuc}, tr} = 
                \left\{
                    \begin{array}{ll}
                        \hat{\Gamma}_{\mr{nuc}, tr}, & \text{if}\quad   |\hat{\Gamma}_{\mr{nuc}, tr}|\leq M \\
                        M\mr{sign}(\hat{\Gamma}_{\mr{nuc}, tr}),   & \text{otherwise }
                    \end{array}
                \right.
            \end{align*} 
            \item [Step 2:] Perform singular value decomposition on $\bar{\Theta}_{\mr{nuc}} := \bar{\Lambda}_{\mr{nuc}}\bar{\Gamma}'_{\mr{nuc}} $, so that $\bar{\Theta}_{\mr{nuc}}/\sqrt{NT} = \bar{U}\bar{D}\bar{V}$, where $\bar{U}\in \mb{R}^{N\times R}$ and $\bar{V}\in \mb{R}^{T\times R}$ are matrices whose columns are the left and right orthonormal singular vectors of $\bar{\Theta}_{\mr{nuc}}$, respectively, and $\bar{D}$ is a diagonal matrix whose diagonal entries are  singular values of $\bar{\Theta}_{\mr{nuc}}/\sqrt{NT}$ (arranged in non-increasing order). We then compute  $(\tilde{\Lambda}_{\mr{nuc}}, \tilde{\Gamma}_{\mr{nuc}})$ as follows: 
            \begin{gather*}
                \tilde{\Lambda}_{\mr{nuc}} =  \sqrt{N} \bar{U} \bar{D}^{1/2} , \quad \tilde{\Gamma}_{\mr{nuc}} = \sqrt{T} \bar{V} \bar{D}^{1/2}.
            \end{gather*}
        \end{enumerate}
        Since the entries of $(\bar{\Lambda}_{\mr{nuc}}, \bar{\Gamma}_{\mr{nuc}}, \Lambda_0^G, \Gamma_0^G)$ are uniformly bounded, we have the following inequalities wpa1:
        \begin{align*}
            \|\bar{\Lambda}_{\mr{nuc}}-\Lambda_0^G\|_{\mr{F}} \leq \|\hat{\Lambda}_{\mr{nuc}}-\Lambda_0^G\|_{\mr{F}} \leq \sqrt{N}  B_1 \left(\varphi_{NT} + \frac{\log(NT)}{\sqrt{\min\{N, T\}}}\right), \\
            \|\bar{\Gamma}_{\mr{nuc}}-\Gamma_0^G\|_{\mr{F}} \leq \|\hat{\Gamma}_{\mr{nuc}}-\Gamma_0^G\|_{\mr{F}} \leq \sqrt{T}  B_1  \left(\varphi_{NT} + \frac{\log(NT)}{\sqrt{\min\{N, T\}}}\right).
        \end{align*} 
        Thus, 
        \begin{align*}
            \|\bar{\Theta}_{\mr{nuc}} - \Theta_0\|_{\mr{F}} \leq & \underbrace{\|\bar{\Lambda}_{\mr{nuc}}\|_{\mr{F}}}_{\leq M\sqrt{NR}}\|\bar{\Gamma}_{\mr{nuc}} - \Gamma_0^G\|_{\mr{F}} + \underbrace{\|\Gamma_0^G\|_{\mr{F}}}_{\leq M\sqrt{TR}} \|\bar{\Lambda}_{\mr{nuc}} - \Lambda_0^G\|_{\mr{F}} \\
            \leq & \underbrace{2M\sqrt{R}  B_1}_{B_2} \sqrt{NT}\left(\varphi_{NT} + \frac{\log(NT)}{\sqrt{\min\{N, T\}}}\right), \quad \text{wpa1}. 
        \end{align*}
        Therefore, the estimation errors of $\bar{\Theta}_{\mr{nuc}}$ and $\hat{\Theta}_{\mr{nuc}}$ differ only by a constant factor.  
        Applying the proof method in Step 3 yields that, wpa1, 
        \begin{align*}
            \frac{1}{\sqrt{N}}\|\tilde{\Lambda}_{\mr{nuc}} - \Lambda_0\|_{\mr{F}}, \frac{1}{\sqrt{T}}\|\tilde{\Gamma}_{\mr{nuc}} - \Gamma_0\|_{\mr{F}} \leq & \underbrace{\left(\frac{16B_2 \Omega_1^{3/2}}{\Omega_R^2} +  \frac{ c_1 R}{\Omega_R^{1/2} }\right)}_{c_2} \left(\varphi_{NT} + \frac{\log(NT)}{\sqrt{\min\{N, T\}}}\right). 
        \end{align*}
        In addition, since $\bar{\Theta}_{\mr{nuc}}$ is the product of two rank-$R$ uniformly bounded matrices and satisfies $ \bar{\Theta}_{\mr{nuc}}= \tilde{\Lambda}_{\mr{nuc}}\tilde{\Gamma}'_{\mr{nuc}}$, each entry in $(\tilde{\Lambda}_{\mr{nuc}}, \tilde{\Gamma}_{\mr{nuc}})$ must be uniformly bounded wpa1. When $\rho_{\lambda}, \rho_{\gamma}$ are sufficiently large (independent of $N, T$), we have $(\tilde{\Lambda}_{\mr{nuc}},  \tilde{\Gamma}_{\mr{nuc}})\in \Phi_{NT}$.   Therefore, we construct new uniformly bounded nuisance estimators $(\tilde{\Lambda}_{\mr{nuc}}, \tilde{\Gamma}_{\mr{nuc}})$ and prove that they achieve the same convergence rate, differing only by a constant. 

        Since constructing a uniformly bounded nuisance estimator is solely for the convenience of theoretical analysis,   we do not distinguish between $(\hat{\Lambda}_{\mr{nuc}}, \hat{\Gamma}_{\mr{nuc}})$ and  $(\tilde{\Lambda}_{\mr{nuc}}, \tilde{\Gamma}_{\mr{nuc}})$ in the rest of the paper, with a slight abuse of notation.  
    \end{proofthm}

    \begin{prooflmm}{lemma:cone}
        The proof is standard in the literature. Since $(\hat{\beta}_{\mr{nuc}}, \hat{\Theta}_{\mr{nuc}})$  solves the nuclear norm regularized optimization problem~\eqref{eq:nnr_definition_formal}, we have 
        \begin{align}\label{eq:lmm_cone_1}
            \mc{L}_{NT}( \beta_0 + \hat{\Delta}_{\beta} ,  \Theta_0 + \hat{\Delta}_{\Theta} ) - \mc{L}_{NT}(\beta_0, \Theta_0)  \leq \frac{\varphi_{NT}}{\sqrt{NT}} (\|\Theta_0\|_{\mr{nuc}} - \|\Theta_0 + \hat{\Delta}_{\Theta}\|_{\mr{nuc}} ). 
        \end{align}
        Consider the first-order Taylor expansion of $\mc{L}_{NT}(\hat{\beta}_{\mr{nuc}}, \hat{\Theta}_{\mr{nuc}}) $ around the true parameter. Since $\mc{L}_{NT}(\cdot, \cdot)$ is convex, we have the following inequality
        \begin{align}\label{eq:lmm_cone_2}
            \mc{L}_{NT}( \beta_0 + \hat{\Delta}_{\beta} ,  \Theta_0 + \hat{\Delta}_{\Theta} ) - \mc{L}_{NT}(\beta_0, \Theta_0)  \geq \nabla_{\beta}\mc{L}_{NT}\left(\beta_0 , \Theta_0\right)'\hat{\Delta} _{\beta}  + \la \nabla_{\Theta}\mc{L}_{NT}\left(\beta_0 , \Theta_0\right), \hat{\Delta}_{\Theta} \ra. 
        \end{align}
        Combining~\eqref{eq:lmm_cone_1} and~\eqref{eq:lmm_cone_2} yields 
        \begin{align*}
            \frac{\varphi_{NT}}{\sqrt{NT}} (\|\Theta_0\|_{\mr{nuc}} - \|\Theta_0 + \hat{\Delta}_{\Theta} \|_{\mr{nuc}} ) - \nabla_{\beta}\mc{L}_{NT}\left(\beta_0 , \Theta_0\right)'\hat{\Delta}_{\beta}  - \la \nabla_{\Theta}\mc{L}_{NT}\left(\beta_0 , \Theta_0\right), \hat{\Delta}_{\Theta } \ra   \geq 0. 
        \end{align*} 
        This implies        
        \begin{align}\label{eq:lmm_cone_3}
            \frac{\varphi_{NT}}{\sqrt{NT}} (\|\Theta_0\|_{\mr{nuc}} - \| \Theta_0 + \hat{\Delta}_{\Theta} \|_{\mr{nuc}} ) +  |\nabla_{\beta}\mc{L}_{NT}\left(\beta_0 , \Theta_0\right)'\hat{\Delta} _{\beta}|  + | \la \nabla_{\Theta}\mc{L}_{NT}\left(\beta_0 , \Theta_0\right), \hat{\Delta}_{\Theta } \ra|   & \geq 0. 
        \end{align}
        The term $|\nabla_{\beta}\mc{L}_{NT}\left(\beta_0 , \Theta_0\right)'\hat{\Delta} _{\beta}|$ is controlled  by 
        \begin{align*}
            |\nabla_{\beta}\mc{L}_{NT}\left(\beta_0 , \Theta_0\right)'\hat{\Delta} _{\beta}| \leqtext{(i)} \|\nabla_{\beta}\mc{L}_{NT}\left(\beta_0 , \Theta_0\right)\| \|\hat{\Delta}_{\beta}\|  \leqtext{(ii)}  \frac{1}{1 + \alpha}\varphi_{NT} \|\hat{\Delta} _{\beta}\|, 
        \end{align*}
        where  inequality (i) follows from the Cauchy-Schwarz inequality, and inequality (ii) holds because of the condition $\varphi_{NT} \geq (1+\alpha)  \|\nabla_{\beta}\mc{L}_{NT}\left(\beta_0 , \Theta_0\right)\|$ as stated in Theorem~\ref{thm:consistency_pre}. Similarly, we can  control  $| \la \nabla_{\Theta}\mc{L}_{NT}\left(\beta_0 , \Theta_0\right), \hat{\Delta}_{\Theta } \ra|$ as follows:  
        \begin{align*}
            | \la \nabla_{\Theta}\mc{L}_{NT}\left(\beta_0 , \Theta_0\right) , \hat{\Delta}_{\Theta } \ra|  \leqtext{(i)} \|\nabla_{\Theta}\mc{L}_{NT}\left(\beta_0 , \Theta_0\right)\|_{\mr{op}} \|\hat{\Delta}_{\Theta}\|_{\mr{nuc}} 
            \leqtext{(ii)}  \frac{1}{1+\alpha}\frac{\varphi_{NT}}{\sqrt{NT}}\|\hat{\Delta}_{\Theta}\|_{\mr{nuc}}. 
        \end{align*}
        where  inequality (i) follows from the Hölder's inequality, as the nuclear norm is the dual norm of the spectral norm, and inequality (ii) holds because of the condition $\varphi_{NT} \geq (1+\alpha)\sqrt{NT} \|\nabla_{\Theta}\mc{L}_{NT}\left(\beta_0 , \Theta_0\right)\|_{\mr{op}}$ as stated in Theorem~\ref{thm:consistency_pre}.
        Thus, inequality~\eqref{eq:lmm_cone_3} can be written as
        \begin{align}\label{eq:lmm_cone_4}
              \left(\|\Theta_0\|_{\mr{nuc}} - \| \Theta_0 + \hat{\Delta}_{\Theta} \|_{\mr{nuc}} \right) + \frac{1}{1+\alpha}   \left(\sqrt{NT}\|\hat{\Delta}_{\beta}\| + \|\hat{\Delta}_{\Theta}\|_{\mr{nuc}}\right) \geq 0. 
        \end{align}
        In addition, we have the following inequalities for the nuclear norm:
        \begin{align*}
            \|\Theta_0 + \hat{\Delta}_{\Theta}\|_{\mr{nuc}} & \eqtext{(i)} \| M_{\Lambda_0 } \Theta_0 M_{\Gamma_0}  +  M_{\Lambda_0 }\hat{\Delta}_{\Theta}M_{\Gamma_0} \| _{\mr{nuc}} + \| \Theta_0 - M_{\Lambda_0 } \Theta_0 M_{\Gamma_0}  + \hat{\Delta}_{\Theta} -  M_{\Lambda_0 }\hat{\Delta}_{\Theta}M_{\Gamma_0} \| _{\mr{nuc}} \\
            & \eqtext{(ii)}  \|   M_{\Lambda_0 }\hat{\Delta}_{\Theta}M_{\Gamma_0} \|_{\mr{nuc}} + \| \Theta_0 + \hat{\Delta}_{\Theta} -  M_{\Lambda_0 }\hat{\Delta}_{\Theta}M_{\Gamma_0} \|_{\mr{nuc}} \\
            & \geqtext{(iii)}  \|   M_{\Lambda_0 }\hat{\Delta}_{\Theta}M_{\Gamma_0} \|_{\mr{nuc}} + \| \Theta_0 \|_{\mr{nuc}} - 
            \|  \hat{\Delta}_{\Theta} -  M_{\Lambda_0 }\hat{\Delta}_{\Theta}M_{\Gamma_0} \|_{\mr{nuc}}. 
        \end{align*}
        Here, equality (i) holds by Lemma~\ref{lemma:property_local_projection}\ref{item:local_projection_1}, equality (ii) follows from the fact that $ M_{\Lambda_0 } \Theta_0 M_{\Gamma_0}  = 0$, and inequality (iii) follows from the triangle inequality.   

        Finally, we combine the nuclear norm inequality and \eqref{eq:lmm_cone_4} to obtain 
        \begin{align*}
            \|  \hat{\Delta}_{\Theta} -  M_{\Lambda_0 }\hat{\Delta}_{\Theta}M_{\Gamma_0} \| _{\mr{nuc}}  - 
            \|  M_{\Lambda_0 }\hat{\Delta}_{\Theta}M_{\Gamma_0} \| _{\mr{nuc}} + \frac{1}{1+\alpha}   \left(\sqrt{NT}\|\hat{\Delta} _{\beta}\| + \|\hat{\Delta}_{\Theta}\|_{\mr{nuc}}\right) \geq 0 \\
            \Rightarrow  \frac{2 + \alpha}{1 + \alpha}\|  \hat{\Delta}_{\Theta} -  M_{\Lambda_0 }\hat{\Delta}_{\Theta}M_{\Gamma_0} \| _{\mr{nuc}}  - 
            \frac{\alpha}{1 + \alpha}\|  M_{\Lambda_0 }\hat{\Delta}_{\Theta}M_{\Gamma_0} \| _{\mr{nuc}} + \frac{1}{1+\alpha}  \sqrt{NT}\|\hat{\Delta} _{\beta}\|  \geq 0 \\
            \frac{2+\alpha}{\alpha}\left(\|  \hat{\Delta}_{\Theta} -  M_{\Lambda_0 }\hat{\Delta}_{\Theta}M_{\Gamma_0} \| _{\mr{nuc}}  
            + \sqrt{NT}\|\hat{\Delta} _{\beta}\| \right)  \geq  \|  M_{\Lambda_0 }\hat{\Delta}_{\Theta}M_{\Gamma_0} \| _{\mr{nuc}}. 
        \end{align*}
        This completes the proof. 
    \end{prooflmm}

\subsection{Proofs of Corollary \ref{corollary:consistency} and Corollary \ref{corollary:consistency_pre}} 

    Since Corollary~\ref{corollary:consistency_pre} extends Corollary~\ref{corollary:consistency} to include predetermined covariates, we provide only the proof of Corollary~\ref{corollary:consistency_pre}, as the proof of Corollary~\ref{corollary:consistency} can be viewed as a special case.

    As stated in Assumption~\ref{assumption:regularity_pre}\ref{item:1_pre}, $\{(Y_{it}, W_{it})\}_{1\leq t\leq T}$ is $\phi$-mixing with a uniformly exponential decay rate across $i$. However, this assumption is stronger than necessary for  the proof of Corollary~\ref{corollary:consistency_pre}. In fact, it can be relaxed to requiring $\{(Y_{it}, W_{it})\}_{1\leq t\leq T}$ is $\alpha$-mixing with a uniformly sufficiently fast polynomial decay rate across $i$. 

    \begin{proofcor}{corollary:consistency_pre}
        It suffices to prove that $$\max\{\|\nabla_{\beta}\mc{L}_{NT}(\beta_0, \Theta_0)\|, \sqrt{NT}\|\nabla_{\Theta}\mc{L}_{NT}(\beta_0, \Theta_0)\|_{\mr{op}} \} = o_p\left(\log(NT)/ \sqrt{\min\{N, T\}}\right). $$ 

       By Assumption~\ref{assumption:regularity_pre}\ref{item:sampling_pre}, $\nabla_{\beta}\mc{L}_{NT}(\beta_0, \Theta_0)$ is the sum of weakly dependent bounded random vectors with zero mean conditional on $(Z, \Lambda_0, \Gamma_0)$. Consequently, we apply \citet[Theorem~1]{kanaya2017convergence} to establish that, wpa1, 
        \begin{align}\label{eq:corollary_pre_1}
            \|\nabla_{\beta}\mc{L}_{NT}(\beta_0, \Theta_0)\| < \log(NT) / \sqrt{NT}. 
        \end{align}
        Note that the $(i, t)$ entry of the matrix $\nabla_{\Theta}\mc{L}_{NT}(\beta_0, \Theta_0) $ is $\dot{\ell}_{it}(X_{it}'\beta_0 + \theta_0)$, and it is straightforward to verify that (i) $\{\dot{\ell}_{it}(X_{it}'\beta_0 + \theta_0)\}_{1\leq i\leq N, 1\leq t\leq T}$ is independent across $i$ and $\phi$-mixing with uniformly exponential decay rate, (ii) $\mb{E}_{Z, \Lambda_0, \Gamma_0}(\dot{\ell}_{it}(X_{it}'\beta_0 + \theta_0)) = 0$ by the first-order condition, and (iii) $\dot{\ell}_{it}(X_{it}'\beta_0 + \theta_0)$ is uniformly bounded across $i, t, N, T$ by Assumption~\ref{assumption:regularity_pre}\ref{item:boundedness_pre} and Assumption~\ref{assumption:regularity_pre}\ref{item:smoothing_pre}. Therefore, employing Lemma~\ref{lemma:independent_entry} gives
        \begin{align*}
            NT \|\nabla_{\Theta}\mc{L}_{NT}(\beta_0, \Theta_0)\|_{\mr{op}} = O_p\left(\log(N + T)\sqrt{\max\{N , T\}}\right). 
        \end{align*}
        Hence,  
        \begin{align}\label{eq:corollary_pre_2}
            \sqrt{NT}\|\nabla_{\Theta}\mc{L}_{NT}(\beta_0, \Theta_0)\| 
            = o_p\left(\log(NT)/\sqrt{\min\{N, T\}}\right). 
        \end{align}
        We then combine \eqref{eq:corollary_pre_1} and \eqref{eq:corollary_pre_2} to complete proof. 
    \end{proofcor}

\subsection{Proof of Lemma~\ref{lemma:sufficient_RSC}}

    Our proof builds on \citet{chernozhukov2019inference} (see Lemma D.3 and Lemma D.4 in their Appendix) and extends their results to accommodate serial correlation. The extension introduces additional technical complexity, particularly in deriving a high-probability upper bound for the empirical process under weak dependence. To address this issue, we apply the concentration inequality of \citet{samson2000concentration} to establish the concentration bound around the expectation of the empirical process. Furthermore, we employ the block method from \citet{yu1994rates}, a new sequence with independent blocks to approximate the original sequence. When the block size is sufficiently large, the dependence between separated blocks becomes negligible, thereby facilitating the theoretical analysis.  
    
    It is worth noting that our proof strategy is not only applicable to models with homogeneous slopes, but with minor modifications, can also be extended to accommodate heterogeneous slopes (e.g., \citet{chernozhukov2019inference}, \citet{ma2022detecting}). We believe that the flexibility of our strategy enhances the applicability of our approach to a broader class of models.

    \begin{prooflmm}{lemma:sufficient_RSC} Recall the definition of the constraints space:
        \begin{gather*}
            \mc{C}_1 = \left\{(\Delta_{\beta}, \Delta_{\Theta})\in (\mb{R}^{d_X}\times \mb{R}^{N\times T})\mid \|M_{\Lambda_0}\Delta_{\Theta}M_{\Gamma_0}\|_{\mr{nuc}} \leq c_0 \left(\sqrt{NT}\|\Delta_{\beta}\| + \|\Delta_{\Theta} - M_{\Lambda_0}\Delta_{\Theta}M_{\Gamma_0}\|_{\mr{nuc}}\right)\right\}, \\
            \mc{C}_2 = \left\{ (\Delta_{\beta}, \Delta_{\Theta})\in (\mb{R}^{d_X}\times \mb{R}^{N\times T})\mid \|\Delta_{\beta}\|^2 + \frac{1}{NT}  \|\Delta_{\Theta}\|_{\mr{F}}^2 \geq  \sqrt{\frac{\log (NT)}{NT}}\right\}. 
        \end{gather*}
        For notational simplicity, let $\mc{C} =  \mc{C}_1\cap \mc{C}_2$. Also, use $\mb{E}_{\mc{V}}(\cdot) = \mb{E}(\cdot\mid \mc{V})$ to denote the conditional expectation, and $\mb{P}_{\mc{V}}(\cdot) = \mb{P}(\cdot\mid \mc{V})$ to denote the conditional probability. 

        \paragraph{Step 1} 
        In this step, we aim to establish a lower bound for $\sum_{i=1}^{N}\sum_{t=1}^{T}\mb{E}_{\mc{V}}(X_{it}' \Delta_{\beta} + \Delta_{\theta_{ it}} )^2$. By Assumption~\ref{assumption:conditional_independence_RSC}\ref{item:conditional_variability_RSC}, there exists a constant $\kappa_0>0$ such that
        \begin{align*}
            \inf_{1\leq i\leq N, 1\leq t\leq T} \sigma_{\min } \left( 
            \begin{pmatrix}
                \mb{E}_{\mc{V}}(X_{it}X_{it}') & \mb{E}_{\mc{V}}(X_{it}) \\
                \mb{E}_{\mc{V}}(X_{it}') & 1
            \end{pmatrix}\right) \geq \kappa_0. 
        \end{align*}
        It follows that
            \begin{equation}\label{eq:lower_bound_expectation}
                \begin{aligned}
                    \sum_{i=1}^{N}\sum_{t=1}^{T}\mb{E}_{\mc{V}}(X_{it}' \Delta_{\beta} + \Delta_{\theta_{it}} )^2 
                    &\geq \sum_{i=1}^{N}\sum_{t=1}^{T}
                    \begin{pmatrix}
                        \Delta'_{\beta} & 
                        \Delta_{\theta_{it}}
                    \end{pmatrix} 
                    \begin{pmatrix}
                        \mb{E}_{\mc{V}}(X_{it}X_{it}') & \mb{E}_{\mc{V}}(X_{it}) \\
                        \mb{E}_{\mc{V}}(X_{it}') & 1
                    \end{pmatrix}
                    \begin{pmatrix}
                        \Delta_{\beta} \\
                        \Delta_{\theta_{it}}
                    \end{pmatrix}\\
                    & \geq \kappa_0\left(NT \|\Delta_{\beta}\|^2 + \|\Delta_{\Theta}\|_{\mr{F}}^2\right). 
                \end{aligned}
            \end{equation} 

        \paragraph{Step 2 (Concentration around the expectation)} 

        For any $\omega>0$, define the constraint set $\mc{N}(\omega)$ as 
        \begin{align*}
            \mc{N}(\omega) := \left\{(\Delta_{\beta},\Delta_{\Theta})\in \mc{C}  \mid \|\Delta_{\beta}\|^2  +  \frac{1}{NT}\|\Delta_{\Theta}\|_{\mr{F}}^2 \leq \omega, \text{ }\|\Delta_{\beta}\|_{\max}\leq 2\rho_{\beta}, \text{ }\|\Delta_{\Theta}\|_{\max} \leq 2 \rho_\theta \right\}. 
        \end{align*}
        and define the empirical process as 
        \begin{align*}
            Z(\omega) = \sup_{(\Delta_{\beta},\Delta_{\Theta})\in \mc{N}(\omega)} \left| \sum_{i=1}^{N}\sum_{t=1}^{T}(X_{it}' \Delta_{\beta} + \Delta_{\theta_{it}} )^2 - \mb{E}_{\mc{V}}\sum_{i=1}^{N}\sum_{t=1}^{T}(X_{it}' \Delta_{\beta} + \Delta_{\theta_{it}} )^2 \right|. 
        \end{align*} 
        It is straightforward to verify that  
        \begin{align*}
            \sup_{(\Delta_{\beta},\Delta_{\Theta})\in \mc{N}(\omega)} |(X_{it}' \Delta_{\beta} + \Delta_{\theta_{it}} )^2- \mb{E}_{\mc{V}}(X_{it}' \Delta_{\beta} + \Delta_{\theta_{it}} )^2| & \leq 2\sup_{(\Delta_{\beta},\Delta_{\Theta})\in \mc{N}(\omega)} (X_{it}' \Delta_{\beta} + \Delta_{\theta_{it}} )^2 \\
            & \leq 4\sup_{\|\Delta_{\beta}\|_{\max}\leq 2\rho_{\beta}, \|\Delta_{\theta}\|_{\max}\leq 2\rho_{\theta} } \left\{(X_{it}' \Delta_{\beta} )^2 + \Delta_{\theta_{it}}^2\right\} \\
            & \leq \underbrace{16(d_X^2 \rho_{X}^2\rho_{\beta}^2 + \rho_{\theta}^2) }_{\sigma}
        \end{align*}
        almost surely. 
                By Assumption~\ref{assumption:conditional_independence_RSC}\ref{item:conditional_weak_dependence_RSC},  conditional on $\mc{V}$, the sequence $\{X_{it}\}_{1\leq t\leq T}$ is $\phi$-mixing with a uniform exponential decay rate. This allows us to apply \citet[Theorem~3]{samson2000concentration} to obtain the following concentration inequality for $Z(\omega)$ around its expectation:
        \begin{align*}
            \mb{P}_{\mc{V}}\left(Z(\omega) \geq \mb{E}_{\mc{V}} Z(\omega) + \delta\right) \leq \exp\left( - L^{-1} \min\left\{\frac{\delta}{\sigma}, \frac{\delta^2}{NT\sigma^2}\right\}\right), \quad \forall \delta >0. 
        \end{align*}
        Here,  $L \in [1, \infty)$ does not depend on $N, T$,  and is only determined by the mixing properties of $X_{it}$ conditional on $\mc{V}$.

        \paragraph{Step 3 (Upper bound for $\mb{E}_{\mc{V}}Z(\omega)$)}
        We follow the block method of \citet{yu1994rates} to derive an upper bound for $\mb{E}_{\mc{V}}Z(\omega)$. Define  $X_i = \left(X_{i1}, X_{i2}, \ldots, X_{iT}\right)$, 
        where $\{X_i\}$ is independent across $i$, and for each $i$, $X_i$ is  $\phi$-mixing with mixing coefficients $\phi(\tau)$. Let  $\tau_{NT} \leq T$ be a positive integer, and $\mu_{NT} = \lfloor \frac{T}{2\tau_{NT}} \rfloor $ denote the largest integer less than or equal to $\frac{T}{2\tau_{NT}}$. For each $X_i$, we divide the sequence into  $2\mu_{NT}$  blocks,  each of length $\tau_{NT}$,   with the remaining part having length at most $2\tau_{NT}$. The subscript notation indicates that the values of $\tau_{NT}$ and $\mu_{NT}$ may depend on $N$ and $T$. 

        We further partition the blocks into two groups: odd-numbered blocks and even-numbered blocks. For notational simplicity, let $\mc{T}^{(0)}_k$ denote the indices of elements in the  $k$-th odd block, and $\mc{T}^{(0)}$ denote the set of indices corresponding to the elements contained in odd-numbered blocks. Similarly, let $\mc{T}^{(1)}_k$ denote the set of indices of elements in the  $k$-th even block, and $\mc{T}^{(1)}$ denote the set of the indices of all elements in even-numbered blocks. Specifically, 
        \begin{align*}
            \mc{T}^{(0)} :=  \bigcup_{k=1}^{\mu_{NT}} \mc{T}^{(0)}_k, \quad  \mc{T}^{(0)}_k = \left\{t \mid 2(k-1)\tau_{NT} + 1\leq  t\leq 2(k-1)\tau_{NT} + \tau_{NT}\right\}, \\
            \mc{T}^{(1)} L=  \bigcup_{k=1}^{\mu_{NT}} \mc{T}^{(1)}_k, \quad  \mc{T}^{(1)}_k = \left\{t \mid (2k-1)\tau_{NT} + 1\leq  t\leq (2k-1)\tau_{NT} + \tau_{NT}\right\}. 
        \end{align*}  
        The corresponding partition of $X_i$ can be written as 
        \begin{gather*}
            X_i^{(0)} := \left(X_i^{(0, 1)},  X_i^{(0, 2)}, 
            \ldots, X_i^{(0, \mu_{NT})}\right), \\
            X_i^{(1)}  := \left( X_i^{(1, 1)} , X_i^{(1, 2)}, 
            \ldots,  X_i^{(1, \mu_{NT})}\right), 
        \end{gather*}
        where for each $1\leq k \leq \mu_{NT}$, 
        \begin{align*}
            X_i^{(0, k)} = \left(X_{it}\mid t\in \mc{T}^{(0)}_k \right),  \quad X_i^{(1, k)} = \left(X_{it}\mid t\in \mc{T}^{(1)}_k \right). 
        \end{align*}
        The remaining terms are collected into $R_i$, where  
        \begin{align*}
            R_i =  \left(X_{it}\mid t \in \mc{R} \right), \quad \mc{R} = \left\{t \mid 2\tau_{NT}\mu_{NT} + 1 \leq t \leq T \right\}
        \end{align*}

        In the next step, for each $i$, we construct a new sequence with independent block structure conditional on $\mc{V}$:
        \begin{align*}
            \widetilde{X}^{(0)}_i = \left(\widetilde{X}_i^{(0, 1)}, \widetilde{X}_i^{(0, 2)}, \ldots, \widetilde{X}_i^{(0, \mu_{NT})} \right), 
        \end{align*}
        such that each block $\widetilde{X}^{(0, k)}_i$ (of size $\tau_{NT}$) is independent of the others conditional on $\mc{V}$. Within each block, $\widetilde{X}^{(0, k)}_i$ follows the same conditional distribution as $X^{(0, k)}_i$. 
        We then construct $\{\widetilde{X}^{(1)}_i\}$ in a  similar way, i.e., 
        \begin{align*}
            \widetilde{X}^{(1)}_i := \left(\widetilde{X}_i^{(1, 1)}, \widetilde{X}_i^{(1, 2)}, \ldots, \widetilde{X}_i^{(1, \mu_{NT})}\right). 
        \end{align*}
        Here, each block is independent with the others conditional on $\mc{V}$. Within each block, $\widetilde{X}^{(1, k)}_i$ follows the same conditional distribution as $X^{(1, k)}_i$. 
        
        Denote $\widetilde{X} = (\widetilde{X}_i^{(0, 1)}, \widetilde{X}_i^{(1, 1)}, \ldots, \widetilde{X}_i^{(0, \mu_{NT})}, \widetilde{X}_i^{(1, \mu_{NT})})$, and for $s \in \{0, 1\}$, define the empirical process
        \begin{align*}
            Z^{(s)}(\omega) & := \sup_{(\Delta_{\beta},\Delta_{\Theta})\in \mc{N}(\omega)} \left| \sum_{i=1}^{N}\sum_{ t\in \mc{T}^{(s)}} (X_{it}' \Delta_{\beta} + \Delta_{\theta_{it}} )^2 - \mb{E}_{\mc{V}}\sum_{i=1}^{N}\sum_{ t\in \mc{T}^{(s)}}(X_{it}' \Delta_{\beta} + \Delta_{\theta_{it}} )^2 \right|,  \\
            \widetilde{Z}^{(s)}(\omega) & := \sup_{(\Delta_{\beta},\Delta_{\Theta})\in \mc{N}(\omega)} \left| \sum_{i=1}^{N}\sum_{ t\in \mc{T}^{(s)}} (\widetilde{X}_{it}' \Delta_{\beta} + \Delta_{\theta_{it}} )^2 - \mb{E}_{\mc{V}}\sum_{i=1}^{N}\sum_{ t\in \mc{T}^{(s)}}(\widetilde{X}_{it}' \Delta_{\beta} + \Delta_{\theta_{it}} )^2 \right|. 
        \end{align*}
        It is straightforward to see that $\mb{E}_{\mc{V}}Z(\omega)$ can be bounded by
        \begin{equation}\label{eq:lemma_RSC_1}
        \begin{aligned}
            \mb{E}_{\mc{V}} Z(\omega) \leq & \mb{E}_{\mc{V}} Z^{(0)}(\omega) + \mb{E}_{\mc{V}} Z^{(1)}(\omega) \\
            \leq & \mb{E}_{\mc{V}} \widetilde{Z}^{(0)}(\omega) +  \mb{E}_{\mc{V}} \widetilde{Z}^{(1)}(\omega) + | \mb{E}_{\mc{V}} \widetilde{Z}^{(0)}(\omega) -  \mb{E}_{\mc{V}}Z^{(0)}(\omega)| + | \mb{E}_{\mc{V}} \widetilde{Z}^{(1)}(\omega) -  \mb{E}_{\mc{V}}Z^{(1)}(\omega)|\\
            & + 2 \sigma N\tau_{NT}. 
        \end{aligned} 
        \end{equation}
        The last term on the right-hand side of the inequality, $2\sigma N\tau_{NT}$,  arises from the remaining terms when $T$ cannot be exactly divided by $\tau_{NT}$. The following lemma is crucial in bounding $| \mb{E}_{\mc{V}} \widetilde{Z}^{(0)}(\omega) -  \mb{E}_{\mc{V}}Z^{(0)}(\omega)|$ and $| \mb{E}_{\mc{V}} \widetilde{Z}^{(1)}(\omega) -  \mb{E}_{\mc{V}}Z^{(1)}(\omega)|$: 
        \begin{lemma}[Based on Lemma 4.1 of \citet{yu1994rates}]\label{lemma:YU-IB}
            For any measurable function $h: \mb{R}^{N \times \tau_{NT}\mu_{NT}} \mapsto \mb{R}$ with bound $M>0$, we have 
            \begin{align*}
                | \mb{E}_{\mc{V}} h(X^{(s)}_{1}, X^{(s)}_{2}, \ldots, X^{(s)}_{N}) - \mb{E}_{\mc{V}} h(\widetilde{X}^{(s)}_{1}, \widetilde{X}^{(s)}_{2}, \ldots, \widetilde{X}^{(s)}_{N}) | \leq M \left(N\mu_{NT} -1 \right) \phi(\tau_{NT}), \quad s = 0,1. 
            \end{align*} 
        \end{lemma}
        To apply Lemma \ref{lemma:YU-IB}, let   
        \begin{align*}
            h^{(s)}(X^{(s)}_{1}, X^{(s)}_{2}, \ldots, X^{(s)}_{N}) : =  \sup_{(\Delta_{\beta},\Delta_{\Theta})\in \mc{N}(\omega)} \left| \sum_{i=1}^{N}\sum_{ t\in \mc{T}^{(s)}} (X_{it}' \Delta_{\beta} + \Delta_{\theta_{it}} )^2 - \mb{E}_{\mc{V}}\sum_{i=1}^{N}\sum_{ t\in \mc{T}^{(s)}}(X_{it}' \Delta_{\beta} + \Delta_{\theta_{it}} )^2 \right| 
        \end{align*}
        It is straightforward to verify that the following inequality holds almost surely for $s = 0,1$: 
        \begin{align*}
            |h^{(s)}(X^{(s)}_{1}, X^{(s)}_{2}, \ldots, X^{(s)}_{N})| \leq & 2 \sup_{(\Delta_{\beta},\Delta_{\Theta})\in \mc{N}(\omega)} \left| \sum_{i=1}^{N}\sum_{ t\in \mc{T}^{(s)}} (X_{it}' \Delta_{\beta} + \Delta_{\theta_{it}} )^2  \right| \\ 
            \leq & 4 \sup_{(\Delta_{\beta},\Delta_{\Theta})\in \mc{N}(\omega)} \left| \sum_{i=1}^{N}\sum_{ t\in \mc{T}^{(s)}} (X_{it}' \Delta_{\beta})^2 + \sum_{i=1}^{N}\sum_{ t\in \mc{T}^{(s)}} \Delta_{\theta_{it}}^2  \right| \\
            \leq & 4 \sup_{(\Delta_{\beta},\Delta_{\Theta})\in \mc{N}(\omega)}   \sum_{i=1}^{N}\sum_{ t\in \mc{T}^{(s)}} (X_{it}' \Delta_{\beta})^2  + 4 \sup_{(\Delta_{\beta},\Delta_{\Theta})\in \mc{N}(\omega)} \sum_{i=1}^{N}\sum_{ t\in \mc{T}^{(s)}} \Delta_{\theta_{it}}^2 \\
            \leq & 4 N \mu_{NT} \tau_{NT}d_X \rho_X^2 \omega  + 4 NT \omega \\ 
            \leq & 2 NT (d_X \rho_X^2  + 2)\omega. 
        \end{align*}
        Thus, applying Lemma \ref{lemma:YU-IB} to $| \mb{E}_{\mc{V}} \widetilde{Z}^{(s)}(\omega) -  \mb{E}_{\mc{V}}Z^{(s)}(\omega)|$ (with $M \leq 2 NT(d_X \rho_X^2  + 2)\omega $) yields
        \begin{equation}\label{eq:lemma_RSC_2}
        \begin{aligned}
            | \mb{E}_{\mc{V}} \widetilde{Z}^{(s)}(\omega) -  \mb{E}_{\mc{V}}Z^{(s)}(\omega)| \leq & 2 NT\left(d_X \rho_X^2  + 2\right)\omega  \left(N\mu_{NT} -1 \right) \phi(\tau_{NT}) \\
            \leq & \underbrace{\left(d_X \rho_X^2  + 2\right)}_{\frac{C_0}{2}} (NT)^2 \frac{\phi(\tau_{NT})}{\tau_{NT}}\omega. 
        \end{aligned}
        \end{equation}
        Intuitively, when $\phi(\tau_{NT})$ decays sufficiently fast as $\tau_{NT}\rightarrow \infty$, we expect that $\phi(\tau_{NT})/ \tau_{NT}\rightarrow 0$ sufficiently fast, so that $| \mb{E}_{\mc{V}} \widetilde{Z}^{(s)}(\omega) -  \mb{E}_{\mc{V}}Z^{(s)}(\omega)|$ can be well controlled. 

        We now turn to establishing a bound for $\mb{E}_{\mc{V}} \widetilde{Z}^{(s)}(\omega)$. Since each block in $\widetilde{X}_i$ is independent with the others conditional on $\mc{V}$, it suffices to study the Rademacher process for each block. More specifically, for each $s$, we can construct a collecion of i.i.d. Rademacher random variables $\{\epsilon_{ik}^{(s)}\mid i = 1,2,\ldots, N, k = 1,2,\ldots, \mu_{NT}\}$, which are independent of $(\widetilde{X}^{(s)}_{1}, \widetilde{X}^{(s)}_{2}, \ldots, \widetilde{X}^{(s)}_{N})$ conditional on $\mc{V}$. Using  symmetrization method, we obtain  
        \begin{equation}\label{eq:lemma_RSC_3}
        \begin{aligned}
            \mb{E}_{\mc{V}} \widetilde{Z}^{(s)}(\omega) & = \mb{E}_{\mc{V}}  \sup_{(\Delta_{\beta},\Delta_{\Theta})\in \mc{N}(\omega)} \left| \sum_{i=1}^{N}\sum_{ t\in \mc{T}^{(s)}} (\widetilde{X}_{it}' \Delta_{\beta} + \Delta_{\theta_{it}} )^2 - \mb{E}_{\mc{V}}\sum_{i=1}^{N}\sum_{ t\in \mc{T}^{(s)}}(\widetilde{X}_{it}' \Delta_{\beta} + \Delta_{\theta_{it}} )^2 \right| \\
            & \leq 2 \mb{E}_{\mc{V}, \epsilon}    \sup_{(\Delta_{\beta},\Delta_{\Theta})\in \mc{N}(\omega)} \left| \sum_{i=1}^{N}\sum_{ k = 1}^{\mu_{NT}} \left(\sum_{t \in \mc{T}^{(s)}_k}(\widetilde{X}_{it}' \Delta_{\beta} + \Delta_{\theta_{it}} )^2\right) \epsilon_{ik}^{(s)}\right|. 
        \end{aligned}
        \end{equation}
        For $s=0$, we have
        \begin{equation}\label{eq:lemma_RSC_4} 
        \begin{aligned}
            & 2 \mb{E}_{\mc{V}, \epsilon}    \sup_{(\Delta_{\beta},\Delta_{\Theta})\in \mc{N}(\omega)} \left| \sum_{i=1}^{N}\sum_{ k = 1}^{\mu_{NT}} \left(\sum_{t \in \mc{T}^{(0)}_k}(\widetilde{X}_{it}' \Delta_{\beta} + \Delta_{\theta_{it}} )^2\right) \epsilon_{ik}^{(0)}\right| \\
            \leqtext{(i)} & 2 \mb{E}_{\mc{V}, \epsilon}    \sup_{(\Delta_{\beta},\Delta_{\Theta})\in \mc{N}(\omega)} \left| \sum_{ \tau = 1}^{\tau_{NT}} \sum_{i=1}^{N} \sum_{k = 1 }^{\mu_{NT}}(\widetilde{X}_{i, 2(k-1)\tau_{NT} + \tau}' \Delta_{\beta} + \Delta_{\theta_{i, 2(k-1)\tau_{NT} + \tau}} )^2 \epsilon_{ik}^{(0)}  \right|  \\
            \leq & 2  \sum_{ \tau = 1}^{\tau_{NT}} \mb{E}_{\mc{V}, \epsilon}  \sup_{(\Delta_{\beta},\Delta_{\Theta})\in \mc{N}(\omega)} \left|\sum_{i=1}^{N} \sum_{k = 1 }^{\mu_{NT}}(\widetilde{X}_{i, 2(k-1)\tau_{NT} + \tau}' \Delta_{\beta} + \Delta_{\theta_{i, 2(k-1)\tau_{NT} + \tau}} )^2 \epsilon_{ik}^{(0)}  \right| \\
            \leqtext{(ii)} & \underbrace{16 (d_X\rho_X\rho_{\beta} + \rho_{\theta})}_{\frac{C_1}{2}}\sum_{ \tau = 1}^{\tau_{NT}}  \mb{E}_{\mc{V}, \epsilon}  \sup_{(\Delta_{\beta},\Delta_{\Theta})\in \mc{N}(\omega)} \left|\sum_{i=1}^{N} \sum_{k = 1 }^{\mu_{NT}}(\widetilde{X}_{i, 2(k-1)\tau_{NT} + \tau}' \Delta_{\beta} + \Delta_{\theta_{i, 2(k-1)\tau_{NT} + \tau}} ) \epsilon_{ik}^{(0)}  \right|  \\
            \leq & \frac{C_1}{2}\sum_{ \tau = 1}^{\tau_{NT}}  \underbrace{\mb{E}_{\mc{V}, \epsilon}  \sup_{(\Delta_{\beta},\Delta_{\Theta})\in \mc{N}(\omega)} \left|\sum_{i=1}^{N} \sum_{k = 1 }^{\mu_{NT}}\widetilde{X}_{i, 2(k-1)\tau_{NT} + \tau}' \Delta_{\beta} \epsilon_{ik}^{(0)}  \right|}_{S_{1, \tau}} \\
            & +\frac{C_1}{2}\sum_{ \tau = 1}^{\tau_{NT}}  \underbrace{\mb{E}_{ \epsilon} \sup_{(\Delta_{\beta},\Delta_{\Theta})\in \mc{N}(\omega)} \left|\sum_{i=1}^{N} \sum_{k = 1 }^{\mu_{NT}}  \Delta_{\theta_{i, 2(k-1)\tau_{NT} + \tau}}\epsilon_{ik}^{(0)}  \right|}_{S_{2, \tau}} . 
        \end{aligned}
        \end{equation}
        Here, we change the order of summation to obtain  inequality (i). Inequality (ii) follows from the contraction property of the Rademacher process (see \citet[Section~4.2]{ledoux2013probability}). To derive an upper bound of $S_{1, \tau}$, for each $\tau$, 
        \begin{equation}\label{eq:lemma_RSC_5}
        \begin{aligned}
            S_{1, \tau} & = \mb{E}_{\mc{V}, \epsilon} \sup_{(\Delta_{\beta},\Delta_{\Theta})\in \mc{N}(\omega)} \left|\sum_{i=1}^{N} \sum_{k = 1 }^{\mu_{NT}}\widetilde{X}_{i, 2(k-1)\tau_{NT} + \tau}' \Delta_{\beta}   \epsilon_{ik}^{(0)}  \right| \\
            & \leqtext{(i)} \mb{E}_{\mc{V}, \epsilon} \sup_{(\Delta_{\beta},\Delta_{\Theta})\in \mc{N}(\omega)}  \left(\left\| \sum_{i=1}^{N} \sum_{k = 1 }^{\mu_{NT}}\widetilde{X}_{i, 2(k-1)\tau_{NT} + \tau}'     \epsilon_{ik}^{(0)}\right\|  \|\Delta_{\beta}\|\right) \\
            & = \mb{E}_{\mc{V}, \epsilon}  \left\| \sum_{i=1}^{N} \sum_{k = 1 }^{\mu_{NT}}\widetilde{X}_{i, 2(k-1)\tau_{NT} + \tau}'    \epsilon_{ik}^{(0)}\right\| \sup_{(\Delta_{\beta},\Delta_{\Theta})\in \mc{N}(\omega)} \|\Delta_{\beta}\|  \\
            & \leqtext{(ii)} \underbrace{\sqrt{2 \pi d_X^3 \rho_X^2 }}_{\sqrt{2}C_2}  \sup_{(\Delta_{\beta},\Delta_{\Theta})\in \mc{N}(\omega)} \|\Delta_{\beta}\|\\
            & \leq \sqrt{2}C_2  \sqrt{N\mu_{NT}}\sqrt{\omega} \\
            & \leq  C_2  \sqrt{\frac{NT\omega}{\tau_{NT}}} , 
        \end{aligned}
        \end{equation}
        where inequality (i) follows from the Cauchy-Schwarz inequality, and inequality (ii) follows from Lemma~\ref{lemma:mul_Hoeffding} using the fact that each element in $\widetilde{X}_{i, 2(k-1)\tau_{NT} + \tau}'    \epsilon_{ik}^{(0)}$ is bounded by $\rho_X$. 

        Let $\Delta^{(0, \tau)}_{\Theta}$ be an $N\times \mu_{NT}$ matrix such that $[\Delta^{(0, \tau)}_{\Theta}]_{ik} = \theta_{i, 2(k-1)\tau_{NT} + \tau}$ for each $\tau =1,2,\ldots, \tau_{NT}$. Let  $E^{(0)}$ denote the $N\times \mu_{NT}$ matrix collecting $\epsilon^{(0)}_{ik}$. Then,  
        \begin{equation}\label{eq:lemma_RSC_6}  
        \begin{aligned}
            S_{2, \tau} =  \mb{E}_{ \epsilon} \sup_{(\Delta_{\beta},\Delta_{\Theta})\in \mc{N}(\omega)} \left|\la \Delta_{\Theta}^{(0, \tau)},  E^{(0)} \ra \right| 
            \leqtext{(i)}   \mb{E}_{\epsilon}  \| E^{(0)} \|_{\mr{op}} \sup_{(\Delta_{\beta},\Delta_{\Theta})\in \mc{N}(\omega)}\|\Delta_{\Theta}^{(0, \tau)} \|_{\mr{nuc}}, 
        \end{aligned}
        \end{equation}
        where inequality (i) comes from the fact that nuclear norm is the duel norm of the operator norm. In addition, by \citet{bandeira2016sharp}, there exists a constant $C_3$ that does not depend on $N$, $T$, or $\tau_{NT}$ such that $\mb{E}_{\epsilon}\| E^{(0)} \|_{\mr{op}} \leq C_3 \sqrt{N + \mu_{NT}}$, based on \cite{bandeira2016sharp}. Furthermore, we have
        \begin{equation}\label{eq:lemma_RSC_7} 
        \begin{aligned}
            \sup_{(\Delta_{\beta},\Delta_{\Theta})\in \mc{N}(\omega)}\|\Delta_{\Theta}^{(0, \tau)} \|_{\mr{nuc}} \leqtext{(i)} & \sup_{(\Delta_{\beta},\Delta_{\Theta})\in \mc{N}(\omega)}\|\Delta_{\Theta}\|_{\mr{nuc}}\\
            \leqtext{(ii)} & \sup_{(\Delta_{\beta},\Delta_{\Theta})\in \mc{N}(\omega)} \left(\|\Delta_{\Theta} - M_{\Lambda_0}\Delta_{\Theta}M_{\Gamma_0}\|_{\mr{nuc}} + \|M_{\Lambda_0}\Delta_{\Theta}M_{\Gamma_0}\|_{\mr{nuc}}\right) \\
            \leqtext{(iii)} &  \sup_{(\Delta_{\beta},\Delta_{\Theta})\in \mc{N}(\omega)}\left( (1 + c_0)\|\Delta_{\Theta} - M_{\Lambda_0}\Delta_{\Theta}M_{\Gamma_0}\|_{\mr{nuc}} + c_0 \sqrt{NT}\|\Delta_{\beta}\|\right) \\
            \leqtext{(iv)} &  \sup_{(\Delta_{\beta},\Delta_{\Theta})\in \mc{N}(\omega)}\left((1+c_0)\sqrt{2R}\|\Delta_{\Theta} - M_{\Lambda_0}\Delta_{\Theta}M_{\Gamma_0}\|_{\mr{F}} + c_0\sqrt{NT}\|\Delta_{\beta}\|\right) \\
            \leqtext{(v)} & \sup_{(\Delta_{\beta},\Delta_{\Theta})\in \mc{N}(\omega)}\left((1+c_0)\sqrt{2R}\|\Delta_{\Theta}\|_{\mr{F}} + c_0 \sqrt{NT}\|\Delta_{\beta}\|\right) \\
            \leqtext{(vi)} & \underbrace{\left((1+c_0)\sqrt{2R} + c_0 \right) }_{C_4}\sqrt{NT\omega} \\
            \leq & C_4 \sqrt{NT\omega }. 
        \end{aligned}
        \end{equation}
        Inequality (i) follows from Lemma~\ref{lemma:SVD_block},  since $\Delta_{\Theta}^{(0, \tau)}$ can be regarded as a submatrix of $\Delta_{\Theta}$. Inequality (ii) is merely an application of triangle inequality, and inequality (iii) comes from $(\Delta_{\beta},\Delta_{\Theta}) \in \mc{C}$. Inequality (iv) holds because $\Delta_{\Theta} - M_{\Lambda_0}\Delta_{\Theta}M_{\Gamma_0}$  is a matrix of rank at most $2R$. Inequality (v) follows from the fact $\|\Delta_{\Theta}\|_{\mr{F}}^2 =  \|\Delta_{\Theta} - M_{\Lambda_0}\Delta_{\Theta}M_{\Gamma_0}\|_{\mr{F}}^2 + \|M_{\Lambda_0}\Delta_{\Theta}M_{\Gamma_0}\|_{\mr{F}}^2$. Finally, inequality (vi) follows directly from $(\Delta_{\beta},\Delta_{\Theta}) \in \mc{N}(\omega)$. 

        The same argument can be applied to the case $s = 1$. Combining equations~\eqref{eq:lemma_RSC_1}---\eqref{eq:lemma_RSC_7}, we derive the following bound:
        \begin{align*}
            \mb{E}_{\mc{V}} \widetilde{Z}^{(s)}(\omega) \leq &\frac{C_1\tau_{NT}}{2}\left(C_2 \sqrt{\frac{NT\omega}{\tau_{NT}}} +  C_3 C_4 \sqrt{NT(N + \mu_{NT})\omega} \right) \\
            \leq & \frac{C_1C_2}{2}\sqrt{NT\tau_{NT}\omega} + \frac{C_1C_3C_4}{2} \sqrt{NT(N+\mu_{NT})\tau_{NT}^2 \omega}, \quad s\in \{0, 1\}. 
        \end{align*}
        Therefore, 
        \begin{align*}
            \mb{E}_{\mc{V}}Z(\omega) \leq &  C_1C_2 \sqrt{NT\tau_{NT}\omega} +  C_1C_3C_4   \sqrt{NT(N+\mu_{NT})\tau_{NT}^2 \omega} + C_0 (NT)^2 \frac{\phi(\tau_{NT})}{\tau_{NT}}\omega  + 2 \sigma N\tau_{NT} \\
            \leq & (C_1C_2 + C_1C_3C_4)\sqrt{NT(N+\mu_{NT})\tau_{NT}^2 \omega} + C_0 (NT)^2 \frac{\phi(\tau_{NT})}{\tau_{NT}} \omega + 2 \sigma N\tau_{NT} \\
            \leq & \frac{\kappa_0}{8} NT\omega  + \left(\frac{8}{\kappa_0} (C_1C_2 + C_1C_3C_4  )^2 + 2\sigma\right)  (N+\mu_{NT})\tau_{NT}^2 + C_0 (NT)^2 \frac{\phi(\tau_{NT})}{\tau_{NT}} \omega. 
        \end{align*}
        When  $\phi(\tau_{NT}) = e^{-\zeta_0 \tau_{NT}}$, let $\tau_{NT} = \frac{2}{\zeta_0}\log (NT)$. Then we have 
        \begin{align*}
            \mb{E}_{\mc{V}}Z(\omega) \leq & \frac{\kappa_0}{8} NT\omega   + \underbrace{\frac{4}{\zeta_0^2}\left( \frac{8}{\kappa_0} (C_1C_2 + C_1C_3C_4)^2  + 2\sigma + C_0\right)}_{ \eta}  (N+T)(\log(NT))^2 \\
            \leq & \frac{\kappa_0}{8} NT\omega   +  \eta  (N+T)(\log(NT))^2. 
        \end{align*}

        Substituting the upper bound for $\mb{E}_{\mc{V}} Z(\omega)$ derived above into the concentration inequality, we obtain
        \begin{align*}
            \mb{P}_{\mc{V}}\left(Z(\omega) \geq \frac{\kappa_0}{8} NT\omega   +  \eta  (N+T)(\log(NT))^2  + \delta\right) \leq \exp\left( -L^{-1} \min\left\{\frac{\delta}{\sigma}, \frac{\delta^2}{NT\sigma^2}\right\}\right), \quad \forall \delta >0. 
        \end{align*}
        Let $\delta = \frac{\kappa_0 }{8} NT\omega$. It follows that 
        \begin{equation}\label{eq:lemma_RSC_8}
        \begin{aligned}
            \mb{P}_{\mc{V}}\left(Z(\omega) \geq \frac{\kappa_0}{4} NT\omega  +  \eta  (N+T)(\log(NT))^2 \right) \leq \exp\left( - \min\left\{\frac{ \kappa_0 NT\omega}{8 L  \sigma}, \frac{\kappa_0^2 NT \omega^2}{64 L \sigma^2}\right\}\right). 
        \end{aligned}
        \end{equation}

    \paragraph{Step 4 (Peeling)}
    For $\ell = 1,2,\ldots$, define 
    \begin{align*}
        \mc{D}_{\ell} := \left\{(\Delta_{\beta},\Delta_{\Theta}) \in \mc{C} \mid  2^{\ell-1} \sqrt{\frac{\log (NT)}{NT}}  \leq \|\Delta_{\beta}\|^2  +  \frac{1}{NT}\|\Delta_{\Theta}\|_{\mr{F}}^2 \leq  2^{\ell} \sqrt{\frac{\log (NT)}{NT}} \right\}
    \end{align*}
    It is readily verified that $\mc{C}  \subset \bigcup_{\ell =1}^{\infty}\mc{D}_{\ell} $. In addition, let $\omega_{\ell} := 2^{\ell}\sqrt{\frac{\log (NT)}{NT}}$. Define  
    \begin{align*}
        \mc{E}_{\ell} := \big\{   Z(\omega_{\ell})\geq  \frac{\kappa_0}{4}NT\omega_{\ell}   +   \eta  (N+T)(\log(NT))^2 \big\}, 
    \end{align*}
    and  
    \begin{align*}
        \tilde{\mc{E}}_{\ell} := \bigg\{
        & \left|\sum_{i=1}^{N}\sum_{t=1}^{T}(X_{it}' \Delta_{\beta} + \Delta_{\theta_{it}} )^2 - \sum_{i=1}^{N}\sum_{t=1}^{T}\mb{E}(X_{it}' \Delta_{\beta} + \Delta_{\theta_{it}} )^2\right | \\
        & \geq \frac{1}{2}\sum_{i=1}^{N}\sum_{t=1}^{T}\mb{E}_{\mc{V}}(X_{it}' \Delta_{\beta} + \Delta_{\theta_{it}} )^2 +  \eta  (N+T)(\log(NT))^2,  \\
        & \exists  (\Delta_{\beta},\Delta_{\Theta}) \in \mc{D}_{\ell}    \bigg\}. 
    \end{align*}
    One can verify that when $\tilde{\mc{E}}_{\ell} $ happens, since 
    \begin{align*}
        \sum_{i=1}^{N}\sum_{t=1}^{T}\mb{E}_{\mc{V}}(X_{it}' \Delta_{\beta} + \Delta_{\theta_{it}} )^2 \geq NT \kappa_0\left( \|\Delta_{\beta}\|^2  +  \frac{1}{NT}\|\Delta_{\Theta}\|_{\mr{F}}^2\right) \geq 2^{\ell-1} \kappa_0     NT \sqrt{\frac{\log (NT)}{NT}}, 
    \end{align*}
    we obtain 
    \begin{align*}
        |\sum_{i=1}^{N}\sum_{t=1}^{T}(X_{it}' \Delta_{\beta} + \Delta_{\theta_{it}} )^2 - \sum_{i=1}^{N}\sum_{t=1}^{T}\mb{E}_{\mc{V}}(X_{it}' \Delta_{\beta} + \Delta_{\theta_{it}} )^2| & \geq \frac{1}{2}\sum_{i=1}^{N}\sum_{t=1}^{T}\mb{E}_{\mc{V}}(X_{it}' \Delta_{\beta} + \Delta_{\theta_{it}} )^2 +  \eta  (N+T)(\log(NT))^2 \\
        \Rightarrow |\sum_{i=1}^{N}\sum_{t=1}^{T}(X_{it}' \Delta_{\beta} + \Delta_{\theta_{it}} )^2 - \sum_{i=1}^{N}\sum_{t=1}^{T}\mb{E}_{\mc{V}}(X_{it}' \Delta_{\beta} + \Delta_{\theta_{it}} )^2| & \geq 2^{\ell-2} \kappa_0     NT \sqrt{\frac{\log NT}{NT}} +  \eta  (N+T)(\log(NT))^2 \\
        \Rightarrow |\sum_{i=1}^{N}\sum_{t=1}^{T}(X_{it}' \Delta_{\beta} + \Delta_{\theta_{it}} )^2 - \sum_{i=1}^{N}\sum_{t=1}^{T}\mb{E}_{\mc{V}}(X_{it}' \Delta_{\beta} + \Delta_{\theta_{it}} )^2| & \geq \frac{\kappa_0}{4}NT\omega_{\ell} + \eta  (N+T)(\log(NT))^2 \\
        \Rightarrow Z(\omega_{\ell}) & \geq \frac{\kappa_0}{4}NT\omega_{\ell} +  \eta  (N+T)(\log(NT))^2. 
    \end{align*}
    Therefore, $\tilde{\mc{E}}_{\ell} \subset \mc{E}_{\ell}$. In addition, 
    \begin{equation}\label{eq:around12}
    \begin{aligned}
        \mb{P}_{\mc{V}}\bigg(& |\sum_{i=1}^{N}\sum_{t=1}^{T}(X_{it}' \Delta_{\beta} + \Delta_{\theta_{it}} )^2 - \sum_{i=1}^{N}\sum_{t=1}^{T}\mb{E}_{\mc{V}}(X_{it}' \Delta_{\beta} + \Delta_{\theta_{it}} )^2|  \\ 
        & \geq  \frac{1}{2}\sum_{i=1}^{N}\sum_{t=1}^{T}\mb{E}(X_{it}' \Delta_{\beta} + \Delta_{\theta_{it}} )^2 +  \eta  (N+T)(\log(NT))^2, \exists  (\Delta_{\beta},\Delta_{\Theta}) \in \mc{C}    \bigg) \\
        \leq & \mb{P}_{\mc{V}}\left(\bigcup_{\ell=1}^{\infty}\tilde{\mc{E}}_{\ell}\right) \leq \sum_{\ell=1}^{\infty}\mb{P}_{\mc{V}}\left(\tilde{\mc{E}}_{\ell}\right) \leq \sum_{\ell=1}^{\infty}\mb{P}_{\mc{V}}\left(\mc{E}_{\ell}\right)  \\
        \leqtext{(i)} &  \sum_{\ell}^{\infty}  \exp\left( - \frac{ \kappa_0 NT\omega_\ell}{8\Phi  \sigma}\right)  +  \sum_{\ell}^{\infty} \exp\left( - \frac{\kappa_0^2 NT \omega_\ell^2}{64\Phi\sigma^2} \right)   \\
        \leq &  \sum_{\ell}^{\infty}  \exp\left( - \frac{ \kappa_0 2^{\ell} \sqrt{NT \log (NT)}}{8\Phi  \sigma}\right)  +  \sum_{\ell}^{\infty} \exp\left( - \frac{\kappa_0^2  4^{\ell}\log (NT)}{64\Phi\sigma^2} \right)   \\
        \rightarrow & 0, 
    \end{aligned}
    \end{equation} 
    where inequality (i) follows from the probability bound~\eqref{eq:lemma_RSC_8}.

    \paragraph{Step 5}
    
    Combining the lower bound~\eqref{eq:lower_bound_expectation} and \eqref{eq:around12}, the following inequality holds for all $(\Delta_{\beta}, \Delta_{\Theta})\in \mc{C}$ with probability approaching $1$:
    \begin{align*}
        |\sum_{i=1}^{N}\sum_{t=1}^{T}(X_{it}' \Delta_{\beta} + \Delta_{\theta_{it}} )^2 - \sum_{i=1}^{N}\sum_{t=1}^{T}\mb{E}(X_{it}' \Delta_{\beta} + \Delta_{\theta_{it}} )^2|  \leq \frac{1}{2}\sum_{i=1}^{N}\sum_{t=1}^{T}\mb{E}(X_{it}' \Delta_{\beta} + \Delta_{\theta_{it}} )^2 +  \eta (N+T)(\log(NT))^2. 
    \end{align*}
    This implies that, wpa1, 
    \begin{align*}
        \Rightarrow \sum_{i=1}^{N}\sum_{t=1}^{T}(X_{it}' \Delta_{\beta} + \Delta_{\theta_{it}} )^2 \geq \frac{1}{2}\sum_{i=1}^{N}\sum_{t=1}^{T}\mb{E}(X_{it}' \Delta_{\beta} + \Delta_{\theta_{it}} )^2 - \eta (N+T)(\log(NT))^2 .
    \end{align*}
    Consequently, 
    \begin{align*}
        \Rightarrow \sum_{i=1}^{N}\sum_{t=1}^{T}(X_{it}' \Delta_{\beta} + \Delta_{\theta_{it}} )^2 \geq \frac{1}{2}\kappa_0\left(NT \|\Delta_{\beta}\|^2 + \vt{\Delta_{\Theta}}_{\mr{F}}^2\right) - \eta (N+T)(\log(NT))^2.  
    \end{align*}
    Since the inequality above holds for any $\mc{V}$, we omit the subscript  $\mc{V}$ and obtain the same bound under $\mb{P}$. 
    Let $\kappa = \frac{1}{2}\kappa_0$. This completes the proof.
    \end{prooflmm}

    \section{Proofs of Local Convexity and Asymptotic Equivalence}

    We begin by clarifying some related issues before proceeding the proofs.

    First, we introduce the following abbreviations to simplify notations.  For any $(\beta, \lambda_i, \gamma_t)$, define  $\dot{\ell}_{it}  = \dot{\ell}(X_{it}'\beta +  \lambda_i' \gamma_t)$ and $\ddot{\ell}_{it}  = \ddot{\ell}(X_{it}'\beta +  \lambda_i' \gamma_t)$. When the log-likelihood is evaluated at the (normalized) true parameters, define $\dot{\ell}_{it}^0  = \dot{\ell}(X_{it}'\beta_0 +  \lambda_{0, i}' \gamma_{0, t})$ and  $ \ddot{\ell}_{it}^0  = \ddot{\ell}(X_{it}'\beta_0 +  \lambda_{0, i}' \gamma_{0, t})$. We further define the following quantities: $\Delta_{\beta} = \beta - \beta_0$, $\Delta_{\gamma_t} = \gamma_t - \gamma^G_{0,t}$, and $\Delta_{\lambda_i} = \lambda_i - \lambda^G_{0,i}$, to denote the deviations of the parameters from their true values. When considering the difference between $\ddot{\ell}_{it}$ and $\ddot{\ell}^0_{it}$, we write 
    \begin{align*}
        \tilde{\Delta}_{Y^*_{it}} = \ddot{\ell}_{it} - \ddot{\ell}_{it}^0 = \tilde{\dddot{\ell}}_{it}(\Delta_{\beta}'X_{it} + \tilde{\Lambda}_i'\Delta_{\gamma_t} + \Delta_{\lambda_i}'\tilde{\gamma_t}), 
    \end{align*}
    which follows from the first-order Taylor expansion. Here, $\tilde{\dddot{\ell}}_{it}$ denotes the third-order derivative of $\ell_{it}$ evaluated at $(\tilde{\beta}, \tilde{\Lambda}, \tilde{\Gamma})$, lying on the line segment between $(\beta, \Lambda, \Gamma)$ and the normalized true parameters 
    $(\beta_0, \Lambda_0^G, \Gamma_0^G)$. Since the parameter space of $(\beta_0, \Lambda_0, \Gamma_0)$ is bounded and all singular values of $G$ are uniformly bounded and strictly positive wpa1, the normalized true parameters 
    $(\beta, \Lambda_0^G, \Gamma_0^G)$ lie within a bounded space wpa1. Thus, 
    $(\tilde{\beta}, \tilde{\Lambda}, \tilde{\Gamma})$  also lie in a bounded space.

    Second, additional effort is required to address the issue that, for some $i$ and $t$, the nuisance parameter estimator may significantly differ from the true nuisance parameter. This challenge arises because the definition of neighborhood $\mc{B}_{\delta_{NT}}$ ~\eqref{eq:neighborhood} only ensures that, for any $(\Lambda, \Gamma)$ in the neighborhood, the distance between $(\Lambda, \Gamma)$ and $(\Lambda_0^G, \Gamma^G_0)$, shrinks in the Frobenius norm, However, it does not guarantee that the nuisance parameter estimates are accurate for each $i$ and $t$. Therefore, for each  $(\beta, \Lambda, \Gamma)\in \mc{B}_{\delta_{NT}}$,  we divide $\Lambda$ into two parts, with subscripts for each part respectively given as
    \begin{equation*}
    \begin{aligned}
        \mc{I}_{NT} = \left\{1\leq i\leq N: \|\lambda_{i} - {\lambda}^G_{0, i} \| \leq \frac{1}{\sqrt{T} \delta_{NT} }  \right\}, \quad \mc{I}^c_{NT} = \{1, \ldots, N\} \backslash \mc{I}_{NT}. 
    \end{aligned}
    \end{equation*}
    Similarly, we divide $\Gamma$ into two parts, with subscripts for each part respectively given as 
    \begin{equation*}
    \begin{aligned}
        \mc{T}_{NT} = \left\{1\leq t\leq T: \|\gamma_{t} - {\gamma}^G_{0, t} \| \leq \frac{1}{\sqrt{N}\delta_{NT}} \right\}, \quad \mc{T}^c_{NT} = \{1, \ldots, T\} \backslash \mc{T}_{NT}. 
    \end{aligned}
    \end{equation*}
    It is straightforward to observe that, when $N\sim T$, for any $i \in \mc{I}_{NT}$, the distance between $\lambda_{i}$ and $\lambda^G_{0, i}$ converges to zero at the rate $T^{-1/8}\log(NT)$. Similarly, for any $t \in \mc{T}_{NT}$, the distance between $\gamma_t$ and $\gamma^G_{0, t}$ also converges to zero as $N, T\rightarrow\infty$ at the rate $N^{-1/8}\log(NT)$. Note that the notations of  $\mc{I}_{NT}$ and $\mc{T}_{NT}$ are slightly imprecise, because for different $(\beta, \Lambda, \Gamma)\in \mc{B}_{\delta_{NT}}$, the corresponding index sets may differ. Nevertheless, following from the Frobenius norm convergence rate (Theorem~\ref{thm:consistency} and Theorem~\ref{thm:consistency_pre}), we conclude that, with probability approaching $1$, for all $(\beta, \Lambda, \Gamma)\in \mc{B}_{\delta_{NT}}$, the sizes of $\mc{I}^c_{NT}$ and $\mc{T}^c_{NT}$ are uniformly bounded by 
    \begin{align}\label{eq:size_division}
        |\mc{I}^c_{NT}| \lesssim NT \delta_{NT}^4, \quad |\mc{T}^c_{NT}| \lesssim NT \delta_{NT}^4. 
    \end{align}
    Therefore, we continue to use this imprecise notation in the following text for simpler notations.
    In addition, our construction~\eqref{eq:size_division} implies that the sizes of $\mc{I}^c_{NT}$ and $\mc{T}^c_{NT}$  are at most of order $\sqrt{N}(\log(NT))^4$ when $N\sim T$.  Hence,  the number of “non-convergent” nuisance estimates grows at a much slower rate  relative to $N$ and $T$. 

    Third, for notational simplicity, we rescale the Hessian matrix $\mc{H}_{NT}$ as $\mc{H} = NT\mc{H}_{NT}$ and suppress its dependence on $(\beta, \Lambda, \Gamma)$ when it does not cause ambiguity. Since $\mc{H}$ differs from $\mc{H}_{NT}$ only by a factor $NT$, we focus on studying the property of $\mc{H}$ in the following text.  Furthermore, we define 
    \begin{align*}
        H = 
        \begin{pmatrix}
            H_{\beta\beta'} & H_{\beta\lambda'} & H_{\beta\gamma'} \\
            H_{\lambda\beta'} & H_{\lambda\lambda'} & H_{\lambda\gamma'} \\
            H_{\gamma \beta'} & H_{\gamma\lambda'} & H_{\gamma\gamma'}
        \end{pmatrix}, \quad 
        F = 
        \begin{pmatrix}
            0 & 0 & 0 \\
            0 & 0 & F_{\lambda\gamma'} \\
            0 & F_{\gamma \lambda'} & 0
        \end{pmatrix} \quad 
        V = 
        \begin{pmatrix}
            0 & 0 & 0 \\
            0 & V_{\lambda\lambda'} & V_{\lambda\gamma'} \\
            0 & V_{\gamma\lambda'} & V_{\gamma\gamma'}
        \end{pmatrix}. 
    \end{align*}
    Here, $H + F$ is the Hessian of the negative log-likelihood function $\mc{L}_{NT}$, and $V$ is Hessian  of the penalty term. 
    Specifically, 
    \begin{equation*}
        \begin{aligned}
            H_{\beta\beta'} & = -\sum_{i=1}^{N}\sum_{t = 1}^{T} \ddot{\ell}_{it} X_{it} X'_{it}, \\
            H_{\beta\lambda'} & = -\left[\sum_{t=1}^{T} \ddot{\ell}_{it} X_{it} \gamma'_t \right]_{i = 1,2,\ldots, N}, \\
            H_{\beta\gamma'} & = -\left[\sum_{i=1}^{N} \ddot{\ell}_{it}X_{it}\lambda'_i\right]_{t = 1,2,\ldots, T},  \\
            H_{\lambda\lambda'} & = -\mr{diag}\left\{\left[\sum_{t=1}^{T}  \ddot{\ell}_{it}\gamma_t  \gamma_t'\right]_{i = 1,2,\ldots, N}\right\},  \\
            H_{\gamma\gamma'} & = -\mr{diag}\left\{\left[\sum_{i=1}^{N}  \ddot{\ell}_{it}\lambda_i  \lambda_i'\right]_{t= 1,2,\ldots, T}\right\},   \\
            H_{\lambda\gamma'} & =  \left[-\ddot{\ell}_{it}\gamma_t\lambda_i'\right]_{i = 1,2,\ldots, N, t = 1,2,\ldots, T},  \\
            V_{\lambda\lambda'} & = \frac{T}{N}\left[ \lambda_i \lambda_{i'}' \right]_{i, i' = 1,2,\ldots, N},  \\
            V_{\lambda\gamma'} & = \left[ -\lambda_i \gamma_{t}' \right]_{i =1,2,\ldots, N, t = 1,2,\ldots, T},  \\
            V_{\gamma\gamma'} & = \frac{N}{T}\left[ \gamma_t   \gamma_{t'}'\right]_{t, t' = 1,2,\ldots, T} , \\
            F_{\lambda\gamma'} & = \left[-\dot{\ell}_{it} \mb{I}_{R}\right]_{i = 1,2,\ldots, N, t = 1,2,\ldots, T}. 
        \end{aligned}
    \end{equation*}
    Here, $H_{\beta\beta'}$ is a $d_X\times d_X$ matrix, $H_{\beta\lambda'}$ is a $d_X\times NR$ matrix consisting of $N$ blocks, each of size $d_X\times R$, and $H_{\beta\gamma'}$ is a $d_X\times TR$ matrix with $T$ blocks, each of size $d_X\times R$. $H_{\lambda\lambda'}, V_{\lambda\lambda'}$ are  $NR\times NR$ matrices consisting of $N^2$ blocks, each of size $R\times R$, $H_{\gamma\gamma'}, V_{\gamma\gamma'}$ are b  $TR\times TR$ matrices with $T^2$ blocks, each of size $R\times R$. In addition, $H_{\lambda\gamma'}, V_{\lambda\gamma'}, F_{\lambda\gamma'}$ are $NR\times TR$ matrices consisting of $NT$ blocks, each with size $R\times R$.  
    
    We use  $\hat{V}$  to denote the matrix  $V$  evaluated at $\hat{\lambda}_{\mr{nuc},i}$ and $\hat{\gamma}_{\mr{nuc},t}$. Specifically, 
    \begin{align*}
        \hat{V}_{\lambda\lambda'} & = \frac{T}{N}\left[ \hat{\lambda}_{\mr{nuc},i} \hat{\lambda}_{\mr{nuc},i'}' \right]_{i, i' = 1,2,\ldots, N\ldots},  \\
        \hat{V}_{\lambda\gamma'} & = \left[ -\hat{\lambda}_{\mr{nuc},i}  \hat{\gamma}_{\mr{nuc},t}' \right]_{i =1,2,\ldots, N, t = 1,2,\ldots, T},  \\
        \hat{V}_{\gamma\gamma'} & = \frac{N}{T}\left[ \hat{\gamma}_{\mr{nuc},t} \hat{\gamma}_{\mr{nuc},t'}'\right]_{t, t' = 1,2,\ldots, T}.  
    \end{align*}
    
    In addition, we use $H_0, F_0, V_0$ when the matrices $  H, F, V$ are evaluated at the normalized true value $(\beta_0, \Lambda_0^G, \Gamma_0^G)$, for example, $H_0 = H(\beta_0, \Lambda^G_0, \Gamma^G_0)$. We also use $\mb{E}_0(\cdot)$ to denote the conditional expectation $\mb{E}_{X, \Lambda_0, \Gamma_0}(\cdot)$ when $X$ is strictly exogenous, or the conditional expectation $\mb{E}_{Z, \Lambda_0, \Gamma_0}(\cdot)$ when we consider predetermined covariates. It is easy to verify that
    \begin{align}\label{eq:VF}
         \mb{E}_{0} F_0 = 0. 
    \end{align}

    By standard calculus,  $\mc{H} \in \mb{R}^{(d_X + R(N+T)) \times (d_X + R(N+T))}$ admits the following decomposition: 
    \begin{align}\label{eq:decomposition}
        \mc{H} = H + F +  \hat{V} . 
    \end{align}
    We  further decompose $H$ as $H = \tilde{H} +\tilde{H}^c $ based on partitions $\mc{I}_{NT}$  and $\mc{T}_{NT}$ such that 
    \begin{equation*}
    \begin{aligned}
        \tilde{H}_{\beta\beta'} & = \sum_{i \in \mc{I}_{NT}, t \in \mc{T}_{NT}} (-\ddot{\ell}_{it}) X_{it} X'_{it},   \\
        \tilde{H}_{\beta\lambda'} & = -\left[\bs{1}(i\in \mc{I}_{NT})\sum_{t\in \mc{T}_{NT}} \ddot{\ell}_{it} X_{it} \gamma'_t \right]_{i = 1,2,\ldots, N}, \\
        \tilde{H}_{\beta\gamma'} & = -\left[\bs{1}(t\in \mc{T}_{NT})\sum_{i\in \mc{I}_{NT}} \ddot{\ell}_{it}X_{it}\lambda'_i\right]_{t = 1,2,\ldots, T},  \\
        \tilde{H}_{\lambda\lambda'} & = -\mr{diag}\left\{\left[\bs{1}(i\in \mc{I}_{NT})\sum_{t\in \mc{T}_{NT}}  \ddot{\ell}_{it}\gamma_t  \gamma_t'\right]_{i = 1,2,\ldots, N}\right\},  \\
        \tilde{H}_{\gamma\gamma'} & = -\mr{diag}\left\{\left[\bs{1}(t\in \mc{T}_{NT})\sum_{i\in \mc{I}_{NT}}  \ddot{\ell}_{it}\lambda_i  \lambda_i'\right]_{t= 1,2,\ldots, T}\right\},   \\
        \tilde{H}_{\lambda\gamma'} & =  -\left[\bs{1}(i \in \mc{I}_{NT}, t\in \mc{T}_{NT})\ddot{\ell}_{it}\gamma_t\lambda_i'\right]_{i = 1,2,\ldots, N, t = 1,2,\ldots, T} . 
    \end{aligned}
    \end{equation*}
    The matrix $V$ can be decomposed in the same way, $V = \tilde{V} + \tilde{V}^c $,  where
    \begin{align*}
        \tilde{V}_{\lambda\lambda'} & = \left[\frac{T}{N}\bs{1}(i, i'\in \mc{I}_{NT})\left(\lambda_{i} \lambda_{i'}' \right)\right]_{i, i' = 1,2,\ldots, N},   \\
        \tilde{V}_{\gamma\gamma'} & = \left[\frac{N}{T}\bs{1}(t, t' \in \mc{T}_{NT})\left(\gamma_{ t} \gamma_{ t'}' \right)\right]_{t, t' = 1,2,\ldots, T} , \\
        \tilde{V}_{\lambda\gamma'} & = \left[-\bs{1}(i\in \mc{I}_{NT}, t\in \mc{T}_{NT}) \lambda_{i}  \gamma_{t}'\right]_{i =1,2,\ldots, N, t = 1,2,\ldots, T}. 
    \end{align*}

\subsection{Proof of Theorem~\ref{thm:convexity_strong} and Theorem~\ref{thm:convexity_strong_pre}}

    As Theorem~\ref{thm:convexity_strong_pre} generalizes Theorem~\ref{thm:convexity_strong} to include predetermined covariates, we provide only the proof of the former, noting that the latter follows as a special case.

    \begin{proofthm}{thm:convexity_strong_pre}\label{proof:thm:convexity_pre}
        It suffices to study the properties of $\mc{H}$ on $\mc{B}_{\delta_{NT}}$. Consider the following decomposition: 
        \begin{equation}\label{eq:Hessian_decomposition} 
        \begin{aligned}
            \mc{H} = &
            \underbrace{ 
            \begin{pmatrix}
                \mb{E}_{0} H_{0, \beta\beta'} & \mb{E}_{0} \tilde{H}_{0, \beta\lambda'} & \mb{E}_{0} \tilde{H}_{0, \beta\gamma'}  \\
                \mb{E}_{0} \tilde{H}_{0, \lambda\beta'} & \mb{E}_0 \tilde{H}_{0, \lambda\lambda'} & \mb{E}_0 \tilde{H}_{0, \lambda\gamma'} \\
                \mb{E}_{0} \tilde{H}_{0, \gamma\beta'} & \mb{E}_0 \tilde{H}_{0, \gamma\lambda'} & \mb{E}_0 \tilde{H}_{0, \gamma\gamma'}
            \end{pmatrix}
            + 
            \begin{pmatrix}
                0 & 0 & 0  \\
                0 &  \tilde{H}^c_{ \lambda\lambda'} & 0 \\
                0 &  0  &   \tilde{H}^c_{ \gamma\gamma'}
            \end{pmatrix} 
            + \mb{E}_0 \hat{V}
            }_{S_1 }
            \\
            & 
            + 
            \underbrace{
            \begin{pmatrix}
                H_{\beta\beta'} - \mb{E}_{0} H_{0, \beta\beta'} & H_{\beta\lambda'} - \mb{E}_{0} \tilde{H}_{0, \beta\lambda'} & H_{ \beta\gamma'} - \mb{E}_{0} \tilde{H}_{0, \beta\gamma'}  \\
                H_{\lambda\beta'} - \mb{E}_{0} \tilde{H}_{ 0, \lambda\beta'} & 0 & 0 \\
                H_{\gamma\beta'}  - \mb{E}_{0} \tilde{H}_{0, \gamma\beta'} & 0 & 0
            \end{pmatrix}
            }_{S_2}
             \\
            & + 
            \underbrace{
            \begin{pmatrix}
                0 & 0 & 0  \\
                0 & \tilde{H}_{\lambda\lambda'} - \mb{E}_0 \tilde{H}_{0, \lambda\lambda'} & H_{ \lambda\gamma'} - \mb{E}_0 \tilde{H}_{0, \lambda\gamma'} \\
                0 & H_{ \gamma\lambda'} - \mb{E}_0 \tilde{H}_{0, \gamma\lambda'} & \tilde{H}_{\gamma\gamma'} - \mb{E}_0 \tilde{H}_{0, \gamma\gamma'}
            \end{pmatrix} 
            }_{S_3}
            +
            (\hat{V} - \mb{E}_0 \hat{V}) 
            + 
            F. 
        \end{aligned}
        \end{equation}
        Let us first examine the first term. It is straightforward to verify that $S_1$  admits the following decomposition:
        \begin{equation}\label{eq:thm_convexity_1}
            \begin{aligned}
            S_1 & =
            \mb{E}_0\tilde{H}_{0} 
            + 
            \mb{E}_0 \tilde{\hat{V}}
            +  
            \underbrace{
            \
            \begin{pmatrix}
                \mb{E}_0 \tilde{H}^c_{\beta\beta'} & 0 & 0 \\
                0 &   0&  0\\
                0 & 0  &  0
            \end{pmatrix} 
            }_{\geq 0}
            +
            \begin{pmatrix}
                0 & 0 & 0 \\
                0 &   \tilde{H}^c_{\lambda\lambda'} &  0\\
               0 & 0  &  \tilde{H}^c_{\gamma\gamma'}
            \end{pmatrix}
            +
            \mb{E}_0 \tilde{\hat{V}}^c 
            \\
            & \geq 
            \mb{E}_0\tilde{H}_{0}  + \mb{E}_0 \tilde{\hat{V}} + 
            \begin{pmatrix}
                0 & 0 & 0 \\
                0 &   \tilde{H}^c_{\lambda\lambda'} &  0\\
               0 & 0  &  \tilde{H}^c_{\gamma\gamma'}
            \end{pmatrix}
            +
            \mb{E}_0 \tilde{\hat{V}}^c. 
        \end{aligned}
        \end{equation}
        Since matrix  $\mb{E}_0\tilde{H}_{0}  + \mb{E}_0 \tilde{\hat{V}}$  can be regarded as the population Hessian matrix indexed by  $i \in \mc{I}_{NT}$  and  $t \in \mc{T}_{NT}$, we obtain 
        \begin{equation}\label{eq:thm_convexity_2}
            \begin{aligned}
                &\mb{E}_0\tilde{H}_{0}  + \mb{E}_0 \tilde{\hat{V}} \\ 
                \geqtext{(i)}   & C  
                \begin{pmatrix}
                    \left(NT - |\mc{T}^c_{NT}| N -  |\mc{I}^c_{NT}| T\right) \mb{I}_{d_X} & 0 & 0  \\
                    0 & \left(T - |\mc{T}^c_{NT}|\right)\mr{diag}\{D_1, \ldots, D_N\} & 0 \\
                    0 & 0 & \left(N - |\mc{I}^c_{NT}|\right)\mr{diag}\{D_{N+1}, \ldots, D_{N+T}\}   
                \end{pmatrix} \\
                \geqtext{(ii)} & \frac{1}{2} C 
                \begin{pmatrix}
                    NT \mb{I}_{d_X} & 0 & 0  \\
                    0 & T\mr{diag}\{D_1, \ldots, D_N\} & 0 \\
                    0 & 0 & N\mr{diag}\{D_{N+1}, \ldots, D_{N+T}\}  
                \end{pmatrix}, 
            \end{aligned} 
        \end{equation} 
        where $D_i = 1_{\{i \in \mc{I}_{NT}\}} \mb{I}_{R}$ for any $i=1,2,\ldots, N$, and  
        $ D_{N+t} = 1_{\{t \in \mc{T}_{NT}\}} \mb{I}_{R}$ for any $t=1,2,\ldots, T$. Inequality (i) follows from Assumption~\ref{assumption:block_pre}, and inequality (ii) follows from the asymptotic assumption on $N$ and $T$. In addition, by Lemma~\ref{lemma:bound_Hc}, 
        we conclude that, there exists a constant $B_1>0$ such that, wpa1,  
        \begin{equation}\label{eq:thm_convexity_3}
        \begin{aligned}
            \begin{pmatrix}
                0 & 0 & 0 \\
                0 &   \tilde{H}^c_{\lambda\lambda'} &  0\\
                0 & 0  &  \tilde{H}^c_{\gamma\gamma'}
            \end{pmatrix}
            \geq B_1
            \min\{N, T\} \mr{diag}\left\{0_{d_X}, \mb{I}_R - D_1, \ldots, \mb{I}_R - D_N, \mb{I}_R  - D_{N+1}, \ldots, \mb{I}_R  -D_{N+T}\right\}. 
        \end{aligned}
        \end{equation}
        Combining equations~\eqref{eq:thm_convexity_1}, ~\eqref{eq:thm_convexity_2}, \eqref{eq:thm_convexity_3}, and $\|\mb{E}_0 \tilde{\hat{V}}^c \|_{\mr{op}} = o_p\left(\min\{N, T\}\right)$\footnote{
            Because by~\eqref{eq:size_division}, 
            \begin{align*}
                \|\tilde{\hat{V}}^c \|_{\mr{op}}^2 \leq \|\tilde{\hat{V}}^c\|_{\mr{F}}^2  = O_p\left((|\mc{I}_{NT}^c| + |\mc{I}_{NT}^c|)(N+T)\right) = O_p(NT\max\{N, T\}\delta_{NT}^4). 
            \end{align*}
            Then, we have $\|\tilde{\hat{V}}^c \|_{\mr{op}} = o_p\left(\min\{N, T\}\right)$ and $\|\mb{E}_0 \tilde{\hat{V}}^c \|_{\mr{op}} = o_p\left(\min\{N, T\}\right)$. 
        }, we  conclude that for any $(\beta, \Lambda, \Gamma)\in \mc{B}_{\delta_{NT}}$, $S_1$ is locally convex wpa1. In addition, let $B_2 = \min\{\frac{1}{2}C, B_1\}$ irrelevant with $N, T$, for any $(\beta, \Lambda, \Gamma)\in \mc{B}_{\delta_{NT}}$, $S_1$ admits asymptotic block structure, i.e.,  
        \begin{equation}\label{eq:thm_convexity_4}
        \begin{aligned}
            S_1 \geq B_2 
            \begin{pmatrix}
                NT\mb{I}_{d_X} & 0 & 0 \\
                0 & T\mb{I}_{NR} & 0 \\
                0 & 0 & N\mb{I}_{TR} 
            \end{pmatrix}, 
            \quad 
            \text{wpa1}. 
        \end{aligned}
        \end{equation}
        Therefore, the optimization problem is strongly convex wpa1 uniformly within $\mc{B}_{\delta_{NT}}$ if the impact (or equivalently, the maximum singular values) of the permutation terms,  $S_2$,  $S_3$, $\hat{V}-\mb{E}_0\hat{V}$, and $F$, is asymptotically negligible (within $\mc{B}_{\delta_{NT}}$) compared to $S_1$ wpa1. Using Lemma~\ref{lemma:bound_S_2}, we prove that  $\frac{1}{2}S_1 + S_2$ is positive definite. In addition, by Lemma~\ref{lemma:bound_S_3} and Lemma~\ref{lemma:bound_F}, we have 
        \begin{align}\label{eq:thm_convexity_6}
            \sup_{(\beta, \Lambda, \Gamma) \in \mc{B}_{\delta_{NT}}}  \|S_3\|_{\mr{op}} = o_p(\min\{N, T\})  , \text{ }  \sup_{(\beta, \Lambda, \Gamma) \in \mc{B}_{\delta_{NT}}}\|F\|_{\mr{op}} = o_p(\min\{N, T\}).   
        \end{align} 
        Using Lemma \ref{lemma:bound_V}, and noting that $(\hat{\Lambda}_{\mr{nuc}}, \hat{\Gamma}_{\mr{nuc}}) \in \mc{B}_{\delta_{NT}}$, we have
        \begin{align}\label{eq:thm_convexity_7}
            \|\hat{V}-\mb{E}_0\hat{V}\|_{\mr{op}} \lesssim  \sup_{(\beta, \Lambda, \Gamma) \in \mc{B}_{\delta_{NT}}} \|V - V_0\|_{\mr{op}} = o_p(\min\{N, T\}). 
        \end{align}
        (We defer the relevant lemmas and their proofs to the end of this subsection). Therefore, combining~\eqref{eq:thm_convexity_4}---\eqref{eq:thm_convexity_7} and applying Weyl's theorem, we conclude that, with probability approaching $1$,  for any $(\beta, \Lambda, \Gamma)\in \mc{B}_{\delta_{NT}}$,   
        \begin{align}
            \mc{H}  \geq \frac{1}{3} B_2
            \begin{pmatrix}
                NT\mb{I}_{d_X} & 0 & 0 \\
                0 & T\mb{I}_{NR} & 0 \\
                0 & 0 & N\mb{I}_{TR} 
            \end{pmatrix}. 
        \end{align}
        Letting $c_5 = \frac{1}{3} B_2$ completes the proof. 
    \end{proofthm}

    \begin{lemma}\label{lemma:bound_Hc}
        Under the assumptions stated in Theorem~\ref{thm:convexity_strong_pre}, there exists a constant $B_1>0$ independent of $N, T$ such that 
        \begin{align*}
            \inf_{(\beta, \Lambda, \Gamma) \in \mc{B}_{\delta_{NT}}}
            \begin{pmatrix}
                0 & 0 & 0 \\
                0 &   H_{\lambda\lambda'} &  0\\
                0 & 0  &  H_{\gamma\gamma'}
            \end{pmatrix}
            \geq B_1
            \begin{pmatrix}
                0_{d_X} & 0 & 0 \\
                0 & T\mb{I}_{NR} & 0 \\
                0 & 0 & N\mb{I}_{TR} 
            \end{pmatrix}, 
            \quad \text{wpa1}. 
        \end{align*}
    \end{lemma}
    \begin{prooflmm}{lemma:bound_Hc}
        Recall 
        \begin{align*}
            H_{\lambda\lambda'} & = -\mr{diag}\left\{\left[\sum_{t=1}^{T}  \ddot{\ell}_{it}\gamma_t  \gamma_t'\right]_{i = 1,2,\ldots, N}\right\},   \\
            H_{\gamma\gamma'} & = -\mr{diag}\left\{\left[\sum_{i=1}^{N}  \ddot{\ell}_{it}\lambda_i  \lambda_i'\right]_{t= 1,2,\ldots, T}\right\} . 
        \end{align*}
        For any $i$-th diagonal block of $H_{\lambda\lambda'}$, with probability approaching $1$, we have 
        \begin{align*}
            \inf_{(\beta, \Lambda, \Gamma) \in \mc{B}_{\delta_{NT}}} \sigma_{\min}\left(\sum_{t=1}^{T}  -\ddot{\ell}_{it}\gamma_t  \gamma_t' \right)  & \gtrsimtext{(i)} \inf_{(\beta, \Lambda, \Gamma) \in \mc{B}_{\delta_{NT}}} \sigma_{\min}\left(\sum_{t=1}^{T}   \gamma_t  \gamma_t' \right)  \\
            & \gtrsim  \inf_{(\beta, \Lambda, \Gamma) \in \mc{B}_{\delta_{NT}}}  \sigma_{\min}\left(\Gamma_0^{G\prime }\Gamma_0^G \right) -  \sup_{(\beta, \Lambda, \Gamma) \in \mc{B}_{\delta_{NT}}} \left\|\Gamma'\Gamma - \Gamma_0^{G\prime }\Gamma_0^G \right\|_{\mr{op}} \\
            & \gtrsimtext{(ii)} T - \sup_{(\beta, \Lambda, \Gamma) \in \mc{B}_{\delta_{NT}}} \left\|\Gamma'\Gamma - \Gamma_0^{G\prime }\Gamma_0^G \right\|_{\mr{F}} \\
            &  \gtrsimtext{(iii)} T - \sup_{(\beta, \Lambda, \Gamma) \in \mc{B}_{\delta_{NT}}} \|\Gamma\|_{\mr{F}}\left\|\Gamma - \Gamma_0^G \right\|_{\mr{F}} -  \sup_{(\beta, \Lambda, \Gamma) \in \mc{B}_{\delta_{NT}}} \|\Gamma_0^G\|_{\mr{F}}\left\|\Gamma - \Gamma_0^G \right\|_{\mr{F}} \\
            &  \gtrsimtext{(iv)} T - \sqrt{T} \sup_{(\beta, \Lambda, \Gamma) \in \mc{B}_{\delta_{NT}}}   \left\|\Gamma - \Gamma_0^G \right\|_{\mr{F}} \\
            & \gtrsimtext{(v)} T - T \delta_{NT} \\
            & \gtrsim T, 
        \end{align*}
        where inequality (i) uses Assumption~\ref{assumption:regularity_pre}\ref{item:smoothing_pre},  inequality (ii) follows from  Assumption~\ref{assumption:regularity_pre}\ref{item:strong_factors_pre}, inequality (iii) follows from the Cauchy-Schwarz inequality, inequality (iv) is based on the uniform boundedness of $\Gamma, \Gamma_0^G$, and inequality (v) follows from Theorem~\ref{thm:consistency_pre}. By the same argument, we also obtain 
        \begin{align*}
            \inf_{(\beta, \Lambda, \Gamma) \in \mc{B}_{\delta_{NT}}} \sigma_{\min}\left(\sum_{i=1}^{N}  -\ddot{\ell}_{it}\lambda_i  \lambda_i' \right) \gtrsim N, \quad \text{wpa1}. 
        \end{align*}
        Therefore, we are able to find a constant $B_1>0$,  independent of $N, T$,  and the proof is complete. 
    \end{prooflmm}

    \begin{lemma}\label{lemma:bound_Hbb}
        Under the assumptions of Theorem \ref{thm:convexity_strong_pre}, we have
        \begin{align*}
            \sup_{(\beta, \Lambda, \Gamma) \in \mc{B}_{\delta_{NT}}}  \left\|H_{\beta\beta'} - \mb{E}_0H_{0, \beta\beta'}\right\|_{\mr{op}} = O_p(NT\delta_{NT}). 
       \end{align*}
    \end{lemma}
    \begin{prooflmm}{lemma:bound_Hbb}
        Note that
        \begin{equation}\label{eq:lemma:bound_Hbb_1}
        \begin{aligned}
            H_{\beta\beta'} - \mb{E}_0H_{0, \beta\beta'} & = -\sum_{i=1}^{N}\sum_{t=1}^{T} (\ddot{\ell}_{it}X_{it}X_{it}') + \sum_{i=1}^{N}\sum_{t=1}^{T} \mb{E}_0 (\ddot{\ell}_{it}^0 X_{it}X_{it}') \\
            & = -\sum_{i=1}^{N}\sum_{t=1}^{T}(\ddot{\ell}_{it}X_{it}X_{it}' - \ddot{\ell}_{it}^0 X_{it}X_{it}') -\sum_{i=1}^{N}\sum_{t=1}^{T}(\ddot{\ell}^0_{it}X_{it}X_{it}' - \mb{E}_0 (\ddot{\ell}_{it}^0 X_{it}X_{it}')). 
        \end{aligned} 
        \end{equation}
        We have 
        \begin{equation}\label{eq:lemma:bound_Hbb_2}
        \begin{aligned}
            & \sup_{(\beta, \Lambda, \Gamma) \in \mc{B}_{\delta_{NT}}}  \left\|\sum_{i=1}^{N}\sum_{t=1}^{T}(\ddot{\ell}_{it}X_{it}X_{it}' - \ddot{\ell}_{it}^0 X_{it}X_{it}') \right\|_{\mr{op}}\\
            \eqtext{(i)} & \sup_{(\beta, \Lambda, \Gamma) \in \mc{B}_{\delta_{NT}}}  \left\| \sum_{i=1}^{N}\sum_{t=1}^{T}
            \underbrace{\tilde{\dddot{\ell}}_{it} \left(\Delta_{\beta}'X_{it} + \tilde{\gamma}'_{t}\Delta_{\lambda_i} + \tilde{\lambda}_{i}'\Delta_{\gamma_t} \right)}_{\tilde{\Delta}_{Y^*_{it}}} X_{it}X_{it}' \right\|_{\mr{op}}\\
            \leqtext{(ii)} &   \sup_{(\beta, \Lambda, \Gamma) \in \mc{B}_{\delta_{NT}}}  \left\| \sum_{i=1}^{N}\sum_{t=1}^{T}\tilde{\dddot{\ell}}_{it} (\Delta_{\beta}'X_{it})     X_{it}X_{it}' \right\|_{\mr{op}} \\
            & + \sup_{(\beta, \Lambda, \Gamma) \in \mc{B}_{\delta_{NT}}}  \left\| \sum_{i=1}^{N}\sum_{t=1}^{T}\tilde{\dddot{\ell}}_{it} \left( \tilde{\gamma}'_{t}\Delta_{\lambda_i} + \tilde{\lambda}_{i}'\Delta_{\gamma_t}\right) X_{it}X_{it}' \right\|_{\mr{op}}  \\
            \lesssimtext{(iii)} & NT \delta_{NT} + \sup_{(\beta, \Lambda, \Gamma) \in \mc{B}_{\delta_{NT}}}  \left\| \sum_{i=1}^{N}\sum_{t=1}^{T}\tilde{\dddot{\ell}}_{it} \left( \tilde{\gamma}'_{t}\Delta_{\lambda_i} + \tilde{\lambda}_{i}'\Delta_{\gamma_t}\right)X_{it}X_{it}' \right\|_{\mr{op}}  \\ 
            \lesssimtext{(iv)} & \sup_{(\beta, \Lambda, \Gamma) \in \mc{B}_{\delta_{NT}}}    \sum_{t=1}^{T}\left\| \sum_{i=1}^{N}\tilde{\dddot{\ell}}_{it}   \tilde{\gamma}_{t}'\Delta_{\lambda_i}   X_{it}X_{it}' \right\|_{\mr{op}}  + \sup_{(\beta, \Lambda, \Gamma) \in \mc{B}_{\delta_{NT}}}    \sum_{i=1}^{N}\left\| \sum_{t=1}^{T}\tilde{\dddot{\ell}}_{it}   \tilde{\lambda}_{i}'\Delta_{\gamma_t}  X_{it}X_{it}' \right\|_{\mr{op}}+ NT \delta_{NT} \\
            \lesssimtext{(v)} & \sum_{t=1}^{T} \sqrt{N}\sup_{(\beta, \Lambda, \Gamma) \in \mc{B}_{\delta_{NT}}} \|\Lambda - \Lambda^G_0\|_{\mr{F}} + \sum_{i=1}^{N} \sqrt{T}\sup_{(\beta, \Lambda, \Gamma) \in \mc{B}_{\delta_{NT}}} \|\Gamma - \Gamma^G_0\|_{\mr{F}} + NT\delta_{NT} \\
            \lesssimtext{(vi)} & NT \delta_{NT}. 
        \end{aligned}
        \end{equation}
        Equation (i) follows from first-order Taylor expansion, where $\tilde{\dddot{\ell}}_{it}$ denotes the third-order derivative of $\ell_{it}$ evaluated at $(\tilde{\beta}, \tilde{\Lambda}, \tilde{\Gamma})$, a point on the line segment between $(\beta, \Lambda, \Gamma)$ and the true parameters. As discussed earlier, $(\tilde{\beta}, \tilde{\Lambda}, \tilde{\Gamma})$ also lies in a bounded space, as both $(\tilde{\beta}, \tilde{\Lambda}, \tilde{\Gamma})$ and true parameter lie in bounded spaces. Equation (ii) follows from the triangle inequality. Equation (iii) is derived based on the uniform boundedness of $\tilde{\dddot{\ell}}_{it}$, $X$, and $(\tilde{\beta}, \tilde{\Lambda}, \tilde{\Gamma})$ (Assumption~\ref{assumption:regularity_pre}\ref{item:boundedness_pre} and \ref{item:smoothing_pre}), and the bound $ \sup_{(\beta, \Lambda, \Gamma) \in \mc{B}_{\delta_{NT}}}\|\beta - \beta_0\| \leq \delta_{NT}$.   Inequality (iv)  again follows from the  triangle inequality.  Inequality (v) is based on the uniform boundedness of $\tilde{\dddot{\ell}}_{it}$, $X$, and $(\tilde{\beta}, \tilde{\Lambda}, \tilde{\Gamma})$, along with the Cauchy-Schwarz inequality. Inequality (vi) follows from the bounds  $\sup_{(\beta, \Lambda, \Gamma) \in \mc{B}_{\delta_{NT}}}\|\Lambda - \Lambda^G_0\|_{\mr{op}}\leq \sqrt{N}\delta_{NT}$ and $\sup_{(\beta, \Lambda, \Gamma) \in \mc{B}_{\delta_{NT}}}\|\Gamma - \Gamma^G_0\|_{\mr{op}}\leq \sqrt{T}\delta_{NT}$. A byproduct that will be frequently used in the subsequent analysis is $$\sup_{(\beta, \Lambda, \Gamma) \in \mc{B}_{\delta_{NT}}}  
        \sum_{i=1}^{N}\sum_{t=1}^{T} \tilde{\Delta}_{Y^*_{it}}^2\leq NT\delta_{NT}^2, $$
        which can be established using a similar method. 

        Under Assumption~\ref{assumption:regularity_pre}\ref{item:sampling_pre} and \ref{assumption:regularity_pre}\ref{item:boundedness_pre}, we apply \citet[Theorem~1]{kanaya2017convergence} to establish that for any $(\beta, \Lambda, \Gamma)$, 
        \begin{equation}\label{eq:lemma:bound_Hbb_3}
            \begin{aligned}
                \left\|\sum_{i=1}^{N}\sum_{t=1}^{T}(\ddot{\ell}^0_{it}X_{it}X_{it}' - \mb{E}_0 (\ddot{\ell}_{it}^0 X_{it}X_{it}'))\right\|_{\mr{op}} = O_p(\sqrt{NT}). 
            \end{aligned}
         \end{equation}
        Combining~\eqref{eq:lemma:bound_Hbb_1},~\eqref{eq:lemma:bound_Hbb_2},~\eqref{eq:lemma:bound_Hbb_3}, and~$\delta_{NT} \lesssim \log(NT)/\sqrt{\min\{N, T\}}$, we obtain  
        \begin{align*}
            \sup_{(\beta, \Lambda, \Gamma) \in \mc{B}_{\delta_{NT}}}  \left\|H_{\beta\beta'} - \mb{E}_0H_{0, \beta\beta'}\right\|_{\mr{op}} = O_p(NT\delta_{NT}).  
        \end{align*}
    \end{prooflmm}

    \begin{lemma}\label{lemma:bound_S_2}
        Under the conditions of Theorem~\ref{thm:convexity_strong}, for any $(\beta, \Lambda, \Gamma) \in \mc{B}_{\delta_{NT}}$,  the matrix $\frac{1}{2}S_1 + S_2$ is positive definite wpa1. 
    \end{lemma}
    \begin{prooflmm}{lemma:bound_S_2}
        As we have demonstrated in the Proof of Theorem~\ref{thm:convexity_strong_pre}, there exists a constant $B_2$ such that 
        \begin{align*}
            S_1 \geq B_2 
            \begin{pmatrix}
                NT\mb{I}_{d_X} & 0 & 0 \\
                0 & T\mb{I}_{NR} & 0 \\
                0 & 0 & N\mb{I}_{TR} 
            \end{pmatrix}
            , 
            \quad
            \text{wpa1. }
        \end{align*}
        It follows that 
        \begin{align*}
            \frac{1}{2} S_1 + S_2 & \geq  \frac{B_2}{2}
            \begin{pmatrix}
                NT\mb{I}_{d_X} & 0 & 0 \\
                0 & T\mb{I}_{NR} & 0 \\
                0 & 0 & N\mb{I}_{TR} 
            \end{pmatrix}
            + 
            \begin{pmatrix}
                H_{\beta\beta'} - \mb{E}_{0} H_{0, \beta\beta'} & H_{\beta\lambda'} - \mb{E}_{0} \tilde{H}_{0, \beta\lambda'} & H_{ \beta\gamma'} - \mb{E}_{0} \tilde{H}_{0, \beta\gamma'}  \\
                H_{\lambda\beta'} - \mb{E}_{0} \tilde{H}_{ 0, \lambda\beta'} & 0 & 0 \\
                H_{\gamma\beta'}  - \mb{E}_{0} \tilde{H}_{0, \gamma\beta'} & 0 & 0
            \end{pmatrix} 
            \\
            & \geqtext{(i)} \begin{pmatrix}
                \frac{B_2}{3} NT \mb{I}_{d_X}  & H_{\beta\lambda'} - \mb{E}_{0} \tilde{H}_{0, \beta\lambda'} & H_{ \beta\beta'} - \mb{E}_{0} \tilde{H}_{0, \beta\gamma'}  \\
                H_{\lambda\beta'} - \mb{E}_{0} \tilde{H}_{ 0, \lambda\beta'} & \frac{B_2}{2} T\mb{I}_{NR} & 0 \\
                H_{\gamma\beta'}  - \mb{E}_{0} \tilde{H}_{0, \gamma\beta'} & 0 & \frac{B_2}{2} N\mb{I}_{TR}
            \end{pmatrix}
            , 
            \quad 
            \text{wpa1}, 
        \end{align*}
        where the inequality (i) follows from Lemma~\ref{lemma:bound_Hbb}, which ensures that $\left\|H_{\beta\beta'} - \mb{E}_0H_{0, \beta\beta'}\right\|_{\mr{op}}= o_p(NT)$. Therefore, the matrix $\frac{1}{2}S_1 + S_2$ is positive definite if its Schur complement 
        \begin{align*}
            \begin{pmatrix}
                \frac{B_2}{2} T\mb{I}_{NR} & 0 \\
                0 & \frac{B_2}{2} N \mb{I}_{TR}
            \end{pmatrix}
            -
            \frac{3}{NT B_2} 
            \begin{pmatrix}
                H_{\lambda\beta'} - \mb{E}_{0} \tilde{H}_{0, \lambda\beta'} \\
                H_{\gamma\beta'} - \mb{E}_{0} \tilde{H}_{0, \gamma\beta'}
            \end{pmatrix}
            \begin{pmatrix}
                H_{\beta\lambda'} - \mb{E}_{0} \tilde{H}_{0, \beta\lambda'} & H_{\beta\gamma'} - \mb{E}_{0} \tilde{H}_{0, \beta\gamma'}
            \end{pmatrix}
        \end{align*}
        is positive definite. Thus, by Weyl's Theorem, it suffices to show that 
        \begin{align*} 
            \sup_{(\beta, \Lambda, \Gamma) \in \mc{B}_{\delta_{NT}}} \left\|
            \begin{pmatrix}
                H_{\lambda\beta'} - \mb{E}_0\tilde{H}_{0, \lambda\beta'} & 
                H_{\gamma\beta'} -  \mb{E}_0\tilde{H}_{0, \gamma\beta'}
            \end{pmatrix}
            \right\|_{\mr{op}}^2 /(NT) = o_p(\min\{N, T\}). 
        \end{align*} 
        Observe that 
        \begin{align*}
            \ddot{\ell}_{it}X_{it}\gamma_t' = ((\ddot{\ell}_{it}X_{it}  - \ddot{\ell}^0_{it}X_{it}) + (\ddot{\ell}^0_{it}X_{it} - \mb{E}_0(\ddot{\ell}^0_{it}X_{it})  + \mb{E}_0(\ddot{\ell}^0_{it}X_{it})) (\Delta_{\gamma_t} + \gamma^{G\prime}_{0, t}) , 
        \end{align*}
        it follows that  
        \begin{align*}
            \ddot{\ell}_{it}X_{it}\gamma_t' - \mb{E}_0(\ddot{\ell}^0_{it}X_{it}) \gamma^{G\prime}_{0, t} = &   \tilde{\dddot{\ell}}_{it}\tilde{\Delta}_{Y^*_{it}} X_{it} \gamma^{G\prime}_{0, t} + (\ddot{\ell}^{0}_{it}X_{it} -  \mb{E}_0 (\ddot{\ell}^{0}_{it}X_{it}))\gamma^{G\prime}_{0, t} +  \tilde{\dddot{\ell}}_{it}\tilde{\Delta}_{Y^*_{it}}X_{it}\Delta'_{\gamma_{t}} \\
            & +   (\ddot{\ell}^{0}_{it}X_{it} - \mb{E}_0( \ddot{\ell}^{0}_{it}X_{it}))\Delta'_{\gamma_{t}} + \mb{E}_0(\ddot{\ell}^{0}_{it}X_{it})\Delta'_{\gamma_{t}}, 
        \end{align*}
        where we use $\tilde{\Delta}_{Y^*_{it}} = \Delta_{\beta}'X_{it} + \tilde{\lambda}_{i}'\Delta_{\gamma_t} + \Delta_{\lambda_i}'\tilde{\gamma}_{t}$ to simplify the expression. 
        Here, $\tilde{\dddot{\ell}}_{it}$ denotes the third-order derivative of $\ell_{it}$ evaluated at $(\tilde{\beta}, \tilde{\Lambda}, \tilde{\Gamma})$, lying on the line segment between $(\beta, \Lambda, \Gamma)$ and the normalized true parameters. Since both $(\tilde{\beta}, \tilde{\Lambda}, \tilde{\Gamma})$ and true parameters lie in bounded spaces, it follows that $(\tilde{\beta}, \tilde{\Lambda}, \tilde{\Gamma})$  is also in a bounded space.
        Since the $X$ and $(\tilde{\beta}, \tilde{\Lambda}, \tilde{\Gamma})$  are uniformly bounded, and $\ell_{it}$ is four times differentiable, $\tilde{\dddot{\ell}}_{it}$ is uniformly bounded by extreme value theorem. 
        In the following proof, it suffices to control the Frobenius norm instead of the spectral norm because 
        \begin{align*}
            \left\| 
            \begin{pmatrix}
                H_{\lambda\beta'} - \mb{E}_0\tilde{H}_{0, \lambda\beta'} \\
                H_{\gamma\beta'} -  \mb{E}_0\tilde{H}_{0, \gamma\beta'}
            \end{pmatrix}
            \right\|_{\mr{op}}
             \leq \left\|
                \begin{pmatrix}
                     H_{ \lambda\beta'} \\
                     H_{ \gamma\beta'} 
                \end{pmatrix} 
                -
                \begin{pmatrix}
                    \mb{E}_0 \tilde{H}_{0, \lambda\beta'} \\
                    \mb{E}_0 \tilde{H}_{0, \gamma\beta'} 
                \end{pmatrix}
            \right\|_{\mr{F}}
            \leq 
            \|H_{\lambda\beta'} - \mb{E}_0 \tilde{H}_{0, \lambda\beta'}\|_{\mr{F}} 
            + 
            \|H_{\gamma\beta'} - \mb{E}_0 \tilde{H}_{ 0, \gamma\beta'}\|_{\mr{F}}. 
        \end{align*}
        We focus on $ \|H_{\lambda\beta'} - \mb{E}_0\tilde{H}_{\lambda\beta'}\|_{\mr{F}}$,  because the bound of $\|H_{\gamma\beta'} - \mb{E}_0 \tilde{H}_{0, \gamma\beta'}\|_{\mr{F}}$ can be obtained using a same method. Note that 
        \begin{equation}\label{eq:bound_A}
        \begin{aligned}
            \|H_{\lambda\beta'} - \mb{E}_0\tilde{H}_{0, \lambda\beta'}\|_{\mr{F}}^2  = &\sum_{i\in \mc{I}_{NT}} \bigg\| \sum_{t\in \mc{T}_{NT}}  \tilde{\dddot{\ell}}_{it}\Delta_{Y^*_{it}} X_{it} \gamma^{G\prime}_{0, t} + \sum_{t\in \mc{T}_{NT}}  (\ddot{\ell}^{0}_{it}X_{it} -  \mb{E}_0 (\ddot{\ell}^{0}_{it}X_{it}))\gamma^{G\prime}_{0, t} \\
            & + \sum_{t\in \mc{T}_{NT}}   \tilde{\dddot{\ell}}_{it}\tilde{\Delta}_{Y^*_{it}}X_{it}\Delta'_{\gamma_{t}} + \sum_{t\in \mc{T}_{NT}}  (\ddot{\ell}^{0}_{it}X_{it} - \mb{E}_0( \ddot{\ell}^{0}_{it}X_{it}))\Delta'_{\gamma_{t}}   \\ 
            & + \sum_{t\in \mc{T}_{NT}}  \mb{E}_0(\ddot{\ell}^{0}_{it}X_{it})\Delta'_{\gamma_{t}} + \sum_{t\in \mc{T}^c_{NT}}\ddot{\ell}_{it}X_{it}\gamma_t' \bigg\|_{\mr{F}}^2 + \sum_{i\in \mc{I}^c_{NT}} \left\|\sum_{t=1}^{T} \ddot{\ell}_{it}X_{it}\gamma_t' \right\|^2_{\mr{F}}  \\
            \leq & 6 \Bigg\{  \underbrace{\sum_{i\in \mc{I}_{NT}}\left\|\sum_{t\in \mc{T}_{NT}} \tilde{\dddot{\ell}}_{it}\tilde{\Delta}_{Y^*_{it}} X_{it} \gamma^{G\prime}_{0, t}  \right\|_{\mr{F}}^2}_{A_1}  + \underbrace{\sum_{i\in \mc{I}_{NT}} \left\| \sum_{t\in \mc{T}_{NT}}  (\ddot{\ell}^{0}_{it}X_{it} -  \mb{E}_0 (\ddot{\ell}^{0}_{it}X_{it}))\gamma^{G\prime}_{0, t} \right\|_{\mr{F}}^2}_{A_2} \\
            & + \underbrace{\sum_{i\in \mc{I}_{NT}} \left\|\sum_{t\in \mc{T}_{NT}}  \tilde{\dddot{\ell}}_{it}\tilde{\Delta}_{Y^*_{it}}X_{it}\Delta'_{\gamma_{t}} \right\|_{\mr{F}}^2}_{A_3} + \underbrace{\sum_{i\in \mc{I}_{NT}} \left\| \sum_{t\in \mc{T}_{NT}}   (\ddot{\ell}^{0}_{it}X_{it} - \mb{E}_0( \ddot{\ell}^{0}_{it}X_{it}))\Delta'_{\gamma_{t}}\right\|_{\mr{F}}^2}_{A_4} \\
            & + \underbrace{\sum_{i\in \mc{I}_{NT}}\left\|\sum_{t\in \mc{T}_{NT}}  \mb{E}_0(\ddot{\ell}^{0}_{it}X_{it})\Delta'_{\gamma_{t}}\right\|_{\mr{F}}^2}_{A_5} + \underbrace{\sum_{i\in \mc{I}_{NT}}\left\|\sum_{t\in \mc{T}^c_{NT}}  \ddot{\ell}_{it}X_{it}\gamma'_{t}\right\|_{\mr{F}}^2}_{A_6} \Bigg\} \\
            & + \underbrace{\sum_{i\in \mc{I}^c_{NT}} \left\|\sum_{t=1}^{T} \ddot{\ell}_{it}X_{it}\gamma_t' \right\|^2_{\mr{F}}}_{A_7}. 
        \end{aligned}
        \end{equation}
        
        The first term $A_1$ is bounded by
        \begin{equation}\label{eq:bound_A_1}
        \begin{aligned}
            \sup_{(\beta, \Lambda, \Gamma) \in \mc{B}_{\delta_{NT}}} A_1 & \leqtext{(i)} \sup_{(\beta, \Lambda, \Gamma) \in \mc{B}_{\delta_{NT}}}  T \sum_{i=1}^{N}\sum_{t=1}^{T} \left\| \tilde{\dddot{\ell}}_{it} X_{it} \gamma'_{0, t}\right\|_{\mr{F}}^2 \tilde{\Delta}_{Y^*_{it}}^2 \\
            & \lesssimtext{(ii)} T  \sup_{(\beta, \Lambda, \Gamma) \in \mc{B}_{\delta_{NT}}} \sum_{i=1}^{N}\sum_{t=1}^{T}(X_{it}'\Delta_{\beta})^2 + T  \sup_{(\beta, \Lambda, \Gamma) \in \mc{B}_{\delta_{NT}}} \sum_{i=1}^{N}\sum_{t=1}^{T} \left(\|\Delta_{\lambda_{i}}\|^2 + \Delta_{\gamma_{t}}\|^2\right)   \\
            & \lesssim NT^2 \sup_{(\beta, \Lambda, \Gamma) \in \mc{B}_{\delta_{NT}}}  \|\beta - \beta_0\|^2 + T^2  \sup_{(\beta, \Lambda, \Gamma) \in \mc{B}_{\delta_{NT}}}  \|\Lambda - \Lambda_0\|^2_{\mr{F}} + TN \sup_{(\beta, \Lambda, \Gamma) \in \mc{B}_{\delta_{NT}}}  \|\Gamma - \Gamma_0\|^2_{\mr{F}}  \\
            & \lesssimtext{(iii)} NT^2\delta_{NT}^2,  
        \end{aligned}      
        \end{equation} 
        where inequality (i) follows from the Cauchy-Schwarz inequality, inequality (ii) uses the boundedness of $X$,$(\beta, \Lambda_0, \Gamma_0)$, $(\tilde{\beta}, \tilde{\Lambda}, \tilde{\Gamma})$, and $\tilde{\dddot{\ell}}_{it}$, and inequality (iii) follows directly from the definition of $\mc{B}_{\delta_{NT}}$. 
        
        For the second term $A_2$, we proceed as follows: 
        \begin{equation}\label{eq:bound_A_2}
        \begin{aligned}
            \sup_{(\beta, \Lambda, \Gamma) \in \mc{B}_{\delta_{NT}}} A_2 \leq &  \sum_{i=1}^{N} \left\| \sum_{t\in \mc{T}_{NT}}  (\ddot{\ell}^{0}_{it}X_{it} - \mb{E}_0 (\ddot{\ell}^{0}_{it}X_{it}))\gamma_{0, t}^{G\prime}\right\|_{\mr{F}}^2 \\
            \leq & \sum_{i=1}^{N} \left\| \sum_{t=1}^{T}  (\ddot{\ell}^{0}_{it}X_{it} - \mb{E}_0 (\ddot{\ell}^{0}_{it}X_{it}))\gamma_{0, t}^{G\prime}- \sum_{t\in \mc{T}_{NT}^c }  (\ddot{\ell}^{0}_{it}X_{it} - \mb{E}_0 (\ddot{\ell}^{0}_{it}X_{it}))\gamma_{0, t}^{G\prime}\right\|_{\mr{F}}^2 \\
            \leq  & 2\sum_{i=1}^{N} \left\| \sum_{t=1}^{T}  (\ddot{\ell}^{0}_{it}X_{it} - \mb{E}_0 (\ddot{\ell}^{0}_{it}X_{it}))\gamma_{0, t}^{G\prime}\right\|_{\mr{F}}^2 + 2  \sum_{i=1}^{N} \left\|  \sum_{t\in \mc{T}_{NT}^c}  (\ddot{\ell}^{0}_{it}X_{it} - \mb{E}_0 (\ddot{\ell}^{0}_{it}X_{it}))\gamma_{0, t}^{G\prime}\right\|_{\mr{F}}^2. 
        \end{aligned}   
        \end{equation}
        Since (i) $ \mb{E}_0 (\ddot{\ell}^{0}_{it}X_{it} - \mb{E}_0 (\ddot{\ell}^{0}_{it}X_{it})) = 0 $,  (ii) $(\ddot{\ell}^{0}_{it}X_{it} -   \mb{E}_0 (\ddot{\ell}^{0}_{it}X_{it})) \gamma_{0, t}^{G\prime}$ is uniformly bounded, and (iii) $\{\ddot{\ell}^{0}_{it} X_{it}\}_{1\leq t\leq T}$ (conditional on $(Z, \Lambda_0, \Gamma_0)$ or $(X, \Lambda_0, \Gamma_0)$) satisfies the mixing condition in Assumption~\ref{assumption:regularity_pre}\ref{item:sampling_pre}, we apply \citet[Theorem~1]{kanaya2017convergence} to obtain 
        \begin{align*}
            \sum_{i=1}^{N} \left\| \sum_{t=1}^{T}  (\ddot{\ell}^{0}_{it}X_{it} - \mb{E}_0 (\ddot{\ell}^{0}_{it}X_{it}))\gamma_{0, t}^{G\prime}\right\|_{\mr{F}}^2 = O_p(NT\log(NT)).  
        \end{align*}
        Additionally, by the uniform boundedness of $(\ddot{\ell}^{0}_{it}X_{it} - \mb{E}_0 (\ddot{\ell}^{0}_{it}X_{it}))\gamma_{0, t}^{G\prime}$, we have  
        \begin{align*}
            \sum_{i=1}^{N} \left\|  \sum_{t\in \mc{T}_{NT}^c}  (\ddot{\ell}^{0}_{it}X_{it} - \mb{E}_0 (\ddot{\ell}^{0}_{it}X_{it}))\gamma_{0, t}^{G\prime}\right\|_{\mr{F}}^2 \lesssim N |\mc{T}_{NT}^c|^2 \lesssim N^3T^2 \delta_{NT}^8. 
        \end{align*}
        Therefore, 
        \begin{align*}
            \sup_{(\beta, \Lambda, \Gamma) \in \mc{B}_{\delta_{NT}}} A_2 = O_p\left(N^3T^2 \delta_{NT}^8\right) . 
        \end{align*}

        The term $A_3$ is bounded by 
        \begin{equation}\label{eq:bound_A_3}
        \begin{aligned}
            \sup_{(\beta, \Lambda, \Gamma) \in \mc{B}_{\delta_{NT}}}  A_3  & \leqtext{(i)}  \sup_{(\beta, \Lambda, \Gamma) \in \mc{B}_{\delta_{NT}}}   T \sum_{i=1}^{N} \sum_{t=1}^{T} \|   \tilde{\dddot{\ell}}_{it}\tilde{\Delta}_{Y^*_{it}}X_{it} \|^2 \|\Delta_{\gamma_{t}}\|^2 \\
            & \leq \sup_{(\beta, \Lambda, \Gamma) \in \mc{B}_{\delta_{NT}}}T\sum_{t=1}^{T} \left\{\sup_{(\beta, \Lambda, \Gamma) \in \mc{B}_{\delta_{NT}}} \sum_{i=1}^{N}   \|\tilde{\dddot{\ell}}_{it}\tilde{\Delta}_{Y^*_{it}}X_{it}\|^2\right\} \|\Delta_{\gamma_{t}}\|^2 \\
            & \lesssimtext{(ii)}   NT^2\delta^2_{NT} \sup_{(\beta, \Lambda, \Gamma) \in \mc{B}_{\delta_{NT}} }  \sum_{t=1}^{T} \|\Delta_{\gamma_{t}}\|^2 \\
            & \lesssimtext{(iii)}  NT^3\delta^4_{NT},   
        \end{aligned}
        \end{equation}
        where inequality (i) uses the Cauchy-Schwarz inequality, inequality (ii) relies on the uniform boundedness of $X$, $(\tilde{\beta}, \tilde{\Lambda}, \tilde{\Gamma})$, and $\tilde{\dddot{\ell}}_{it}$,  as well as the fact that  $\sup_{(\beta, \Lambda, \Gamma) \in \mc{B}_{\delta_{NT}}}  \sum_{t=1}^{T} \tilde{\Delta}_{Y^*_{it}}^2 \leq \sup_{(\beta, \Lambda, \Gamma) \in \mc{B}_{\delta_{NT}}}  \sum_{i=1}^{N}\sum_{t=1}^{T} \tilde{\Delta}_{Y^*_{it}}^2 \lesssim NT \delta_{NT}^2$ (this can be established using the same method as in the Proof of Lemma~\ref{lemma:bound_Hbb}).   Inequality (iii) follows directly from the definition of $\mc{B}_{\delta_{NT}}$.
        Using similar arguments, we obtain that 
        \begin{align}\label{eq:bound_A_4}
            \sup_{(\beta, \Lambda, \Gamma) \in \mc{B}_{\delta_{NT}}}  A_4, A_5 \lesssim NT^2\delta^2_{NT} .
        \end{align}
        We establish the upper bounds of $A_6$ and $A_7$ by the uniform boundedness condition:
        \begin{equation}\label{eq:bound_A_6}
        \begin{aligned}
            \sup_{(\beta, \Lambda, \Gamma) \in \mc{B}_{\delta_{NT}}}  A_6 &\lesssim N |\mc{T}^c_{NT}|^2 \lesssim N^ 3 T^2 \delta_{NT}^8 \\
            \sup_{(\beta, \Lambda, \Gamma) \in \mc{B}_{\delta_{NT}}}  A_7 &\lesssim T^2 |\mc{I}^c_{NT}| \lesssim N T^3 \delta_{NT}^4. 
        \end{aligned}
        \end{equation}
        Thus, combining~\eqref{eq:bound_A}, \eqref{eq:bound_A_1}, \eqref{eq:bound_A_2}, \eqref{eq:bound_A_3}, \eqref{eq:bound_A_4}, and \eqref{eq:bound_A_6}, we  conclude that,  wpa1, 
        \begin{align*}
            \sup_{(\beta, \Lambda, \Gamma) \in \mc{B}_{\delta_{NT}}}  \|H_{\lambda\beta'} - \mb{E}_0H_{0, \lambda\beta'}\|_{\mr{F}}^2 & = O_p\left(\max\{NT^2\delta_{NT}^2, NT^3\delta_{NT}^4, N^2 T \delta_{NT}^4, N^ 3 T^2 \delta_{NT}^8 \}\right)   \\
            & = O_p\left( \max\{N^4, T^4\} \delta_{NT}^4 \right). 
        \end{align*}
        Based on similar arguments we also show that, wpa1,  
        \begin{align*}
            \sup_{(\beta, \Lambda, \Gamma) \in \mc{B}_{\delta_{NT}}} \|H_{\gamma\beta'} - \mb{E}_0 H_{0, \gamma\beta'}\|_{\mr{F}}^2 
            & \lesssim = O_p\left( \max\{N^4, T^4\} \delta_{NT}^4 \right). 
        \end{align*}
        Therefore, we conclude that 
        \begin{equation*}
        \begin{aligned}
            \sup_{(\beta, \Lambda, \Gamma) \in \mc{B}_{\delta_{NT}}} \left\|
            \begin{pmatrix}
                H_{\lambda\beta'} - \mb{E}_0\tilde{H}_{0, \lambda\beta'} & 
                H_{\gamma\beta'} -  \mb{E}_0\tilde{H}_{0, \gamma\beta'}
            \end{pmatrix}
            \right\|_{\mr{op}}^2/(NT) & = o_p\left(\min\{\sqrt{N}, \sqrt{T}\}\left(\log(NT)\right)^4\right) \\
            & = o_p(\min\{N, T\}). 
        \end{aligned}   
        \end{equation*}
        This complete the proof. 
    \end{prooflmm}

    \begin{lemma}\label{lemma:bound_S_3}
        Under the conditions of Theorem~\ref{thm:convexity_strong_pre}, the maximum singular value of $S_3$ in~\eqref{eq:Hessian_decomposition} satisfies 
        \begin{align*}
            \sup_{(\beta, \Lambda, \Gamma) \in \mc{B}_{\delta_{NT}}} \|S_3\|_{\mr{op}} = o_p(\min\{N, T\}). 
        \end{align*}
    \end{lemma}
    \begin{prooflmm}{lemma:bound_S_3}
        Recall that 
        \begin{align*}
            S_3 = \begin{pmatrix}
                \tilde{H}_{\lambda\lambda'} & H_{\lambda\gamma'} \\
                H_{\gamma\lambda'} & \tilde{H}_{\gamma\gamma'} 
            \end{pmatrix} 
            -
            \begin{pmatrix}
                \mb{E}_0\tilde{H}_{0, \lambda\lambda'} & \mb{E}_0 \tilde{H}_{0, \lambda\gamma'} \\
                \mb{E}_0\tilde{H}_{0, \gamma\lambda'} & \mb{E}_0\tilde{H}_{0, \gamma\gamma'} 
            \end{pmatrix}. 
        \end{align*}
        The operator norm of $S_3$  is bounded by
        \begin{align*}
            \|S_3\|_{\mr{op}} \leq 
            \underbrace{
            \left\|
            \begin{pmatrix}
                \tilde{H}_{\lambda\lambda'} -  \mb{E}_0 \tilde{H}_{0, \lambda\lambda'} & 0 \\
                0 & \tilde{H}_{\gamma\gamma'} - \mb{E}_0 \tilde{H}_{0, \gamma\gamma'} 
            \end{pmatrix}
            \right\|_{\mr{op}}
            }_{A_1}
            + 
            \underbrace{
            \left\|
            \begin{pmatrix}
                    0  & H_{\lambda\gamma'} - \mb{E}_0 \tilde{H}_{0, \lambda\gamma'} \\
                H_{\gamma\lambda'} - \mb{E}_0\tilde{H}_{0, \gamma\lambda'} & 0
            \end{pmatrix}
            \right\|_{\mr{op}}
            }_{A_2}. 
        \end{align*}
        \paragraph{Step 1} In this step, we aim to derive an upper bound for $A_1$. We focus on bounding the operator norm of $\tilde{H}_{\lambda\lambda'} - \mb{E}_0 \tilde{H}_{0, \lambda\lambda'}$, since the same argument applies to $\tilde{H}_{\gamma\gamma'} - \mb{E}_0 \tilde{H}_{0, \gamma\gamma'}$. Since $A_1$ has a block-diagonal structure and, for any $i\in \mc{I}_{NT}^c$, the $i$-th $R\times R$ diagonal block of $\tilde{H}_{\lambda\lambda'} - \mb{E}_0 \tilde{H}_{0, \lambda\lambda'}$ is a zero matrix, it suffices to consider the $i$-diagonal $R\times R$ block for $i\in \mc{I}_{NT}$. 
        
        For any $i\in \mc{I}_{NT}$, we have 
        \begin{align*}
            [\tilde{H}_{\lambda\lambda'} - \mb{E}_0 \tilde{H}_{0, \lambda\lambda'}]_{i} = & \sum_{t\in \mc{T}_{NT}}\left((-\ddot{\ell}_{it})\gamma_t \gamma_{t}' -  \mb{E}_0 (-\ddot{\ell}^0_{it})\gamma_{0, t}^{G} \gamma_{0, t}^{G\prime} \right) \\
            = &  -\sum_{t\in \mc{T}_{NT}} \ddot{\ell}_{it}\gamma^G_{0, t}\Delta_{\gamma_t}' - \sum_{t\in \mc{T}_{NT}}\ddot{\ell}_{it}\Delta_{\gamma_t}\gamma_{t}' - \sum_{t\in \mc{T}_{NT}}(\ddot{\ell}_{it} -  \ddot{\ell}^0_{it} )\gamma_{0, t}^G \gamma_{0, t}^{G\prime}  - \sum_{t\in \mc{T}_{NT}}(\ddot{\ell}^0_{it} -  \mb{E}(\ddot{\ell}^0_{it}) )\gamma^G_{0, t}  \gamma_{0, t}^{G\prime}. 
        \end{align*}
        Thus,  
        \begin{align*}
            \sup_{(\beta, \Lambda, \Gamma) \in \mc{B}_{\delta_{NT}}}\|[\tilde{H}_{\lambda\lambda'} - \mb{E}_0 \tilde{H}_{0, \lambda\lambda'}]_{i}\|_{\mr{op}} \leq   &  \underbrace{\sup_{(\beta, \Lambda, \Gamma) \in \mc{B}_{\delta_{NT}}} \left\|\sum_{t\in \mc{T}_{NT}} \ddot{\ell}_{it}\gamma_{0, t}^G\Delta_{\gamma_t}' \right\|_{\mr{op}}}_{Q1}  +  \underbrace{\sup_{(\beta, \Lambda, \Gamma) \in \mc{B}_{\delta_{NT}}} \left\|\sum_{t\in \mc{T}_{NT}} \ddot{\ell}_{it}\Delta_{\gamma_t} \gamma_{ t}'\right\|_{\mr{op}}}_{Q2}  \\ 
            &   \underbrace{\sup_{(\beta, \Lambda, \Gamma) \in \mc{B}_{\delta_{NT}}} \left\| \sum_{t\in \mc{T}_{NT}}(\ddot{\ell}_{it} -  \ddot{\ell}^0_{it} )\gamma_{0, t}^G \gamma_{0, t}^{G\prime} \right\|_{\mr{op}}}_{Q3} \\
            & + \underbrace{\sup_{(\beta, \Lambda, \Gamma) \in \mc{B}_{\delta_{NT}}} \left\| \sum_{t\in \mc{T}_{NT}}(\ddot{\ell}^0_{it} -  \mb{E}(\ddot{\ell}^0_{it}) )\gamma_{0, t}^G  \gamma_{0, t}^{G\prime}  \right\|_{\mr{op}}}_{Q4}. 
        \end{align*}
        
        The upper bound of $Q_1$ is derived as follows: 
        \begin{equation}\label{eq:bound_1}
            \begin{aligned}
                Q_1 & \leq \sup_{(\beta, \Lambda, \Gamma) \in \mc{B}_{\delta_{NT}}} \left\|\sum_{t\in \mc{T}_{NT}} \ddot{\ell}_{it}\gamma^G_{0, t}\Delta_{\gamma_t}' \right\|_{\mr{F}} \\
                & \leqtext{(i)} \sqrt{|\mc{T}_{NT}|} \sup_{(\beta, \Lambda, \Gamma) \in \mc{B}_{\delta_{NT}}} \left(\sum_{t\in \mc{T}_{NT}}  \left\|\ddot{\ell}_{it}\gamma^G_{0, t}\Delta_{\gamma_{t}}'\right\|_{\mr{F}}^2\right)^{1/2} \\
                & \leq \sqrt{T} \sup_{(\beta, \Lambda, \Gamma) \in \mc{B}_{\delta_{NT}}} \left(\sum_{t=1}^{T}  \left\|\ddot{\ell}_{it}\gamma^G_{0, t}\Delta_{\gamma_{t}}'\right\|_{\mr{F}}^2\right)^{1/2} \\
                & \lesssimtext{(ii)} \sqrt{T} \sup_{(\beta, \Lambda, \Gamma) \in \mc{B}_{\delta_{NT}}}  \|\Gamma - \Gamma_0^G\|_{\mr{F}} \\
                & \lesssimtext{(iii)}  T \delta_{NT}, 
            \end{aligned}
        \end{equation} 
        where inequality (i) follows from the Cauchy-Schwarz inequality, inequality (ii) relies on the assumption that ${\ell}_{it}, \gamma_{0,t}$ are uniformly bounded, and inequality (iii) follows directly from the definition of $\mc{B}_{\delta_{NT}}$. The upper bound of $Q_2$ is derived via a similar method: 
        \begin{equation}\label{eq:bound_2}
            \begin{aligned}
                Q_2 & \leq \sup_{(\beta, \Lambda, \Gamma) \in \mc{B}_{\delta_{NT}}} \left\|\sum_{t\in \mc{T}_{NT}} \ddot{\ell}_{it}\Delta_{\gamma_t}\gamma_{t}' \right\|_{\mr{F}} \\
                & \leq \sqrt{|\mc{T}_{NT}|} \sup_{(\beta, \Lambda, \Gamma) \in \mc{B}_{\delta_{NT}}} \left(\sum_{t\in \mc{T}_{NT}} \left\|\ddot{\ell}_{it}\Delta_{\gamma_{t}}\gamma_{ t}'\right\|_{\mr{F}}^2\right)^{1/2} \\
                & \leq \sqrt{T} \sup_{(\beta, \Lambda, \Gamma) \in \mc{B}_{\delta_{NT}}} \left(\sum_{t=1}^{T} \left\|\ddot{\ell}_{it}\Delta_{\gamma_{t}}\gamma_{ t}'\right\|_{\mr{F}}^2\right)^{1/2} \\
                & \lesssimtext{(i)} \sqrt{T} \sup_{(\beta, \Lambda, \Gamma) \in \mc{B}_{\delta_{NT}}}  \|\Gamma - \Gamma_0^G \|_{\mr{F}} \\
                & \lesssim  T \delta_{NT}, 
            \end{aligned}
        \end{equation} 
        where inequality (i) uses the assumption that ${\ell}_{it}, \gamma_{t}$ are uniformly bounded. We drive the bound of $Q_3$ by
        \begin{equation}\label{eq:bound_3}
            \begin{aligned}
                Q_3 \leq &  \sup_{(\beta, \Lambda, \Gamma) \in \mc{B}_{\delta_{NT}}}  \left\|\sum_{t\in \mc{T}_{NT}}(\ddot{\ell}_{it} -  \ddot{\ell}^0_{it} ) \gamma_{0, t}^G \gamma_{0, t}^{G\prime} \right\|_{\mr{F}}  \\
                \lesssimtext{(i)} & \sup_{(\beta, \Lambda, \Gamma) \in \mc{B}_{\delta_{NT}}} \left|\sum_{t\in \mc{T}_{NT}}\tilde{\Delta}_{Y^*_{it}} \right|   \\
                \lesssim &\sup_{(\beta, \Lambda, \Gamma) \in \mc{B}_{\delta_{NT}}}  \sum_{t\in \mc{T}_{NT}}  |X_{it}'\Delta_{\beta} | + \sup_{(\beta, \Lambda, \Gamma) \in \mc{B}_{\delta_{NT}}}  \sum_{t\in \mc{T}_{NT}}  \|\lambda_i\| \|\gamma_t -\gamma^G_{0, t} \|  + \sup_{(\beta, \Lambda, \Gamma) \in \mc{B}_{\delta_{NT}}}  \sum_{t\in \mc{T}_{NT}}  \|\gamma_{0, t}\| \|\lambda_{i} - \lambda^G_{0, i} \| \\
                \lesssim &\sup_{(\beta, \Lambda, \Gamma) \in \mc{B}_{\delta_{NT}}}  \sum_{t=1}^{T}  |X_{it}'\Delta_{\beta} | + \sup_{(\beta, \Lambda, \Gamma) \in \mc{B}_{\delta_{NT}}}  \sum_{t=1}^{T}  \|\lambda_i\| \|\gamma_t -\gamma^G_{0, t} \|  + \sup_{(\beta, \Lambda, \Gamma) \in \mc{B}_{\delta_{NT}}}  \sum_{t= 1}^{T}  \|\gamma_{0, t}\| \|\lambda_{i} - \lambda^G_{0, i} \| \\
                \lesssim & T \sup_{(\beta, \Lambda, \Gamma) \in \mc{B}_{\delta_{NT}}} \|\Delta_{\beta}\|_2  + \sqrt{T}  \sup_{(\beta, \Lambda, \Gamma) \in \mc{B}_{\delta_{NT}}}\|\Gamma - \Gamma_0^G \|_{\mr{F}} + T \|\lambda_i - \lambda^G_{0, i}\| \\
                \lesssim & T\delta_{NT} + T  \delta_{NT} +  \frac{\sqrt{T}}{\delta_{NT}}\\
                = & o_p(\min\{N, T\}). 
            \end{aligned}
        \end{equation}

        Therefore, combining~\eqref{eq:bound_1}, \eqref{eq:bound_2}, and \eqref{eq:bound_3}, we obtain 
        \begin{equation*}
            \begin{aligned}
                \sup_{(\beta, \Lambda, \Gamma) \in \mc{B}_{\delta_{NT}}}  \|[\tilde{H}_{\lambda\lambda'} - \mb{E}_0\tilde{H}_{0, \lambda\lambda'}]_{i}\|_{\mr{op}} = o_p(\min\{N, T\}). 
            \end{aligned}
        \end{equation*}
        By the same method, we have  
        \begin{equation*}
            \begin{aligned}
                \sup_{(\beta, \Lambda, \Gamma) \in \mc{B}_{\delta_{NT}}}  \|[\tilde{H}_{\gamma\gamma'} - \mb{E}_0\tilde{H}_{0, \gamma\gamma'}]_{t}\|_{\mr{op}} =  o_p(\min\{N, T\}). 
            \end{aligned}
        \end{equation*}
        Consequently, 
        \begin{align*}
            \sup_{(\beta, \Lambda, \Gamma) \in \mc{B}_{\delta_{NT}}}\|A_1\|_{\mr{op}}  = o_p(\min\{N, T\}). 
        \end{align*}

        \paragraph{Step 2} In this step, we bound the influence of $A_2$. We focus on $H_{\lambda\gamma'} - \mb{E}_0  \tilde{H}_{0, \lambda\gamma'}$ and consider the expansion of its $(i, t)$-th block, $\ddot{\ell}_{it}\lambda_i\gamma'_t - \mb{E}_0 (\ddot{\ell}^0_{it})\lambda_{0,i}^G\gamma^{G\prime}_{0, t}$.  When $i\in \mc{I}_{NT}$ and $t\in \mc{T}_{NT}$,  
        \begin{align*}
            \ddot{\ell}_{it}\lambda_i\gamma'_t - \mb{E}_0(\ddot{\ell}^0_{it})\lambda_{0,i}^G\gamma^{G\prime}_{0, t}  &  = \underbrace{\ddot{\ell}_{it} (\lambda_i \gamma'_t - \lambda_{0,i}^G\gamma^{G\prime}_{0, t}) }_{H_{1, it}} + \underbrace{(\ddot{\ell}_{it} - \ddot{\ell}^0_{it})\lambda_{0,i}^G\gamma^{G\prime}_{0, t}}_{H_{2, it}} + \underbrace{(\ddot{\ell}^0_{it} - \mb{E}_{\phi_0}(\ddot{\ell}^0_{it}))\lambda_{0,i}^G\gamma^{G\prime}_{0, t} }_{H_{3, it}}. 
        \end{align*} 
        For notational simplicity we collect the terms above into matrices $H_1, H_2, H_3\in \mb{R}^{RN\times RT}$. 
        Note that   
        \begin{equation*}
        \begin{aligned}
            & \sup_{(\beta, \Lambda, \Gamma) \in \mc{B}_{\delta_{NT}}}\|H_{\lambda\gamma'} - \mb{E}_0  \tilde{H}_{0, \lambda\gamma'}\|^2_{\mr{op}}   \\
            \leq & \sup_{(\beta, \Lambda, \Gamma) \in \mc{B}_{\delta_{NT}}} (\|H_1\|^2_{\mr{F}} + \|H_2\|^2_{\mr{F}} + \|H_3\|^2_{\mr{op}}) +  \sup_{(\beta, \Lambda, \Gamma) \in \mc{B}_{\delta_{NT}}}\left(\sum_{i=1}^{N}\sum_{t\in \mc{T}_{NT}} (\ddot{\ell}_{it}\lambda_i\gamma'_t)^2 + \sum_{i\in \mc{I}_{NT}}\sum_{t=1}^{T} (\ddot{\ell}_{it}\lambda_i\gamma'_t)^2\right). 
        \end{aligned}
        \end{equation*}
        It follows that,  wpa1, 
        \begin{equation}\label{eq:bound_4}
        \begin{aligned}
            \sup_{(\beta, \Lambda, \Gamma) \in \mc{B}_{\delta_{NT}}}\|H_1\|^2_{\mr{F}} 
            \lesssimtext{(i)} & \sup_{(\beta, \Lambda, \Gamma) \in \mc{B}_{\delta_{NT}}}\sum_{i\in \mc{I}_{NT}}\|\lambda_i\|^2\sum_{t\in \mc{T}_{NT}} \|\gamma_t - \gamma^G_{0, t}\|^2 \\
            & + \sup_{(\beta, \Lambda, \Gamma) \in \mc{B}_{\delta_{NT}}}\sum_{t\in \mc{T}_{NT} }\|\gamma^G_{0, t}\|^2\sum_{i\in \mc{I}_{NT}} \|\lambda_i - \lambda^G_{0, i}\|^2 \\
            \leq   & \sup_{(\beta, \Lambda, \Gamma) \in \mc{B}_{\delta_{NT}}}\sum_{i=1}^{N}\|\lambda_i\|^2\sum_{t=1}^{T} \|\gamma_t - \gamma^G_{0, t}\|^2 \\
            & + \sup_{(\beta, \Lambda, \Gamma) \in \mc{B}_{\delta_{NT}}}\sum_{t=1}^{T}\|\gamma_{0, t}^G\|^2\sum_{i=1}^{N} \|\lambda_i - \lambda^G_{0, i}\|^2 \\
            \lesssimtext{(ii)} & N \sup_{(\beta, \Lambda, \Gamma) \in \mc{B}_{\delta_{NT}}}\|\Gamma - \Gamma_0^G\|_{\mr{F}}^2 + T \sup_{(\beta, \Lambda, \Gamma) \in \mc{B}_{\delta_{NT}}}\|\Lambda - \Lambda_0^G\|_{\mr{F}}^2 \\
            \lesssim & NT\delta_{NT}^2, 
        \end{aligned}
        \end{equation}
        Here, inequality (i) follows from the Cauchy-Schwarz inequality, and (ii) follows directly from the definition of $\mc{B}_{\delta_{NT}}$. Also, wpa1,  
        \begin{equation}\label{eq:bound_5}
        \begin{aligned}
            \sup_{(\beta, \Lambda, \Gamma) \in \mc{B}_{\delta_{NT}}}\|H_2\|^2_{\mr{F}} 
            \lesssim & \sup_{(\beta, \Lambda, \Gamma) \in \mc{B}_{\delta_{NT}}}\sum_{i\in \mc{I}_{NT}} \sum_{t\in \mc{T}_{NT}} \tilde{\Delta}^2_{Y^*_{it}} \\
            \leq  & \sup_{(\beta, \Lambda, \Gamma) \in \mc{B}_{\delta_{NT}}}\sum_{i=1}^{N}\sum_{t=1}^{T} \tilde{\Delta}^2_{Y^*_{it}} \\
            \lesssimtext{(i)} & NT\delta_{NT}^2, 
        \end{aligned}
        \end{equation} 
        where inequality (i) is from the proof of Lemma~\ref{lemma:bound_Hbb}. 
        In addition, based on the Lemma~\ref{lemma:independent_entry}, we  show that,  wpa1,  
        \begin{align}\label{eq:bound_6}
            \sup_{(\beta, \Lambda, \Gamma) \in \mc{B}_{\delta_{NT}}} \|H_3\|^2_{\mr{op}} \lesssim \max\{N, T\}\log((NT))^2 \lesssim NT \delta_{NT}^2 
        \end{align}
        Finally,  wpa1, 
        \begin{equation}\label{eq:bound_7}
        \begin{aligned}
            \sup_{(\beta, \Lambda, \Gamma) \in \mc{B}_{\delta_{NT}}}\left(\sum_{i=1}^{N}\sum_{t\in \mc{T}_{NT}^c} \|\ddot{\ell}_{it}\lambda_i\gamma'_t\|_{\mr{F}}^2 + \sum_{i\in \mc{I}_{NT}^c }\sum_{t=1}^{T} \|\ddot{\ell}_{it}\lambda_i\gamma'_t\|_{\mr{F}}^2\right) \lesssim & N|\mc{T}_{NT}^c| + T|\mc{I}_{NT}^c| 
            \lesssimtext{(i)} & \max\{N^3, T^3\}\delta_{NT}^4 
        \end{aligned}
        \end{equation}
        where inequality (i) follows from~\eqref{eq:size_division}. 
        Therefore, combining equations~\eqref{eq:bound_4}--- \eqref{eq:bound_7} we  conclude that  $\sup_{(\beta, \Lambda, \Gamma) \in \mc{B}_{\delta_{NT}}}\|A_2\|_{\mr{op}}=  o_p(\min\{N, T\}) $. Consequently, 
        \begin{align*}
            \sup_{(\beta, \Lambda, \Gamma) \in \mc{B}_{\delta_{NT}}} \|S_3\|_{\mr{op}}  = o_p(\min\{N, T\}). 
        \end{align*}
        This completes the proof. 
    \end{prooflmm}

    \begin{lemma}\label{lemma:bound_F}
        Under the conditions of Theorem~\ref{thm:convexity_strong_pre}, the operator norm of $F$ in~\eqref{eq:Hessian_decomposition} satisfies
        \begin{align*}
            \sup_{(\beta, \Lambda, \Gamma) \in \mc{B}_{\delta_{NT}}} \|F\|_{\mr{op}} = o_p(\min\{N, T\}). 
        \end{align*}
    \end{lemma}
    \begin{prooflmm}{lemma:bound_F}
        Consider the expansion of $\dot{\ell}_{it}$ (noting that $\mb{E}_0\dot{\ell}_{it}^0 = 0$): 
        \begin{align*}
            \dot{\ell}_{it} = (\dot{\ell}_{it} - \dot{\ell}^{0}_{it})+ (\dot{\ell}^0_{it} - \mb{E}\dot{\ell}_{it}^0). 
        \end{align*}
        The corresponding matrix decomposition of $F$ is given by
        \begin{align*}
            F = 
            \underbrace{\begin{pmatrix}
                0 & 0 & 0 \\
                0 & 0 & F_{1, \lambda \gamma'} \\
                0 & F_{1, \gamma\lambda'} & 0
            \end{pmatrix}}_{F_1}
            + 
            \underbrace{\begin{pmatrix}
                0 & 0 & 0 \\
                0 & 0 & F_{2, \lambda \gamma'} \\
                0 & F_{2, \gamma\lambda'} & 0
            \end{pmatrix}}_{F_2}, 
        \end{align*}
        where 
        \begin{align*}
            F_{1, \lambda \gamma'} & = \left[(\dot{\ell}_{it} - \dot{\ell}^{0}_{it})\mb{I}_R\right]_{i=1,2,\ldots, N, t = 1,2,\ldots, T},  \\
            F_{2, \lambda \gamma'} & = \left[(\dot{\ell}^0_{it} - \mb{E}_0 \dot{\ell}_{it}^0)\mb{I}_R\right]_{i=1,2,\ldots, N, t = 1,2,\ldots, T}. 
        \end{align*}
        The operator norm of $F$  is therefore bounded by 
        \begin{align*}
            \sup_{(\beta, \Lambda, \Gamma) \in \mc{B}_{\delta_{NT}}}  \|F\|_{\mr{op}} \leq \sup_{(\beta, \Lambda, \Gamma) \in \mc{B}_{\delta_{NT}}}  \|F_1\|_{\mr{op}} +   \sup_{(\beta, \Lambda, \Gamma) \in \mc{B}_{\delta_{NT}}}  \|F_2\|_{\mr{op}}. 
        \end{align*}
        We conclude that $\sup_{(\beta, \Lambda, \Gamma) \in \mc{B}_{\delta_{NT}}}  \|F_1\|_{\mr{op}} = o_p(\min\{N, T\})$ using the same argument as in the proof of Lemma~\ref{lemma:bound_S_2} and $\sup_{(\beta, \Lambda, \Gamma) \in \mc{B}_{\delta_{NT}}}  \|F_2\|_{\mr{op}} = o_p(\max \{N, T\})$ by Lemma~\ref{lemma:independent_entry}. Therefore, when $N\sim T$, we have
        \begin{align*}
            \sup_{(\beta, \Lambda, \Gamma) \in \mc{B}_{\delta_{NT}}}  \|F\|_{\mr{op}}  = o_p(\min\{N, T\}). 
        \end{align*}
        This completes the proof. 
    \end{prooflmm}

    \begin{lemma}\label{lemma:bound_V}
        Under the conditions of Theorem~\ref{thm:convexity_strong_pre}, we have 
        \begin{align*}
            \sup_{(\beta, \Lambda, \Gamma) \in \mc{B}_{\delta_{NT}}} \|V - V_0\|_{\mr{op}} = o_p(\max\{N, T\}). 
        \end{align*}
    \end{lemma}
    \begin{prooflmm}{lemma:bound_V}
        Recall that 
        \begin{align*}
            V_{\lambda\lambda'} & = \frac{T}{N}\left[ \lambda_i \lambda_{i'}' \right]_{i, i' = 1,2,\ldots, N },  \\
            V_{\lambda\gamma'} & = \left[ -\lambda_i  \gamma_{t}' \right]_{i =1,2,\ldots, N, t = 1,2,\ldots, T},  \\
            V_{\gamma\gamma'} & = \frac{N}{T}\left[ \gamma_t   \gamma_{t'}'\right]_{t, t' = 1,2,\ldots, T} .  
        \end{align*}
        In addition,  
        \begin{align*}
            \sup_{(\beta, \Lambda, \Gamma) \in \mc{B}_{\delta_{NT}}} \|V - V_0\|_{\mr{op}} \leq & \sup_{(\beta, \Lambda, \Gamma) \in \mc{B}_{\delta_{NT}}} \|V - V_0\|_{\mr{F}} \\
            \leq & \sup_{(\beta, \Lambda, \Gamma) \in \mc{B}_{\delta_{NT}}}\|V_{\lambda\lambda'} - V_{0, \lambda\lambda'} \|_{\mr{F}} + \sup_{(\beta, \Lambda, \Gamma) \in \mc{B}_{\delta_{NT}}}\|V_{\gamma\gamma'} - V_{0, \gamma\gamma'} \|_{\mr{F}} \\
            & + 2 \sup_{(\beta, \Lambda, \Gamma) \in \mc{B}_{\delta_{NT}}}\|V_{\lambda\gamma'} - V_{0, \lambda\gamma'} \|_{\mr{F}}. 
        \end{align*}
        We first show that, wpa1, 
        \begin{align*}
            \sup_{(\beta, \Lambda, \Gamma) \in \mc{B}_{\delta_{NT}}} \|V_{\lambda\lambda'} - V_{0, \lambda\lambda'} \|_{\mr{F}}^2  \leq &  \frac{T^2}{N^2 } \sup_{(\beta, \Lambda, \Gamma) \in \mc{B}_{\delta_{NT}}} \sum_{i=1}^{N} \sum_{i'=1}^{N}  \left\|\lambda_i \lambda_{i'}' -  \lambda_{i ,0}^{G} \lambda_{i', 0}^{G\prime} \right\|_{\mr{F}}^2  \\
            \lesssimtext{(i)} & \frac{T^2}{N^2 }  \sup_{(\beta, \Lambda, \Gamma) \in \mc{B}_{\delta_{NT}}} \sum_{i=1}^{N} \sum_{i'=1}^{N} \left(\|\lambda_{i'}\|^2 \|\lambda_{i} - \lambda^G_{0, i}\|^2 +  \| \lambda^G_{0, i}\|^2 \|\lambda_{i'} - \lambda^G_{0, i'}\|^2  \right) \\
            \lesssimtext{(ii)} & \frac{T^2}{N }  \sup_{(\beta, \Lambda, \Gamma) \in \mc{B}_{\delta_{NT}}}\|\Lambda - \Lambda_0^G\|_{\mr{F}}^2 \\
            \lesssim & T^2 \delta_{NT}^2,  
        \end{align*}
        where inequality (i) employs the Cauchy-Schwarz inequality, and inequality (ii) uses the uniform boundedness of $\Lambda$ and $\Lambda_0^{G}$. Similarly, we have, wpa1, 
        \begin{align*}
            \sup_{(\beta, \Lambda, \Gamma) \in \mc{B}_{\delta_{NT}}} \|V_{\gamma\gamma'} - V_{0, \gamma\gamma'} \|_{\mr{F}}^2 \lesssim  N^2 \delta_{NT}^2. 
        \end{align*}
        In addition, wpa1, 
        \begin{align*}
            \sup_{(\beta, \Lambda, \Gamma) \in \mc{B}_{\delta_{NT}}} \|V_{\lambda\gamma'} - V_{0, \lambda\gamma'} \|_{\mr{F}}^2 \lesssimtext{(i)} & \sum_{i=1}^{N}\sum_{t=1}^{T} \left(\|\lambda_i\|^2 \|\gamma_{t} - \gamma^G_{0, t}\|^2 +  \| \gamma^G_{0, t}\|^2 \|\lambda_{i} - \lambda^G_{0, i}\|^2\right) \\
            \lesssimtext{(ii)} & N \sup_{(\beta, \Lambda, \Gamma) \in \mc{B}_{\delta_{NT}}} \|\Gamma - \Gamma_0^G\|_{\mr{F}}^2 + T \sup_{(\beta, \Lambda, \Gamma) \in \mc{B}_{\delta_{NT}}}  \|\Lambda - \Lambda_0^G\|_{\mr{F}}^2 \\
            \lesssim & NT\delta^2_{NT}, 
        \end{align*}
        where inequality (i) again employs the Cauchy-Schwarz inequality, and inequality (ii) uses the uniform boundedness of $\Gamma$ and $\Gamma_0^{G}$. Finally, 
        \begin{align*}
            \sup_{(\beta, \Lambda, \Gamma) \in \mc{B}_{\delta_{NT}}} \|V - V_0\|_{\mr{op}} \lesssim \max\{N, T\}\delta_{NT} = o_p(\max\{N, T\}). 
        \end{align*}
        This completes the proof. 
    \end{prooflmm}

\subsection{Proof of Lemma \ref{lemma:sufficient_convexity}}

    \begin{prooflmm}{lemma:sufficient_convexity}
        Let $\mb{P}_{\mc{U}}(\cdot): = \mb{P}\left(\cdot \mid \mc{U}\right)$ denote the probability conditional on $\mc{U}$ defined in Assumption~\ref{assumption:conditional_independence_hessian}, and let $\mb{E}_{\mc{U}}(\cdot): = \mb{E}\left(\cdot \mid \mc{U}\right)$ denote the expectation conditional on $\mc{U}$. By~\eqref{eq:VF} and~\eqref{eq:decomposition}, we have 
        \begin{align*}
            NT \mb{E}_0 \mc{H}(\beta_0, \Lambda_0^G, \Gamma_0^G) = \mb{E}_{\mc{U}}\mb{E}_0 H_0 + (\mb{E}_0 H_0 - \mb{E}_{\mc{U}}\mb{E}_0 H_0 ) + V_0 + (\mb{E}_0\hat{V} - V_0). 
        \end{align*}
        In addition, by Lemma~\ref{lemma:bound_V}, we have $\|\mb{E}_0\hat{V} - V_0\|_{\mr{op}}\lesssim \sup_{\beta, \Lambda, \Gamma\in \mc{B}_{\delta_{NT}}}\|V - V_0\|_{\mr{op}} = o_p(\max\{N, T\})$. Thus, 
        it suffices to show that the minimum eigenvalue of $\mb{E}_{\mc{U}}\mb{E} H_0 + V_0$ is strictly positive with order $\min\{N, T\}$, and the operator norm of the perturbation $(\mb{E}_0 H_0 - \mb{E}_{\mc{U}}\mb{E}_0 H_0 )$ is sufficiently small relative to $\mb{E}_{\mc{U}}\mb{E} H_0 + V_0$.  
        
        \paragraph{Step 1} In this step, we show that minimum eigenvalue of $\mb{E}_{\mc{U}}\mb{E}_0 H_0 + V_0$ is strictly positive and of order $\min\{N, T\}$. For any $i, t$, let 
        $\nabla^2 \ell_{it}^0$ be the $(d_X + R(N+T)) \times (d_X + R(N+T))$ Hessian matrix with respect to $(\beta, \Lambda, \Gamma)$. The matrix $\nabla^2(-\ell_{it}^0)$ can be interpreted as the sample Hessian based on a single observation $(Y_{it}, X_{it})$. For notational simplicity, define $ L^0_{it}: = \mb{E}_{\mc{U}}\mb{E}_0(\ell_{it}^0)$.  It is immediate that both $\nabla^2 (-\ell_{it}^0)$ and  $\nabla^2(-L^0_{it})$ are positive semi-definite.  We now consider the following rescaled version of $\nabla^2(-L_{it}^0)$:  
        \begin{align*}
            \widecheck{ \nabla^2(-L^0_{it})}  = 
            \begin{pmatrix}
                \sqrt{\nu}\nabla_{\beta\beta'}^2(-L^0_{it}) & \nabla_{\beta\lambda'}^2(-L^0_{it})  & \nabla_{\beta\gamma'}^2(-L^0_{it}) \\
                \nabla_{\lambda\beta'}^2(-L^0_{it}) & \sqrt{\nu}\nabla_{\lambda\lambda'}^2(-L^0_{it})   & \sqrt{\nu}\nabla_{\lambda\gamma'}^2(-L^0_{it}) \\
                \nabla_{\gamma\beta'}^2 (-L^0_{it}) & \sqrt{\nu}\nabla_{\gamma\lambda'}^2(-L^0_{it})  & \sqrt{\nu} \nabla_{\gamma\gamma'}^2(-L^0_{it})  
            \end{pmatrix}, 
        \end{align*}
        where $\nu$ is defined as in Assumption~\ref{assumption:conditional_independence_hessian}\ref{item:conditional_variability_hessian}. Since $\sqrt{\nu}\mb{E}_{\mc{U}}(\mb{E}_0(\ddot{\ell}^0_{it} X_{it}X_{it}'))$ is positive semi-definite, it follows that $\sqrt{\nu}\nabla_{\beta\beta'}^2(-L^0_{it}) $ is positive semi-definite as well. Thus, by Schur's lemma, $\widecheck{ \nabla^2(-L^0_{it})}$ is positive semi-definte if its Schur's complement $\widecheck{ \nabla^2(-L^0_{it})}\backslash \sqrt{\nu}\nabla_{\beta\beta'}^2(-L^0_{it})$ is positive semi-definite. The latter statement is true because
        \begin{align*} 
            & \begin{pmatrix}
                \mb{E}_{\mc{U}}(\mb{E}_0(-\ddot{\ell}_{it}^0)) \gamma_{t,0} \gamma_{t,0}' & \mb{E}_{\mc{U}}(\mb{E}_0(-\ddot{\ell}_{it}^0)) \gamma_{t,0} \lambda_{i, 0}' \\
                \mb{E}_{\mc{U}}(\mb{E}_0(-\ddot{\ell}_{it}^0)) \lambda_{i, 0} \gamma_{t,0}'  & \mb{E}_{\mc{U}}(\mb{E}_0(-\ddot{\ell}_{it}^0)) \lambda_{i, 0} \lambda_{i,0}' 
            \end{pmatrix}  \\
            & -  
            \mb{E}_{\mc{U}}(\mb{E}_0(-\ddot{\ell}_{it}^0X_{it}'))\left(\sqrt{\nu}\mb{E}_{\mc{U}}(\mb{E}_0(-\ddot{\ell}^0_{it} X_{it}X_{it}'))\right)^{-1}\mb{E}_{\mc{U}}(\mb{E}_0(-\ddot{\ell}_{it}^0X_{it}))
            \begin{pmatrix}
                 \gamma_{0, t}\gamma_{0, t}' &  \gamma_{0, t}\lambda_{0, i}' \\
                 \lambda_{0, i}\gamma_{0, t}'  &   \lambda_{0, i} \lambda_{0, i}'
            \end{pmatrix}  \\
            = &  \left(\sqrt{\nu} \mb{E}_{\mc{U}}(\mb{E}_0(-\ddot{\ell}_{it}^0))  - \mb{E}_{\mc{U}}(\mb{E}_0(-\ddot{\ell}_{it}^0X_{it}'))\left(\sqrt{\nu}\mb{E}_{\mc{U}}(\mb{E}_0(-\ddot{\ell}^0_{it} X_{it}X_{it}'))\right)^{-1}\mb{E}_{\mc{U}}(\mb{E}_0(-\ddot{\ell}_{it}^0X_{it}))\right)
            \begin{pmatrix}
                 \gamma_{0, t}\gamma_{0, t}' &  \gamma_{0, t}\lambda_{0, i}' \\
                 \lambda_{0, i}\gamma_{0, t}'  &   \lambda_{0, i} \lambda_{0, i}'
            \end{pmatrix} \\
            \geq & 0, 
        \end{align*}
        where the last inequality comes from Assumption~\ref{assumption:conditional_independence_hessian}\ref{item:conditional_variability_hessian}. Thus, 
        \begin{align*}
            \mb{E}_{\mc{U}}\mb{E}_0H_0 = \sum_{i=1}^{N}\sum_{t=1}^{T} \nabla^2(-L^0_{it}) 
            & \geq \sum_{i=1}^{N}\sum_{t=1}^{T} \nabla^2(-L^0_{it}) - \sum_{i=1}^{N}\sum_{t=1}^{T} \underbrace{\widecheck{ \nabla^2(-L^0_{it})}}_{\geq 0}  \\
            & = (1-\sqrt{\nu}) \mb{E}_{\mc{U}}\mb{E}_0
            \begin{pmatrix}
                H_{0, \beta\beta'} & 0 & 0 \\
                0 & H_{0, \lambda\lambda'} & H_{0, \lambda\gamma'}\\
                0 & H_{0, \gamma\lambda'} & H_{0, \gamma\gamma'}\\
            \end{pmatrix} \\
            & =
            (1-\sqrt{\nu}) \mb{E}_{\mc{U}}\mb{E}_0
            \begin{pmatrix}
                H_{0, \beta\beta'} & 0 & 0 \\
                0 & 0 & 0 \\
                0 & 0 & 0\\
            \end{pmatrix} 
            + (1-\sqrt{\nu}) \mb{E}_{\mc{U}}\mb{E}_0
            \begin{pmatrix}
                0 & 0 & 0 \\
                0 & H_{0, \lambda\lambda'} & H_{0, \lambda\gamma'}\\
                0 & H_{0, \gamma\lambda'} & H_{0, \gamma\gamma'}\\
            \end{pmatrix}. 
        \end{align*}
        Hence, using the similar argument from  \citet[Lemma~2]{chen2021nonlinear},  there exists a constant $c>0$,  independent of $N, T$,   such that 
        \begin{align*}
            (1-\sqrt{\nu}) \mb{E}_{\mc{U}}\mb{E}_0
            \begin{pmatrix}
                0  & 0 & 0 \\
                0 & H_{0, \lambda\lambda'} & H_{0, \lambda\gamma'}\\
                0 & H_{0, \gamma\lambda'} & H_{0, \gamma\gamma'}\\
            \end{pmatrix} +  V_0 & \geq 
            c
            \begin{pmatrix}
                0 & 0 & 0  \\
                0 & T \mb{I}_{R N}& 0 \\
                0 & 0& N \mb{I}_{R T}
            \end{pmatrix},  
        \end{align*}
        and 
        \begin{align*}
            (1-\sqrt{\nu}) \mb{E}_{\mc{U}}\mb{E}_0
            \begin{pmatrix}
                H_{0, \beta\beta'} & 0 & 0 \\
                0 & 0 & 0 \\
                0 & 0 & 0\\
            \end{pmatrix} 
            \geq 
            c
            \begin{pmatrix}
                NT\mb{I}_{d_X} & 0 & 0 \\
                0 & 0 & 0 \\
                0 & 0 & 0\\
            \end{pmatrix}.
        \end{align*}
        Consequently, we must have 
        \begin{align*}
            \mb{E}_{\mc{U}}\mb{E}_0H_0 + V_0 \geq c
            \begin{pmatrix}
                NT\mb{I}_{d_X} & 0 & 0  \\
                0 & T \mb{I}_{R N}& 0 \\
                0 & 0& N \mb{I}_{R T}
            \end{pmatrix},  
        \end{align*}
        whose minimum eigenvalue is positive with order $\min\{N, T\}$. 

        \paragraph{Step 2} In this step, we show that the perturbation $(\mb{E} H_0 - \mb{E}_{\mc{U}}\mb{E} H_0 )$ is sufficiently small relative to $\mb{E}_{\mc{U}}\mb{E}H_0 + bV_0$. Note that 
        \begin{equation*}
        \begin{aligned}
            \mb{E}_0 H_0 + bV_0  = & \mb{E}_{\mc{U}}\mb{E}_0 H_0 + bV_0+ (E_0 H_0 - \mb{E}_{\mc{U}}\mb{E}_0 H_0)  \\
            \geq &  
            c
            \begin{pmatrix}
                0  & 0 & 0  \\
                0 & T \mb{I}_{R N}& 0 \\
                0 & 0& N \mb{I}_{R T}
            \end{pmatrix} 
            + 
            \frac{c}{2}
            \begin{pmatrix}
                NT\mb{I}_{d_X} & 0 & 0  \\
                0 &  0 & 0 \\
                0 & 0& 0
            \end{pmatrix}  
            +
            \frac{c}{2}
            \begin{pmatrix}
                NT\mb{I}_{d_X} & 0 & 0  \\
                0 &  0 & 0 \\
                0 & 0& 0
            \end{pmatrix}
            \\
            & 
            +  ( \mb{E}_0 - \mb{E}_{\mc{U}}\mb{E}_0)  
            \begin{pmatrix}
                H_{0, \beta \beta'} & H_{0, \beta \lambda'} & H_{0, \beta \gamma'} \\
                H_{0, \lambda \beta'} & 0 & 0 \\
                H_{0, \gamma \beta'} & 0 & 0
            \end{pmatrix}
            + (\mb{E}_0 - \mb{E}_{\mc{U}}\mb{E}_0 )
            \begin{pmatrix}
                0 & 0 & 0 \\
                0 &  H_{0, \lambda \lambda'} &  H_{0, \lambda \gamma'}  \\
                0 & H_{0, \gamma \lambda'} & H_{0, \gamma \gamma'}
            \end{pmatrix} \\
            = & 
            \underbrace{
            \begin{pmatrix}
                \frac{c}{2}NT\mb{I}_{d_X} + ( \mb{E}_0 - \mb{E}_{\mc{U}}\mb{E}_0)H_{0, \beta \beta'}  & ( \mb{E}_0 - \mb{E}_{\mc{U}}\mb{E}_0)H_{0, \beta \lambda'}  &  ( \mb{E}_0 - \mb{E}_{\mc{U}}\mb{E}_0)H_{0, \beta \gamma'}   \\
                ( \mb{E}_0 - \mb{E}_{\mc{U}}\mb{E}_0)H_{0, \lambda\beta'} & 0& 0 \\
                ( \mb{E}_0 - \mb{E}_{\mc{U}}\mb{E}_0)H_{0,  \gamma \beta'} & 0& 0
            \end{pmatrix}
            }_{A_1}
            \\
            & 
            + 
            \underbrace{
            ( \mb{E}_0 - \mb{E}_{\mc{U}}\mb{E}_0)  
            \begin{pmatrix}
                0 & 0 & 0 \\
                0 &  H_{0, \lambda \lambda'} &  0  \\
                0 & 0 & H_{0, \gamma \gamma'}
            \end{pmatrix} 
            }_{A_2} 
            + 
            \underbrace{
            ( \mb{E}_0 - \mb{E}_{\mc{U}}\mb{E}_0)  
            \begin{pmatrix}
                0 & 0 & 0 \\
                0 &  0 &  H_{0, \lambda \gamma'}  \\
                0 & H_{0, \gamma \lambda'} & 0
            \end{pmatrix}
             }_{A_3} \\
            & + c
            \begin{pmatrix}
                \frac{NT}{2}\mb{I}_{d_X} & 0 & 0  \\
                0 &  T \mb{I}_{R N}& 0 \\
                0 & 0& N \mb{I}_{R T}
            \end{pmatrix} . 
        \end{aligned}
        \end{equation*} 
        We first study the eigenvalues of $A_1$. For the upper-left block of $A_1$, we establish a lower bound for its minimum eigenvalue as follows:
        \begin{align*}
            \sigma_{\min}\left(\frac{1}{2}NT\mb{I}_{d_X} + ( \mb{E}_0 - \mb{E}_{\mc{U}}\mb{E}_0)H_{0, \beta \beta'}\right)  \gtrsimtext{(i)} NT + O_p(\sqrt{NT}) \gtrsim NT, \quad \textbf{wpa1}, 
        \end{align*}
        where inequality (i) follows from the fact that $( \mb{E}_0 - \mb{E}_{\mc{U}}\mb{E}_0) H_{0, \beta \beta'} = O_p(\sqrt{NT})$ by Assumption~\ref{assumption:conditional_independence_hessian}\ref{item:conditional_weak_dependence_hessian} and \citet[Theorem~1]{kanaya2017convergence}. By the same argument, each entry of $( \mb{E}_0 - \mb{E}_{\mc{U}}\mb{E}_0) H_{\beta\lambda'}$ is of order $O_p(\sqrt{T})$, and each entry of $( \mb{E}_0 - \mb{E}_{\mc{U}}\mb{E}_0) H_{\beta\gamma'}$ is of order $O_p(\sqrt{N})$. By applying the Gershgorin's Circle Theorem, we conclude that every eigenvalue of $A_1$ falls into one of two categories: (i) either it is positive and of order $O_p(NT)$,  or (ii) it is of order $O_p(\sqrt{T}), O_p(\sqrt{N})$ (with an unspecified sign). Therefore, by the asymptotic assumption about $N, T$, the minimum eigenvalue of $A_1$ must be of order $\min\{\sqrt{N}, \sqrt{T}\}$.

        In addition, by Assumption~\ref{assumption:conditional_independence_hessian}\ref{item:conditional_weak_dependence_hessian} and \citet[Theorem~1]{kanaya2017convergence}, each $R$-dimensional block of $(\mb{E}_0-\mb{E}_{\mc{U}}\mb{E}_0 ) H_{0, \lambda \lambda'}$ is of order $O_p(\sqrt{T})$. By the same argument, each $R$-dimensional diagonal block of 
        $(\mb{E}_0 - \mb{E}_{\mc{U}}\mb{E}_0) H_{0,\gamma\gamma'}$ is of order $O_p(\sqrt{N})$. Since $A_2$ has a diagonal block structure, we conclude that the minimum eigenvalue of $A_2$ is of order $o_p(\min\{N, T\})$ as well. 
        
        To bound the eigenvalues of $A_3$, by Assumption~\ref{assumption:conditional_independence_hessian} and Lemma~\ref{lemma:independent_entry}, the operator norm of $A_3$ is of order $O_p(\sqrt{\max\{N, T\}}\log(NT))$, and hence the minimum eigenvalue of $A_3$ must be of order $o_p(\min\{N, T\})$. Combining the previous result, Weyl's theorem, and the asymptotic assumption about $N, T$, we conclude that the minimum eigenvalue of $\mb{E}_0\mc{H}_{0}$ is 
        \begin{align*}
            \sigma_{\min}(\mb{E}_0\mc{H}_{0})  \gtrsim \min\{N, T\}. 
        \end{align*}
        Since $\mc{H}_{0, NT} = \frac{1}{NT}\mc{H}_0$, it follows immediately that $\sigma_{\min}(\mb{E}_0\mc{H}_{0, NT})  \gtrsim (\max\{N, T\})^{-1}$.  
    \end{prooflmm}

    \section{Proof of Section \ref{sec:implementation}}

    \begin{proofthm}{thm:algorithm_convergence}
        Recall that 
        \begin{align*}
            (\hat{\beta}_{\mr{nuc}}, \hat{\Theta}_{\mr{nuc}}) \in \argmin_{\beta\in \mb{R}^{d_X}, \Theta\in \mb{R}^{N\times T}} \left\{\mc{L}_{NT}(\beta, \Theta) + \frac{\varphi_{NT}}{\sqrt{NT}}\|\Theta\|_{\mr{nuc}}\right\}. 
        \end{align*}
        Let $h_{NT}(\Theta) = \frac{\varphi_{NT}}{\sqrt{NT}}\|\Theta\|_{\mr{nuc}}$.  The proximal gradient operator is defined as 
        \begin{align*}
            \mr{prox}_{s_{\theta}h_{NT}}(\tilde{\Theta})  = \argmin_{\Theta } \left\{h_{NT}(\tilde{\Theta}) + \frac{1}{2}\|\tilde{\Theta} - \Theta\|_{\mr{F}}^2\right\}. 
        \end{align*}
        By Lemma~\ref{lemma:quadratic_upperbound}, for any  feasible $(\beta_1, \Theta_1), (\beta_2, \Theta_2)$ in the optimization~\eqref{eq:nnr_definition_formal}, there exist constants $L_{\beta}, L_{\theta}$ such that 
        \begin{align*}
            \mc{L}_{NT}(\beta_2, \Theta_2) \leq &  \mc{L}_{NT}(\beta_1, \Theta_1 ) +  \nabla_{\beta} \mc{L}_{NT}(\beta_1, \Theta_1 )'(\beta_2 -\beta_1) + \la \nabla_{\Theta} \mc{L}_{NT}(\beta_1, \Theta_1 ), (\Theta_2 -\Theta_1) \ra \\
            & + \frac{L_{\beta}}{2} \|\beta_2 -\beta_1 \|^2 + \frac{L_{\theta}}{2NT}\|\Theta_2 - \Theta_1 \|_{\mr{F}}^2. 
        \end{align*}
        We apply the standard proximal gradient argument with the quadratic upper bound established in Lemma~\ref{lemma:quadratic_upperbound}. Define $G_{s_{\theta}}(\beta,\Theta) =  \frac{1}{s_{\theta}} (\Theta - \mr{prox}_{s_{\theta}h_{NT}}(\Theta - s_{\theta}\nabla_{\Theta}\mc{L}_{NT}(\beta, \Theta) )) $. The parameters update is given by  
        \begin{align*}
            \begin{pmatrix}
                \beta_2 \\
                \Theta_2
            \end{pmatrix}
            =
            \begin{pmatrix}
                \beta_1 -s_{\beta}\nabla_{\beta} \mc{L}_{NT}(\beta_1, \Theta_1) \\
                \Theta_1 -s_{\theta}G_{s_{\theta}}(\beta_1,\Theta_1)
            \end{pmatrix}. 
        \end{align*}
        By Lemma~\ref{lemma:quadratic_upperbound}, we obtain 
        \begin{align*}
            \mc{L}_{NT}(\beta_2, \Theta_2) \leq & \mc{L}_{NT}(\beta_1, \Theta_1) - s_{\beta}\| \nabla_{\beta} \mc{L}_{NT}(\beta_1, \Theta_1 )\|^2 - s_{\theta} \la \nabla_{\Theta} \mc{L}_{NT}(\beta_1, \Theta_1 ), G_{s_{\theta}}(\beta_1,\Theta_1) \ra \\
            & + \frac{L_{\beta}s_{\beta}^2 }{2}\| \nabla_{\beta} \mc{L}_{NT}(\beta_1, \Theta_1 )\|^2 + \frac{L_{\theta}s_{\theta}^2 }{2NT} \|G_{s_{\theta}}(\beta_1,\Theta_1)\|_{\mr{F}}^2. 
        \end{align*}
        When $s_{\beta} \leq \frac{1}{L_{\beta}}$ and $s_{\theta} \leq \frac{NT}{L_{\theta} }$, 
        \begin{align*}
            \mc{L}_{NT}(\beta_2, \Theta_2) \leq & \mc{L}_{NT}(\beta_1, \Theta_1) - s_{\theta} \la \nabla_{\Theta} \mc{L}_{NT}(\beta_1, \Theta_1 ), G_{s_{\theta}}(\beta_1,\Theta_1) \ra \\
            & - \frac{ s_{\beta} }{2}\| \nabla_{\beta} \mc{L}_{NT}(\beta_1, \Theta_1 )\|^2 + \frac{s_{\Theta} }{2} \|G_{s_{\theta}}(\beta_1,\Theta_1)\|_{\mr{F}}^2. 
        \end{align*}
        Since $\mc{L}_{NT}(\cdot, \cdot)$ and $h_{NT}(\cdot)$ are convex, we have 
        \begin{align*}
            h_{NT}(\Theta_2)  & \leq h_{NT}(\Theta_0) -   \la G_{s_{\theta}}(\beta_1, \Theta_1) - \nabla_{\Theta}\mc{L}_{NT}(\beta_1, \Theta_1),  \Theta_0 - \Theta_1 + s_{\theta}G_{s_{\theta}}(\beta_1, \Theta_1)\ra \\
            \mc{L}_{NT}(\beta_1, \Theta_1 ) & \leq   \mc{L}_{NT}(\beta_0, \Theta_0)  -  \nabla_{\beta} \mc{L}_{NT}(\beta_1, \Theta_1 )'(\beta_0 -\beta_1) - \la \nabla_{\Theta} \mc{L}_{NT}(\beta_1, \Theta_1 ), (\Theta_0 -\Theta_1) \ra. 
        \end{align*}
        Adding the above inequalities yields  
        \begin{align*}
            \mc{L}_{NT}(\beta_2, \Theta_2) + h_{NT}(\Theta_2)  \leq &  \mc{L}_{NT}(\beta_0, \Theta_0) + h_{NT}(\Theta_0) \\
            & + \nabla_{\beta}L_{NT}(\beta_1, \Theta_1)'(\beta_1 - \beta_0) + \la G_{s_{\theta}}(\beta_1, \Theta_1), \Theta_1 - \Theta_0\ra \\
            & - \frac{ s_{\beta} }{2}\| \nabla_{\beta} \mc{L}_{NT}(\beta_1, \Theta_1 )\|^2 - \frac{s_{\theta} }{2} \|G_{s_{\theta}}(\beta_1,\Theta_1)\|_{\mr{F}}^2. 
        \end{align*}
        Let $\mc{\psi}^{(k)} = \mc{L}_{NT}(\beta^{(k)}, \Theta^{(k)}) + h_{NT}(\Theta^{(k)}) $ and $\mc{\psi}_0 = \mc{L}_{NT}(\beta_0, \Theta_0) + h_{NT}(\Theta_0)$ for notational simplicity. Then 
        \begin{align*}
            \mc{\psi}^{(k+1)} - \mc{\psi}_0 
            \leq  \frac{1}{2s_{\beta}}\left(\|\beta^{(k)} - \beta_0\|^2 - \|\beta^{(k+1)} - \beta_0 \|^2\right) + \frac{1}{2s_{\theta}} \left(\|\Theta^{(k)} - \Theta_0\|^2_{\mr{F}} - \|\Theta^{(k+1)} - \Theta_0 \|^2_{\mr{F}} \right). 
        \end{align*}
        Summing over $k$ gives 
        \begin{align*}
            \mc{\psi}^{(k)} - \mc{\psi}_0 \leq \frac{1}{k}\sum_{i=0}^{k}(\mc{\psi}^{(i)} - \mc{\psi}_0 ) \leq  \frac{1}{2 k s_{\beta}}\|\beta^{(0)} - \beta_0\|^2 + \frac{1}{2ks_{\theta}}\|\Theta^{(0)} - \Theta_0\|^2_{\mr{F}}. 
        \end{align*}
        Hence, when $s_{\beta} \leq \frac{1}{L_{\beta}}$ and $s_{\theta} \leq \frac{NT}{L_{\theta} }$, the proximal gradient method converges to the global minimizer with rate $O(1/k)$. 
    \end{proofthm}

    \begin{lemma}[Quadratic upper bound]\label{lemma:quadratic_upperbound}
        For any feasible $(\beta_1, \Theta_1), (\beta_2, \Theta_2)$ in the optimization~\eqref{eq:nnr_definition_formal}, we have 
        \begin{align*}
            \mc{L}_{NT}(\beta_2, \Theta_2) \leq &  \mc{L}_{NT}(\beta_1, \Theta_1 ) +  \nabla_{\beta} \mc{L}_{NT}(\beta_1, \Theta_1 )'(\beta_2 -\beta_1) + \la \nabla_{\Theta} \mc{L}_{NT}(\beta_1, \Theta_1 ), (\Theta_2 -\Theta_1) \ra \\
            & + \frac{L_{\beta}}{2} \|\beta_2 -\beta_1 \|^2 + \frac{L_{\theta}}{2NT}\|\Theta_2 - \Theta_1 \|_{\mr{F}}^2, 
        \end{align*}
        with $ L_{\beta} = 2 d_X b_{\max} \rho_X^2$ and $L_{\theta} =  2b_{\max}$. 
    \end{lemma}
    \begin{prooflmm}{lemma:quadratic_upperbound} Note that 
        \begin{align*}
            \mc{L}_{NT}(\beta_2, \Theta_2) = & \mc{L}_{NT}(\beta_1, \Theta_1 ) +  \nabla_{\beta} \mc{L}_{NT}(\beta_1, \Theta_1 )'(\beta_2 -\beta_1) + \la \nabla_{\Theta} \mc{L}_{NT}(\beta_1, \Theta_1 ),  (\Theta_2 -\Theta_1)\ra  \\
            & + \frac{1}{2} (\beta_2' -\beta'_1, \mr{vec}(\Theta_2)' - \mr{vec}(\Theta_2)')\nabla^2 \mc{L}_{NT}(\tilde{\beta}, \tilde{\Theta})(\beta_2' -\beta'_1, \mr{vec}(\Theta_2)' - \mr{vec}(\Theta_2)')', 
        \end{align*}
        where 
        \begin{align*}
            \nabla^2 \mc{L}_{NT}(\beta, \Theta ) = 
            \begin{pmatrix}
                H_{\beta\beta'} & H_{\beta\theta'} \\
                H_{\theta\beta'} & H_{\theta\theta'} 
            \end{pmatrix} \geq 0, 
        \end{align*}
        and 
        \begin{align*}
            H_{\beta\beta'} & = -\frac{1}{NT}\sum_{i=1}^{N}\sum_{t=1}^{T} X_{it}\ddot{\ell}_{it} X'_{it},  \\
            H_{\beta\theta'} & = -\frac{1}{NT}\left[X_{it}\ddot{\ell}_{it}\right]_{i=1,\ldots, N, t = 1, \ldots, T},  \\
            H_{\theta\theta'} & =-\frac{1}{NT}\mr{diag}\left\{ \ddot{\ell}_{it}\right\}_{i=1,\ldots, N, t = 1, \ldots, T}. 
        \end{align*}
        It is straightforward to verify that 
        \begin{align*}
            \begin{pmatrix}
                2 H_{\beta\beta'} & 0 \\
                0 & 2 H_{\theta\theta'} 
            \end{pmatrix}
            - 
            \begin{pmatrix}
                H_{\beta\beta'} & H_{\beta\theta'} \\
                H_{\theta\beta'} & H_{\theta\theta'} 
            \end{pmatrix}
            = 
            \begin{pmatrix}
                H_{\beta\beta'} & -H_{\beta\theta'} \\
                -H_{\theta\beta'} & H_{\theta\theta'} 
            \end{pmatrix}
            \geq 0. 
        \end{align*}
        The last inequality holds because (i) $\nabla^2 \mc{L}_{NT}(\beta, \Theta )\geq 0$ and (ii) flipping the sign of off-diagonal part of the matrix does not change the eigenvalues of the matrix. 
        Therefore, 
        \begin{align*}
            \begin{pmatrix}
                \beta_2 - \beta_1 \\
                \mr{vec}(\Theta_2) - \mr{vec}(\Theta_1)
            \end{pmatrix}'
            \nabla^2 \mc{L}_{NT}(\tilde{\beta}, \tilde{\Theta} )
            \begin{pmatrix}
                \beta_2 - \beta_1 \\
                \mr{vec}(\Theta_2) - \mr{vec}(\Theta_1)
            \end{pmatrix} \leq 2 \sigma_{\max}(H_{\beta\beta'})\|\beta_2 - \beta_1\|^2 + 2\sigma_{\max}(H_{\theta\theta'})\|\Theta_2 - \Theta_1\|_{\mr{F}}^2. 
        \end{align*}
        Since $\sigma_{\max}(H_{\beta\beta'}) \leq d_X b_{\max} \rho_X^2$ and $\sigma_{\max}(H_{\theta\theta'}) \leq \frac{b_{\max}}{NT}$, the result follows. 
    \end{prooflmm}

    \begin{proofthm}{thm:algorithm_convergence_local}
        The proof consists of two steps. In the first step, we show that if we start from $(\beta^{(0)}, \Lambda^{(0)}, \Gamma^{(0)}) = (\hat{\beta}_{\mr{nuc}}, \hat{\Lambda}_{\mr{nuc}} ,\hat{\Gamma}_{\mr{nuc}}) \in \mc{B}_{\delta_{NT}}$, then, under properly chosen step sizes $(s_{\beta}, s_{\lambda}, s_{\gamma})$, the updated estimators $(\beta^{(1)}, \Lambda^{(1)}, \Gamma^{(1)})$  also remain in the neighborhood $\mc{B}_{\delta_{NT}}$. In the second step, by establishing a local quadratic upper bound for the optimization problem~\eqref{eq:FE_estimator_formal}, we show that the sequence of updated estimators will always remain within the neighborhood $\mc{B}_{\delta_{NT}}$. Additionally, after each iteration, the estimators become closer to the global minimizer than in the previous step. 
        
        \paragraph{Step 1}  First, note that when evaluated at $(\Lambda^{(0)},\Gamma^{(0)})$, we have 
        \begin{align*}
            \nabla_{\lambda} \|\hat{\Lambda}'_{\mr{nuc}}\Lambda /N - \Gamma' \hat{\Gamma}_{\mr{nuc}}/ T \|_{\mr{F}}^2 = 0, \quad  \nabla_{\gamma}\|\hat{\Lambda}'_{\mr{nuc}}\Lambda / N - \Gamma' \hat{\Gamma}_{\mr{nuc}}/ T\|_{\mr{F}}^2 = 0. 
        \end{align*}
        Therefore,  
        \begin{align*}
            \beta^{(1)} = &   \beta^{(0)} - s_{\beta} \nabla_{\beta}\mc{L}_{NT}(\beta^{(0)}, \Lambda^{(0)}, \Gamma^{(0)}),  \\
            \Lambda^{(1)} = &   \Lambda^{(0)} - s_{\lambda} \nabla_{\lambda}\mc{L}_{NT}(\beta^{(0)}, \Lambda^{(0)}, \Gamma^{(0)}),  \\
            \Gamma^{(1)} = &   \Gamma^{(0)} - s_{\gamma} \nabla_{\gamma}\mc{L}_{NT}(\beta^{(0)}, \Lambda^{(0)}, \Gamma^{(0)}). 
        \end{align*}
        We have the following inequality when $s_{\beta}\lesssim 1$:   
        \begin{align*}
            \|\beta^{(1)} - \beta^{(0)}\| = & \frac{s_{\beta}}{NT}\left\| \sum_{i=1}^{N}\sum_{t=1}^{T}\dot{\ell}_{it}(\beta^{(0)\prime } X_{it} + \lambda^{(0)\prime}_i \gamma^{(0)}_t)X_{it}\right\|  \\
            \leqtext{(i)} & \frac{s_{\beta}}{NT} \left\| \sum_{i=1}^{N}\sum_{t=1}^{T} \dot{\ell}^0_{it} X_{it} \right\| + \frac{s_{\beta}}{NT} \left\| \sum_{i=1}^{N}\sum_{t=1}^{T}\tilde{\ddot{\ell}}_{it} (\beta^{(0)} - \beta_0 )'X_{it}X_{it}\right\| \\
            & + \frac{s_{\beta}}{NT} \left\| \sum_{i=1}^{N}\sum_{t=1}^{T}\tilde{\ddot{\ell}}_{it} \lambda^{(0)\prime}_i (\gamma_{0, t}^G - \gamma^{(0)}_t )X_{it} \right\| +  \frac{s_{\beta}}{NT} \left\| \sum_{i=1}^{N}\sum_{t=1}^{T}\tilde{\ddot{\ell}}_{it} \gamma_{t}^{(0)\prime} (\lambda_{0, i}^G - \lambda^{(0)}_i )X_{it} \right\|, 
        \end{align*}
        where inequality (i) follows from a Taylor expansion and the triangle inequality. Here, $\tilde{\ddot{\ell}}_{it}$ is an abbreviation for $\ddot{\ell}_{it}(\tilde{\beta}'X_{it} + \tilde{\lambda}_i \tilde{\gamma}_t)$, where $(\tilde{\beta}, \tilde{\Lambda}, \tilde{\Gamma})$ lies on the line segment between $(\beta^{(0)}, \Lambda^{(0)}, \Gamma^{(0)})$ and $(\beta_0, \Lambda_0, \Gamma_0)$. Each of the four terms above can be bounded using the same argument as in the proof of Theorem~\ref{thm:convexity_strong_pre}. 
        
        Specifically, combining moment condition $\mb{E}_{Z, \Lambda_0, \Gamma_0} (\dot{\ell}_{it}^{0}X_{it}) = 0$, sampling assumption (Assumption~\ref{assumption:regularity_pre}\ref{item:sampling_pre}), and  \citet[Theorem~1]{kanaya2017convergence} yields 
        \begin{align*}
            \frac{s_{\beta}}{NT} \left\| \sum_{i=1}^{N}\sum_{t=1}^{T} \dot{\ell}^0_{it} X_{it} \right\| = O_p\left(\frac{\log(NT)}{\sqrt{NT}}\right). 
        \end{align*}
        In addition, since $\{X_{it}\}_{1\leq i\leq N, 1\leq t\leq T}$, $(\tilde{\beta}, \tilde{\Lambda}, \tilde{\Gamma})$, $(\beta^{(0)}, \Lambda^{(0)}, \Gamma^{(0)})$,  and $(\beta_0, \Lambda_0, \Gamma_0)$ are uniformly bounded, we obtain   
        \begin{align*}
            \frac{s_{\beta}}{NT} \left\| \sum_{i=1}^{N}\sum_{t=1}^{T}\tilde{\ddot{\ell}}_{it} (\beta^{(0)} - \beta_0 )'X_{it}X_{it}\right\| \lesssim \|\beta^{(0)} - \beta_0\|  & \eqtext{(i)} o_p(\delta_{NT}) , \\
            \frac{s_{\beta}}{NT} \left\| \sum_{i=1}^{N}\sum_{t=1}^{T}\tilde{\ddot{\ell}}_{it} \lambda^{(0)\prime}_i (\gamma_{0, t}^G - \gamma^{(0)}_t )X_{it} \right\| \lesssim \frac{1}{\sqrt{T}} \|\Gamma^{(0)} - \Gamma_0^G\|_{\mr{F}} & \eqtext{(ii)} o_p(\delta_{NT}),  \\
            \frac{s_{\beta}}{NT} \left\| \sum_{i=1}^{N}\sum_{t=1}^{T}\tilde{\ddot{\ell}}_{it} \gamma_{t}^{(0)\prime} (\lambda_{0, i}^G - \lambda^{(0)}_i )X_{it} \right\| \lesssim \frac{1}{\sqrt{N}} \|\Lambda^{(0) }- \Lambda_0^G\|_{\mr{F}} & \eqtext{(iii)} o_p(\delta_{NT}). 
        \end{align*}
        Here, inequalities (i), (ii), and (iii) are obtained by the similar argument as in the proof of Theorem~\ref{thm:consistency_pre}. Therefore, we conclude that   
        \begin{equation}\label{eq:algorithm_beta_bound}
            \|\beta^{(1)} - \beta^{(0)}\| = o_p(\delta_{NT}). 
        \end{equation}
        
        In addition, when $s_{\lambda}\sim N$,   
        \begin{align*}
            \|\Lambda^{(1)} - \Lambda^{(0)}\|_{\mr{F}}^2 \lesssim & \frac{1}{T^2}\sum_{i=1}^{N} \left\| \sum_{t=1}^{T}\dot{\ell}_{it}(\beta^{(0) \prime} X_{it} + \lambda^{(0)\prime}_i \gamma^{(0)}_t)\gamma^{(0)}_t\right\|^2 \\
            \leqtext{(i)} & \frac{1}{T^2} \sum_{i=1}^{N} \left\|  \sum_{t=1}^{T} \dot{\ell}^0_{it} \gamma_{0, t}\right\|^2 +  \frac{1}{T^2} \sum_{i=1}^{N} \left\|  \sum_{t=1}^{T} \dot{\ell}^0_{it} (\gamma_{0, t} - \gamma^{(0)}_{t})\right\|^2  \\
            & + \frac{1}{T^2} \sum_{i=1}^{N}\left\| \sum_{t=1}^{T}\tilde{\ddot{\ell}}_{it} (\beta^{(0)} - \beta_0 )'X_{it} \gamma^{(0)}_t \right\|^2 \\
            & + \frac{1}{T^2} \sum_{i=1}^{N}\left\|  \sum_{t=1}^{T}\tilde{\ddot{\ell}}_{it} \lambda^{(0)\prime}_i (\gamma_{0, t}^G - \gamma^{(0)}_t ) \gamma^{(0)}_t\right\|^2 +  \frac{1}{T^2} \sum_{i=1}^{N} \left\|  \sum_{t=1}^{T}\tilde{\ddot{\ell}}_{it} \gamma_{t}^{(0)\prime} (\lambda_{0, i}^G - \lambda^{(0)}_i )\gamma^{(0)}_t \right\|^2, 
        \end{align*}
        where inequality (i) follows from the Taylor expansion and the triangle inequality. Similarly, we establish the following bounds: 
        \begin{align*}
            \frac{1}{T^2} \sum_{i=1}^{N} \left\|  \sum_{t=1}^{T} \dot{\ell}^0_{it} \gamma_{0, t}\right\|^2 \eqtext{(i)} O_p\left(\frac{N\log(T)^2}{T}\right) & = o_p(N\delta_{NT}^2 ), \\
            \frac{1}{T^2} \sum_{i=1}^{N} \left\|  \sum_{t=1}^{T} \dot{\ell}^0_{it} (\gamma_{0, t} - \gamma^{(0)}_{t})\right\|^2 \lesssim \frac{N}{T} \|\Gamma^{(0)} - \Gamma_0^G\|_{\mr{F}}^2 & \eqtext{(ii)} o_p(N\delta_{NT}^2 ),  \\
            \frac{1}{T^2} \sum_{i=1}^{N}\left\| \sum_{t=1}^{T}\tilde{\ddot{\ell}}_{it} (\beta^{(0)} - \beta_0 )'X_{it} \gamma^{(0)}_t \right\|^2 \lesssim N \|\beta^{(0)} - \beta_0 \|_{\mr{F}}^2 & \eqtext{(iii)} o_p(N\delta_{NT}^2 ),   \\
            \frac{1}{T^2} \sum_{i=1}^{N}\left\|  \sum_{t=1}^{T}\tilde{\ddot{\ell}}_{it} \lambda^{(0)\prime}_i (\gamma_{0, t}^G - \gamma^{(0)}_t ) \gamma^{(0)}_t\right\|^2  \lesssim \frac{N}{T} \|\Gamma^{(0)} - \Gamma_0^G\|_{\mr{F}}^2 & \eqtext{(iv)} o_p(N\delta_{NT}^2 ),  \\
            \frac{1}{T^2} \sum_{i=1}^{N} \left\|  \sum_{t=1}^{T}\tilde{\ddot{\ell}}_{it} \gamma_{t}^{(0)\prime} (\lambda_{0, i}^G - \lambda^{(0)}_i )\gamma^{(0)}_t \right\|^2 \lesssim \|\Lambda^{(0)} - \Lambda_0^G\|_{\mr{F}}^2 & \eqtext{(v)} o_p(N\delta_{NT}^2 ),  
        \end{align*}
        where inequality (i) comes from the sampling assumption (Assumption~\ref{assumption:regularity_pre}\ref{item:sampling_pre}) and \citet[Theorem~1]{kanaya2017convergence}. Inequalities (ii)---(v) are obtained by the same argument as in the proof of Theorem~\ref{thm:consistency_pre}. Therefore, 
        \begin{equation}\label{eq:algorithm_lambda_bound}
            \frac{1}{\sqrt{N}}\|\Lambda^{(1)} - \Lambda^{(0)}\| = o_p(\delta_{NT}). 
        \end{equation}
        Also, by a similar argument, we have
        \begin{equation}\label{eq:algorithm_gamma_bound}
            \frac{1}{\sqrt{T}}\|\Gamma^{(1)} - \Gamma^{(0)}\| = o_p(\delta_{NT}). 
        \end{equation}
        Combining equations~\eqref{eq:algorithm_beta_bound}, \eqref{eq:algorithm_lambda_bound}, and \eqref{eq:algorithm_gamma_bound}, we show that if the optimization starts from $(\beta^{(0)}, \Lambda^{(0)}, \Gamma^{(0)}) = (\hat{\beta}_{\mr{nuc}}, \hat{\Lambda}_{\mr{nuc}} ,\hat{\Gamma}_{\mr{nuc}}) \in \mc{B}_{\delta_{NT}}$, then, under properly chosen step sizes $(s_{\beta}, s_{\lambda}, s_{\gamma})$ as in Theorem~\ref{thm:algorithm_convergence_local},  the updated estimators $(\beta^{(1)}, \Lambda^{(1)}, \Gamma^{(1)})$ also remain in the neighborhood $\mc{B}_{\delta_{NT}}$.
    
        \paragraph{Step 2} 
        By Theorem~\ref{thm:convexity_strong_pre}, the optimization problem is strongly convex. Consequently, with properly chosen $(s_{\beta}, s_{\lambda}, s_{\gamma})$ as in Theorem~\ref{thm:algorithm_convergence_local},  the first-step iterate $(\beta^{(1)}, \Lambda^{(1)}, \Gamma^{(1)})$ not only remains within the neighborhood $\mc{B}_{\delta_{NT}}$ but also moves closer to the global minimizer (up to an orthogonal transformation) than $(\beta^{(0)}, \Lambda^{(0)}, \Gamma^{(0)})$. 

        The argument in Step~1 can be applied directly to all subsequent iterations. It is therefore straightforward to verify that each iterate remains within $\mc{B}_{\delta_{NT}}$ and moves closer to the global minimizer (up to an orthogonal transformation) than the previous iterate.  This completes the proof. %
    \end{proofthm}

    \section{Technical Lemmas}

    \begin{lemma}\label{lemma:mul_Hoeffding}
        Let $\{x_{it} \in \mb{R}^{d_X}\mid i=1, \ldots,N, t = 1, \ldots, T\}$ be a collection of real vectors satisfying   $|x_{it, d}| \leq \rho_X $ for all $i=1, \ldots,N$, $t = 1, \ldots, T$, and $d =1, \ldots, d_X$. Let  $\epsilon_{it}$ be Rademacher random variables independent across $i$ and $t$. Then,  
        \begin{align*}
            \mb{E}_{ \epsilon} \|\sum_{i=1}^{N}\sum_{t=1}^{T}x_{it} \epsilon_{it}\|_2  \leq  D_1 \sqrt{NT}, 
        \end{align*}
        where $D_1 = \sqrt{2\pi d_X^3\rho_X^2}$. 
    \end{lemma}
    \begin{proof}
        Since $|x_{it,d}| \le \rho_X$, the random variable $x_{it,d}\epsilon_{it}$ is independent, mean-zero, and sub-Gaussian with parameter $\rho_X^2$. Therefore, for any $d = 1,\ldots,d_X$ and any $\delta>0$, Hoeffding's inequality yields
        \begin{align*}
            \mb{P}_{\epsilon} \left(|\sum_{i=1}^{N}\sum_{t=1}^{T}x_{it, d} \epsilon_{it}|\geq \delta\right)\leq 2 e^{-\frac{\delta^2}{2 NT \rho_X^2}}. 
        \end{align*}
        Hence, by the union bound,
        \begin{align*}
            \left\{\|\sum_{i=1}^{N}\sum_{t=1}^{T}x_{it} \epsilon_{it}\|_2 \geq \delta\right\} &\subset \bigcup_{d=1}^{d_X}\{|\sum_{i=1}^{N}\sum_{t=1}^{T}x_{it, d} \epsilon_{it}| \geq \frac{\delta}{\sqrt{d_X}}\} \\
            \Rightarrow \mb{P}_{\epsilon}\left(\|\sum_{i=1}^{N}\sum_{t=1}^{T}x_{it} \epsilon_{it}\|_2 \geq \delta \right) &\leq \mb{P}_{\epsilon}\left(\bigcup_{d=1}^{d_X}\{|\sum_{i=1}^{N}\sum_{t=1}^{T}x_{it, d} \epsilon_{it}| \geq \frac{\delta}{\sqrt{d_X}}\}\right) \\
            & \leq\sum_{d=1}^{d_X} \mb{P}_{\epsilon}\left(|\sum_{i=1}^{N}\sum_{t=1}^{T}x_{it, d} \epsilon_{it}| \geq \frac{\delta}{\sqrt{d_X}}\right) \\
            & \leq 2d_X \exp\left\{-\frac{\delta^2}{2NTd_X\rho_X^2 }\right\}. 
        \end{align*}
        Therefore, 
        \begin{align*}
            \mb{E}_{\epsilon} \|\sum_{i=1}^{N}\sum_{t=1}^{T}x_{it} \epsilon_{it}\|_2  = \int_{0}^{\infty} \mb{P}_{\epsilon} (\|\sum_{i=1}^{N}\sum_{t=1}^{T}x_{it} \epsilon_{it}\|_2 \geq \delta )\mr{d}\delta \leq \int_{0}^{\infty}2 d_X e^{-\frac{\delta^2}{2 NT d_X \rho_X^2}}\mr{d}\delta  = D_1 \sqrt{NT}, 
        \end{align*}
        where $D_1 = \sqrt{2\pi d_X^3\rho_X^2}$. 
    \end{proof}

    \begin{lemma}\label{lemma:independent_entry}
        Consider a random matrix $Z\in \mb{R}^{N \times T}$ with uniformly bounded entries and $\mb{E}Z = 0$.  Assume that the rows of $Z$ are independent. 
        For each row $i$, $\{Z_{it}\}_{1\leq t\leq T}$ is $\alpha$-mixing with mixing coefficient $\alpha_i(\tau)\rightarrow 0$ as $\tau \rightarrow \infty$, where 
                \begin{align*}
                    \alpha_i(\tau) = \sup_{t} \sup_{A \in \mc{A}^{i}_{t},  B\in \mc{B}^{i}_{t + \tau}}|\mb{P}(A\cap  B) - \mb{P}(A)\mb{P}(B)|, 
                \end{align*} 
                where $\mc{A}^{i}_{t}$ is the sigma-field generated by $\{\ldots, Z_{i,t-1}, Z_{i,t}, \}$, and $ \mc{B}^{i}_{t+\tau}$ is the sigma-field generated by $\{Z_{i,t+\tau}, Z_{i,t+\tau+1}, \ldots\}$. Assume further that the mixing coefficients satisfy a uniform polynomial decay condition: there exist constants $\beta > 2$ and $C > 0$ such that $\sup_{1\leq i\leq N}\alpha_i(\tau) \leq C\tau^{-\beta }$. 
        Then, wpa1,  
        \begin{align*}
            \|Z\|_{\mr{op}} \lesssim \sqrt{\max\{N, T\}}\log (N + T)
        \end{align*}
    \end{lemma}
    \begin{prooflmm}{lemma:independent_entry}
        We employ the rectangular matrix Bernstein inequality   (\citet{tropp2012user}), stated as follows: 
        \begin{theorem}[\citet{tropp2012user}]\label{thm:op_norm_sum_matrices}
            Consider a finite sequence $\{Z_k\}$ of independent, random matrices with dimensions $N\times T$, Assume that each random matrix satisfies
            \begin{align*}
                \mb{E}Z_k = 0, \quad \|Z_k\|_{\mr{op}} \leq D. 
            \end{align*}
            Define $\sigma^2 = \max\left\{\left\|\sum_{k} \mb{E}(Z_kZ_k')\right\|_{\mr{op}}, \left\|\sum_{k} \mb{E}(Z_k'Z_k)\right\|_{\mr{op}}\right\}$. 
            Then for all $\delta>0$, we have 
            \begin{align*}
                \mb{P}\left( \|\sum_{k}Z_k\|_{\mr{op}}\geq \delta \right) \leq (N + T)\exp\left\{\frac{-\delta^2}{2\sigma^2 + \frac{2}{3} D\delta}\right\}. 
            \end{align*}. 
        \end{theorem}
        
        To apply Theorem \ref{thm:op_norm_sum_matrices}, let   $Z_i$  denote the  $N \times T$  matrix where   the  $i$-th row of $Z_i$ is the  $i$-th row of  $Z$, and all other rows are zero. Since the entries of $Z$ are  uniformly bounded, there exists a constant $a_1>0$,  irrelevant with $N, T$, such that
        \begin{align}\label{eq:llm:independent_entry_1}
            \|Z_{i}\|_{\mr{op}} \leq a_1 \sqrt{T} \leq a_1\sqrt{\max\{N, T\}}, \quad \forall i=1,2,\ldots, N. 
        \end{align}
        In addition, it is straightforward to verify that $\sum_{i=1}^{N} \mb{E}(Z_iZ_i')$ has a diagonal structure. Thus, there exists a constant $a_2>0$, irrelevant with $N, T$, such that
        \begin{align}\label{eq:llm:independent_entry_2}
            \|\sum_{i=1}^{N} \mb{E}(Z_iZ_i')\|_{\mr{op}} \leq a_2 T . 
        \end{align}
        To bound the term $\left\|\sum_{i=1}^{N} \mb{E}(Z_i'Z_i)\right\|_{\mr{op}}$, observe that  
        \begin{align*}
            \mb{E}(Z_i'Z_i) = 
            \begin{pmatrix}
                \gamma_i(1, 0) & \gamma_i(1, 1) & \ldots & \gamma_i(1, T-1) \\
                \gamma_i(2, -1) & \gamma_i(2, 0) & \ldots & \gamma_i(1, T-2) \\
                \vdots & \vdots & \ddots  & \vdots \\
                \gamma_i(T, 1- T ) & \gamma_i(T, 2 - T) & \ldots & \gamma_i(T, 0)
            \end{pmatrix}, 
        \end{align*}
        where $\gamma_i(t, t + \tau) = \mr{Cov}(Z_{i, t}, Z_{i, t + \tau })$ for any $1\leq t\leq T$. By the Gershgorin circle theorem, we show that $\|\mb{E}(Z_i'Z_i)\|_{\mr{op}}$ is uniformly bounded by a constant $a_3>0$. Indeed,  
        \begin{align*}
            \|\mb{E}(Z_i'Z_i)\|_{\mr{op}} & \leq \max_{1\leq t \leq T} \left\{\gamma_i(t, 0) + \sum_{\tau=1 - t}^{-1 }\left|\gamma_i(t, \tau)\right| + \sum_{\tau= 1 }^{T-t}\left|\gamma_i(t, \tau)\right|  \right\} . 
        \end{align*}
        In addition, since the sequence is $\alpha$-mixing with a uniformly polynomial decay rate $\sup_{1\leq i\leq N}\alpha_i(\tau) \leq C\tau^{-\beta }$, it follows from  \cite[Proposition~2.7]{fan2008nonlinear} that 
        \begin{align*}
            \left|\gamma_i(t, \tau)\right| \leq 4 \rho_Z^2 \alpha_i(|\tau| )^{\frac{1}{2}} = 4 \rho_Z^2 C^{\frac{1}{2}}|\tau|^{-\frac{\beta}{2}}, 
        \end{align*}
        where $\rho_Z>0$ is a constant,  independent of $N, T$,  such that $|Z_{it}|\leq \rho_Z$ for all $i, t, N, T$. Therefore,  
        \begin{align}\label{eq:llm:independent_entry_3}
            \|\mb{E}(Z_i'Z_i)\|_{\mr{op}} \leq  \max_{1\leq t \leq T}   \gamma_i(t, 0) + \sum_{\tau=1}^{T} 8 \rho_Z^2 C^{\frac{1}{2}}\tau^{-\frac{\beta}{2}} \leq \max_{1\leq t \leq T}   \gamma_i(t, 0) + 8 \rho_Z^2 C^{\frac{1}{2}}\sum_{\tau=1}^{\infty} \tau^{-\frac{\beta}{2}}. 
        \end{align}
        It is straightforward to verify that $ \max_{1\leq t \leq T}   \gamma_i(t, 0)\leq 4\rho_Z^2$, and $\sum_{\tau=1}^{\infty} \tau^{-\frac{\beta}{2}}<\infty$ as $\beta>2$. Therefore, there exists a constant $a_3 >0$, independent of $N, T$,  such that $\|\mb{E}(Z_i'Z_i)\|_{\mr{op}}\leq a_3$,  and additionally, $\max_{1\leq i\leq N}\|\mb{E}(Z_i'Z_i)\|_{\mr{op}}\leq a_3$.  
        Hence, we have 
        \begin{align*}
            \|\sum_{i=1}^{N}\mb{E}(Z_i'Z_i)\|_{\mr{op}} \leq \sum_{i}^{N} \| \mb{E}(Z_i'Z_i)\|_{\mr{op}}\leq a_3N. 
        \end{align*}
        Let $a_4 = 2\max\{a_2, a_3\}$. Combining~\eqref{eq:llm:independent_entry_2} and \eqref{eq:llm:independent_entry_3}, we obtain  
        \begin{align}\label{eq:llm:independent_entry_4}
            \sigma^2 = \max\left\{\|\sum_{i}^{N} \mb{E}(Z_iZ_i')\|_{\mr{op}}, \|\sum_{i}^{N} \mb{E}(Z_i'Z_i)\|_{\mr{op}} \right\} \leq a_4 \max\{N, T\}. 
        \end{align}
        
        Finally, since $\{Z_i\}_{1\leq i\leq N}$ are independent  and $Z = \sum_{i=1}^{N} Z_{i}$,  we  apply Theorem \ref{thm:op_norm_sum_matrices} together with \eqref{eq:llm:independent_entry_1} and \eqref{eq:llm:independent_entry_4} to obtain 
        \begin{align*}
            \mb{P}\left(\left\|Z\right\|_{\mr{op}}\geq \delta \right) & \leq 2\max\{N, T\}\exp\left\{-\frac{\delta^2}{2a_4\max\{N, T\} + \frac{2}{3} a_1 \sqrt{\max\{N, T\}}\delta}\right\} . 
        \end{align*}
        Let $\delta =  \max\{4\sqrt{a_4}, 4\sqrt{\frac{2}{3}}a_1\}\sqrt{\max\{N, T\}} \log (N+T)$. Then   
        \begin{align*}
            \mb{P}\left(\left\|Z\right\|_{\mr{op}}\geq \delta \right) & 
            \leq 2 \max\{N, T\} \exp\left\{-2\log(N + T)\right\} \leq 2 \exp\left\{-\log(N+T)\right\} \rightarrow 0. 
        \end{align*}
        Therefore, we conclude that $\|Z\|_{\mr{op}}\lesssim \sqrt{\max\{N, T\}}\log (N+T)$ with probability approaching  $1$. 
    \end{prooflmm}

    \begin{lemma}\label{lemma:eigenvalue_ABB0} 
       For any block matrix $M\in \mb{R}^{(n+m)\times (n+m)}$
            \begin{align*}
                M = 
                \begin{pmatrix}
                    A & B\\
                    B' & 0
                \end{pmatrix}, 
            \end{align*}
            where $A\in \mb{R}^{n\times n}$ is a positive definite matrix, $B\in \mb{R}^{n\times m}$, we have  
            \begin{align}\label{eq:eigenvalue_ABB0}
                \sigma_{\min} (M) \geq \frac{1}{2}\left(\sigma_{\min}(A) - \sqrt{\sigma^2_{\min}(A) + 4 s_{\max}^2(B)}\right). 
            \end{align}
            where $\sigma(\cdot)$ denotes the eigenvalue and $s(\cdot)$ denotes the singular value. 
    \end{lemma}
    \begin{prooflmm}{lemma:eigenvalue_ABB0}
        When $B = 0$,  inequality~\eqref{eq:eigenvalue_ABB0} holds trivially since $\sigma_{\min}(M) = 0$.  For $B \neq 0$, the matrix $M$ admits negative eigenvalues; denote one of them as $\varsigma <0$. 
        Let $(x',y')'$ be a corresponding eigenvector. By the definition of eigenvalue, we have 
        \begin{align*}
            M(x', y')' = \varsigma (x', y')' \Rightarrow 
            \left\{
                \begin{array}{l}
                    Ax + By = \varsigma x  \\
                    B'x = \varsigma y 
                \end{array}
            \right. . 
        \end{align*}
        Since $\varsigma \neq 0$, substituting $y = \varsigma^{-1} B'x$ into $Ax + By = \varsigma x$ gives 
        \begin{gather*}
            \varsigma^2 x'x - \varsigma x'Ax - x'BB'x = 0 \quad \Rightarrow \quad \varsigma = \frac{1}{2}\left(\frac{x'Ax}{x'x } - \sqrt{\left(\frac{x'Ax}{x'x }\right)^2  + 4 \frac{x'BB'x}{x'x}}\right), 
        \end{gather*}
        where we select the negative root since $\varsigma < 0$.  It is straightforward to verify that $\varsigma$ is increasing in  $\frac{x'Ax}{x'x }$ and decreasing in $\frac{x'BB'x}{x'x}$. Since 
        \begin{align*}
            \frac{x'Ax}{x'x } \geq \min_{x\neq 0 }\frac{x'Ax}{x'x } = \sigma_{\min}(A), \quad\frac{x'BB'x}{x'x} \leq \max_{x} \frac{x'BB'x}{x'x} = s^2_{\max}(B), 
        \end{align*}
        we  conclude that 
        \begin{align*}
            \sigma_{\min} (M) \geq \frac{1}{2}\left(\sigma_{\min}(A) - \sqrt{\sigma^2_{\min}(A) + 4 s_{\max}^2(B)}\right). 
        \end{align*}
        This completes the proof. 
    \end{prooflmm}

    \begin{lemma}\label{lemma:SVD_block}
        For any block matrix $A = [A_1, A_2]$, where $A_1 \in \mb{R}^{n\times m_1}, A_2 \in \mb{R}^{n\times m_2}$, we have $s_{r}(A)\geq s_{r}(A_1)$, for $r = 1,\ldots, \min\{n, m_1  \}$. Consequently, $\|A\|_{\mr{nuc}}\geq \|A_1\|_{\mr{nuc}}$. 
    \end{lemma}
    \begin{prooflmm}{lemma:SVD_block}
        We use the Courant-Fischer variational characterization of singular values. Let $u = (u_1', u_2')'$, where $u_1\in \mb{R}^{m_1}$, $u_2\in \mb{R}^{m_2}$. Then
        \begin{align*}
            s_r(A)  = \min_{\mr{dim}(V)=r-1} \max_{u\perp V, \|u\|=1} \|Au\|. 
        \end{align*}
        Since $A = [A_1,A_2]$, for any $u_1\in\mathbb{R}^{m_1}$,define $\tilde{u}=(\tilde{u}_1',0)' \in \mathbb{R}^{m_1+m_2}$. Then $A \tilde{u}= A_1 \tilde{u}_1$. Therefore, for any $(r-1)$-dimensional subspace $V_1\subset\mathbb{R}^{m_1}$,
        let
        \begin{align*}
            \tilde{V} := \{(\tilde{v}_1',0)' : \tilde{v}_1\in V_1\}\subset \mathbb{R}^{m_1+m_2}.
        \end{align*}
        Then 
        \begin{align*}
            \max_{u \perp \tilde{V},  \|u\|=1} \|Au\| \geq \max_{ u_1 \perp V_1,  \|u_1\|=1} \|A_1 u_1\| 
        \end{align*}
        Taking the minimum over all $(r-1)$-dimensional subspaces gives
        \begin{align*}
            s_r(A) = \min_{\dim(V)=r-1}\max_{ u\perp V, \|u\|=1}\|Au\| \geq \min_{\dim(V_1)=r-1} \max_{u_1\perp V_1, \|u_1\|=1}\|A_1 u_1\| = s_r(A_1). 
        \end{align*}
        Finally,
        \begin{align*}
            \|A\|_{\mr{nuc}}=\sum_{r=1}^{\min\{n,m_1+m_2\}} s_r(A)
            \geq \sum_{r=1}^{\min\{n,m_1 \}} s_r(A)
            \geq\sum_{r=1}^{\min\{n,m_1\}} s_r(A_1) =\|A_1\|_{\mr{nuc}}.
        \end{align*}

        This completes the proof. 
    \end{prooflmm}

\end{document}